\documentclass[11pt]{article}
\usepackage[utf8]{inputenc}
\usepackage{biblatex}
\usepackage{amsthm}
\usepackage{dsfont}
\usepackage{amsmath}
\usepackage{amssymb}
\usepackage{upgreek}
\usepackage{esint}
\usepackage{esvect}
\usepackage{graphicx}
\usepackage[margin=1in]{geometry}
\usepackage{bm}
\graphicspath{ {img/} }
\usepackage{mathtools}
\usepackage{mathrsfs}
\usepackage{nicefrac}
\usepackage[ruled,vlined]{algorithm2e}
\usepackage{enumitem}
\usepackage{thmtools, thm-restate}
\usepackage[dvipsnames]{xcolor}
\usepackage{hyperref}
\hypersetup{
    colorlinks=true,
    linkcolor=blue,
    citecolor=violet   
    }
\usepackage[nameinlink]{cleveref}
\usepackage{array}

\usepackage{physics}

\newtheorem{theorem}{Theorem}[section]
\newtheorem{lemma}[theorem]{Lemma}
\newtheorem{proposition}[theorem]{Proposition}
\theoremstyle{definition}
\newtheorem{definition}[theorem]{Definition}
\theoremstyle{definition}
\newtheorem{construction}{Construction}
\theoremstyle{definition}
\newtheorem{example}{Example}

\theoremstyle{definition}

\theoremstyle{plain}
\newtheorem{corollary}[theorem]{Corollary}
\theoremstyle{remark}
\newtheorem{remark}{Remark}
\theoremstyle{definition}



\newcommand{\C}{\mathbb{C}}
\newcommand{\Z}{\mathbb{Z}}

\newcommand{\F}{\mathbb{F}}

\newcommand{\Q}{\mathbb{Q}}

\DeclarePairedDelimiter\ceil{\lceil}{\rceil}

\newcommand{\supp}{\mathrm{supp}}
\newcommand{\CSS}{\mathrm{CSS}}
\renewcommand{\ev}{\mathrm{ev}}
\newcommand{\Syl}{\mathrm{Syl}}
\newcommand{\Gal}{\mathrm{Gal}}
\newcommand{\Gr}{\mathrm{Gr}}

\allowdisplaybreaks

\title{\textbf{Quantum Hierarchical Locally Recoverable Codes}}
\author{Venkatesan Guruswami\thanks{University of California, Berkeley, and the Simons Institute for the Theory of Computing. {\tt venkatg@berkeley.edu.} Research supported by a grant from Fujitsu Limited, Inc., ONR grant N00014-24-1-2491, a Simons Investigator award, and a UC Noyce initiative award.} \and Rutuja Kshirsagar \thanks{Fujitsu Research of America, Inc. {\tt rkshirsagar@fujitsu.com.}} \and Pranav Trivedi\thanks{University of California, Berkeley, and Fujitsu Research of America, Inc. {\tt pranavtrivedi@berkeley.edu.}}}
\date{}
\begin{document}

\maketitle
\thispagestyle{empty}
\begin{abstract}
% Erasure recovery is a fundamental aspect of classical data storage, especially for the reliability of large-scale distributed systems. Extending these ideas to the quantum domain provides valuable insights for developing robust and scalable quantum memories. In quantum information theory, erasures—where the positions of lost qubits are known—are generally easier to correct than arbitrary errors, which require both detection and correction. This distinction simplifies recovery protocols and enhances fault-tolerance. Recent experimental advances have demonstrated that detected quantum errors can be converted into erasures in physical systems such as neutral-atom–based and superconducting qubit architectures, highlighting the practical importance of studying quantum erasure recovery. 

%\VG{Could we skip or shorten the first paragraph; it's good for the intro but not sure is needed in the abstract?}

Quantum locally recoverable codes (QLRCs) have recently gained attention as a framework for achieving efficient quantum storage with local recovery capabilities. Analogous to their classical counterparts, QLRCs allow a lost qudit to be reconstructed using only a small subset of other qudits, thereby reducing the resource and operational overhead in recovery. In this work, we extend the study of QLRCs by considering $(r,\delta)$ QLRCs characterized by locality parameter $r$ and local distance $\delta \geq 2$. We present constructions of both random and explicit $(r,\delta)$ QLRCs, including explicit families based on the quantum Tamo--Barg construction. We also present an efficient decoding algorithm for these quantum Tamo--Barg codes.

Furthermore, we introduce quantum \emph{hierarchical} locally recoverable codes (QHLRCs), which extend local recovery to multiple hierarchical levels. For any integer $h\geq 2$, we construct both random and explicit $h$-level QHLRCs, the latter being $h$-level quantum Tamo--Barg codes, and establish a Singleton-like bound for these codes using a CSS framework built from dual-containing classical codes. These results advance the theoretical foundations of quantum erasure recovery and contribute to the design of efficient quantum storage architectures.
\end{abstract}
\newpage

\thispagestyle{empty}
%\section*{Table of contents}

\tableofcontents
\thispagestyle{empty}
\newpage
\setcounter{page}{1}

\section{Motivation and Contributions}\label{sec:motivation}
\subsection{Introduction}

Erasure recovery plays a central role in maintaining data reliability in classical storage systems, particularly those that are distributed across multiple physical locations or devices. In such systems, data loss can occur due to node failures, hardware errors, or network disruptions. To mitigate these failures, erasure codes are widely employed, enabling the reconstruction of lost data from the remaining intact portions. As the size and complexity of distributed systems grow, so does the need for codes that can efficiently recover data with minimal redundancy and computational overhead.

Extending these concepts to quantum systems is essential for developing robust quantum memories—an indispensable component of large-scale quantum computing and quantum communication networks. Quantum memories must not only preserve fragile quantum information over time but also allow recovery from physical qubit losses, known as erasures. Consequently, understanding and adapting erasure recovery mechanisms to the quantum setting is a fundamental step toward fault-tolerant quantum information processing.

Classically, erasure recovery operates under the assumption that the location of erased symbols is known to the decoder. Once the erased positions are identified, the remaining symbols and parity constraints can be used to reconstruct the missing information efficiently. The same principle applies to quantum erasure recovery, where it is assumed that the locations of the lost qubits are known to the recovery operation. In both classical \cite{Cover2006} and quantum \cite{gottesman1997stabilizercodesquantumerror, Grassl_1997} contexts,  erasures are inherently easier to recover than general errors, since the decoder does not need to infer where errors have occurred. Furthermore, recent experiments have shown that detected errors can be converted into erasures in certain physical implementations of quantum hardware. This conversion has been demonstrated in neutral-atom systems \cite{Grassl_1997,Wu_2022} and superconducting qubit architectures \cite{Teoh_2023}, making erasure-based quantum recovery schemes particularly practical and appealing.

In the classical regime, locally recoverable codes (LRCs), introduced in \cite{LRC_intro}, have emerged as a particularly effective solution for efficient erasure recovery. The central idea behind LRCs is locality: each symbol of a codeword can be recovered by accessing only a small subset of other symbols, rather than the entire codeword. Formally, an $(n, k, r)$ LRC over the finite field $\mathbb{F}_q$ encodes a message of $k$ symbols into a codeword of length $n$ such that each codeword symbol has a recovery set of at most $r$ other symbols. The parameter $r$ thus captures the locality of the code, representing the size of the largest subset needed for recovery. Typically, $1 \leq r \leq k$, and smaller $r$ values indicate more efficient local recovery. Such local recovery is critical in distributed storage for efficient repair of a failed server, which is highly common. At the same time, we want the LRCs to have good (ideally near-optimal) distance, subject to the locality constraints, so that we have fault-tolerance to a larger number of failures (erasures) in the worst-case. 

The notion of locality can be further generalized by introducing the $(r, \delta)$-LRC \cite{KPLK14}, which enhances fault tolerance within each recovery group. In an $(r, \delta)$-LRC, each local group can correct up to $\delta - 1$ erasures, allowing the system to recover even when multiple erasures occur within a single locality group. This generalization strikes a balance between local redundancy and global robustness, providing a flexible framework for erasure recovery in large-scale distributed storage systems.

Explicit constructions achieving optimal trade-offs between distance and locality have been provided by Tamo and Barg \cite{tamo.barg.lrc}. Their families of LRCs meet the Singleton-like bound for local recoverability~\cite{LRC_intro}. Ballentine, Barg, and Vladut \cite{ballentine.barg.vladut.hlrc} further demonstrated that these constructions naturally extend to the hierarchical setting.

Building upon these classical developments, \emph{quantum} locally recoverable codes (QLRCs) were introduced by Golowich and Guruswami \cite{golowich2025quantum} to explore the notion of locality in the context of \emph{quantum} error correction. QLRCs are designed so that lost or corrupted qudits can be recovered using only a small subset of other qudits, preserving the quantum analog of locality. 

Subsequent research has expanded on the theory of QLRCs, including \cite{bu2025quantumlocallyrecoverablecode, galindo2026quantumrdeltalocallyrecoverablebch, galindo2025optimalquantumlocallyrecoverable, galindo2024quantumrdeltalocallyrecoverablecodes, li2025improvedboundsoptimalconstructions,  li2025optimalquantumlrcshermitian,luo2023boundsconstructionsquantumlocally,sharma2024quantumlocallyrecoverablecodes, zhou2025optimalquantumrdeltalocallyrepairable}. In this work, we focus on the $(r, \delta)$-QLRCs as defined in \cite{galindo2024quantumrdeltalocallyrecoverablecodes}. Our contributions extend the random QLRC and (folded) Quantum Tamo--Barg constructions introduced in \cite{golowich2025quantum} to the general $(r, \delta)$ setting. The folded variant of the explicit constructions improves the asymptotic rate and distance trade-off when compared to the unfolded explicit construction. While proving that our constructions indeed give $(r,\delta)$ QLRCs, computing the dimension and locality parameters for these codes is straightforward, but computing a bound on the minimum distance (\Cref{thm:qtb_distance}) of these codes is trickier and %enforces using a modified technical approach based on utilizing the $(r,\delta)$ parameters. %Subsequently, the distance proofs are
substantially more involved than in \cite{golowich2025quantum}. The constructions in \cite{golowich2025quantum} still arise as a special case of our framework when $\delta = 2$. 

Analogous to their classical counterparts, we generalize QLRCs to quantum hierarchical locally recoverable codes (QHLRCs), characterized by a hierarchy of locality and distance parameters $((r_1, \delta_1), \ldots, (r_h, \delta_h))$ for $h \geq 2$. QHLRCs inherit the multi-tier recovery structure of HLRCs, enabling both efficient local recovery and robust multi-erasure correction. 

We propose a random construction of QHLRCs (\Cref{sec:random_r_delta_hierarchical_qlrc}), along with a Singleton-like bound on their minimum distance, extending the trade-off relations known for classical codes to the quantum domain. Furthermore, we present an explicit construction of quantum hierarchical Tamo--Barg codes (\Cref{sec:r_delta_hierarchical_qTB}), which naturally generalize the $(r, \delta)$-Quantum Tamo--Barg codes while introducing a hierarchical recovery structure. We also apply a folding operation and obtain improved parameters in the folded Quantum Tamo--Barg codes (\Cref{sec:folded_qtb} and \Cref{sec:folded_hqtb}).
Finally, we also develop an efficient decoding algorithm for the Quantum Tamo--Barg codes.

This unified framework bridges the gap between locality-based code design in classical and quantum settings. By generalizing locality and hierarchical recovery to quantum codes, our constructions introduce new techniques which will prove useful towards developing scalable, efficient, and fault-tolerant quantum storage and computation systems capable of addressing reliability challenges in quantum technologies. It should be noted that not all notions related to classical locality extend to the quantum setting. One such notion is that of availability, as shown in \cite[Theorem 9]{golowich2025quantum}. Hence, hierarchy is the natural generalization of locality in the quantum setting. 

We now present some important results proved in this work. 

\subsection{Main Contributions}
Let $C_X,C_Z\subseteq \F_q^n$ be classical linear codes with $C_Z^\perp\subseteq C_X$. The corresponding CSS code \cite{Calderbank_1996,Steane_1996} is denoted $\CSS(C_X,C_Z)$. For our constructions, the most important case is a CSS code constructed from a dual-containing classical code: $\mathcal Q=\CSS(C,C)$ with $C^\perp\subseteq C$. The quantum dimension is $\dim \mathcal{Q} = 2\dim C-n$, and the CSS distance is \(d(\mathcal Q)=\min\{\mathrm{wt}(c):c\in C\setminus C^\perp\}.\)

A quantum code is an $(r,\delta)$-QLRC if every coordinate lies in a set $J$ of size at most $r+\delta-1$ such that every erasure pattern $I\subseteq J$ with $|I|\le \delta-1$ can be recovered by an operation supported on $J$. For CSS codes constructed from dual-containing classical codes, we use a known equivalence for quantum $(r,\delta)$ locality \cite{galindo2024quantumrdeltalocallyrecoverablecodes}: if $C^\perp\subseteq C$ and $d(C^\perp)\ge \delta$, then $\CSS(C,C)$ is an $(r,\delta)$-QLRC if and only if $C$ is a classical $(r,\delta)$-LRC. Thus, for the explicit constructions it suffices to use dual-containing classical LRCs.

In order to extend the notion of locality to the hierarchical setting, we need CSS codes to be closed under puncturing (restricting the code to a subset of coordinates). This property is in the same broad direction as recent puncturing results for quantum codes \cite{Grasslpuncturingstabilizer,Gundersen2025puncturing}, but we specifically isolate the CSS case.

\begin{lemma}[restatement of \Cref{lem:puncturing_css}]
\label{restatement lem:puncturing_css}
Let $\mathcal Q=\CSS(C_X,C_Z)$ be a CSS code with $C_Z^\perp\subseteq C_X$, and let $I\subseteq \{1,\cdots,n\}$ be a coordinate set. Then \(\mathcal Q|_I:=\CSS((C_X)|_I,(C_Z)|_I) \)
is again a well-defined CSS code. In particular, the punctured classical pair still satisfies the CSS orthogonality conditions.
\end{lemma}

% \begin{proof}[Proof idea]
% Write $\pi_I$ for puncturing to $I$ and $\sigma_I$ for shortening to $I$. The standard duality identity
% \[
%     (\pi_I C)^\perp=\sigma_I(C^\perp)
% \]
% gives
% \[
%     ((C_Z)|_I)^\perp=\sigma_I(C_Z^\perp)\subseteq \sigma_I(C_X)\subseteq (C_X)|_I,
% \]
% and the other inclusion is identical after exchanging $X$ and $Z$. Thus the punctured pair still defines a CSS code.
% \end{proof}

% This lemma is more than a technicality: it lets us define a local quantum recovery group by restricting the global CSS code to the group and checking that the restricted object is still a CSS code. It is also the mechanism behind the recursive definition of hierarchical locality: after puncturing to a level-$1$ group, the resulting punctured code can itself be required to have lower-level quantum locality.

The $h$-level QHLRC definition is recursive, paralleling classical hierarchical locality \cite{sasidharan2015hlrc,ballentine.barg.vladut.hlrc}. A code has hierarchical locality $((r_1,\delta_1),\ldots,(r_h,\delta_h))$ if every coordinate lies in a level-$1$ group of size at most $r_1+\delta_1-1$ that corrects $\delta_1-1$ erasures, and the punctured code on that group is itself an $(h-1)$-level QHLRC with parameters $((r_2,\delta_2),\ldots,(r_h,\delta_h))$. The constructions assume nested parameters
\(r_1\ge \cdots\ge r_h\ge \delta_1\ge\cdots\ge\delta_h\ge 2,\)
such that \( n_{l+1}\mid n_l,\) where \( n_l:=r_l+\delta_l-1.\) 

We present a Singleton-like bound for QHLRCs to which we will compare our constructions.

\begin{theorem}[informal restatement of \Cref{thm:singleton_qlrc} and \Cref{prop:singleton.bound.qhlrc}]\label{thm:singleton_informal} If $\mathcal{Q}$ is an $h$-level QHLRC with rate $R$, then its relative distance $\rho = d(\mathcal{Q})/N$ must satisfy \begin{align}\label{eqn:relative_singleton}
    \rho \leq \frac{1-R}{2} - \Omega\left(\frac{\delta_1}{r_h}\right).
\end{align}
\end{theorem}
For exact constants hidden in the $\Omega(\delta_1/r_h)$ term, we refer the reader to \Cref{thm:singleton_qlrc} and \Cref{prop:singleton.bound.qhlrc}

\Cref{table:constructions} presents a summary of rate, distance, and alphabet size tradeoffs achieved by our constructions.
\begin{table}[htbp]
\centering
\begin{tabular}{|c|c|c|}
   \hline
   Code  &  Distance & Alphabet Size\\
   \hline
   Random QLRC  & \(\rho \geq \frac{1-R}{2} - O\left(\frac{\delta}{r}\right)\) & $2^{O\left(\frac{r+\delta-1}{\delta-1}\right)}$\\
   \hline
   Random QHLRC  & \(\rho \geq \frac{1-R}{2} - O\left(\frac{\delta_1}{r_h}\right)\) & $2^{O\left(\frac{r_h+\delta_h-1}{\delta_1-1}\right)}$\\
   \hline
   $(r,\delta)$ QTB & $\rho \geq \Omega\left(\frac{1-R}{\delta}\right) - O\left(\frac{\delta}{r}\right)$ & $N+1$\\
   \hline
   $h$-level QTB & $\rho \geq \Omega\left(\frac{1-R}{\delta_1^h}\right) - O\left(\frac{\delta_1}{r_h}\right)$& $N+1$\\
   \hline
   Folded $(r,\delta)$ QTB & $\rho \geq \frac{1-R}{2} - O\left(\sqrt{\frac{\delta}{r}}\right)$ & $O(N(r+\delta-1)^2)$\\
   \hline
   Folded $h$-level QTB & $\rho \geq \frac{1-R}{2} - O\left(\sqrt{\frac{\delta_1}{r_h}}\right)$ & $O(N(r+\delta-1)^2)$\\
   \hline
\end{tabular}
\caption{Rate-Distance Tradeoffs}
\label{table:constructions}
\end{table}

The randomized construction is still $O(\delta_1/r_h)$ below \labelcref{eqn:relative_singleton} because the constant hidden in $O(\delta_1/r_h)$ is larger than $\Omega(r_1/\delta_h)$. The parameters of the randomized construction improve upon the folded QTB codes at cost of explicitness. There is no efficient algorithm to certify the distance of the random codes nor is there an efficient algorithm to correct arbitrary errors. The QTB codes and their folded variants are explicit so their distance bounds are guaranteed and we present an efficient algorithm to correct arbitrary errors.

The rest of \Cref{sec:motivation} is devoted to an overview of the constructions and proofs. In particular, we will also compare our proof techniques with the techniques used in \cite{golowich2025quantum} and explain where the more general $(r,\delta)$ parameters and the hierarchy structure require more involved analysis.

\subsection{\texorpdfstring{Random $(r,\delta)$-QLRCs and random QHLRCs}{Random (r,delta)-QLRCs and random QHLRCs}}\label{subsec:random_constructions}

A random $(r,\delta)$-QLRC construction extends the random QLRC construction of \cite[Section~4.1]{golowich2025quantum}. We build two parity-check matrices $H_X,H_Z$. Each matrix contains local Vandermonde blocks supported on disjoint groups of size $r+\delta-1$, giving $\delta-1$ local checks per group. The $H_Z$ blocks are chosen as orthogonal Vandermonde-like blocks, so the local row spaces are mutually orthogonal. The construction then appends $\ell$ random rows to $H_X$ and $\ell$ random rows to $H_Z$, always sampling from the orthogonal complement of the other row span. This preserves the CSS condition and produces
\[
    \mathcal Q=\CSS(\ker H_X,\ker H_Z),
    \qquad
    k=N-2\left(m(\delta-1)+\ell\right),
\]
where $m = N/(r+\delta-1)$. The local Vandermonde blocks imply $(r,\delta)$ locality.

The main probabilistic estimate refines the standard random-CSS union bound by counting only supports that are compatible with local distance $\delta$. Let $\mathcal{N}^{(\delta)}(N,w)$ be the number of weight-$w$ supports in which every nonempty local block has weight at least $\delta$:
\[
\mathcal{N}^{(\delta)}(N,w)
=\sum_{s=1}^{\lfloor w/\delta\rfloor}\binom ms
\sum_{\substack{w_1+\cdots+w_s=w\\ w_i\ge\delta}}
\prod_{i=1}^s\binom{r+\delta-1}{w_i}.
\]
Define
\[
    H_q^{(\delta)}(\rho)
    :=\limsup_{N\to\infty}\frac1N\log_q\left(\sum_{w\le \rho N}\mathcal{N}^{(\delta)}(N,w)(q-1)^w\right).
\]
The resulting high-probability statement is the following.

\begin{proposition}[restatement of \Cref{prop:distance_random_r_delta_qlrc}]
\label{restatement prop:distance_random_r_delta_qlrc}
For sufficiently large $N$, if $\ell\ge (H_q^{(\delta)}(\rho)+2\epsilon)N$, then with probability at least $1-2q^{-\epsilon N}$ the resulting code has distance at least $\rho N$.
\end{proposition}

\begin{proof}[Proof idea]
For any fixed vector $y\notin C_X^\perp$, each random $Z$-check eliminates it with probability at least $1-1/q$, so the chance that $y\in C_Z\setminus C_X^\perp$ is at most $q^{-\ell}$, and symmetrically for $C_X\setminus C_Z^\perp$. The local Vandermonde checks already exclude support patterns whose nonempty local blocks have size less than $\delta$; hence the union bound only ranges over the refined count $\mathcal{N}^{(\delta)}(N,w)$, rather than over all supports of weight $w$. The entropy condition on $\ell$ makes the union bound at weights bounded above by $\rho N$ which is at most $2q^{-\epsilon N}$.
\end{proof}

A generalization of this framework yields random $h$-level QHLRCs. The local blocks are now nested, and the number of deterministic local checks is
\[
    M=m_h(\delta_h-1)+\sum_{l=1}^{h-1}m_l(\delta_l-\delta_{l+1}),
    \qquad m_l=\frac{N}{r_l+\delta_l-1}.
\]
The dimension is $k=N-2(M+\ell)$. A recursive generating function counts hierarchically admissible supports: set
\[
B_h(z)=\sum_{t=\delta_h}^{n_h}\binom{n_h}{t}z^t,
\]
and for $l<h$ let
\[
B_l(z)=\sum_{t=\delta_l}^{n_l}[z^t](1+B_{l+1}(z))^{n_l/n_{l+1}}z^t.
\]
Then $\mathcal N^{(\boldsymbol\delta)}(N,w)=[z^w](1+B_1(z))^{N/n_1}$. Replacing the one level count, $\mathcal N^{(\delta)}$, by the hierarchical count, $\mathcal N^{(\boldsymbol\delta)}$, in the union bound gives the hierarchical analogue.

\begin{theorem}[restatement of \Cref{prop:distance_random_qhlrc}]
\label{restatement prop:distance_random_qhlrc}
Let
\[
    \mathcal H_q^{(\boldsymbol\delta)}(\rho)
    :=\limsup_{N\to\infty}\frac1N\log_q
    \left(\sum_{w\le \rho N}\mathcal N^{(\boldsymbol\delta)}(N,w)(q-1)^w\right).
\]
For sufficiently large \(N\), if $\ell\ge (\mathcal H_q^{(\boldsymbol\delta)}(\rho)+2\epsilon)N$, then the random $h$-level QHLRC has distance at least $\rho N$ with probability at least $1-2q^{-\epsilon N}$.
\end{theorem}

\begin{proof}[Proof idea]
The probabilistic part is unchanged from \Cref{restatement prop:distance_random_r_delta_qlrc}. The only difference is the support count. A low-weight nonzero word that survives the deterministic hierarchical checks must be nonempty in a nested pattern of local blocks, and every nonempty level-$l$ block must have weight at least $\delta_l$. The generating functions $B_l(z)$ enumerate exactly these admissible patterns, so the same random-CSS union bound gives the result.
\end{proof}

\subsection{\texorpdfstring{Explicit $(r,\delta)$ quantum Tamo--Barg codes}{Explicit (r,delta) quantum Tamo--Barg codes}}

Let $q$ be a prime power, let $n:=r+\delta-1$ divide $q-1$, and let $\Omega_{n}\subseteq \F_q^*$ be the subgroup of $n$th roots of unity. For $S\subseteq [q-1]$, write
\[
    \F_q[X]^S=\left\{\sum_{i\in S}f_iX^i:f_i\in\F_q\right\},
    \qquad
    \ev(f)=(f(x))_{x\in\F_q^*}.
\]
Define residue sets
\[
    S_+=\bigcup_{j=1}^{\delta-1}(j+n\Z),
    \qquad
    S_-=\bigcup_{j=1}^{\delta-1}(-j+n\Z).
\]
The $(r,\delta)$-QTB code is $\mathcal Q=\CSS(C,C)$, where $C=\ev(\F_q[X]^S)$ and
\[
    S=\left([\ell]\setminus S_-\right)\cup\left([q-1]\cap S_+\right).
\]
for $q/2 \leq \ell \leq q-1$. The first part is the usual Tamo--Barg exponent set \cite{tamo.barg.lrc}, while the second part is added to ensure $C$ is dual-containing and to supply local checks. Since $\ell\ge q/2$, $C^\perp\subseteq C$. The dimension is
\[
 k=1+\left|\{q - \ell \leq i \leq \ell - 1 : i \not\in (S_{+} \cup S_{-})\}\right|
\]
and hence
\[
    k=1+(2\ell-q)\left(1-\frac{2(\delta-1)}{r+\delta-1}\right)+O(\delta).
\]

The locality proof identifies the dual part
\(
    B^\perp=\ev(\F_q[X]^{[q-1]\cap S_+})
\)
with functions that, on each coset $\alpha\Omega_{n}$, agree with a polynomial in $\omega$, a primitive $n$th root of unity, of degree at most $\delta-1$ and with no constant term. Therefore, on each coset, $B^\perp$ contains $\delta-1$ independent Vandermonde checks supported entirely on that coset. Since each coset has size $r+\delta-1$, any $\delta-1$ erasures in the coset are locally recoverable. By the CSS/classical equivalence, $\CSS(C,C)$ is an $(r,\delta)$-QLRC.

\subsubsection{The non-vanishing theorem}

A key ingredient in proving our distance bound for the $(r,\delta)$-QTB is a theorem about a polynomial forced to vanish at consecutive roots of unity. This argument is not necessary in the case $\delta = 2$, but we require it to adapt the distance proof technique of \cite{golowich2025quantum}. We first proof the theorem over $\C$ and apply a reduction to obtain the result over finite fields.

\begin{theorem}[restatement of \Cref{thm:Q_b_has_delta-1_roots}]
\label{restatement thm:Q_b_has_delta-1_roots}
Let $r\ge\delta\ge 3$, set $n=r+\delta-1$, and let $\zeta \in \C$ be a primitive $n$th root of unity. For $b\in\{\delta-1,\ldots,n-1\}$, let
\[
    Q_b(Y)=Y^b+\sum_{t=0}^{\delta-2}v_tY^t
\]
be the unique polynomial satisfying $Q_b(\zeta^t)=0$ for $t=0,\ldots,\delta-2$. Then these are the only roots of $Q_b$ among the $n$th roots of unity:
\[
    Q_b(\zeta^s)\ne 0
    \qquad\text{for all }s=\delta-1,\ldots,n-1.
\]
\end{theorem}

\begin{proof}[Proof idea]
The proof first expresses $Q_b(\zeta^s)$ as a quotient of two generalized Vandermonde determinants. By Jacobi's bialternant formula and the Jacobi--Trudi identity \cite{cauchy1815memoire,Jacobi1841,trudi1862teoria,Macdonald_2008}, this quotient becomes
\[
    Q_b(\zeta^s)
    =h_{b-\delta+1}(1,\zeta,\ldots,\zeta^{\delta-2},\zeta^s)
      \prod_{t=0}^{\delta-2}(\zeta^s-\zeta^t),
\]
where $h_m$ is the complete homogeneous symmetric polynomial. The second factor is nonzero for $s\ge \delta-1$, so the task is to prove non-vanishing of the specialized $h_m$. The non-vanishing of \(h_{b-\delta+1}(1,\zeta,\ldots,\zeta^{\delta-2},\zeta^s)\) follows from \Cref{lem:homogenous_poly_roots_of_unity}. Its proof utilizes a generating-function argument for complete homogeneous symmetric polynomials \cite{Macdonald_2008} which gives
\[
    h_j(1,\zeta,\ldots,\zeta^{\delta-2},\zeta^s)
    =[T^j]\prod_{\substack{u=\delta-1\\u\ne s}}^{n-1}(1-\zeta^uT)
    \quad (0\le j\le r-1).
\]
A reciprocal symmetry allows us to consider $j\le (r-1)/2$. Replacing $\zeta^s$ by $z$ and expressing the homogeneous polynomials as a complex polynomial and applying a form of Enestr\"om--Kakeya \cite{Gardner2014} proves the desired non-vanishing property.
\end{proof}

\begin{corollary}[restatement of \Cref{cor:Q_b_finite_field_version_outside_finitely_many_char}]
\label{restatement cor:Q_b_finite_field_version_outside_finitely_many_char}
The theorem transfers to finite fields outside finitely many characteristics. For fixed $(r,\delta)$ define
\[
    \mathcal M_{r,\delta}
    :=\prod_{m=0}^{r-1}\prod_{s=\delta-1}^{n-1}
    \operatorname{Res}\left(h_m(1,X,\ldots,X^{\delta-2},X^s),\Phi_{n}(X)\right),
\] where $\Phi_n(X)$ is the $n^{th}$ cyclotomic polynomial.
\Cref{restatement thm:Q_b_has_delta-1_roots} implies $\mathcal M_{r,\delta}\ne0$. Therefore, if $\operatorname{char}(\F_q)\nmid\mathcal M_{r,\delta}$ and $n\mid(q-1)$, the same non-vanishing conclusion holds over $\F_q$.
\end{corollary}

\begin{proof}[Proof idea]
\Cref{restatement thm:Q_b_has_delta-1_roots} implies that the relevant specialized symmetric polynomials have no common root with the cyclotomic polynomial $\Phi_{n}$. Equivalently, the associated resultants are nonzero integers. Reducing modulo a prime that does not divide their product preserves non-vanishing over finite fields.
\end{proof}

This is stronger than what follows from the composite-order uncertainty principle: for prime $n$, Tao's uncertainty principle \cite{tao2003} gives the result immediately because $|Q_b|\le\delta$, but for composite $n$ Meshulam's uncertainty bound \cite{MESHULAM200663} is generally too weak. Thus, $Q_b$ is a concrete family for which the composite-order uncertainty principle is not tight enough, while the algebraic-combinatorial proof still succeeds. This is an important consequence because \cite{golowich2025quantum} relies on the primality of $r+1$ in the case of $\delta = 2$ (which is equivalent to using the uncertainty principle). However, if we try to use Tao's uncertainty principle, then the hierarchy (\Cref{subsec:explicit_hqlrc}) necessarily collapses and as we noted Meshulam's uncertainty principle is too weak.

\subsubsection{Distance of QTB codes}
In the following theorem, we demonstrate a lower bound on the minimum distance of $(r,\delta)$ QTB codes. 
\begin{theorem}[restatement of \Cref{thm:qtb_distance}]
\label{restatement thm:qtb_distance}
Assume either $\delta\ge3$ and $\operatorname{char}(\F_q)\nmid\mathcal M_{r,\delta}$, or $\delta=2$ and $r+1$ is prime. Then the $(r,\delta)$-QTB code has distance at least
\[
\frac{q-1}{2}\left(
\frac{1}{\delta-1}+\frac{r}{r+\delta-1}
-\sqrt{\left(\frac{r}{r+\delta-1}-\frac{1}{\delta-1}\right)^2
+\frac{4r}{(\delta-1)(r+\delta-1)}\cdot\frac{\ell-1}{q-1}}
\right).
\]
\end{theorem}

\begin{proof}[Proof idea]
Take a nonzero $\ev(f)\in C\setminus C^\perp$ and decompose $f=g+h$, where $h$ is the piecewise low-degree dual part (as characterized in \Cref{lem:low.wt.parity.checks}) and $g$ is supported on $[\ell]\setminus(S_+\cup S_-)$. For each admissible $i\in\{\delta-1,\ldots,n-1\}$, the polynomial $Q_i$ defines a linear transform
\(g_i(X)=\omega^{-i}g(\omega^iX)+\sum_{t=0}^{\delta-2}v_{i,t}\omega^{-t}g(\omega^tX). \)
The transform annihilates the dual part $h$, and the non-vanishing theorem implies each $g_i$ is nonzero. Hence
\(
    G(X)=\prod_{i=\delta-1}^{n-1}g_i(X)
\)
is nonzero and has degree at most $r(\ell-1)$.

Let $n = r+\delta -1$. The lower bound on the number of roots of $G$ is combinatorial. On a coset $A=\alpha\Omega_{n}$, let $w_A$ be the weight of $f$ on $A$. Define $y_+ := \max \{y,0\}$. The number of starting points of a consecutive block of $\delta-1$ zeros is at least $(n-(\delta-1)w_A)_+$, and for each such starting point at least $(r-w_A)_+$ admissible transforms also see a zero. Thus the root contribution from $A$ is at least
\(
    \psi(w_A):=(r-w_A)_+\left(n-(\delta-1)w_A\right)_+.
\)
The function $\psi$ is convex and nonincreasing on $[0,\infty)$ (\Cref{lem:psi_convex_monotone}), so Jensen's inequality converts the sum over cosets into a function of the total weight. Comparing this root lower bound with $\deg G\le r(\ell-1)$ and solving the resulting quadratic gives the distance bound.
\end{proof}

\subsection{Explicit hierarchical quantum Tamo--Barg codes}\label{subsec:explicit_hqlrc}

In case of hierarchy, for $l = 1,\ldots, h$ define
\[
S_{l,+}=\bigcup_{j=1}^{\delta_l-1}(j+n_l\Z),
\qquad
S_{l,-}=\bigcup_{j=1}^{\delta_l-1}(-j+n_l\Z),
\qquad
S_+=\bigcup_l S_{l,+},\quad S_-=\bigcup_l S_{l,-}.
\]
The $h$-level HQTB code is $\mathcal Q=\CSS(C,C)$ with
\(    C=\ev(\F_q[X]^S),\)
    where
    \( S=([\ell]\setminus S_- )\cup([q-1]\cap S_+).
\)
A disjointness lemma (\Cref{lem:disjointness}) shows $S_{u,+}\cap S_{v,-}=\emptyset$ for all nested levels, using the divisibility chain $n_h\mid\cdots\mid n_1$ and the inequalities on $r_l,\delta_l$. Consequently, the same dual-containment argument as in the one-level case gives $C^\perp\subseteq C$. The dimension is
\( k=1+\left|\{q-\ell\le i\le \ell-1:i\notin S_+\cup S_-\}\right|
\) and asymptotically
\[
 k=1+(2\ell-q)\left(1-2\sum_{l=1}^{h-1}\frac{\delta_l-\delta_{l+1}}{n_l}
           -2\frac{\delta_h-1}{n_h}\right)+O(\delta_1).
\]
The locality proof identifies $B^\perp=\ev(\F_q[X]^{[q-1]\cap S_+})$ as a sum of piecewise low-degree polynomial spaces over the nested cosets. Each level-$l$ coset contributes $\delta_l-1$ independent parity checks, and the nesting of root-of-unity subgroups gives nested repair groups. Hence the code is an $h$-level QHLRC.

The proof of the minimum distance of the hierarchical code iterates the one-level root-counting method. For each level $l$, set
\(
    \psi_l(t):=(r_l-t)_+\left(n_l-(\delta_l-1)t\right)_+.
\)
Define recursively
\[
    \Psi_h(t)=\psi_h(t),
    \qquad
    \Psi_l(t)=\psi_l\left(n_l-\frac{n_l}{r_{l+1}n_{l+1}}\Psi_{l+1}(t)\right)
    \quad(l=h-1,\ldots,1),
\]
and
\[
    \Theta_h(t)=\frac{q-1}{n_1}\left(\prod_{l=2}^h r_l\right)\Psi_1(t).
\]
We present a lower bound on the minimum distance of the HQTB code in the following theorem. 
\begin{theorem}[restatement of \Cref{thm:h_level_hqtb_distance_bound}]
\label{restatement thm:h_level_hqtb_distance_bound}
If \Cref{restatement cor:Q_b_finite_field_version_outside_finitely_many_char} holds at every level, and
\[
    \tau_h=\inf\left\{t\in[0,n_h]:\Theta_h(t)\le \left(\prod_{l=1}^h r_l\right)(\ell-1)\right\},
\]
then
\[
    d(\mathcal Q)\ge \frac{q-1}{n_h}\tau_h.
\]
\end{theorem}

\begin{proof}[Proof idea]
The proof recursively uses the proof of \Cref{restatement thm:qtb_distance} at each level of the hierarchy. For complete details see the proof of \Cref{thm:h_level_hqtb_distance_bound}.
\end{proof}

The bound presented here is implicitly stated. An explicit bound is presented in \Cref{cor:h_level_nested_square_root}. As an example, for $h = 2$, we can unravel the bound in \labelcref{eqn:hqtb_distance_bound} to see that \begin{equation}\label{restatement eqn:two_level_hqtb_distance_bound}
    d(\mathcal{Q}) \geq  \frac{q-1}{n_2}\cdot\frac{n_2 + (\delta_2-1)r_2 - \sqrt{(n_2 + (\delta_2-1)r_2)^2 - 4(\delta_2-1)r_2n_2\frac{y_1}{n_1}}}{2(\delta_2-1)}
\end{equation} is the explicit bound where \begin{equation*}
    y_1 = \frac{n_1 + (\delta_1-1)r_1 - \sqrt{(n_1 + (\delta_1-1)r_1)^2 - 4(\delta_1-1)r_1n_1\left(1-\frac{\ell-1}{q-1}\right)}}{2(\delta_1-1)}.
\end{equation*} Additionally, $y_1$ is exactly the one-level bound in \Cref{restatement thm:qtb_distance}.

An important conclusion is that the recursive bound is not automatically stronger than the top-level one-level QTB bound. If $y_l$ denotes the explicit inverse recursion obtained by solving $\psi_l(t)=\theta$, then
\[
    \frac{y_h}{n_h}\le \frac{y_{h-1}}{n_{h-1}}\le\cdots\le\frac{y_1}{n_1}.
\]
This leads us to view hierarchy primarily as improving repair structure and not necessarily as improving the global distance lower bound.

\paragraph{A small two-level hierarchy example.}
For a concrete instance, take \(
    (r_1,\delta_1)=(9,4),(r_2,\delta_2)=(4,3). \)
Then $n_1=12$, $n_2=6$, and $n_2\mid n_1$.  The smallest field satisfying $12\mid(q-1)$ and the finite-characteristic non-vanishing hypotheses for these two levels is $\F_{25}$.  Exact CSS-distance computations for length $24$ give, for example,
\[
\begin{array}{c|cc|cc}
\ell
& k_{\mathrm{1lev}} & d_{\mathrm{1lev}}
& k_{\mathrm{hier}} & d_{\mathrm{hier}}\\
\hline
18 & 6 & 7 & 2 & 7\\
19 & 8 & 6 & 4 & 4\\
20 & 10 & 5 & 4 & 4
\end{array}
\]
where the one-level code uses only the top parameters $(9,4)$ and the hierarchical code has the two-level repair structure.  The example illustrates the main tradeoff: hierarchy creates smaller fallback repair groups of size $6$ inside top-level groups of size $12$, but it need not improve the global minimum distance.

\subsection{Folded QTB and folded HQTB codes}

Folding groups consecutive evaluation positions into larger alphabet symbols. For the one-level code, following the folding perspective for Quantum Tamo--Barg codes \cite{golowich2025quantum}, take $s\mid (q-1)/(r+\delta-1)$ and (for a fixed generator $\omega_{q-1} \in \mathbb{F}_q^*$) group
\(
    F_{\omega_{q-1}^{is}}=
    \{\omega_{q-1}^{is},\omega_{q-1}^{is+1},\ldots,\omega_{q-1}^{is+s-1}\}
\)
into one symbol over $\F_q^s$. Since each coset $\alpha\Omega_{r+\delta-1}$ intersects each folded block in at most one scalar coordinate (\Cref{lem:coset_block_intersection}), folding preserves the $(r,\delta)$ repair groups and preserves rate. The trivial distance lower bound is the unfolded distance divided by $s$, but we prove a stronger distance bound.

\begin{theorem}[restatement of \Cref{thm:fqtb_distance}]
\label{restatement thm:fqtb_distance}
Let
\(
    \lambda:=1-\frac{\ell-1}{q-1}.
\) For prime $n=r+\delta-1$ satisfying \Cref{cor:uncertainty_prime_finite_field}, the folded QTB distance is at least
\(
    \frac{q-1}{s}(\lambda-\epsilon),
\)
where
\[
\epsilon=
\max_{1\le m_f\le n}
\min\left\{
\lambda\frac{m_f-1}{n},
\max_{\max\{1,m_f-(\delta-1)\}\le m_g\le m_f}
\left(\frac{\delta-1}{m_g}+\frac{m_g-1}{s}\right)
\right\},
\] $m_f$ is the number of the nonzero coefficients corresponding to exponents modulo $n$ of a codeword $f = g + h \in C\setminus C^\perp$, and $m_g$ is defined analogously.
\end{theorem}

\begin{proof}[Proof idea]
The proof has two cases: one uses an uncertainty principle over finite fields and the other case uses a determinant polynomial. If the coefficient support modulo $n$ is small, the uncertainty principle forces many nonzero evaluations on each nonzero coset. If the nondual support is large, a determinant polynomial, as in the folded QTB analysis of \cite{golowich2025quantum}, detects many folded zero blocks; the determinant has high multiplicity at low-rank matrices, which converts folded zeros into root multiplicity.
\end{proof}
If $s\ge 2n^2$, this yields an asymptotic bound on the relative distance 
\[
\frac{r+1}{r+\delta-1}\left(\frac12-\frac R2\cdot\frac{r+\delta-1}{r-\delta+1}\right)
-\frac32\sqrt{\frac{\delta-1}{r+\delta-1}\left(\frac12-\frac R2\cdot\frac{r+\delta-1}{r-\delta+1}\right)},
\]
where $R$ is the rate.
The paper also extends folding to hierarchical QTB codes. Choose $s\mid(q-1)/n_1$, which implies $s\mid(q-1)/n_l$ for every level. The folded code has alphabet $\F_q^s$, block length $(q-1)/s$, and the same rate as the unfolded HQTB code. The same coset-block intersection argument applies at every level, so the folded hierarchical code is an $h$-level QHLRC.

The bound on the distance of the folded HQTB is proven similarly to that of the folded QTB distance. We split into two cases: one using an uncertainty principle and another applying the determinant polynomial. The determinant polynomial argument is applied at each level, but the uncertainty principle is only applied to the top level. A detailed description is presented in \Cref{thm:fhqtb_distance}. We note that for the folded HQTB code, we need to use Meshulam's uncertainty principle as it works for composite order rather than Tao's stronger uncertainty principle as in the the one-level folded QTB distance analysis. The primality assumption in Tao's uncertainty principle would collapse the hierarchical structure to one level. 

\subsection{Decoding}
We conclude our summary with a decoding algorithm for $(r,\delta)$ QTB codes. 
\begin{theorem}[restatement of \Cref{thm:dec.r.delta.qtb}]
\label{restatement thm:dec.r.delta.qtb}
The $(r,\delta)$-QTB code can be decoded from any error pattern of size less than \[\frac{q-1}{4}\left(
\frac{1}{\delta-1}+\frac{r}{r+\delta-1}
-\sqrt{\left(\frac{r}{r+\delta-1}-\frac{1}{\delta-1}\right)^2
+\frac{4r}{(\delta-1)(r+\delta-1)}\cdot\frac{\ell}{q-1}}
\right).\] in $(q^{O(\delta)}\operatorname{poly}(r,q))$ time
\end{theorem}

\begin{proof}[Proof idea]
The decoding algorithm reuses the transformed-polynomial operators from the distance proof. Given a received word $a$, for each admissible $i\in\{\delta-1,\ldots,r+\delta-2\}$ the algorithm forms
\(
    a_i(X)=\omega^{-i}a(\omega^iX)+\sum_{t=0}^{\delta-2}v_{i,t}\omega^{-t}a(\omega^tX).
\)
For a true codeword $f=g+h$, the transform annihilates $h$ and maps $g$ to a Reed--Solomon codeword of degree $<\ell$ whose coefficient support is unchanged by the non-vanishing theorem, \Cref{restatement cor:Q_b_finite_field_version_outside_finitely_many_char}. A counting argument using the same convex function $\psi$ shows that, if the original error has weight at most half the proved QTB distance, then for some $i$ the word $a_i$ lies within the Reed--Solomon list-decoding radius of the corresponding transform of $g$. The algorithm list-decodes $a_i$ \cite{GS98}, inverts the diagonal transform using $Q_i(\omega^{j-1})\ne 0$, and chooses the candidate closest modulo $B^\perp$.

The distance to $B^\perp$ can be computed efficiently because $B^\perp$ separates over local cosets and consists of polynomials of degree at most $\delta-1$ with no constant term on each coset. This subroutine runs in $((q-1)q^{\delta-1}\delta)$ time. Overall, the algorithm decodes up to
\[
\frac{q-1}{4}\left(
\frac{1}{\delta-1}+\frac{r}{r+\delta-1}
-\sqrt{\left(\frac{r}{r+\delta-1}-\frac{1}{\delta-1}\right)^2
+\frac{4r}{(\delta-1)(r+\delta-1)}\cdot\frac{\ell}{q-1}}
\right)
\]
errors in $(q^{O(\delta)}\operatorname{poly}(r,q))$ time. Standard CSS decoding \cite{Calderbank_1996,gottesman1997stabilizercodesquantumerror,Steane_1996} then gives the corresponding quantum decoder.
\end{proof}

\subsection{Use of AI}

ChatGPT Pro was primarily used for proof verification and improving the presentation of the paper. It noted important mistakes in \Cref{const:random_qhlrc} and the proof of \Cref{thm:qtb_distance} which were then corrected independent of GPT Pro. Additionally, it suggested fixing various typos, making variable names uniform across sections, and adding some minor hypotheses in theorem statements which were already implicitly used in the proofs. All the main ideas, constructions, and proofs are entirely human-generated and everything in the paper was written by the authors.

\subsection{Open problems}
We end this section by listing some open questions.
\begin{enumerate}
    \item The proven QTB and hierarchical QTB distance bounds are not tight as illustrated in \Cref{ex:two_level_qtb_vs_one_level_true_distance}. The repeated use of Jensen's inequality and the degree growth compared to the root count weakens the bound. This points to a natural open problem: can we find an alternative proof for the distance bound for the hierarchical case that uses the reduced nondual support $[\ell]\setminus(S_+\cup S_-)$ more efficiently and better utilizes the distribution of zeros among repair groups?

    \item Can we construct explicit QHLRCs meeting the quantum Singleton-like bound?

    \item Can we extend the QTB decoder to the hierarchical setting such that we can correct errors up to half the proven distance bound?
    % \begin{enumerate}
%     \item Proving that puncturing a CSS code preserves the CSS orthogonality condition.
%     \item A CSS-based framework for $(r,\delta)$ QLRCs and $h$-level QHLRCs, including puncturing definitions and a Singleton-like bound for CSS codes built from dual-containing classical HLRCs.
%     \item Random $(r,\delta)$ QLRCs and random QHLRCs with high-probability distance bounds.
%     \item Explicit $(r,\delta)$ QTB codes and $h$-level HQTB codes from evaluation codes on $\F_q^*$, computing their dimensions, and proving their local and hierarchical repair properties through piecewise polynomial dual checks.
%     \item Proving a new non-vanishing theorem for specific polynomials $Q_b$ that vanish at $\delta-1$ consecutive roots of unity. The proof uses techniques involving complete homogeneous symmetric functions, complex analysis, resultants, and finite-field reduction \cite{cox2005using, Gardner2014, gelfand1994discriminants, Jacobi1841, Macdonald_2008, neukirch1999algebraic, trudi1862teoria}. The theorem bypasses the limitations of composite-order uncertainty principles and is used in the unfolded QTB and HQTB distance proofs.
%     \item Establishing distance lower bounds for unfolded QTB and HQTB codes.
%     \item Establishing folded QTB and folded hierarchical QTB codes. Folding preserves rate and locality, and the folded distance proofs combine uncertainty principles with determinant-polynomial root multiplicity.
%     \item Proving an efficient decoder for QTB codes.
% \end{enumerate}
\end{enumerate}

\section{Preliminaries}
\label{sec:prelim}
In this section, we introduce the necessary notation, preliminary results and definitions. Given a positive integer $t$, $[t]:=\{ 0, 1,  \dots, t-1 \}$.  We write \(\Z\), \(\Q\), and \(\C\) for the integers, rational numbers, and complex numbers, respectively. If \(p\) is a prime, we write \(\F_p\) for the finite field with \(p\) elements. We write $\mathbb{F}_q$ for the finite field where $q = p^m$ is a power of some prime $p$ and positive integer $m$. Let $\mathbb{F}_q^n$ denote the $n$-dimensional vector space over the field $\mathbb{F}_q$. Let $\mathbb{F}_q^*$ denote the set of non-zero elements of $\mathbb{F}_q$. The set of $m \times n$ matrices with entries in $\mathbb{F}_q$ is denoted by $\mathbb{F}_q^{m \times n}$. The entry in the $i^{th}$ row and $j^{th}$ column of a matrix $A \in \mathbb{F}_q^{m \times n}$ is denoted by $A_{ij}$. 

Let ${C} \subseteq \F_q^n$ be an $[n,k,d]_q$ classical linear code with length $n$, dimension $k$ and minimum distance $d$. The weight of a codeword, $\mathrm{wt}(c)$ is the (Hamming) distance between the codeword, $c$ and the zero codeword. Let $G \in \mathbb{F}_q^{k \times n}$ and $H \in \mathbb{F}_q^{n-k \times n}$ denote the generator matrix and parity check matrix of ${C}$, respectively. Given an $[n,k]$ linear code ${C},$ the subspace of $\mathbb{F}_q^n$ containing all those vectors that are orthogonal to every codeword in ${C}$ forms the \textit{dual code, ${C}^\perp$} of code ${C}$. That is, 
\({C}^\perp := \left\{\mathbf{u} \in \mathbb{F}_q^n: \mathbf{u\cdot v} =0, \forall\mathbf{v} \in {C}\right\}.\)

For any classical code $C$ of length $n$ over $\F_q$, a set of indices $I \subseteq \{1,\ldots,n\}$ such that $|I| = \ell$, and a projection map $\pi_I(C): \F_q^n \rightarrow \F_q^\ell$, we define the punctured code, \[C|_I = \pi_I(C) = \{ {\bf c}|_I = (c_j)_{j\in I} \mid {\bf c} = (c_1,\ldots,c_n) \in C\}.\] The shortened code is defined to be \[C^I =\sigma_I(C) = \{ {\bf c}|_I \mid {\bf c} = (c_1,\ldots,c_n) \in C \land \supp({\bf c}) \subseteq I\},\] where $\supp({\bf c}) = \{j \mid c_j \neq 0\}$. Note $\pi_I(C^\perp) = (\sigma_I(C))^\perp$, where $C^\perp$ is the dual of $C$.

An $(r,\delta)$ LRC is defined as follows:
\begin{definition}[$(r,\delta)$-LRCs (\cite{KPLK14})]
For every $i \in \{1,\ldots, n\}$, the $i^{th}$ code symbol $c_i$ of ${C}$ is said to have $(r,\delta)$ locality,
$\delta \geq 2$, if there exists a punctured code of ${C}$ with support containing $i$,
whose length is at most $r + \delta - 1$, and whose minimum distance is at least $\delta$, that is, there exists a subset $S_i \subseteq \{1,\ldots, n\}$  such that
\begin{itemize}
\item $i \in S_i$, $|S_i| \leq r + \delta -1$, and
\item $d_{min} ({C}|_{S_i}) \geq \delta$, where ${C}|_{S_i}$ denotes the code obtained when ${C}$ is punctured to the set of coordinates corresponding to $S_i$.
\end{itemize}
\end{definition}

An $h$-level HLRC is defined as follows:
\begin{definition}[$((r_{1},\delta_{1}), \ldots, (r_{h},\delta_{h}))$-HLRCs (\cite{sasidharan2015hlrc})]
An $[n,k,d]$ code ${C}$ is a code with a $h$-level hierarchical locality having locality parameters $(r_{1},\delta_{1}), \ldots, (r_{h},\delta_{h})$ if for every $1 \leq i \leq n$, there exists a set $I \subseteq \{1,\ldots, n\}$ of cardinality at most $r_1+ \delta_1 -1$ such that $i \in I$ and the punctured code ${C}|_{I}$ satisfies the following:
\begin{itemize}
    \item $\dim({C}|_{I}) \leq r_1$,
    \item $d({C}|_{I}) \geq \delta_1$,
    \item the punctured code ${C}|_{I}$ is a code with $(h-1)$-level hierarchical locality having local parameters $(r_{2},\delta_{2}), \ldots, (r_{h},\delta_{h})$.
\end{itemize}
\end{definition}

Let $\mathbb{F}_q[X]$ be the ring of polynomials over $\mathbb{F}_q$. Then for $S \subseteq \mathbb{Z}^+$, let \[\mathbb{F}_q[X]^S = \left\{\sum_{i\in S}f_iX^i : \forall i \in S,  f_i \in \mathbb{F}_q\right\}.\] Consider the following polynomial evaluation map \(\ev : \mathbb{F}_q[X]^S \rightarrow \mathbb{F}_q^{\mathbb{F}_q^*} \cong \mathbb{F}_q^{q-1},\) where $\ev(f) = (f(x))_{x \in \mathbb{F}_q^*}$. Note that $C = \ev(\mathbb{F}_q[X]^S)$ represents an evaluation code of length $q-1$ and dimension $|S|$. When $S = [\ell]$ for some $\ell \in [q]$, the code $C$ is a Reed-Solomon code. 

Let $\mathbb{C}^q$ be a $q$-dimensional vector space over the complex field $\mathbb{C}$. Let $\mathcal{Q} \subseteq (\mathbb{C}^{q})^{{\otimes n}} = \mathbb{C}^q \otimes \ldots \otimes \mathbb{C}^q$ be an $[[n,k,d]]_q$ quantum code of length $n$, dimension $k$, and minimum distance $d$. We specify the distinction between classical and quantum parameters through the text wherever it is not clear from context. Given two classical codes, $C_X \subseteq \F_q^n$ and $C_Z \subseteq \F_q^n$ such that $C_Z^\perp \subseteq C_X$, the code $\mathcal{Q} = \mathrm{CSS}(C_X,C_Z)$ represents a quantum CSS code as defined in \cite{Calderbank_1996,Steane_1996}.

Now we turn our attention to some definitions and results that are necessary to understand the main constructions in this manuscript.

\subsection{Field theory}\label{sec:field_theory}

Here, we collect some standard field-theoretic notation and results used throughout the paper. For more details, we refer the reader to \cite{neukirch1999algebraic}.

Let \(K\) be a field. A \emph{field extension} \(L/K\) is an inclusion of fields \(K\subseteq L\). It is
called \emph{finite} if \(L\) is finite-dimensional as a vector space over \(K\).
In this case, the \emph{degree} of the extension is
\(
[L:K]:=\dim_K L.
\)

A field extension \(L/K\) is called \emph{algebraic} if every element of \(L\) is a root of a nonzero polynomial with coefficients in \(K\). An algebraic field extension \(L/K\) is called \emph{separable} if every element of \(L\) has
separable minimal polynomial (i.e. coprime to its formal derivative) over \(K\). An algebraic field extension \(L/K\) is called \emph{normal} if every irreducible polynomial over \(K\) that has a root in \(L\) splits into linear factors over \(L\). An algebraic field extension which is both normal and separable is called \emph{Galois}. We denote the Galois group of a Galois extension \(L/K\) by \(\Gal(L/K)\).

An \emph{algebraic closure} of \(K\) is a field
\(\overline K\) containing \(K\) such that \(\overline K\) is algebraic over \(K\)
and every nonconstant polynomial in \(\overline K[X]\) has a root in
\(\overline K\). Equivalently, every nonconstant polynomial in \(K[X]\) splits
into linear factors over \(\overline K\). Whenever an algebraic closure is needed, we fix one and denote it by \(\overline K\).

Let \(L/K\) be a finite field extension. The \emph{field norm} from \(L\) to \(K\)
is the map
\[
N_{L/K}:L\longrightarrow K
\]
defined as follows. For \(x\in L\), let \(m_s: L\longrightarrow L\) be the \(K\)-linear map given by multiplication by \(x\): \(y\longmapsto xy\).
Then
\[
N_{L/K}(x):=\mathrm{det}_K(m_x).
\]
If \(L/K\) is separable, then the norm admits the equivalent description
\[
N_{L/K}(x)=\prod_{\sigma:L\hookrightarrow \overline K}\sigma(x),
\]
where the product is over all \(K\)-embeddings of \(L\) into a fixed algebraic
closure \(\overline K\). Here a \(K\)-embedding means an injective field
homomorphism that restricts to the identity on \(K\). In particular, if \(L/K\)
is Galois, then
\[
N_{L/K}(x)=\prod_{\sigma\in \Gal(L/K)}\sigma(x).
\]
When \(K\) is perfect, every finite extension of \(K\) is separable, so the
embedding formula for the norm applies to all finite extensions of \(K\).

For a positive integer \(n\), an element \(\zeta\) of a field is called an
\(n\)th root of unity if \(\zeta^n=1\).
It is called a \emph{primitive} \(n\)th root of unity if \(n\) is the smallest
positive integer such that \(\zeta^n=1\), or equivalently, if \(\zeta\) has
multiplicative order \(n\).

The \(n\)th \emph{cyclotomic polynomial} is the polynomial
\[
\Phi_n(X)
:=
\prod_{\substack{1\leq k\leq n\\ \gcd(k,n)=1}}
\left(X-e^{2\pi i k/n}\right)
\in \C[X].
\]
It is a standard fact that \(\Phi_n(X)\in \Z[X]\), and that
\(
X^n-1=\prod_{d\mid n}\Phi_d(X).
\)
The degree of \(\Phi_n(X)\) is \(\varphi(n)\), where \(\varphi\) denotes Euler's
totient function. We write \(\zeta_n\) for a fixed primitive \(n\)th root of
unity. Thus the roots of \(\Phi_n(X)\) over \(\C\) are precisely the primitive
\(n\)th roots of unity.

We end this section with a definition describing the set of subspaces of a vector space and its size.

\begin{definition}\label{def:grassmannian}
Let \(V\) be a finite-dimensional vector space over \(\F_q\), and let
\(0\le k\le \dim_{\F_q}V\). The \emph{Grassmannian} of \(k\)-dimensional
subspaces of \(V\), denoted
\(\Gr_{\F_q}(k,V),\)
is the set
\[
\Gr_{\F_q}(k,V)
:=
\{U\subseteq V: U\text{ is an }\F_q\text{-linear subspace and }
\dim_{\F_q}U=k\}.
\]
When \(V=\F_q^n\), we also write
\(
\Gr_{\F_q}(k,n):=\Gr_{\F_q}(k,\F_q^n).
\)
Since \(\F_q\) is finite, \(\Gr_{\F_q}(k,V)\) is a finite set. If
\(\dim_{\F_q}V=n\), then
\[
|\Gr_{\F_q}(k,V)|
=
\binom{n}{k}_q
:=
\prod_{i=0}^{k-1}
\frac{q^{n-i}-1}{q^{k-i}-1},
\]
where \(\binom{n}{k}_q\) is the Gaussian binomial coefficient.
\end{definition}

\subsection{Puncturing CSS codes}
\label{subsec:puncturing_css_codes}
As with classical codes, there is a puncturing operation associated to quantum
codes. For general stabilizer codes, puncturing is most naturally described as
a shortening operation on the stabilizer group. Let $\mathcal{P}_n$ be the Pauli group on $n$ qudits and $S\subseteq \mathcal P_n$
be a stabilizer group. If
\[
    P = \lambda P_1\otimes \cdots \otimes P_n \in \mathcal P_n,
\]
where each $P_i$ is a single-qudit Pauli operator and $\lambda$ is a phase, then
for a subset $I\subseteq \{1,\ldots,n\}$ we write
\[
    P|_I := \lambda \bigotimes_{i\in I} P_i
\]
whenever $P_j=I$ for all $j\notin I$. In other words, we only restrict $P$ to
$I$ if $P$ acts trivially outside $I$.

The puncturing of the stabilizer group $S$ to the coordinate set $I$ is defined
by
\[
    S|_I
    :=
    \{P|_I : P\in S,\ \mathrm{supp}(P)\subseteq I\}.
\]
Thus, one first keeps only those stabilizers which act trivially on the deleted
coordinates, and then restricts them to the remaining coordinates. Since $S$ is
abelian, $S|_I$ is also abelian. Moreover, $S|_I$ contains no nontrivial phase:
if $P|_I$ were a nontrivial phase, then $P$ itself would be a nontrivial phase in
$S$, which is impossible. Hence $S|_I$ is again a valid stabilizer group, now on
$|I|$ physical qudits. This is the stabilizer-code puncturing operation \cite{Grasslpuncturingstabilizer, Gundersen2025puncturing}.

For the constructions, we only need this operation for quantum CSS codes. Let
\[
    \mathcal Q = \mathrm{CSS}(C_X,C_Z)
\]
be a CSS code, where $C_Z^\perp\subseteq C_X$. Its stabilizer group is
\[
    S_{\mathcal Q}
    =
    \{X^aZ^b : a\in C_X^\perp,\ b\in C_Z^\perp\}.
\]
We define the puncturing of $\mathcal Q$ to the coordinate set
$I\subseteq \{1,\ldots,n\}$ by puncturing the associated classical codes:
\[
    \mathcal Q|_I
    :=
    \mathrm{CSS}(C_X|_I,C_Z|_I).
\]
Here $C|_I=\pi_I(C)$ denotes the classical puncturing of $C$ to the coordinates
in $I$.

\begin{lemma}\label{lem:puncturing_css}
    Let $\mathcal Q=\mathrm{CSS}(C_X,C_Z)$ be a quantum CSS code with
    $C_Z^\perp\subseteq C_X$. For any subset
    $I\subseteq \{1,\ldots,n\}$, the punctured code
    \[
        \mathcal Q|_I
        :=
        \mathrm{CSS}(C_X|_I,C_Z|_I)
    \]
    is again a well-defined quantum CSS code. Moreover, this definition agrees
    with the stabilizer-code puncturing operation described above.
\end{lemma}

\begin{proof}
    Let $\pi_I$ denote puncturing to the coordinate set $I$, so that
    $C|_I=\pi_I(C)$. Let $\sigma_I(C)$ denote shortening to $I$, namely
    \[
        \sigma_I(C)
        :=
        \{c|_I : c\in C,\ c_j=0 \text{ for all } j\notin I\}.
    \]
    We use the standard duality relation between puncturing and shortening:
    \(
        (C|_I)^\perp
        =
        \sigma_I(C^\perp).
    \)

    First, we check the CSS orthogonality condition. Since
    $C_Z^\perp\subseteq C_X$, we have
    \[
        (C_Z|_I)^\perp
        =
        \sigma_I(C_Z^\perp)
        \subseteq
        \sigma_I(C_X)
        \subseteq
        C_X|_I.
    \]
    and by duality, we obtain 
    \(
        (C_X|_I)^\perp
        \subseteq
        C_Z|_I.
    \)
    Therefore $C_X|_I$ and $C_Z|_I$ satisfy the CSS orthogonality conditions, so
    $\mathrm{CSS}(C_X|_I,C_Z|_I)$ is a well-defined quantum CSS code.

    It remains to compare this definition with stabilizer puncturing. The
    stabilizer group of $\mathcal Q$ is
    \[
        S_{\mathcal Q}
        =
        \{X^aZ^b : a\in C_X^\perp,\ b\in C_Z^\perp\}.
    \]
    Its stabilizer puncturing to $I$ is
    \[
        S_{\mathcal Q}|_I
        =
        \{(X^aZ^b)|_I :
        a\in C_X^\perp,\ b\in C_Z^\perp,\ 
        a_j=b_j=0 \text{ for all } j\notin I\}.
    \]
    Equivalently,
    \[
        S_{\mathcal Q}|_I
        =
        \{X^{a'}Z^{b'} :
        a'\in \sigma_I(C_X^\perp),\
        b'\in \sigma_I(C_Z^\perp)\}.
    \]
    On the other hand, the stabilizer group of
    $\mathrm{CSS}(C_X|_I,C_Z|_I)$ is
    \[
        \{X^{a'}Z^{b'} :
        a'\in (C_X|_I)^\perp,\
        b'\in (C_Z|_I)^\perp\}.
    \]
    Using $(C|_I)^\perp=\sigma_I(C^\perp)$, this becomes
    \(
        \{X^{a'}Z^{b'} :
        a'\in \sigma_I(C_X^\perp),\
        b'\in \sigma_I(C_Z^\perp)\},
    \)
    which is exactly $S_{\mathcal Q}|_I$. Thus, puncturing the classical codes
    $C_X$ and $C_Z$ agrees with puncturing the associated CSS stabilizer group.
\end{proof}

\subsection{QLRCs and QHLRCs}
\label{subsec:definitions_qlrc_qhlrc}
Before defining quantum locally recoverable codes, we recall what it means for a
quantum code to correct a set of erasures. Let
\(
 \mathcal E=\{X^aZ^b:a,b\in \mathbb F_q\}
\)
be the single-qudit Pauli error basis. For a density matrix
$\rho$ on a single qudit, define the completely depolarizing channel
\[
 \tau(\rho)
 :=
 \frac{1}{q^2}\sum_{E\in \mathcal E} E\rho E^\dagger 
\]
where \(E^\dagger\) is the Hermitian dual of \(E\). Thus $\tau$ applies each single-qudit Pauli error with equal probability. For a
set of coordinates $I\subseteq \{1,\ldots,n\}$, let $\tau^I$ denote the channel
on $(\mathbb C^q)^{\otimes n}$ which applies $\tau$ to the qudits in $I$ and
acts as the identity on the remaining qudits. Equivalently,
\[
 \tau^I(\rho)
 =
 \frac{1}{q^{2|I|}}
 \sum_{\substack{E\in \mathcal E^{\otimes n}\\ \mathrm{supp}(E)\subseteq I}}
 E\rho E^\dagger .
\]

Let $\mathcal Q\subseteq (\mathbb C^q)^{\otimes n}$ be a quantum code. We say
that $\mathcal Q$ corrects erasures at $I$ if there exists a trace-preserving
quantum operation $\mathcal R_{\mathcal Q,I}$ such that
\[
 \mathcal R_{\mathcal Q,I}\circ \tau^I(\ket{\phi}\bra{\phi})
 =
 \ket{\phi}\bra{\phi}
\]
for all $\ket{\phi}\in \mathcal Q$.

Equivalently, $\mathcal Q$ corrects erasures at $I$ if it can correct every
Pauli error whose support is contained in $I$, under the assumption that the
decoder knows the erased set $I$. This is the usual erasure model: the locations
of the errors are known, but the errors themselves are unknown.

We now impose a locality constraint on the recovery operation. The idea is that
erasures on a set $I$ should be recoverable by using only a small set of qudits
$J$ containing $I$.

\begin{definition}[$(I,J)$-QLRC {\cite{galindo2024quantumrdeltalocallyrecoverablecodes}}]
\label{def:I_J_qlrc}
Let $\emptyset\neq I\subsetneq J\subseteq \{1,\ldots,n\}$. A quantum code
$\mathcal Q\subseteq (\mathbb C^q)^{\otimes n}$ is said to be an
$(I,J)$-QLRC if there exists a trace-preserving quantum operation
$\mathcal R_{\mathcal Q,I}^J$, acting nontrivially only on the qudits indexed
by $J$, such that
\[
\mathcal R_{\mathcal Q,I}^J\circ \tau^I(\ket{\phi}\bra{\phi})
=
\ket{\phi}\bra{\phi}
\]
for all $\ket{\phi}\in \mathcal Q$.
\end{definition}

Thus, if $\mathcal Q$ is an $(I,J)$-QLRC, then erasures on the coordinates in
$I$ can be recovered locally using only the qudits in $J$.

\begin{definition}[$(r,\delta)$-QLRC {\cite{galindo2024quantumrdeltalocallyrecoverablecodes}}]
\label{def:r_delta_qlrc}
Let $r,\delta\geq 2$ be positive integers. A quantum error-correcting code
$\mathcal Q\subseteq (\mathbb C^q)^{\otimes n}$ is an $(r,\delta)$-QLRC if, for
each coordinate $i\in \{1,\ldots,n\}$, there exists a set
$J\subseteq \{1,\ldots,n\}$ containing $i$ such that
\[
\delta\le |J|\le r+\delta-1
\]
and such that, for every nonempty subset \(I\subsetneq J\) with
\(
|I|\le \delta-1,
\)
the code \(\mathcal Q\) is \((I,J)\)-locally recoverable.
\end{definition}

Equivalently, since \(|J|\ge\delta\), it suffices to require this for subsets
\(I\subsetneq J\) with \(|I|=\delta-1\). In other words, the definition states that every coordinate belongs to a local recovery set $J$ of size at
most $r+\delta-1$, and any erasure pattern of size at most $\delta-1$ inside
$J$ can be corrected using only the qudits in $J$. This is the quantum analogue
of classical $(r,\delta)$-locality.

The following result relates this notion of quantum locality
to classical locality for dual-containing codes.

\begin{theorem}[{\cite[Theorem 28]{galindo2024quantumrdeltalocallyrecoverablecodes}}]
\label{thm:css_qlrc_iff_clrc}
Let $C\subseteq \mathbb F_q^n$ be a linear code. Assume that
$C^\perp\subseteq C$ and $\delta\leq d(C^\perp)$, where $d(C^\perp)$ is the
minimum distance of $C^\perp$. Then
\(
\mathcal Q=\mathrm{CSS}(C,C)
\)
is an $(r,\delta)$-QLRC if and only if $C$ is an $(r,\delta)$ classical LRC.
\end{theorem}

We now extend the established $(r,\delta)$-QLRCs to a quantum hierarchical setting, analogous to classical HLRCs. For simplicity, we first define $2$-level QHLRCs.

\begin{definition}[$((r_1,\delta_1),(r_2,\delta_2))$-QHLRC]\label{def:2_level_qhlrc}
     Let $r_1,\delta_1, r_2,\delta_2 \in \Z^+$ such that $r_1\geq r_2$ and $\delta_1\geq \delta_2 \geq 2$. Let $\mathcal{Q} \subseteq (\C^{q})^{\otimes n}$ be a stabilizer quantum code. Fix two level locality parameters $(r_1,\delta_1), (r_2,\delta_2)$. We say $\mathcal{Q}$ is a $((r_1,\delta_1), (r_2,\delta_2))$-QHLRC if:
        \begin{enumerate}
            \item For every $i \in \{1,\ldots, n\}$ there exists a recovery set $J_1(i) \subseteq \{1,\ldots, n\}$ containing $i$ with $\delta_1 \leq |J_1(i)| \leq r_1 + \delta_1 -1$ such that for all $\emptyset \neq I_1 \subsetneq J_1(i)$ with $|I_1| \leq \delta_1-1$, $\mathcal{Q}$ is an $(I_1,J_1(i))$-QLRC. In other words, $\mathcal{Q}$ is an $(r_1,\delta_1)$-QLRC.
            \item The punctured code $\mathcal{Q}|_{J_1(i)}$ is itself an $(r_2,\delta_2)$-QLRC.
        \end{enumerate}
        %There exists a disjoint partitioning $\mathcal{J}_1 = \{J_{1,1}, \ldots, J_{1,m_1}\}$ of $\{1,\ldots, n\}$ such that for each group of coordinates $J_{1,j}$, the punctured code $\mathcal{Q}|_{J_{1,j}}$ is itself an $(r_2,\delta_2)$ quantum LRC.
       
\end{definition} Using \Cref{def:2_level_qhlrc}, we make the following observation: let $\mathcal{Q}$ be a $((r_1,\delta_1), (r_2,\delta_2))$-QHLRC. Fix a coordinate $i$ and an erasure pattern $I$ containing $i$. Let $J_1(i)$ be the corresponding recovery set. Let $J_2(i)$ be a recovery set that arises from $\mathcal{Q}|_{J_1(i)}$. If $I \subseteq J_2(i)$ and $|I| \leq \delta_2 - 1$ then we can recover qudits corresponding to $I$ using the qudits corresponding to $J_2(i) \setminus I$. Otherwise, assuming $|I| \leq \delta_1 - 1$ and $I \subseteq J_1(i)$, we can recover qudits corresponding to $I$ using the qudits corresponding to $J_1(i) \setminus I$.

Note that a $1$-level QHLRC is simply a QLRC. Additionally, it is not too difficult to generalize the previous definition to any number of hierarchical levels.

\begin{definition}[$h$-level QHLRC]\label{def:h_level_qhlrc}
Let $2 \leq h \in \Z^+$, and $ r_l,\delta_l \in \Z^+$ for all $l \in \{1,\ldots,h\}$. Consider $h$ level locality parameters $(r_1,\delta_1), \ldots, (r_h,\delta_h)$ with \[r_1\geq \cdots \geq r_h\quad\text{and}\quad \delta_1\geq \cdots\geq \delta_h \geq 2.\] A stabilizer quantum code $\mathcal{Q} \subseteq (\C^q)^{\otimes n}$ is said to be an $h$-level $((r_1,\delta_1), \ldots, (r_h,\delta_h))$-QHLRC if the following conditions are satisfied:
    \begin{enumerate}
        \item For every $i \in \{1,\ldots, n\}$ there exists a recovery set $J_1(i) \subseteq \{1,\ldots, n\}$ containing $i$ with $\delta_1 \leq |J_1(i)| \leq r_1 + \delta_1 -1$ such that for all $\emptyset \neq I_1 \subsetneq J_1(i)$ with $|I_1| \leq \delta_1-1$, $\mathcal{Q}$ is an $(I_1,J_1(i))$-QLRC. In other words, $\mathcal{Q}$ is an $(r_1,\delta_1)$-QLRC.
        %\item There exists a sequence of disjoint partitions $\mathcal{J}_1,\ldots, \mathcal{J}_{h-1}$ such that $\mathcal{J}_1$ partitions $\{1,\ldots, n\}$ into disjoint subsets and $\mathcal{J}_l$ is a set of disjoint sets which further partition the sets in $\mathcal{J}_{l-1}$ for $l = 2, \ldots, h-1$.
        %\item For each $l \in \{1,\ldots, h-1\}$, and for each $J \in \mathcal{J}_l$, the punctured code $Q|_{J}$ is an $(h-l)$-level $((r_{l+1},\delta_{l+1}),\ldots, (r_h,\delta_h))$-QHLRC.
        \item The punctured code $\mathcal{Q}|_{J_1(i)}$ is an $(h-1)$-level $((r_{2},\delta_{2}),\ldots, (r_h,\delta_h))$-QHLRC.
    \end{enumerate}
\end{definition} Notice that the second condition implies that, after suitably puncturing $h-1$ times, the resulting code is a $1$-level $(r_h,\delta_h)$-QHLRC which is simply an $(r_h,\delta_h)$-QLRC. 

% We end the section with the following theorem which we show that given a dual containing $h$-level classical HLRC, the corresponding quantum CSS code $\mathcal{Q} = \mathrm{CSS}(C,C)$ is an $h$-level QHLRC.

% \begin{theorem}\label{thm:css_qhlrc_iff_chlrc}
%     Let $C \subseteq \mathbb{F}_q^n$ be a linear code. Assume that $C^\perp \subseteq C$. Fix hierarchical locality and distance parameters $(r_l,\delta_l)_{l=1,\ldots,h}$ and assume $\delta_1 \leq d(C^\perp)$ where $d(C^\perp)$ is the minimum distance of $C^\perp$. Then $\mathcal{Q} = \mathrm{CSS}(C,C)$ is an $h$-level $((r_1,\delta_1), \ldots, (r_h,\delta_h))$-QHLRC if and only if $C$ is an $h$-level $((r_1,\delta_1), \ldots, (r_h,\delta_h))$ classical HLRC.
% \end{theorem}

% \begin{proof}
%     We induct on $h$. If $h = 1$, then we reduce to \Cref{thm:css_qlrc_iff_clrc}. Suppose the claim holds for $(h-1)$-level hierarchical codes. By \Cref{thm:css_qlrc_iff_clrc}, $C$ is an $(r_1,\delta_1)$ classical LRC if and only if $\mathcal{Q}$ is an $(r_1,\delta_1)$-QLRC.

%     Fix a coordinate $i$ and corresponding recovery set $J(i)$. By \Cref{def:classical_hlrc}, $C|_{J(i)}$ is an $(h-1)$-level classical HLRC. Since puncturing preserves the CSS property (\Cref{lem:puncturing_css}), $\mathcal{Q}|_{J(i)}$ is an $(h-1)$-level QHLRC. Therefore, $\mathcal{Q}$ satisfies \Cref{def:h_level_qhlrc}.
% \end{proof}

We end this section by proving two Singleton-type bounds.
\begin{theorem}\label{thm:singleton_qlrc}
    Let $\mathcal{Q}$ be a QLRC of block length $n$, dimension $k$, distance $d$ with locality parameters $r, \delta$. Assume that for each coordinate $i \in \{1,\ldots, n\}$, there exists a recovery set $J(i)$ containing $i$ such that $|J(i)| = r + \delta -1$. Further assume that for all $i,j \in \{1,\ldots, n\}$ either $J(i) = J(j)$ or $J(i) \cap J(j) = \emptyset$. Then \[k \leq \left(1 - \frac{2(\delta-1)}{r+\delta-1}\right)n-2\left(d-1-(\delta-1)\left\lceil\frac{d-1}{r}\right\rceil\right).\] 
\end{theorem}
\begin{proof}
    Follows from \cite[Theorem 36]{golowich2025quantum}.
\end{proof}
The above bound can also be extended to the hierarchical setting, but since the explicit constructions in this work are of the form $\CSS(C,C)$ we give a specialized Singleton-like bound for CSS codes constructed from a dual-containing classical HLRC. When \(h=1\), this recovers the same bound as \cite[Theorem 30]{galindo2024quantumrdeltalocallyrecoverablecodes}.
\begin{proposition} \label{prop:singleton.bound.qhlrc}
    Consider a classical $h$-level $((r_1,\delta_1),\ldots, (r_h,\delta_h))$-HLRC $C$ such that $C^\perp \subseteq C$ and $d(C^\perp) \geq \delta_1$. Then $\dim C = (N+k)/2$ and $\mathcal{Q} = \mathrm{CSS}(C,C)$ is an $[[N,k,\geq d(C)]]$ quantum $h$-level HLRC satisfying \begin{equation}\label{eqn:singleton_bound_css_qhlrc}
         k + 2d(C) \leq N+2 -  2\sum_{l = 1}^{h-1}\left(\left\lceil{\frac{N+k}{2r_l}}\right\rceil - 1\right)(\delta_l - \delta_{l+1}) - 2\left(\left\lceil{\frac{N+k}{2r_h}}\right\rceil - 1\right)(\delta_h - 1)
    \end{equation}
\end{proposition}

\begin{proof}
    Denote the parameters of $C$ as $[N,k',d(C)]$. By the classical HLRC Singleton-like bound \cite[Theorem 3.1]{sasidharan2015hlrc}, we have \begin{align}
        d(C) &\leq N- k' + 1 - \sum_{l = 1}^{h-1}\left(\left\lceil{\frac{k'}{r_l}}\right\rceil - 1\right)(\delta_l - \delta_{l+1}) - \left(\left\lceil{\frac{k'}{r_h}}\right\rceil - 1\right)(\delta_h - 1).
    \end{align} We know by the CSS construction that $k = k' + k' - N$ so $k' = (N+k)/2$. Plugging this into the above bound, we have \begin{align*}
        d(C) \leq N - \frac{N+k}{2} + 1 - \sum_{l = 1}^{h-1}\left(\left\lceil{\frac{N+k}{2r_l}}\right\rceil - 1\right)(\delta_l - \delta_{l+1}) - \left(\left\lceil{\frac{N+k}{2r_h}}\right\rceil - 1\right)(\delta_h - 1)
    \end{align*} and upon rearranging we get \begin{align*}
        \frac{N+k}{2} + d(C) &\leq N + 1 - \sum_{l = 1}^{h-1}\left(\left\lceil{\frac{N+k}{2r_l}}\right\rceil - 1\right)(\delta_l - \delta_{l+1}) - \left(\left\lceil{\frac{N+k}{2r_h}}\right\rceil - 1\right)(\delta_h - 1)\\
        k + 2d(C) &\leq N+2 -  2\sum_{l = 1}^{h-1}\left(\left\lceil{\frac{N+k}{2r_l}}\right\rceil - 1\right)(\delta_l - \delta_{l+1}) - 2\left(\left\lceil{\frac{N+k}{2r_h}}\right\rceil - 1\right)(\delta_h - 1),
    \end{align*} as desired.
\end{proof} 

\subsection{Decoding algorithms}
Here we recall a few results that are necessary to prove the efficiency of our decoding algorithm, \Cref{alg:decoding_r_delta_qtb}. Given a Reed-Solomon code with block length $q-1$ and rate $R = \ell/(q-1)$, there exists an efficient list decoding algorithm for the Reed-Solomon code which can be summarized in the following statement.

\begin{theorem}[\cite{GS98}]
\label{thm:RS.list.dec}
The Reed-Solomon code with parameters $q,\ell$ has an $e$-list-decoding algorithm that runs in time $q^{O(1)}$ for $e = (q-1)(1-\sqrt{R})$.
\end{theorem}

Furthermore, it is well-known that to decode a CSS code $\mathcal{Q}=CSS(C_X,C_Z)$ it is sufficient to have classical decoders for $C_X,C_Z$.

\begin{proposition}[\cite{Calderbank_1996,Steane_1996}]
Let $\mathcal{Q}=\CSS(C_X,C_Z)$ be a CSS code of block length $n$ over $\mathbb{F}_q^s$ of size $|\mathbb{F}_q^s|$. Let $e\geq 0$ be an integer such that for each permutation $(\alpha,\beta)$ of $(X,Z)$, there exists a classical decoding algorithm $\mathrm{Dec}_\alpha$ that takes as input a classical corrupted codeword $c+b$  for some $c \in C_\alpha$ and some corruption $b \in (\mathbb{F}_q^s)^n$ of Hamming weight $|b| \leq e$, and outputs some $c' \in C_\alpha$ such that $c'-c \in C_\beta^\perp$.
Then \(\mathcal Q\) has a decoding algorithm \(\mathrm{Dec}\) that recovers from errors of weight \(e\), so \(\mathcal{Q}\) has distance \(d \geq 2e  + 1\). Furthermore, if each $\mathrm{Dec}_\alpha$ has running time \(T_\alpha(n,a)\), then \(\mathrm{Dec}\) has running time
\(T_X(n,a) + T_Z (n,a) + O(n^3 \mathrm{poly log} \, a).\)
\end{proposition}

\subsection{The resultant}\label{sec:resultant}

In this section we recall the Sylvester matrix and the resultant of two
univariate polynomials. We also record the description of the resultant in terms
of roots in an algebraic closure, its interpretation as a field norm in the
irreducible case, and a reduction-modulo-\(p\) criterion that will be used later. For a more in-depth treatment, see
\cite{cox2005using, gelfand1994discriminants, neukirch1999algebraic}.
We will use the field-theoretic notation fixed in
\Cref{sec:field_theory}. 

Let \(K\) be a field, and let
\[
f(X)=a_mX^m+a_{m-1}X^{m-1}+\cdots+a_0,
\qquad
g(X)=b_nX^n+b_{n-1}X^{n-1}+\cdots+b_0
\]
be nonzero polynomials in \(K[X]\) of degrees \(m\) and \(n\), respectively.

\begin{definition}[Sylvester matrix]
The \emph{Sylvester matrix} of \(f\) and \(g\), denoted \(\Syl(f,g)\), is the
\((m+n)\times(m+n)\) matrix obtained by writing \(n\) shifted copies of the
coefficient vector of \(f\), followed by \(m\) shifted copies of the coefficient
vector of \(g\). Explicitly,
\[
\Syl(f,g)=
\begin{pmatrix}
a_m & a_{m-1} & \cdots & a_0 & 0 & \cdots & 0\\
0 & a_m & a_{m-1} & \cdots & a_0 & \ddots & \vdots\\
\vdots & \ddots & \ddots & \ddots & & \ddots & 0\\
0 & \cdots & 0 & a_m & a_{m-1} & \cdots & a_0\\[4pt]
b_n & b_{n-1} & \cdots & b_0 & 0 & \cdots & 0\\
0 & b_n & b_{n-1} & \cdots & b_0 & \ddots & \vdots\\
\vdots & \ddots & \ddots & \ddots & & \ddots & 0\\
0 & \cdots & 0 & b_n & b_{n-1} & \cdots & b_0
\end{pmatrix},
\]
where the first block consists of \(n\) rows and the second block consists of
\(m\) rows.
\end{definition}

\begin{definition}[Resultant]
The \emph{resultant} of \(f\) and \(g\) is
\(
\Res(f,g):=\det \Syl(f,g).
\)
\end{definition}

Since the entries of \(\Syl(f,g)\) are polynomial expressions in the coefficients
of \(f\) and \(g\), the resultant is itself a polynomial expression in those
coefficients, with integer coefficients. In particular, if \(f,g\in \Z[X]\), then
\(\Res(f,g)\in \Z\).

\begin{proposition}[Root formula for the resultant]
\label{prop:root_formula_resultant}
Let \(K\) be a field. Let \(f,g\in K[X]\) be nonzero polynomials of degrees
\(m\) and \(n\), respectively, and write
\[
f(X)=a_m\prod_{i=1}^m (X-\alpha_i),
\qquad
g(X)=b_n\prod_{j=1}^n (X-\beta_j)
\]
in \(\overline K[X]\), where the roots are counted with multiplicity. Then
\[
\Res(f,g)
=
a_m^n b_n^m
\prod_{i=1}^m\prod_{j=1}^n(\alpha_i-\beta_j).
\]
Equivalently,
\[
\Res(f,g)
=
a_m^n\prod_{i=1}^m g(\alpha_i)
=
(-1)^{mn}b_n^m\prod_{j=1}^n f(\beta_j).
\]
\end{proposition}

\begin{proof}
The equivalence of the displayed formulas follows from the fact that
\[
g(\alpha_i)=b_n\prod_{j=1}^n(\alpha_i-\beta_j)
\quad \text{and} \quad
f(\beta_j)=a_m\prod_{i=1}^m(\beta_j-\alpha_i).
\]
The agreement of these expressions with \(\det \Syl(f,g)\) is the standard
determinantal formula for the resultant.
\end{proof}

The most important immediate consequence is the following.

\begin{corollary}\label{cor:resultant_common_root}
Let \(K\) be a field, and let \(f,g\in K[X]\) be nonzero polynomials. Then
\(
\Res(f,g)=0
\)
if and only if \(f\) and \(g\) have a common root in \(\overline K\).
Equivalently, \(\Res(f,g)=0\) if and only if \(f\) and \(g\) have a nonconstant
common factor in \(K[X]\).
\end{corollary}

\begin{proof}
By \Cref{prop:root_formula_resultant},
\[
\Res(f,g)
=
a_m^n b_n^m
\prod_{i=1}^m\prod_{j=1}^n(\alpha_i-\beta_j).
\]
Since \(a_m\neq 0\) and \(b_n\neq 0\), this product vanishes if and only if
\(\alpha_i=\beta_j\) for some \(i,j\), that is, if and only if \(f\) and \(g\)
share a root in \(\overline K\). The final equivalence is standard: two polynomials in \(K[X]\) have a common root
in \(\overline K\) if and only if they have a nonconstant common factor in
\(K[X]\).
\end{proof}

We next record the norm interpretation of the resultant in the irreducible case.

\begin{proposition}[Resultant as a norm]
\label{prop:resultant_as_norm}
Let \(K\) be a field, let
\(
f(X)=a_mX^m+\cdots+a_0\in K[X],\) and \(g(X)\in K[X],
\)
and suppose that \(f\) is irreducible and separable over \(K\). Let \(\alpha\) be
a root of \(f\) in \(\overline K\). Then
\[
\Res(f,g)
=
a_m^{\deg g}\,
N_{K(\alpha)/K}\left(g(\alpha)\right).
\]
In particular, if \(f\) is monic, then
\(
\Res(f,g)
=
N_{K(\alpha)/K}\left(g(\alpha)\right).
\)
\end{proposition}

\begin{proof}
Let \(m=\deg f\) and \(n=\deg g\). Write
\[
f(X)=a_m\prod_{i=1}^m(X-\alpha_i)
\]
in \(\overline K[X]\), where \(\alpha_1,\dots,\alpha_m\) are the roots of \(f\).
Since \(f\) is irreducible and separable, these roots are distinct and are
precisely the images of \(\alpha\) under the \(K\)-embeddings
\(
\sigma:K(\alpha)\hookrightarrow \overline K.
\)
By \Cref{prop:root_formula_resultant},
\[
\Res(f,g)
=
a_m^n\prod_{i=1}^m g(\alpha_i).
\]
On the other hand, by the embedding formula for the norm,
\[
N_{K(\alpha)/K}\left(g(\alpha)\right)
=
\prod_{\sigma:K(\alpha)\hookrightarrow \overline K}
\sigma\left(g(\alpha)\right)
=
\prod_{\sigma:K(\alpha)\hookrightarrow \overline K}
g\left(\sigma(\alpha)\right)
=
\prod_{i=1}^m g(\alpha_i).
\]
Combining the two formulas gives
\(
\Res(f,g)
=
a_m^n\,
N_{K(\alpha)/K}\left(g(\alpha)\right).
\)
The second statement follows by setting \(a_m=1\).
\end{proof}

A particularly important special case is when \(f\) is a cyclotomic polynomial.

\begin{corollary}[Cyclotomic norm formula]
\label{cor:cyclotomic_norm_formula}
Let \(n\geq 1\), let \(\Phi_n(X)\in \Z[X]\) be the \(n\)th cyclotomic polynomial,
and let \(\zeta_n\in \C\) be a primitive \(n\)th root of unity. Then for every
\(A(X)\in \Z[X]\),
\[
\Res(\Phi_n(X),A(X))
=
N_{\Q(\zeta_n)/\Q}\left(A(\zeta_n)\right).
\]
% Equivalently,
% \[
% \Res(A(X),\Phi_n(X))
% =
% (-1)^{\deg(A)\varphi(n)}
% N_{\Q(\zeta_n)/\Q}\left(A(\zeta_n)\right),
% \]
% where \(\varphi\) is Euler's totient function.
\end{corollary}

\begin{proof}
The polynomial \(\Phi_n(X)\) is monic and irreducible over \(\Q\), and
\(\zeta_n\) is one of its roots. Hence the formula follows immediately
from \Cref{prop:resultant_as_norm}. 
% The second formula follows from the symmetry
% relation
% \[
% \Res(A,\Phi_n)=(-1)^{\deg(A)\deg(\Phi_n)}\Res(\Phi_n,A),
% \]
% together with \(\deg(\Phi_n)=\varphi(n)\).
\end{proof}

We will also need the following basic reduction criterion.

\begin{proposition}
\label{prop:resultant_mod_p}
Let \(f,g\in \Z[X]\), and let \(p\) be a prime. If the reductions of \(f\) and
\(g\) modulo \(p\) have a common root in an algebraic closure of \(\F_p\), then
\(
p\mid \Res(f,g).
\)

Equivalently, if \(p\nmid \Res(f,g)\), then the reductions of \(f\) and \(g\)
modulo \(p\) have no common root in an algebraic closure of \(\F_p\).
\end{proposition}

\begin{proof}
Let \(m=\deg f\) and \(n=\deg g\). Let
\(\overline f,\overline g\in\F_p[X]\) denote the reductions of \(f\) and \(g\)
modulo \(p\). We write
\[
\Res_{m,n}(\overline f,\overline g)
:=
\det \Syl_{m,n}(\overline f,\overline g),
\]
where \(\Syl_{m,n}\) is the Sylvester matrix formed using the degree bounds
\(\deg \overline f\le m\) and \(\deg \overline g\le n\), i.e. after padding
the coefficient vectors with leading zeroes if necessary.

Suppose that \(\overline f\) and \(\overline g\) have a common root
\(\alpha\in \overline{\F}_p\). Consider the padded Sylvester matrix
\(
\Syl_{m,n}(\overline f,\overline g).
\)
With the convention that coefficient rows are ordered from highest degree to
lowest degree, this matrix has rows given by the coefficient vectors of
\(
X^{n-1}\overline f,\ldots,\overline f,
X^{m-1}\overline g,\ldots,\overline g,
\)
each padded to length \(m+n\).

Let
\(
v_\alpha:=
[\alpha^{m+n-1},\alpha^{m+n-2},\ldots,\alpha,1]^T.
\)
Then dotting the coefficient row of any polynomial \(P\) of degree at most
\(m+n-1\) with \(v_\alpha\) gives \(P(\alpha)\). Hence each shifted
\(\overline f\)-row satisfies
\[
\operatorname{coeff}(X^k\overline f)\cdot v_\alpha
=
(X^k\overline f)(\alpha)
=
\alpha^k\overline f(\alpha)
=
0,
\]
and similarly each shifted \(\overline g\)-row satisfies
\[
\operatorname{coeff}(X^k\overline g)\cdot v_\alpha
=
(X^k\overline g)(\alpha)
=
\alpha^k\overline g(\alpha)
=
0.
\]
Thus,
\(
\Syl_{m,n}(\overline f,\overline g)v_\alpha=0.
\)
Since \(v_\alpha\neq 0\), the padded Sylvester matrix is singular. Therefore
\[
\Res_{m,n}(\overline f,\overline g)
=
\det \Syl_{m,n}(\overline f,\overline g)
=
0
\]
in \(\F_p\).

On the other hand, \(\Res_{m,n}\) is an integer polynomial in the coefficients
of its two inputs. Therefore, reducing the entries of the integer Sylvester
matrix modulo \(p\) gives
\[
\Res_{m,n}(\overline f,\overline g)
\equiv
\Res_{m,n}(f,g)
\pmod p.
\]
Since the left-hand side is \(0\), we obtain
\(
\Res(f,g)\equiv 0\pmod p.
\)
Equivalently,
\(
p\mid \Res(f,g).
\) The final statement is the contrapositive.
\end{proof}

For later use, let us also note the following reformulation of
\Cref{prop:resultant_mod_p}. If \(f,g\in \Z[X]\) and
\(p\nmid \Res(f,g)\), then the reductions of \(f\) and \(g\) modulo \(p\) are
coprime in \(\F_p[X]\). In particular, they cannot have a common root in any
extension field of \(\F_p\).

\subsection{Uncertainty principles}\label{sec:uncertainty_principles}

In this section, we recall an uncertainty principle for finite abelian groups. The first proposition is due to
Meshulam \cite{MESHULAM200663} for finite abelian groups of arbitrary order. Meshulam's composite-order uncertainty principle generalizes the sharper prime-order uncertainty principle due to Tao \cite{tao2003}. We will also record a finite-field version of the composite-order uncertainty principle, valid outside finitely many characteristics and derive an immediate corollary for the finite-field version of the prime-order uncertainty principle.

Let \(n\ge 2\), let \(\zeta_n\in \C\) be a primitive \(n\)th root of unity, and let
\(
\Omega_n:=\{\zeta_n^j : j=0,\dots,n-1\}.
\)
Given a nonzero polynomial
\[
f(X)=\sum_{i=0}^{n-1} f_iX^i \in \C[X]
\qquad\text{with}\qquad \deg f<n,
\]
we define
\[
|f|:=|\{\,i\in\{0,\dots,n-1\}:f_i\neq 0\,\}|
\]
and
\[
|\ev(f)|_{\Omega_n}:=|\{\,\xi\in \Omega_n:f(\xi)\neq 0\,\}|.
\]

Equivalently, if one identifies \(f\) with its coefficient vector in \(\C^{\mathbb Z/n\mathbb Z}\),
then \(|f|\) is the support size of that vector and \(|\ev(f)|_{\Omega_n}\) is the support size of its
discrete Fourier transform.

\begin{proposition}[Meshulam\cite{MESHULAM200663}]\label{prop:uncertainty_composite_complex}
Let \(0\neq f(X)\in \C[X]\) with \(\deg f<n\). Let \(d_1<d_2\) be consecutive divisors of \(n\)
such that
\(
d_1\le |f|\le d_2.
\)
Then
\[
|\ev(f)|_{\Omega_n}
\ge
\frac{n}{d_1d_2}\left(d_1+d_2-|f|\right).
\]
Equivalently, the number of roots of \(f\) in \(\Omega_n\) is at most
\[
n-\frac{n}{d_1d_2}\left(d_1+d_2-|f|\right).
\]
\end{proposition}

\begin{corollary}\label{cor:uncertainty_composite_finite_field}
Fix \(n\ge 2\). Then there exists a nonzero integer \(\mathcal U_n\) such that the following holds.
If \(\F_q\) is a finite field of characteristic \(p\) with \(p\nmid \mathcal U_n\) and \(n \mid (q-1)\),
then every nonzero polynomial \(f(X)\in \F_q[X]\) of degree \(<n\) satisfies
\[
|\ev(f)|_{\Omega_n}
\ge
\frac{n}{d_1d_2}\left(d_1+d_2-|f|\right),
\]
where \(d_1<d_2\) are the consecutive divisors of \(n\) such that
\(
d_1\le |f|\le d_2.
\)
Equivalently, the number of roots of \(f\) in \(\Omega_n\subseteq \F_q^*\) is at most
\[
n-\frac{n}{d_1d_2}\left(d_1+d_2-|f|\right).
\]
\end{corollary}

\begin{proof}
For each \(k\in\{1,\dots,n\}\), let \(u_n(k)\) denote the Meshulam lower bound
\[
u_n(k):=\frac{n}{d_1d_2}(d_1+d_2-k),
\]
where \(d_1<d_2\) are the consecutive divisors of \(n\) such that \(d_1\le k\le d_2\).

Fix \(k\), and let \(A,B\subseteq \{0,\dots,n-1\}\) be such that \(|A| =k \) and \(|B| < u_n(k).\)
Consider the \((n-|B|)\times k\) matrix
\[
V_{A,B}(X):=\left(X^{ab}\right)_{a\in B^c,\ b\in A},
\]
where \(B^c=\{0,\dots,n-1\}\setminus B\), and define \(V_{A,B}(\zeta_n)\) by specializing \(X=\zeta_n\).

We claim that \(V_{A,B}(\zeta_n)\) has full column rank \(k\). Indeed, if not, then there would exist
a nonzero coefficient vector \(c=(c_b)_{b\in A}\in \C^A\) such that
\(
V_{A,B}(\zeta_n)c=0.
\)
Defining
\[
f(X):=\sum_{b\in A} c_b X^b,
\]
we would have \(|f|=|A|=k\), and the equation \(V_{A,B}(\zeta_n)c=0\) implies \(f(\zeta_n^a)=0\) for all $a \in B^c$.
Hence,
\[
|\ev(f)|_{\Omega_n}\le |B|<u_n(k),
\]
contradicting \Cref{prop:uncertainty_composite_complex}. Thus \(V_{A,B}(\zeta_n)\) has full column rank.

It follows that some \(k\times k\) minor of \(V_{A,B}(\zeta_n)\) is nonzero. Choose one such minor and
denote the corresponding determinant polynomial by \(\Delta_{A,B}(X)\in \Z[X]\). Then
\(
\Delta_{A,B}(\zeta_n)\neq 0.
\)
Since \(\zeta_n\) is a root of \(\Phi_n(X)\), it follows that \(\Delta_{A,B}(X)\) and \(\Phi_n(X)\)
have no common root over \(\C\), and therefore
\[
\Res(\Delta_{A,B}(X),\Phi_n(X))\neq 0.
\]

Now define
\[
\mathcal U_n:=
\prod_{\substack{1\le k\le n\\ A,B\subseteq\{0,\dots,n-1\}\\ |A|=k,\ |B|<u_n(k)}}
\Res(\Delta_{A,B}(X),\Phi_n(X)).
\]
Since there are only finitely many pairs \((A,B)\), this is a finite product of nonzero integers, hence \(\mathcal U_n\neq 0\).

Now let \(\F_q\) be a finite field of characteristic \(p\) such that \(p\nmid \mathcal U_n\) and
\(n\mid(q-1)\), and let \(\omega\in \F_q^*\) be a primitive \(n\)th root of unity. Suppose, for contradiction,
that there exists a nonzero polynomial
\[
f(X)=\sum_{i=0}^{n-1} f_iX^i\in \F_q[X]
\qquad\text{with}\qquad
\deg f<n
\]
such that 
\(
|\ev(f)|_{\Omega_n} < u_n(|f|).
\) Let
\(
A:=\{\,i:f_i\neq 0\,\},\)
and \(
B:=\{\,a\in\{0,\dots,n-1\}:f(\omega^a)\neq 0\,\}.
\)
Then \(|A|=|f|\), and by assumption
\[
|B| = |\ev(f)|_{\Omega_n} < u_n(|f|) = u_n(|A|).
\]
For every \(a\in B^c\), we have
\[
\sum_{b\in A} f_b\omega^{ab}=0.
\]
Thus the coefficient vector \((f_b)_{b\in A}\neq 0\) lies in the kernel of \(V_{A,B}(\omega)\), so
\(V_{A,B}(\omega)\) does not have full column rank. Therefore every \( |A|\times |A| \) minor of
\(V_{A,B}(\omega)\) vanishes, including the chosen one:
\(
\Delta_{A,B}(\omega)=0.
\)

Since \(\omega\) is a primitive \(n\)th root of unity and \(p\nmid n\), the reduction of \(\Phi_n(X)\)
modulo \(p\) also vanishes at \(\omega\). Hence the reductions of \(\Delta_{A,B}(X)\) and \(\Phi_n(X)\)
have a common root in an algebraic closure of \(\F_p\). By \Cref{prop:resultant_mod_p},
\[
p\mid \Res(\Delta_{A,B}(X),\Phi_n(X)).
\]
But this resultant is one of the factors of \(\mathcal U_n\), so this implies \(p\mid \mathcal U_n\),
contrary to hypothesis. This contradiction proves the corollary.
\end{proof}

\begin{corollary}\label{cor:uncertainty_prime_finite_field}
Fix a prime \(n\). Then there exists a nonzero integer \(\mathcal U_n'\) such that the following holds.
If \(\F_q\) is a finite field of characteristic \(p\) with
\(p\nmid \mathcal U_n'\) and \(n \mid (q-1)\),
then every nonzero polynomial \(f(X)\in \F_q[X]\) of degree \(<n\) satisfies
\[
|f|+|\ev(f)|_{\Omega_n}\ge n+1.
\]
Equivalently,
\[
|\ev(f)|_{\Omega_n}\ge n+1-|f|,
\]
and the number of roots of \(f\) in \(\Omega_n\) is at most \(|f|-1\).
\end{corollary}

\begin{proof}
When \(n\) is prime, the only divisors of \(n\) are \(1\) and \(n\). Thus,
\Cref{cor:uncertainty_composite_finite_field} yields
\[
|\ev(f)|_{\Omega_n}
\ge
\frac{n}{1\cdot n}(1+n-|f|)
=
n+1-|f|,
\]
which is equivalent to the stated inequality. One may take \(\mathcal U_n'=\mathcal U_n\).
\end{proof}

\begin{remark}\label{rem:uncertainty_used_by_GG}
The prime-order finite-field uncertainty principle in
\Cref{cor:uncertainty_prime_finite_field} is the finite field version of the uncertainty principle proven by Tao \cite{tao2003}. It is also the same statement used by Golowich and Guruswami in their paper \cite[Proposition 66]{golowich2025quantum} on quantum locally recoverable codes, where it is attributed to Goldstein--Guralnick--Isaacs \cite{Goldstein_Guralnick_Isaacs_2005}.
\end{remark}

This completes the preliminaries. In the next section, we present randomized constructions of $(r,\delta)$-QLRCs and $h$-level QHLRCs achieving good rate--distance--locality tradeoffs.

\section{\texorpdfstring{Random $(r,\delta)$-QLRCs}{Random (r, delta)-QLRCs}}
\label{sec:random_r_delta_qlrc}
In this section, we construct random $(r,\delta)$-QLRCs. We begin with a lemma that we will use repeatedly in the random constructions.

\begin{lemma}\label{lem:vandermonde_local_mds}
Let \(\alpha_1,\ldots,\alpha_n\in\F_q\) be distinct. For an integer \(t\le n\), consider the \(t\times n\) Vandermonde matrix
\[
V_t(\alpha_1,\ldots,\alpha_n)
=
\begin{bmatrix}
1&1&\cdots&1\\
\alpha_1&\alpha_2&\cdots&\alpha_n\\
\vdots&\vdots&\ddots&\vdots\\
\alpha_1^{t-1}&\alpha_2^{t-1}&\cdots&\alpha_n^{t-1}
\end{bmatrix}.
\]
Then any \(t\) columns of \(V_t(\alpha_1,\ldots,\alpha_n)\) are linearly
independent.
\end{lemma}

\begin{proof}
Choose any \(t\) columns, say those indexed by \(j_1,\ldots,j_t\). The
corresponding determinant is
\[
\det(\alpha_{j_b}^{a-1})_{1\le a,b\le t}
=
\prod_{1\le u<v\le t}(\alpha_{j_v}-\alpha_{j_u}),
\]
which is nonzero because the \(\alpha_i\)'s are distinct.
\end{proof}

\begin{construction}[Random $(r,\delta)$-QLRC]
\label{const:random_r_delta_qlrc}
We are given a block length $N$, locality parameter $r$, %\leq n/2$
distance parameter $\delta \leq r$, and finite field $\mathbb{F}_q$ with $q \geq n:= r+\delta-1$. Assume that $m := N/n$ is an integer. Choose $\ell \in [N/2 - m(\delta-1)]$. We define a random QLRC via a CSS code $\mathcal{Q} = \mathrm{CSS}(C_X, C_Z)$ which is sampled as follows. Initialize $G_i$ to be the $(\delta - 1) \times n$ Vandermonde matrix \[G_i = \begin{bmatrix}
    1&1&\cdots &1\\
    \alpha_{i,1}&\alpha_{i,2}&\cdots &\alpha_{i,n}\\
    \vdots&\vdots&\ddots&\vdots\\
    \alpha_{i,1}^{\delta-2}&\alpha_{i,2}^{\delta-2}&\cdots &\alpha_{i,n}^{\delta-2}
\end{bmatrix}\] where for each $i = 1,\ldots, m$, every $\alpha_{i,j} \in \mathbb{F}_q$ is distinct for $j = 1, \ldots, n$. Let $H_X$ be a block diagonal matrix comprised of $m$ matrices: \[H_X = \begin{bmatrix}
    G_1\\
    &G_2\\
    &&\ddots\\
    &&&G_m
\end{bmatrix}.\] 
Let $\beta_{i,l} = \prod_{j\neq l}\frac{1}{\alpha_{i,l}-\alpha_{i,j}}$  for $l = 1,\ldots, n$. Let $H_i$ be an $r \times n$ Vandermonde-like matrix \begin{align*}H_i &= \begin{bmatrix}
    1&1&\cdots &1\\
    \alpha_{i,1}&\alpha_{i,2}&\cdots &\alpha_{i,n}\\
    \vdots&\vdots&\ddots&\vdots\\
    \alpha_{i,1}^{r-1}&\alpha_{i,2}^{r-1}&\cdots &\alpha_{i,n}^{r-1}
\end{bmatrix}\cdot\begin{bmatrix} \beta_{i,1} \\&\beta_{i,2}\\&&\ddots\\&&&\beta_{i,n}\end{bmatrix}\end{align*} and let $H_Z$ be the block matrix \[H_Z = \begin{bmatrix}
    (H_1)_{\delta-1}\\&(H_2)_{\delta-1}\\&&\ddots\\&&&(H_m)_{\delta-1}
\end{bmatrix}\] where $(H_i)_{\delta-1}$ means the first $\delta-1$ rows of $H_i$. We then add $\ell$ random rows to $H_X$ and $H_Z$ subject to the orthogonality conditions, as follows:
\begin{enumerate}
    \item for $i = 1, \ldots, \ell$:
        \begin{enumerate}
            \item let $v$ be a uniformly random vector sampled from $\text{row-span}(H_Z)^\perp\setminus \text{row-span}(H_X)$
    
            \item append $v$ to $H_X$ and rename the resulting matrix as $H_X$ i.e. $H_X = \begin{bmatrix}
                H_X\\v
            \end{bmatrix}$
        \end{enumerate}
    \item for $i = 1, \ldots, \ell$:
        \begin{enumerate}
            \item let $v$ be a uniformly random vector sampled from $\text{row-span}(H_X)^\perp\setminus \text{row-span}(H_Z)$

            \item append $v$ to $H_Z$ and rename the resulting matrix as $H_Z$  i.e. $H_Z = \begin{bmatrix}
            H_Z\\v
            \end{bmatrix}$
        \end{enumerate}
\end{enumerate}
We have sampled matrices $H_X, H_Z \in \mathbb{F}_q^{(m(\delta-1) + \ell) \times N}$ with orthogonal row spaces. Let $C_X = \ker H_X$, $C_Z = \ker H_Z$ are such that we obtain a well-defined CSS code $\mathcal{Q} = \mathrm{CSS}(C_X, C_Z)$.
\end{construction}

\begin{lemma}
\label{lem:dual_vandermonde_blocks}
For positive integers $r \geq \delta$, let \(n=r+\delta-1\), and let
\(\alpha_1,\ldots,\alpha_n\in\F_q\) be distinct. For $l=1,\ldots,n$ define
\[
\beta_l:=\prod_{\substack{j=1\\j\neq l}}^n
\frac{1}{\alpha_l-\alpha_j}.
\]
Let
\[
G=
\begin{bmatrix}
1&1&\cdots&1\\
\alpha_1&\alpha_2&\cdots&\alpha_n\\
\vdots&\vdots&\ddots&\vdots\\
\alpha_1^{\delta-2}&\alpha_2^{\delta-2}&\cdots&\alpha_n^{\delta-2}
\end{bmatrix}
\]
and
\[
H=
\begin{bmatrix}
1&1&\cdots&1\\
\alpha_1&\alpha_2&\cdots&\alpha_n\\
\vdots&\vdots&\ddots&\vdots\\
\alpha_1^{r-1}&\alpha_2^{r-1}&\cdots&\alpha_n^{r-1}
\end{bmatrix}
\begin{bmatrix}
\beta_1\\
&\beta_2\\
&&\ddots\\
&&&\beta_n
\end{bmatrix}.
\]
Then
\[
GH^T=0.
\]
Moreover, any \(\delta-1\) columns of \(G\) are linearly independent, and any
\(\delta-1\) columns of the first \(\delta-1\) rows of \(H\) are linearly
independent.
\end{lemma}

\begin{proof}
The column-independence statement for \(G\) follows from
\Cref{lem:vandermonde_local_mds}. The same statement for the first
\(\delta-1\) rows of \(H\) follows from \Cref{lem:vandermonde_local_mds} after
scaling each column by the nonzero scalar \(\beta_l\).

It remains to prove \(GH^T=0\). The \((a,b)\)-entry of \(GH^T\), where
\(0\le a\le \delta-2\) and \(0\le b\le r-1\), is
\[
\sum_{l=1}^n \beta_l \alpha_l^{a+b}.
\]
Since
\[
0\le a+b\le (\delta-2)+(r-1)=r+\delta-3=n-2,
\]
it is enough to prove that
\[
\sum_{l=1}^n \beta_l\alpha_l^m=0
\qquad
\text{for }0\le m\le n-2.
\]

This is the standard Lagrange interpolation identity. Indeed, let
\[
L_l(X)=
\prod_{\substack{j=1\\j\neq l}}^n
\frac{X-\alpha_j}{\alpha_l-\alpha_j}
\]
be the \(l\)-th Lagrange basis polynomial. The leading coefficient of
\(L_l(X)\) is \(\beta_l\). For \(0\le m\le n-2\), the polynomial \(X^m\) has
degree \(<n-1\), so its Lagrange interpolation expansion
\(
X^m=\sum_{l=1}^n \alpha_l^m L_l(X)
\)
has zero coefficient of \(X^{n-1}\) on the left-hand side. The coefficient of
\(X^{n-1}\) on the right-hand side is
\(
\sum_{l=1}^n \alpha_l^m\beta_l.
\)
Therefore,
\(
\sum_{l=1}^n \beta_l\alpha_l^m=0
\)
for \(0\le m\le n-2\), proving \(GH^T=0\).
\end{proof}

We now prove that \Cref{const:random_r_delta_qlrc} has the desired properties of locality and good distance.
\begin{lemma}
    The random QLRC $\mathcal{Q}$ from \Cref{const:random_r_delta_qlrc} is a QLRC of locality $r$ and dimension $k = N - 2(m(\delta -1) + \ell)$. Furthermore, each coordinate $s\in\{1,\ldots, N\}$ has recovery set $J_s = \{a(r+\delta -1)+1,\ldots, a(r+\delta -1) + r+\delta-1\}$ for $a = \ceil{s/(r+\delta-1)} -1$.
\end{lemma}
\begin{proof}
First consider the deterministic local rows before the \(\ell\) random rows are
added. The matrix \(H_X\) contains, on each local group
\(
J_a:=\{a(r+\delta-1)+1,\ldots,a(r+\delta-1)+r+\delta-1\},
\)
a copy of the \((\delta-1)\times(r+\delta-1)\) Vandermonde matrix \(G_a\).
Similarly, \(H_Z\) contains, on the same group, the first \(\delta-1\) rows of
the scaled Vandermonde matrix \(H_a\).

By \Cref{lem:vandermonde_local_mds}, any \(\delta-1\) columns of \(G_a\) are
linearly independent. Since the first \(\delta-1\) rows of \(H_a\) are obtained
from a Vandermonde matrix by nonzero column scalings, any \(\delta-1\) of their
columns are also linearly independent. Therefore, on every local group \(J_a\),
both the \(X\)- and \(Z\)-local parity checks can recover any erasure pattern of
size at most \(\delta-1\). Equivalently, the punctured local \(X\)- and
\(Z\)-classical codes have local distance at least \(\delta\).

The local \(X\)- and \(Z\)-row spaces are orthogonal by
\Cref{lem:dual_vandermonde_blocks}. Since distinct local groups have disjoint
supports, the deterministic local parts of \(H_X\) and \(H_Z\) are globally
orthogonal. The additional random rows are sampled from the orthogonal
complement of the current opposite row span, so after every appended row the row
spans of \(H_X\) and \(H_Z\) remain orthogonal. Hence
\(
C_X^\perp=\operatorname{rowspan}(H_X)\subseteq C_Z
\)
and the CSS code \(\mathcal Q=\CSS(C_X,C_Z)\) is well-defined.

The local recovery sets are precisely the groups \(J_a\). Thus
\(\mathcal Q\) is an \((r,\delta)\)-QLRC. Finally, each of \(H_X\) and \(H_Z\) has \(m(\delta-1)+\ell\) rows. The
Vandermonde local rows are linearly independent across disjoint supports, and
the random rows are sampled outside the current row span, so the rank of each
matrix is \(m(\delta-1)+\ell\). Therefore
\(
\dim C_X=\dim C_Z=N-(m(\delta-1)+\ell),
\)
and the CSS dimension is
\[
k
=
\dim C_X+\dim C_Z-N
=
N-2(m(\delta-1)+\ell).
\]
The condition \(\ell\in[N/2-m(\delta-1)]\) ensures this dimension is
positive.
\end{proof}

Now we bound the distance of the random $(r,\delta)$-QLRCs that we constructed.

\begin{proposition}\label{prop:distance_random_r_delta_qlrc}
    For sufficiently large \(N\), given any $\rho, \epsilon > 0$, the distance $d(\mathcal{Q})$ of a random $(r,\delta)$-QLRC $\mathcal{Q}$ from \Cref{const:random_r_delta_qlrc} with parameters, $N, r, \delta,$ and $N/2 - m(\delta-1)>\ell \geq (H_q^{(\delta)}(\rho) + 2\epsilon)N$ (if one exists) over the alphabet $\mathbb{F}_q$ satisfies \begin{equation}
        \Pr[d(\mathcal{Q}) \geq \rho N] > 1 - 2q^{-\epsilon N},
    \end{equation} where for $y \in \mathbb{F}_q^N \setminus \{0\}$ with weight $w = |y|$, 
    \begin{equation}
        H_q^{(\delta)}(\rho) = \limsup_{N\to \infty} \frac{1}{N}\log_q\left(\sum_{w\leq\rho N}\mathcal{N}^{(\delta)}(N,w)(q-1)^w\right) 
    \end{equation} and \begin{equation}\label{eqn:supports_of_weight_w} \mathcal{N}^{(\delta)}(N,w) = \sum_{s=1}^{\lfloor w/\delta \rfloor} \binom{m}{s}\sum_{\substack{w_1+\cdots +w_s = w\\ w_i\geq \delta}}\prod_{i=1}^s\binom{r+\delta-1}{w_i}\end{equation} and the \(\limsup\) is taken over \(N\) divisible by \(r+\delta-1\).
\end{proposition}
\begin{proof}
    %This proof follows the proof of Proposition 40 from \cite{golowich2025quantum}. The main change is in the union bound.
    
    Fix a nonzero vector $y \in \mathbb{F}_q^N$. Let $H_Z^{(i)}$ denote the matrix after the $i$th iteration of step 2 in \Cref{const:random_r_delta_qlrc}. If $y \in C_Z\setminus C_X^\perp$ then for $0 \leq i \leq \ell -1$, $y \in \ker H_Z^{(i)}$ and $y \in \ker H_Z^{(i+1)}$. Since $y \not \in C_X^\perp$, exactly $1/q$-fraction of the vectors in $C_X$ are orthogonal to $y$. It follows if $y \in \ker H_Z^{(i)}$ then less than $1/q$-fraction of the vectors in $C_X\setminus \text{row-span}(H_Z^{(i)})$ are orthogonal to $y$. Formally, \[\Pr[y \in \ker H_Z^{(i+1)}\setminus C_X^\perp \mid y \in \ker H_Z^{(i)}\setminus C_X^\perp] < \frac{1}{q}.\] Hence, \begin{align*}
        \Pr[y \in C_Z\setminus C_X^\perp] &= \Pr[y \not \in C_X^\perp] \cdot \Pr[y \in \ker H_Z^{(0)}\mid y \not \in C_X^\perp]\\ &\qquad\cdot \prod_{i=0}^{\ell-1}\Pr[y \in \ker H_Z^{(i+1)}\setminus C_X^\perp\mid y \in \ker H_Z^{(i)}\setminus C_X^\perp]\\ &< q^{-\ell}.
    \end{align*}
    We now explain why the union bound only needs to count supports appearing in
    \(\mathcal{N}^{(\delta)}(N,w)\). Let
    \(
    y\in C_Z\setminus C_X^\perp
    \)
    be a nonzero vector. Since \(y\in C_Z=\ker H_Z\), it satisfies all local
    \(Z\)-checks. Consider any local group \(J_a\). If
    \(
    0<|\operatorname{supp}(y)\cap J_a|<\delta,
    \)
    then the restriction of the local \(Z\)-check matrix to these nonzero
    coordinates has full column rank by \Cref{lem:vandermonde_local_mds}. Hence
    the only vector supported on those coordinates and satisfying the local checks
    is the zero vector, a contradiction. Therefore every nonempty local block of
    \(\operatorname{supp}(y)\) has size at least \(\delta\). Thus the possible
    supports of \(y\) of weight \(w\) are counted by \(\mathcal{N}^{(\delta)}(N,w)\).
    The same argument applies to vectors in \(C_X\setminus C_Z^\perp\).
    
    Applying the union bound over all $y \in \mathbb{F}_q^N$ such that $|y| \leq \rho N$ gives \begin{align*}
        \Pr[\exists\, y \in C_Z\setminus C_X^\perp, |y| \leq \rho N] &< q^{-\ell}\sum_{w\leq\rho N}\mathcal{N}^{(\delta)}(N,w)(q-1)^w\\
        & \leq q^{-\ell}q^{(H_q^{(\delta)}(\rho)+\epsilon)N}\\
        & \leq q^{-\epsilon N}
    \end{align*}
    for sufficiently large $N$
    and by symmetry \[\Pr[\exists\, y \in C_X\setminus C_Z^\perp, |y| \leq \rho N]<q^{-\epsilon N}.\] Therefore, \[\Pr[d(\mathcal{Q}) \leq \rho N] < 2q^{-\epsilon N},\] as desired.
\end{proof}
\begin{remark}
    The proof of Proposition 40 from \cite{golowich2025quantum} works for \Cref{prop:distance_random_r_delta_qlrc}, but gives a weaker estimate because it overcounts the number of low-weight codewords. It would require  $\ell$ more random rows to drive the union bound to zero so existence/distance threshold and rate-vs-distance tradeoff will be worse.
\end{remark}

\section{\texorpdfstring{Random $((r_1,\delta_1),\ldots,(r_h,\delta_h))$-QHLRCs}{Random ((r1,delta1),\ldots,(rh,deltah))-QHLRCs}}
\label{sec:random_r_delta_hierarchical_qlrc}
Now we move on to constructing random $h$-level quantum hierarchical LRCs. The main difference from the one-level construction is that the deterministic local \(X\) and \(Z\) checks must be chosen recursively so that nested local checks remain mutually orthogonal.

\subsection{Constructing a random QHLRC}
\begin{definition}
\label{def:strong_local_mds}
Let \(J\) be a finite set of coordinates, let \(U\subseteq \F_q^J\) be a subspace, and let \(a,s\) be positive integers. For \(E \subseteq J\), define \[\F_q^J[E] = \{x \in \F_q^J \mid \supp(x) \subseteq E\}.\] We say that \(U\) has the
\((a,s)\)-correctable property on \(J\) if:
\begin{enumerate}
\item for every \(E\subseteq J\) with \(|E|\le a\),
\(
\dim(U|_E)=|E|;
\)
\item for every \(E\subseteq J\) with \(|E|\le s\),
\(
U\cap \F_q^J[E]=\{0\}.
\)
\end{enumerate}
\end{definition}

Equivalently, the second condition says that the restriction map \[\pi_{J\setminus E}: U \to \F_q^{J\setminus E}\] is injective for every \(E\subseteq J\) with \(|E| \leq s\). It is also equivalent to \(\dim(U^\perp|_E) = |E|\) for \(|E| \leq s\).

\begin{construction}[$h$-level QHLRC]\label{const:random_qhlrc}
    Suppose we are given a block length $N$, locality and distance parameters $(r_1,\delta_1), (r_2,\delta_2),\ldots, (r_h,\delta_h)$ such that \(r_1\geq\cdots\geq r_h\quad\text{and}\quad \delta_1\geq\cdots\geq \delta_h\geq 2,\) and a finite field $\mathbb{F}_q$, with $q$ sufficiently large. Let $n_l  := r_l + \delta_l -1$ and suppose $\delta_1 \leq n_h/2$. 
    Further assume that $m_l := N/n_l$ is an integer for $l = 1,\ldots, h$ and $n_h\mid n_{h-1}\mid\cdots\mid n_1\mid N$.

    We will set some notation for the construction. For a level-\(l\) block  where \( b=0,\ldots,m_l-1\), write    \(J_{l,b}:=\{bn_l+1,bn_l+2,\ldots,(b+1)n_l\}\). Thus level-\((l+1)\) blocks refine level-\(l\) blocks. We first construct deterministic local row spaces
    \(U^{\mathrm{loc}},V^{\mathrm{loc}}\subseteq \F_q^N\)
    which will become the local parts of \(H_X\) and \(H_Z\).

    \medskip
    \noindent\textbf{Bottom level.}
    For each level-\(h\) block \(J_{h,b}\), choose \(n_h\) distinct field elements
    \(
    \alpha_{b,1},\ldots,\alpha_{b,n_h}\in\F_q.
    \)
    Let \(U_{h,b}\subseteq \F_q^{J_{h,b}}\) be the row space of the
    \((\delta_h-1)\times n_h\) Vandermonde matrix
    \[
    \begin{bmatrix}
    1&1&\cdots&1\\
    \alpha_{b,1}&\alpha_{b,2}&\cdots&\alpha_{b,n_h}\\
    \vdots&\vdots&\ddots&\vdots\\
    \alpha_{b,1}^{\delta_h-2}&
    \alpha_{b,2}^{\delta_h-2}&
    \cdots&
    \alpha_{b,n_h}^{\delta_h-2}
    \end{bmatrix}.
    \]
    Let
    \[
    \eta_{b,t}:=
    \prod_{\substack{u=1\\u\neq t}}^{n_h}
    \frac{1}{\alpha_{b,t}-\alpha_{b,u}}.
    \]
    Let \(V_{h,b}\subseteq \F_q^{J_{h,b}}\) be the row space of the scaled
    Vandermonde matrix
    \[
    \begin{bmatrix}
    1&1&\cdots&1\\
    \alpha_{b,1}&\alpha_{b,2}&\cdots&\alpha_{b,n_h}\\
    \vdots&\vdots&\ddots&\vdots\\
    \alpha_{b,1}^{\delta_h-2}&
    \alpha_{b,2}^{\delta_h-2}&
    \cdots&
    \alpha_{b,n_h}^{\delta_h-2}
    \end{bmatrix}
    \begin{bmatrix}
    \eta_{b,1}\\
    &\eta_{b,2}\\
    &&\ddots\\
    &&&\eta_{b,n_h}
    \end{bmatrix}.
    \]
    By \Cref{lem:dual_vandermonde_blocks}, the Vandermonde space generated by
    degrees \(0,\ldots,\delta_h-2\) is orthogonal to the full scaled dual
    Vandermonde space generated by degrees \(0,\ldots,r_h-1\). Since \(V_{h,b}\)
    is contained in this scaled dual space, we have \(U_{h,b}\perp V_{h,b}.\)

    Moreover, both \(U_{h,b}\) and \(V_{h,b}\) have the
    \((\delta_h-1,\delta_1-1)\)-correctable property on \(J_{h,b}\). Indeed,
    the Vandermonde property gives
    \[
    \dim(U_{h,b}|_E)=\dim(V_{h,b}|_E)=|E|
    \qquad
    (|E|\le \delta_h-1).
    \]
    Also, each of \(U_{h,b}\) and \(V_{h,b}\) is a scaled generalized
    Reed--Solomon row space of dimension \(\delta_h-1\) and length \(n_h\), hence
    has minimum distance \(n_h-(\delta_h-1)+1 = r_h + 1\) and 
    \[
    r_h + 1 > \delta_1 - 1 \iff n_h > \delta_1 + \delta_h - 3.
    \]
    Since \(\delta_h\le\delta_1\le n_h/2\), we have \(n_h \geq 2 \delta_1 > \delta_1 + \delta_h - 3.\) Therefore, neither row space has a nonzero word supported
    on at most \(\delta_1-1\) coordinates.

    \medskip
    \noindent\textbf{Higher levels.}
    Suppose that for all levels \(u>l\) we have already constructed local
    \(X\)- and \(Z\)-row spaces supported on level-\(u\) blocks, and that within
    each level-\((l+1)\) block the accumulated \(X\)- and \(Z\)-spaces are
    orthogonal.
    
    Fix a level-\(l\) block \(J_{l,b}\). Let \(U_{>l,b}\)
    be the sum of all previously constructed \(X\)-row spaces supported inside \(J_{l,b}\), and let \(V_{>l,b}\) be the analogous sum of previously constructed \(Z\)-row spaces. By induction,
    \(U_{>l,b}\perp V_{>l,b}.\)   
    
    Set \(\Delta_l:=\delta_l-\delta_{l+1}. \)
    Choose a \(\Delta_l\)-dimensional subspace \(U_{l,b}^{+}\subseteq V_{>l,b}^{\perp}\) such that
    \[
    \dim(U_{>l,b}+U_{l,b}^{+})
    =
    \dim U_{>l,b}+\Delta_l,
    \]
    and such that \(U_{>l,b}+U_{l,b}^{+}\) has the \((\delta_l-1,\delta_1-1)\)-correctable property on \(J_{l,b}\). Then choose a \(\Delta_l\)-dimensional subspace \( V_{l,b}^{+}\subseteq (U_{>l,b}+U_{l,b}^{+})^\perp \)
    such that
    \[
    \dim(V_{>l,b}+V_{l,b}^{+})
    =
    \dim V_{>l,b}+\Delta_l,
    \]
    and such that \(V_{>l,b}+V_{l,b}^{+}\) has the \((\delta_l-1,\delta_1-1)\)-correctable property on \(J_{l,b}\).
    
    The existence of such \(U_{l,b}^{+}, V_{l,b}^{+}\) follows from
    \Cref{thm:generic_nested_local_extension} which is proved at the end of the section.
    
    % we
    % can sample the required subspaces from the indicated orthogonal complements and
    % reject if one of the finitely many local rank conditions fails. For sufficiently
    % large \(q\), the success probability is positive. Alternatively, 
    We perform an exhaustive
    search over the relevant finite Grassmannian (\Cref{def:grassmannian}) which gives a deterministic finite procedure. Enumerate all \(U_{l,b}^+\in \Gr_{\F_q}(\Delta_l,V_{>l,b}^{\perp})\) and choose one such that \(U_{l,b}^+\cap U_{>l,b}=\{0\}\) and \(U_{l,b}^++ U_{>l,b}\) has the \((\delta_l-1,\delta_1-1)\)-correctable property. Then enumerate all \(V_{l,b}^+\in \Gr_{\F_q}(\Delta_l,(U_{>l,b}+U_{l,b}^+)^\perp)\) and choose one such that \(V_{l,b}^+\cap V_{>l,b}=\{0\}\) and \(V_{l,b}^++ V_{>l,b}\) has the \((\delta_l-1,\delta_1-1)\)-correctable property.
    
    After performing this for every level-\(l\) block and every
    \(l=h-1,\ldots,1\), let \(U^{\mathrm{loc}}\) be the span of all constructed
    \(X\)-local spaces and let \(V^{\mathrm{loc}}\) be the span of all constructed
    \(Z\)-local spaces. These spaces are all constructed deterministically.
    
    By construction, \(U^{\mathrm{loc}}\perp V^{\mathrm{loc}}.\)
    Let
    \[
    M
    :=
    m_h(\delta_h-1)
    +
    \sum_{l=1}^{h-1}m_l(\delta_l-\delta_{l+1}),
    \]
    then
    \(
    \dim U^{\mathrm{loc}}=\dim V^{\mathrm{loc}}=M.
    \)
    Choose bases of \(U^{\mathrm{loc}}\) and \(V^{\mathrm{loc}}\), and use them as
    the initial rows of \(H_X\) and \(H_Z\), respectively. 
\medskip

 Finally, choose an integer \(\ell < N/2 - M.\) We then add $\ell$ random rows to $H_X$ and $H_Z$ subject to the orthogonality conditions, as follows:
\begin{enumerate}
    \item for $i = 1, \ldots, \ell$:
        \begin{enumerate}
            \item let $v$ be a uniformly random vector sampled from $\text{row-span}(H_Z)^\perp\setminus \text{row-span}(H_X)$
    
            \item append $v$ to $H_X$ and rename the resulting matrix as $H_X$ i.e. $H_X = \begin{bmatrix}
                H_X\\v
            \end{bmatrix}$
        \end{enumerate}
    \item for $i = 1, \ldots, \ell$:
        \begin{enumerate}
            \item let $v$ be a uniformly random vector sampled from $\text{row-span}(H_X)^\perp\setminus \text{row-span}(H_Z)$
    
            \item append $v$ to $H_Z$ and rename the resulting matrix as $H_Z$ i.e. $H_Z = \begin{bmatrix}
                H_Z\\v
            \end{bmatrix}$
        \end{enumerate}
\end{enumerate}
We have sampled matrices $H_X, H_Z \in \mathbb{F}_q^{(M + \ell)\times N}$ with orthogonal row spaces. Let $C_X = \ker H_X$, $C_Z = \ker H_Z$ and obtain a well-defined CSS code $\mathcal{Q} = \mathrm{CSS}(C_X, C_Z)$.
\end{construction}

\subsection{Dimension, locality, and distance}

We will begin with computing the dimension and proving the locality properties of the quantum HLRC $\mathcal{Q}$.
\begin{lemma}
    The random QHLRC $\mathcal{Q}$ from \Cref{const:random_qhlrc} is a quantum HLRC of hierarchical locality $(r_1,\delta_1),\ldots, (r_h,\delta_h)$ and dimension $k = N - 2(M+\ell)$ where \begin{equation}M = m_h(\delta_h-1) + \sum_{l = 1}^{h-1} m_l(\delta_l-\delta_{l+1}).\end{equation}
\end{lemma}\begin{proof}
    Since $r_1,\ldots, r_h$ and $\delta_1,\ldots, \delta_h$ are decreasing sequences, $m_1,\ldots, m_h$ is an increasing sequence. Therefore, \begin{equation}
    m_h(\delta_h-1) + \sum_{l=1}^{h-1} m_l(\delta_l-\delta_{l+1}) \leq m_h(\delta_1 - 1)
    \end{equation}
    and \begin{equation}\label{eqn:rhlrc_distance_restriction}
    m_h(\delta_1 -1) < N/2 \iff \delta_1-1< n_h/2.
    \end{equation} The latter inequality is enforced in \Cref{const:random_qhlrc}. By construction, the deterministic local \(X\)-row space has dimension
    \[
    M
    =
    m_h(\delta_h-1)
    +
    \sum_{l=1}^{h-1}m_l(\delta_l-\delta_{l+1}),
    \]
    and the same is true for the deterministic local \(Z\)-row space. Each random
    global row is chosen outside the current row span, so it increases rank by one.
    Hence, \(\rank H_X=\rank H_Z=M+\ell.\)
    Therefore, \(\dim C_X=\dim C_Z=N-(M+\ell),\)
    and the CSS dimension is
    \[
    k
    =
    \dim C_X+\dim C_Z-N
    =
    N-2(M+\ell).
    \]
    
    It remains to prove hierarchical locality. Fix a level \(l\) block \(J_{l,b}\).
    By construction, the \(X\)-local row space supported in \(J_{l,b}\) has the
    \(\delta_l-1\) local MDS property:
    \[
    \dim (U|_E)=|E|
    \qquad
    \text{for every }E\subseteq J_{l,b},\ |E|\le \delta_l-1.
    \]
    Thus any \(\delta_l-1\) erasures in \(J_{l,b}\) can be recovered using
    \(X\)-type local checks supported inside \(J_{l,b}\). The same statement holds
    for the \(Z\)-local row space.
    
    Since \(n_h\mid n_{h-1}\mid\cdots\mid n_1,\)
    each level-\((l+1)\) block is contained in a level-\(l\) block. Therefore the
    punctured code on a level-\(l\) block inherits the lower-level repair structure
    from the blocks contained inside it. Hence the code is an \(h\)-level QHLRC with
    locality parameters \(((r_1,\delta_1),\ldots,(r_h,\delta_h)).\)
    \end{proof}
We now count the possible supports of logical operators compatible with the
hierarchical local checks. Define
\[
B_h(z)
=
\sum_{t=\delta_h}^{n_h}\binom{n_h}{t}z^t.
\]
For \(l=h-1,\ldots,1\), define recursively
\[
B_l(z)
=
\sum_{t=\delta_l}^{n_l}
[z^t]\left(1+B_{l+1}(z)\right)^{n_l/n_{l+1}}z^t.
\]
Finally, define
\begin{equation}\label{eqn:supports_of_weight_w_hierarchical} \mathcal{N}^{(\boldsymbol{\delta})}(N,w) = [z^w](1 + B_1(z))^{N/n_1}.\end{equation}
When \(h=1\), \labelcref{eqn:supports_of_weight_w_hierarchical} recovers \labelcref{eqn:supports_of_weight_w}:
\[
\mathcal N^{(\delta_1)}(N,w)
=
[z^w]\left(
1+
\sum_{t=\delta_1}^{n_1}\binom{n_1}{t}z^t
\right)^{N/n_1} = \sum_{s=1}^{\lfloor w/\delta_1 \rfloor} \binom{m_1}{s}\sum_{\substack{w_1+\cdots +w_s = w\\ w_i\geq \delta_1}}\prod_{i=1}^s\binom{n_1}{w_i}.
\]

\begin{lemma}
\label{lem:hierarchical_admissible_supports}
Let \(y\in C_Z\setminus C_X^\perp\) or \(y\in C_X\setminus C_Z^\perp.\) If \(y\neq 0\), then for every level-\(l\) block \(J_{l,b}\), either
\(
\operatorname{supp}(y)\cap J_{l,b}=\emptyset,
\)
or
\(
|\operatorname{supp}(y)\cap J_{l,b}|\ge \delta_l.
\)
Consequently, the number of possible supports of such vectors of weight \(w\)
is at most \(\mathcal N^{(\boldsymbol\delta)}(N,w).\)
\end{lemma}

\begin{proof}
We prove the statement for \(y\in C_Z\setminus C_X^\perp\); the other case is
identical. Since \(y\in C_Z=\ker H_Z,\)
the vector \(y\) satisfies all deterministic local \(Z\)-checks.

Fix a level-\(l\) block \(J_{l,b}\). Let $E = \operatorname{supp}(y)\cap J_{l,b}$ and suppose
\(0<|E|\le \delta_l-1.\) The local \(Z\)-row space supported in \(J_{l,b}\) has the
\((\delta_l-1,\delta_1-1)\)-correctable property, so
\(
\dim(V|_E)=|E|.
\)
Thus, the only vector supported on \(E\) and satisfying all these local checks is
the zero vector. This contradicts the definition of \(E\). Therefore every
nonempty intersection with \(J_{l,b}\) has size at least \(\delta_l\).

The recursive generating function counts exactly such supports. At the bottom, a nonempty level-\(h\) block must have weight at least \(\delta_h\), so
its generating function is
\[
B_h(z)=\sum_{t=\delta_h}^{n_h}\binom{n_h}{t}z^t.
\]
A level-\(l\) block consists of \(n_l/n_{l+1}\) level-\((l+1)\) blocks. Each
child block is either empty or has a support counted by \(B_{l+1}(z)\). Hence
the generating function before imposing the level-\(l\) threshold is
\[
(1+B_{l+1}(z))^{n_l/n_{l+1}}.
\]
However, we only want the nonempty configurations which also meet the threshold of $\delta_l$. For $l=1,\ldots,h-1$ define \[B_l(z) = \sum_{t = \delta_l}^{n_l}[z^t](1 + B_{l+1}(z))^{\frac{n_l}{n_{l+1}}} \cdot z^t\] where $[z^t] f(z)$ is the coefficient of $z^t$ in $f(z)$. This recursively defines $B_{h-1}, \ldots, B_1$. At
the top level, the whole length \(N\) consists of \(N/n_1\) level-\(1\) blocks,
each either empty or counted by \(B_1(z)\). Therefore the number of admissible
supports of weight \(w\) is
\[
[z^w](1+B_1(z))^{N/n_1}
=
\mathcal N^{(\boldsymbol\delta)}(N,w),
\] as desired.
\end{proof}

We use the number of admissible supports to define an entropy-like function and give a probabilistic argument for the distance bound.
\begin{proposition}\label{prop:distance_random_qhlrc}
    For all sufficiently large \(N\), given any $\rho, \epsilon > 0$, the distance $d(\mathcal{Q})$ of a random $h$-level QHLRC $\mathcal{Q}$ from \Cref{const:random_qhlrc} with parameters $N$, $(r_l,\delta_l)_{l=1,\ldots,h}$, and $N/2 - M > \ell \geq (H_q^{(\boldsymbol{\delta})}(\rho) + 2\epsilon)N$ (if one exists) over the alphabet $\mathbb{F}_q$ satisfies \begin{equation}
        \Pr[d(\mathcal{Q}) \geq \rho N] > 1 - 2q^{-\epsilon N}.
    \end{equation} where \begin{equation}
        H_q^{(\boldsymbol{\delta})}(\rho) = \limsup_{N\to \infty} \frac{1}{N}\log_q\left(\sum_{w\leq\rho N}\mathcal{N}^{(\boldsymbol{\delta})}(N,w)(q-1)^w\right).
    \end{equation} and the \(\limsup\) is taken over \(N\) divisible by \(n_1\).
\end{proposition}
\begin{proof}
    Fix a nonzero vector $y \in \mathbb{F}_q^N$. Let $H_Z^{(i)}$ denote the matrix after the $i$th iteration of step 2 in \Cref{const:random_qhlrc}. If $y \in C_Z\setminus C_X^\perp$ then for $0 \leq i \leq \ell -1$, $y \in \ker H_Z^{(i)}$ and $y \in \ker H_Z^{(i+1)}$. Since $y \not \in C_X^\perp$, exactly $1/q$-fraction of the vectors in $C_X$ are orthogonal to $y$. It follows if $y \in \ker H_Z^{(i)}$ then less than $1/q$-fraction of the vectors in $C_X\setminus \text{row-span}(H_Z^{(i)})$ are orthogonal to $y$. Formally, \[\Pr[y \in \ker H_Z^{(i+1)}\setminus C_X^\perp \mid y \in \ker H_Z^{(i)}\setminus C_X^\perp] < \frac{1}{q}.\] Hence, \begin{align*}
        \Pr[y \in C_Z\setminus C_X^\perp] &= \Pr[y \not \in C_X^\perp] \cdot \Pr[y \in \ker H_Z^{(0)}\mid y \not \in C_X^\perp]\\ &\qquad\cdot \prod_{i=0}^{\ell-1}\Pr[y \in \ker H_Z^{(i+1)}\setminus C_X^\perp\mid y \in \ker H_Z^{(i)}\setminus C_X^\perp]\\ &< q^{-\ell}.
    \end{align*}
    Now union bound over all possible supports of weight \(w\le \rho N\). By
    \Cref{lem:hierarchical_admissible_supports}, any logical vector of weight \(w\)
    has one of at most \(\mathcal N^{(\boldsymbol\delta)}(N,w)\) admissible supports. For each support of size \(w\), there are at most \((q-1)^w\)
    nonzero vectors supported on it. Therefore,
    applying the union bound over all $y \in \mathbb{F}_q^N$ with weight, $w$ at most $\rho N$ gives \begin{align*}
        \Pr[\exists\, y \in C_Z\setminus C_X^\perp, |y| \leq \rho N] &< q^{-\ell}\sum_{w\leq\rho N}\mathcal{N}^{(\boldsymbol\delta)}(N,w)(q-1)^w\\
        & \leq q^{-\ell}q^{(H_q^{(\boldsymbol\delta)}(\rho) + \epsilon)N}\\
        & \leq q^{-\epsilon N}
    \end{align*} for sufficiently large $N$
    and by symmetry \[\Pr[\exists\, y \in C_X\setminus C_Z^\perp, |y| \leq \rho N]<q^{-\epsilon N}.\] Therefore, \[\Pr[d(\mathcal{Q}) \leq \rho N] < 2q^{-\epsilon N},\] as desired.
\end{proof}

\begin{remark}\label{rem:strict_hierarchical}
    In the hierarchical setting, for the most general setting, we are given locality and distance parameters $(r_1,\delta_1), \ldots, (r_h, \delta_h)$ such that $r_1 \geq \cdots \geq r_h \geq 1$ and $\delta_1 \geq \cdots \geq \delta_h \geq 2$. We imposed the restriction $n_h \mid n_{h-1} \mid \cdots \mid n_1 \mid N$ so the groups nest. We want each level to meaningfully correct additional erasures so if there is some $1 \leq l \leq h-1$ such that $\delta_l = \delta_{l+1}$ then level $l$ does not correct any more erasures than level $l+1$. In the parity check view, we would not be adding additional parity checks corresponding to level $l$ because $\delta_l - \delta_{l+1} = 0$. Consequently, it is reasonable to assume $\delta_1 > \cdots > \delta_h\geq 2$. Suppose there is some $1 \leq l \leq h-1$ such that $r_l = r_{l+1}$. We have \[r_l + \delta_l - 1 = r_{l+1} + \delta_{l+1} - 1 + (\delta_l - \delta_{l+1})\] and by the divisibility condition $r_{l+1} + \delta_{l+1} - 1\mid r_l + \delta_l -1$ so $r_{l+1} + \delta_{l+1} - 1 \mid \delta_l - \delta_{l+1}$ which is a contradiction. Therefore, we may also assume $r_1 > \cdots > r_h$.
\end{remark}

\subsection{Existence of subspaces}
We now prove the existence of the desired \(U_{l,b}^+, V_{l,b}^+\) used in \Cref{const:random_qhlrc}. The proof is quite technical so we first present two inequalities and use them to justify the existence of the desired subspaces.
\begin{theorem}
\label{thm:generic_nested_local_extension}
Let \(J\) be a level-\(l\) block of length \(n_l\), and set \(\Delta_l:=\delta_l-\delta_{l+1}.\) Suppose that inside \(J\) we have already constructed lower-level row spaces \(U_{>l},V_{>l}\subseteq \F_q^J\) such that \[U_{>l}\perp V_{>l} \quad \text{and} \quad \dim U_{>l} = \dim V_{>l}.\] Further suppose both \(U_{>l}\) and \(V_{>l}\) have the \((\delta_{l+1}-1, \delta_1-1)\)-correctable property on every level-\((l+1)\) child block inside $J$. More specifically, suppose that \(J\) is partitioned into level-\((l+1)\) child blocks
\[
J=\bigsqcup_i J_i,
\]
and that
\[
U_{>l}=\bigoplus_i U_i,
\qquad
V_{>l}=\bigoplus_i V_i,
\]
where \(U_i,V_i\subseteq \F_q^{J_i}\), extended by zero outside \(J_i\), satisfy
\(
U_i\perp V_i
\)
and both \(U_i\) and \(V_i\) have the
\((\delta_{l+1}-1,\delta_1-1)\)-correctable property on \(J_i\). Assume \(q\) is sufficiently large. Then there exists a
\(\Delta_l\)-dimensional subspace \(U_l^+\subseteq V_{>l}^{\perp}\) such that 
\[
\dim(U_{>l}+U_l^+)=\dim U_{>l}+\Delta_l,
\]
and \(U_{>l-1} := U_{>l}+U_l^+\) has the \((\delta_l-1,\delta_1-1)\)-correctable property on $J$.

After choosing such a \(U_l^+\), there exists a \(\Delta_l\)-dimensional subspace \(V_l^+\subseteq (U_{>l}+U_l^+)^\perp\) such that \[\dim(V_{>l}+V_l^+)=\dim V_{>l}+\Delta_l,\] and \(V_{>l-1}:=V_{>l}+V_l^+\) has the \((\delta_l-1,\delta_1-1)\)-correctable property on \(J\). Consequently,
\[
U_{>l-1}\perp V_{>l-1} \quad \text{and} \quad \dim U_{>l-1} = \dim V_{>l-1}.
\]
\end{theorem}

\begin{proof}
We first record two consequences of the inductive hypotheses.

\noindent \textbf{Rank-Deficiency Bound:} Let \(E \subseteq J\) such that $|E| \leq \delta_l-1$. Decompose \(E\) among the level-\((l+1)\) child blocks: \[ E=\bigsqcup_i E_i. \]
By the level-\((l+1)\) correctable property,
\[
\dim(U_{>l}|_{E_i})
\ge
\min\{|E_i|,\delta_{l+1}-1\}
\]
which implies 
\[
\dim(U_{>l}|_E)
\ge
\sum_i \min\{|E_i|,\delta_{l+1}-1\}.
\]
Hence,
\[
|E|-\dim(U_{>l}|_E)
\le
\sum_i (|E_i|-\delta_{l+1}+1)_+.
\] where $y_+ = \max\{y,0\}$.
If the right-hand side is nonzero, then
\[
\sum_i (|E_i|-\delta_{l+1}+1)_+
\le
|E|-\delta_{l+1}+1
\le
\delta_l -1 - \delta_{l+1}+1
=
\Delta_l.
\]
If it is zero, the same inequality is immediate. Thus,
\(
|E|-\dim(U_{>l}|_E)\le \Delta_l.
\)
The same argument gives
\(
|E|-\dim(V_{>l}|_E)\le \Delta_l.
\)

\medskip
\noindent \textbf{Dimension Margin Inequality:} By hypothesis, \(D := \dim U_{>l} = \dim V_{> l}.\) Inside a level-\(l\) block, the accumulated lower-level dimension is 
\[
D = \frac{n_l}{n_h}(\delta_h-1) + \sum_{u=l+1}^{h-1} \frac{n_l}{n_u}(\delta_u - \delta_{u+1})
\] and dividing by $n_l$, we get 
\[
\frac{D}{n_l} = \frac{\delta_h-1}{n_h} + \sum_{u=l+1}^{h-1} \frac{\delta_u - \delta_{u+1}}{n_u} \leq \frac{\delta_{l+1}-1}{n_h}.
\]
Thus,
\begin{equation}\label{eq:n_l-2D}
n_l - 2D  \geq n_l\left(1-\frac{2(\delta_{l+1}-1)}{n_h}\right) = \frac{n_l}{n_h}\left(n_h - 2(\delta_{l+1}-1)\right).
\end{equation}

Decompose \(J\) among the level-\((l+1)\) child blocks: \[ J=\bigsqcup_i J_i,\] then \(V_{>l}\) decomposes as a direct sum over the \(J_i\): \[V_{>l} = \bigoplus_i V_i\] where each $V_i \subseteq \F_q^{J_i}$. By induction hypothesis, each \(V_i\) has the \((\delta_{l+1}-1,\delta_1-1)\)-correctable property on \(J_i\). Hence, for every \(E_i\subseteq J_i\), with \(|E_i| \leq \delta_1-1\), we have \(V_i \cap \F_q^J[E_i] = \{0\}\). Now take any subset \(E \subseteq J\) with \(|E| \leq \delta_1-1\) and write \(E_i = E \cap J_i\), then \(|E_i| \leq \delta_1-1.\) Suppose \(v \in V_{>l} \cap \F_q^J[E]\). We can decompose \(v = \sum_i v_i\) for \(v_i \in V_i\). Since \(v\) is supported inside $E$, each component is supported inside $E_i$. Thus, \[v_i \in V_i \cap \F_q^J[E_i],\] but \(|E_i| \leq \delta_1-1\) so $v_i = 0$. Hence, \(v = 0\) and \(V_{>l} \cap \F_q^J[E] =\{0\}\),

Recall by the hypotheses of the construction, \(n_h\mid n_l\) and \(\delta_1-1 < n_h/2\). For every \(E \subseteq J\) with \(|E| \leq \delta_1-1\), let
\[
\pi_{J\setminus E}: \F_q^J \to \F_q^{J\setminus E}
\]
be the coordinate projection. Since \(V_{>l}\cap \F_q^J[E]=\{0\}\) for
\(|E|\le \delta_1-1\), we have
\begin{align*}
\dim \pi_{J\setminus E}(V_{>l}^{\perp})
&= \dim V_{>l}^\perp - \dim \ker (\pi_{J\setminus E}|_{V_{>l}^\perp})\\
&= n_l - D - \dim (V_{>l}^\perp \cap \F_q^J[E])\\
%&= n_l - D - \dim \ker(V_{>l}|_E)\\
&= n_l-D-|E|+\dim(V_{>l}|_E).
\end{align*} The last equality follows from recognizing that \(V_{>l}^\perp \cap \F_q^J[E] \cong (V_{>l}|_E)^\perp\). 

Also, since \(U_{>l}\cap \F_q^J[E]=\{0\}\), \[\dim \pi_{J\setminus E}(U_{>l}) = \dim U_{>l} - \dim \ker(\pi_{J\setminus E}|_{U_{>l}}) = D - \dim (U_{>l} \cap \F_q^J[E]) = D.\]
Therefore,
\[
\dim \pi_{J\setminus E}(V_{>l}^{\perp})
-
\dim \pi_{J\setminus E}(U_{>l})
=
n_l-2D-|E|+\dim(V_{>l}|_E).
\]
By the level-\(l+1\) correctable property,
\[
\dim(V_{>l}|_E)\ge \min\{|E|,\delta_{l+1}-1\}
\]
so
\[
|E|-\dim(V_{>l}|_E)\le \delta_1-1 - \delta_{l+1} + 1 = \delta_1 -\delta_{l+1}
\]
and 
\[
\dim \pi_{J\setminus E}(V_{>l}^{\perp})
-
\dim \pi_{J\setminus E}(U_{>l})
\geq
n_l-2D-\delta_1 + \delta_{l+1}.
\]
By \labelcref{eq:n_l-2D}, 
\[
n_l-2D\ge \frac{n_l}{n_h}\left(n_h - 2(\delta_{l+1}-1)\right) \geq n_h - 2(\delta_{l+1}-1)
\] and
\[
n_h - 2(\delta_{l+1}-1) = n_h - 2(\delta_1-1) + 2(\delta_1-\delta_{l+1}) \geq 2(\delta_1-\delta_{l+1}) \geq \delta_1 + \delta_l - 2\delta_{l+1}
\]
Consequently,
\[
\dim \pi_{J\setminus E}(V_{>l}^{\perp})
-
\dim \pi_{J\setminus E}(U_{>l})
\ge
\delta_l - \delta_{l+1}
=
\Delta_l.
\]
The same argument, with \(U_{>l}\) and \(V_{>l}\) interchanged, gives
\[
\dim \pi_{J\setminus E}(U_{>l}^{\perp})
-
\dim \pi_{J\setminus E}(V_{>l})
\ge
\Delta_l.
\]

\medskip
\noindent \textbf{Existence:} We prove the existence of \(U_l^+\); the reasoning for \(V_l^+\) is identical after \(U_l^+\) has been chosen. We choose \(U_l^+\) from the Grassmannian \(\Gr_{\F_q}(\Delta_l,V_{>l}^{\perp}).\) Equivalently, after fixing a basis of \(V_{>l}^{\perp}\), we may parameterize an ordered
\(\Delta_l\)-tuple of vectors in \(V_{>l}^{\perp}\) by a matrix of variables
\[
W\in \F_q^{\Delta_l\times \dim V_{>l}^{\perp}}.
\] We impose three types of conditions.

First, we require \(U_l^+\cap U_{>l}=\{0\}.\) This is equivalent to
\[
\dim(U_{>l}+U_l^+)=\dim U_{>l}+\Delta_l,
\]
and is the nonvanishing of some \((D + \Delta_l)\times(D+\Delta_l)\) minor in the
coordinates \(W\). Such a determinant polynomial is not identically zero because
\[
\dim(V_{>l}^{\perp}/U_{>l})=n_l-2D\ge \Delta_l.
\]
Hence there exist \(w_1,\ldots,w_{\Delta_l}\in V_{>l}^{\perp}\) whose images
are linearly independent modulo \(U_{>l}\). For this choice, the matrix obtained
by adjoining the \(w_i\)'s to a fixed basis of \(U_{>l}\) has rank
\(\dim U_{>l}+\Delta_l\). Therefore at least one full-rank minor of that matrix
is nonzero at this choice, and the corresponding minor polynomial is not
identically zero.

Second, for every \(E\subseteq J\) with \(|E|\le \delta_l-1\), we require
\[
\dim((U_{>l}+U_l^+)|_E)=|E|.
\]
The rank-deficiency bound proved above says that at most \(\Delta_l\) new
directions are needed on \(E\). Since \(V_{>l} \cap \F_q^J[E] = \{0\}\), the
restriction \((V_{>l}^{\perp})|_E\) has dimension \(|E|\). In other words, \((V_{>l}^{\perp})|_E = \F_q^E\). We can choose $\Delta_l$ vectors from \(V_{>l}^{\perp}\) whose restrictions can complete the dimension of
\(U_{>l}|_E\) to \(|E|\). Hence, for each fixed
\(E\), at least one \(|E|\times |E|\) minor of the restricted matrix is a
nonzero polynomial in the entries of \(W\).

Third, for every \(E\subseteq J\) with \(|E|\le \delta_1-1\), we require
\[
(U_{>l}+U_l^+)\cap \F_q^J[E]=\{0\}.
\]
Equivalently, restriction to \(J\setminus E\) is injective on
\(U_{>l}+U_l^+\), i.e.
\[
\dim((U_{>l}+U_l^+)|_{J\setminus E})
=
\dim(U_{>l}+U_l^+).
\]
The
dimension margin inequality above says that inside \(\pi_{J\setminus E}(V_{>l}^{\perp})\) there are at least \(\Delta_l\) dimensions available modulo \(\pi_{J\setminus E}(U_{>l}).\) Thus, the injectivity condition is also the nonvanishing of a suitable maximal minor.

Thus all required conditions are finitely many polynomial nonvanishing
conditions in the entries of \(W\). Let \(F(W)\) be the product of one nonzero
polynomial for each condition. Then \(F\) is a nonzero polynomial. By the
Schwartz--Zippel lemma,
\[
\Pr[F(W)=0]\le \frac{\deg F}{q}
\]
where the probability is over random choices of entries of \(W\). Therefore, for \(q>\deg F\), there exists a choice of entries of \(W\) such that
\(F(W)\neq 0\). The span of the \(\Delta_l\) chosen vectors in \(V_{>l}^\perp\) is the desired subspace \(U_l^+\).

Now set \(U_{>l-1}:=U_{>l}+U_l^+.\) By construction, \(U_{>l-1}\) has the
\((\delta_l-1,\delta_1-1)\)-correctable property and \(\dim U_{>l-1}=D+\Delta_l.\) We now choose \(V_l^+\) from \(\Gr_{\F_q}(\Delta_l,U_{>l-1}^{\perp}).\) The rank-deficiency bound for \(V_{>l}\) was already proved:
\[
|E|-\dim(V_{>l}|_E)\le \Delta_l
\qquad (|E|\le \delta_l-1).
\]
Also, since \(U_{>l-1}\) has no nonzero word supported on a set of size at most
\(\delta_1-1\), the restriction of \(U_{>l-1}^{\perp}\) to any such \(E\) has
full rank. It remains to check the dimension margin inequality. For \(|E|\le \delta_1-1\), we have 
\begin{align*}
    \dim \pi_{J\setminus E}(U_{>l-1}^{\perp}) &= \dim U_{>l-1}^{\perp} - \dim \ker(\pi_{J\setminus E}|_{U_{>l-1}^{\perp}})\\
    &= n_l - D - \Delta_l - \dim(U_{>l-1}^{\perp} \cap \F_q^J[E])\\
    %&= n_l - D- \Delta_l - \dim \ker(U_{>l-1}|_E)\\
    &= n_l - D - \Delta_l - |E| + \dim(U_{>l-1}|_E)
\end{align*} and since \(V_{>l} \cap \F_q^J[E] = \{0\}\)
\[
\dim \pi_{J\setminus E}(U_{>l-1}^{\perp})
-
\dim \pi_{J\setminus E}(V_{>l})
=
n_l-(D+\Delta_l)-D-|E|+\dim(U_{>l-1}|_E).
\]
Since \(U_{>l-1}\) has the \((\delta_l-1,\delta_1-1)\)-correctable property, for every \(E\subseteq J\) with \(|E| \leq \delta_1-1\), we have \[\dim (U_{>l-1}|_E) \geq \min \{|E|, \delta_l-1\}\] so
\[
|E|-\dim(U_{>l-1}|_E)\le \delta_1-1 - (\delta_l-1) = \delta_1 - \delta_l.
\]
Therefore,
\[
\dim \pi_{J\setminus E}(U_{>l-1}^{\perp})
-
\dim \pi_{J\setminus E}(V_{>l})
\ge
n_l-2D-\Delta_l-(\delta_1-\delta_l).
\]
Using
\[
n_l-2D\ge \delta_1+\delta_l-2\delta_{l+1},
\]
we get
\[
n_l-2D-\Delta_l-(\delta_1-\delta_l)
\ge
\delta_l-\delta_{l+1}
=
\Delta_l.
\]
Thus, the same Schwartz--Zippel argument produces a subspace \(V_l^+\subseteq U_{>l-1}^{\perp}\) such that \(V_l^+\cap V_{>l}=\{0\}\) and \(V_{>l-1}:=V_{>l}+V_l^+\) has the \((\delta_l-1,\delta_1-1)\)-correctable property.

Finally, \(U_l^+\subseteq V_{>l}^{\perp}\) and \(V_l^+\subseteq U_{>l-1}^\perp=(U_{>l}+U_l^+)^\perp\), so by construction,
\[
U_{>l-1}\perp V_{>l-1} \quad \text{and} \quad \dim U_{>l-1} = \dim V_{>l-1}. \] This proves the theorem.
\end{proof}

\section{\texorpdfstring{$(r,\delta)$ Quantum Tamo--Barg codes}{(r,delta) Quantum Tamo--Barg codes}}
\label{sec:r_delta_qTB}
Here we extend the quantum Tamo--Barg code definition in \cite{golowich2025quantum} to a $(r,\delta)$ quantum Tamo--Barg (QTB) code. Our parameters are adjusted to match the convention of classical $(r,\delta)$ Tamo--Barg code constructions presented in \cite{tamo.barg.lrc}.

\begin{definition}[$(r,\delta)$ quantum Tamo--Barg code]
\label{def:r.delta.qTB.qlrc}
Let $p$ be a prime number and $m$ be a positive integer, and $q = p^m$. Given an integer $\delta \geq 2$, a locality parameter $r\geq \delta $ such that $(r + \delta-1)|(q-1)$, and an integer $q/2\leq \ell \leq q-1$, the $(r,\delta)$ quantum Tamo--Barg code is defined to be the CSS code $\mathcal{Q} = \mathrm{CSS}(C,C)$ with $C = \ev(\F_q[X]^S)$, where \begin{align*}S &= \{i \in [\ell]\mid i \not\equiv -j \bmod(r+\delta-1) \text{ for any } j \in\{1,\ldots,\delta-1\}\}\\&\qquad \cup \{i \in [q-1]\mid i \equiv j \bmod(r+\delta-1) \text{ for some } j \in\{1,\ldots,\delta-1\}\}.\end{align*}
\end{definition}

In what follows, we will require this lemma \cite[Lemma 54]{golowich2025quantum} which we restate without proof. Consider classical codes $A_X,A_Z,B_X,B_Z$ in $\mathbb{F}_q^n$ such that $A_X^\perp \subseteq A_Z$ and $B_X^\perp \subseteq B_Z$.
\begin{lemma}\label{lem:css_from_pair_of_css}
For CSS codes $\mathcal{A} = \mathrm{CSS}(A_X , A_Z)$ and $\mathcal{B} = \mathrm{CSS}(B_X, B_Z)$ of block length $n$ over $\F_q$,
there exists a CSS code $\mathcal{Q} = \mathrm{CSS}(C_X , C_Z)$ given by
\begin{equation}
   C_X = (A_X \cap B_X) + B^\perp_Z
\end{equation}
\begin{equation}
   C_Z = (A_Z \cap B_Z) + B^\perp_X,
\end{equation}
so that
\begin{equation}
C^\perp_X = (A^\perp_X \cap B_Z) + B^\perp_X
\end{equation}
\begin{equation}
    C^\perp_Z = (A^\perp_Z \cap B_X) + B^\perp_Z.
\end{equation}
\end{lemma}

Consider the following sets:
\begin{align*}
    S_{+} &= \bigcup_{j=1}^{\delta-1} (j+(r+\delta-1)\Z),\\
    S_{-} &= \bigcup_{j=1}^{\delta-1} (-j+(r+\delta-1)\Z).
\end{align*}
Then \[S = \left([\ell] \cap ([q-1]\setminus S_{-})\right)\cup ([q-1]\cap S_{+}) = \left([\ell] \setminus S_{-}\right)\cup ([q-1]\cap S_{+}).\]
%\VG{Above and elsewhere $[\ell] \cap ([q-1]\setminus S_{-}))$ can be replaced by $[\ell] \setminus S_{-}$.}
We construct the $(r,\delta)$ quantum Tamo--Barg codes as follows:

\begin{lemma}
\label{lem:qtb.construction}
For integers $ r\geq \delta \geq 2$ such that $(r + \delta -1)|(q - 1)$, and an integer $q/2\leq \ell \leq q-1$, let 
\begin{equation}
    \label{eq.a}
    A = \ev(\F_q[X]^{[\ell]}),
\end{equation}
\begin{equation}
    \label{eq.b}
    B = \ev(\F_q[X]^{[q-1]\setminus S_{-}}).
\end{equation}
Then $A \cap B$ is a $(r,\delta)$ TB code. Furthermore, 
\begin{equation}
    \label{eq.a.perp}
    A^\perp = \ev(\F_q[X]^{[q-\ell]\setminus \{0\}}) \text { and }
\end{equation}
\begin{equation}
    \label{eq.b.perp}
    B^\perp = \ev(\F_q[X]^{[q-1]\cap S_{+}}) \subseteq B.
\end{equation}
If $\ell \geq q/2$, then $A^\perp \subseteq A$. Letting $C = (A \cap B) + B^\perp$, we obtain that $\mathrm{CSS}(C,C)$ is a QTB code with 
\begin{equation}
    \label{eq.c.perp}
    C^\perp = (A^\perp \cap B) + B^\perp = \ev(\F_q[X]^T) \subseteq C
\end{equation}   
where
\begin{equation}\label{eq.t}
    T = \left(([q-\ell]\setminus \{0\}) \cap ([q-1]\setminus S_{-})\right) \cup ([q-1]\cap S_{+}).
\end{equation}
\end{lemma}

\begin{proof}
We first prove the equality in \labelcref{eq.a.perp} by dimension counting. Notice that $\dim A = \ell$, so $\dim A^\perp = q- 1- \ell$. Since $\dim \ev(\F_q[X]^{[q-\ell]\setminus \{0\}}) = q-\ell - 1$, it suffices to prove $\ev(\F_q[X]^{[q-\ell]\setminus \{0\}}) \subseteq A^\perp$. To that end, by choosing a monomial basis for $\F_q[X]^{[\ell]}$ and $\F_q[X]^{[q-\ell]\setminus \{0\}}$, it is sufficient to prove $\ev(X^i)\cdot\ev(X^s) =0$ for every $i \in [\ell]$ and $s \in [q-\ell]\setminus\{0\}$. We have $\ev(X^i)\cdot\ev(X^s) = \sum_{\alpha \in \F_q^*} \alpha^{i+s} = 0$ for $i+s \not\equiv 0 \mod{q-1}$.

To see this last result, let $\omega_{q-1}$ be a generator for $\mathbb{F}_q^*$. Then $\alpha = \omega_{q-1}^j$ for some $j$. Therefore, \begin{align*}
    \sum_{\alpha\in \mathbb{F}_q^*} \alpha^{i+s} = \sum_{j=1}^{q-1}\omega_{q-1}^{j(i+s)} = \sum_{j=1}^{q-1}(\omega_{q-1}^{(i+s)})^j = \frac{(\omega_{q-1}^{i+s})^q - \omega_{q-1}^{i+s}}{\omega_{q-1}^{i+s} - 1} = \frac{\omega_{q-1}^{i+s} - \omega_{q-1}^{i+s}}{\omega_{q-1}^{i+s} - 1}= 0
\end{align*} because $i+s \not \equiv 0 \pmod{q-1}$ so the denominator is not zero.
%\VG{The more common way to see this is to replace $\alpha$ by $\gamma^j$ for primitive $\gamma$, and the use the geometric formula.}To see this last result, first observe that \begin{align*}X^{q-1}-1 &= \prod_{\alpha\in\mathbb{F}_q^*}(X-\alpha)\\&= \sum_{i=0}^{q-1}(-1)^{q-1-i}e_{q-1-i}X^i\end{align*} where $e_{q-1-i}$ are the elementary symmetric polynomials in $q-1$ variables evaluated at $\alpha\in\mathbb{F}_q^*$. Therefore, $e_{q-1-i} = 0$ for $i = 1,\ldots, q-2$. Applying Newton's formulas\cite[(2.11')]{Macdonald_2008} relating power sums to elementary symmetric polynomials gives the result.

Now we prove the equality in \labelcref{eq.b.perp}. It is straightforward to see that \[\dim B = q-1 - \frac{q-1}{r+\delta -1}(\delta-1)\] and \[\dim B^\perp = \frac{q-1}{r+\delta -1}(\delta-1).\] Notice, $\dim \ev(\F_q[X]^{[q-1]\cap S_{+}}) = \dim B^\perp$ so it suffices to prove \[\ev(\F_q[X]^{[q-1]\cap S_{+}}) \subseteq B^\perp.\] Choose a monomial basis for $\F_q[X]^{[q-1]\setminus S_{-}}$ and $\F_q[X]^{[q-1]\cap S_{+}}$. For all $i \in [q-1]\setminus S_{-}$ and $s \in [q-1]\cap S_{+}$, we have $\ev(X^i)\cdot\ev(X^s) = \sum_{\alpha \in \F_q^*} \alpha^{i+s} = 0$ for $i+s \not\equiv 0 \mod{q-1}$.

The inclusion $B^\perp \subseteq B$ follows from \begin{align*}\delta \leq r &\iff \delta -1 <\frac{r+\delta-1}{2}\\ &\implies S_{-}\cap S_{+} = \emptyset\end{align*} and the remaining arguments follow immediately using \Cref{lem:css_from_pair_of_css}.
\end{proof}

We now compute the dimension of the $(r,\delta)$ quantum Tamo--Barg code defined in \Cref{def:r.delta.qTB.qlrc}. 
\begin{lemma}\label{lem:dim.r.delta.TB}
The $(r,\delta)$ quantum Tamo--Barg code $\mathcal{Q} = \mathrm{CSS}(C,C)$ with parameters $q, r, \ell, \delta$ as defined in \Cref{def:r.delta.qTB.qlrc} has dimension
\begin{align*}k &= 1 + |\{q - \ell \leq i \leq \ell - 1 : i \not\in (S_{+} \cup S_{-})\}|\\
&= 1 + (2\ell - q)\left(1-\frac{2(\delta-1)}{r+\delta-1}\right) + \epsilon\end{align*} for some $\epsilon \in [-2(\delta-1),2(\delta-1)]$.
\end{lemma}

\begin{proof}
Recall $S \subseteq [q-1]$ as defined in \Cref{def:r.delta.qTB.qlrc}, $T \subseteq S$ as defined in \Cref{lem:qtb.construction} such that $C = \ev(\F_q[X]^{S})$ and $C^{\perp} = \ev(\F_q[X]^{T})$. Then, the dimension of $\mathcal{Q}$ can be computed as follows:
\begin{align*}
    \dim(\mathcal{Q}) &= \dim(C) - \dim(C^{\perp}) \\
    &= |S \setminus T| \\
    &= \left| \{0\} \cup \{q - \ell \leq i \leq \ell - 1 : i \not\in (S_{+} \cup S_{-})\}\right|.
\end{align*} which proves the first equality. For the second equality, we inspect \[\{q - \ell \leq i \leq \ell - 1 : i \not\in (S_{+} \cup S_{-})\}.\] If $(r+\delta-1) \mid (2\ell-q)$ then this set has size exactly \[(2\ell - q)\left(1-\frac{2(\delta-1)}{r+\delta-1}\right)\] because we exclude $2(\delta-1)$ elements in every block of $(r+\delta-1)$. Note, this exclusion will leave a nontrivial number of elements precisely because \[2(\delta-1) < r + \delta -1 \iff \delta - 1< r \iff \delta \leq r.\]If $(r+\delta-1) \nmid (2\ell-q)$, we differ from \[(2\ell - q)\left(1-\frac{2(\delta-1)}{r+\delta-1}\right)\] by at most $2(\delta-1)$ in absolute value and the second equality follows.
\end{proof}

We need to show that the code, $\mathcal{Q} = \mathrm{CSS}(C,C)$ defined in \Cref{def:r.delta.qTB.qlrc} is an $(r,\delta)$-QLRC. By \Cref{thm:css_qlrc_iff_clrc}, it suffices to show that $C^{\perp}$ contains low-weight parity checks whose supports cover all $q-1$ code components. As defined in \Cref{lem:qtb.construction}, $B^{\perp} \subseteq C^{\perp}$ so it suffices to show that $B^{\perp}$ contains low-weight parity checks. The result is summarized in \Cref{lem:low.wt.parity.checks} and \Cref{cor:locality_qtb}.

\begin{lemma}
\label{lem:low.wt.parity.checks}
Consider the code $B^\perp = \ev(\F_q[X]^{[q-1]\cap S_{+}})$ as defined in \Cref{lem:qtb.construction}. Let $\Omega_{r+\delta-1} = \{x \in \F_q^* : x^{r+\delta-1} = 1\}$ be the subgroup of $(r+\delta-1)$th roots of unity. Then a function $f : \F_q^* \rightarrow \F_q$ lies in $B^\perp$ if and only if for every coset $\alpha\Omega_{r+\delta-1}$ for some $\alpha \in \mathbb{F}_q^{*}$, there exists a polynomial $P_\alpha(X) \in \mathbb{F}_q[X]$ with $\deg P_\alpha \leq \delta-1$ and no constant term such that \[f(\alpha \omega) = P_\alpha(\omega)\] for all $\omega \in \Omega_{r+\delta-1}$.
\end{lemma}

Before we prove the lemma, we make some observations about the statement. Since $P_\alpha$ has no constant term, $P_\alpha(0) = 0$. An equivalent formulation would be to say that $f$ lies in $B^\perp$ if and only if $f$ is a piecewise polynomial in $\omega$ of degree $\leq \delta -1$ on each coset $\alpha \Omega_{r+\delta-1}$ and that polynomial has no constant term.

\begin{proof} By dimension counting, it suffices to show the forward direction.

 $(\implies)$ Let $n := r + \delta - 1$. Fix $f(X) = \sum_{i\in[q-1]} f_iX^i$ with $\ev(f) \in B^\perp$. By definition, $f_i$ for $i \not\in [q-1]\cap S_{+}$ must be zero. Grouping terms by residues modulo $n$, we can write \[f(X) = \sum_{j=1}^{\delta-1}X^jF_j(X^{n})\] where $F_j \in \mathbb{F}_q[Y]$ is a polynomial. Fix a coset $\alpha\Omega_{n}$. For any $x\in \alpha\Omega_{n}$, we can write $x = \alpha\omega$ for some $\omega\in\Omega_{n}$. Since $\omega^{n} = 1$, we have \[f(x) = f(\alpha\omega) = \sum_{j=1}^{\delta-1}(\alpha\omega)^jF_j((\alpha\omega)^{n}) = \sum_{j=1}^{\delta-1}(\alpha\omega)^jF_j(\alpha^{n}) = \sum_{j=1}^{\delta-1}(\alpha^jF_j(\alpha^{n}))\omega^j.\] The final sum is a polynomial in $\omega$ of degree $\leq \delta-1$ and no constant term. Therefore, we can define $P_\alpha$ to be \[P_\alpha(X) = \sum_{j=1}^{\delta-1}(\alpha^jF_j(\alpha^{n}))X^j\] and it is clear that $\deg P_\alpha \leq \delta-1$, $P_\alpha$ has no constant term, and $f(\alpha\omega) = P_\alpha(\omega)$ for all $\omega \in \Omega_{r+\delta-1}$.

\end{proof}

\begin{corollary}\label{cor:locality_qtb}
The $(r,\delta)$-QTB code defined in \Cref{def:r.delta.qTB.qlrc} is a QLRC with locality $r$ such that $\delta-1$ erasures are corrected by each local repair group.
\end{corollary}

\begin{proof}
    By \Cref{thm:css_qlrc_iff_clrc}, it is sufficient to show that $C$ is classically $(r,\delta)$-locally recoverable. Let $n := r + \delta -1$ and fix $\alpha\Omega_{n}$. By \Cref{lem:qtb.construction}, we know $C^\perp \supseteq B^\perp = \ev(\F_q[X]^{[q-1]\cap S_{+}})$ and \Cref{lem:low.wt.parity.checks} implies that $B^\perp$ contains the $\delta-1$ functions \[f_{\alpha,j}(x) = \begin{cases}
        x^j,& x\in \alpha\Omega_{n}\\
        0,&x\not\in\alpha\Omega_{n}
    \end{cases}\] for $j = 1,\ldots, \delta -1$. These are linearly independent. Their restriction to $\alpha\Omega_{n}$ forms a $(\delta-1) \times n$ Vandermonde-type matrix \[\begin{bmatrix}
        x_1 & x_2 & \cdots & x_{n}\\
        x_1^2 & x_2^2 & \cdots & x_{n}^2\\
        \vdots & \vdots & \ddots & \vdots\\
        x_1^{\delta-1} & x_2^{\delta-1} & \cdots &x_{n}^{\delta-1}
    \end{bmatrix}.\] Hence, any $\delta -1$ columns are linearly independent so $d(C|_{\alpha\Omega_{n}}) \geq \delta$. Since $|\alpha\Omega_{n}| = r+\delta-1$, this shows $C$ is an $(r,\delta)$-LRC.

    Additionally, since puncturing can only decrease the distance and $C^\perp \subseteq C$, we have the following \[d(C^\perp) \geq d(C) \geq d(C|_{\alpha \Omega_{n}}) \geq \delta.\] Therefore, by \Cref{thm:css_qlrc_iff_clrc}, $\mathcal{Q} = \mathrm{CSS}(C,C)$ is an $(r,\delta)$-QLRC.
\end{proof}
\subsection{Non-vanishing theorem}
It remains to compute a lower bound on the distance of the $(r,\delta)$-QTB. For the distance bound, we will need the following theorem on the number of roots of unity at which a particular polynomial vanishes. In order to prove this theorem, we rely on proving a particular combinatorial fact on homogeneous symmetric polynomials (\Cref{lem:homogenous_poly_roots_of_unity}) which, to the best of our knowledge, does not appear in the current literature and may be of independent interest.

\begin{theorem}\label{thm:Q_b_has_delta-1_roots}
    Let $r\geq \delta \geq 3$ and let $n = r + \delta -1$. Let $\zeta \in \mathbb{C}$ be a primitive $n$th root of unity. For each $b \in \{\delta - 1,\ldots, n - 1\}$, let \[Q_b(Y) = Y^b + \sum_{t = 0}^{\delta-2}v_tY^t\] be a polynomial satisfying $Q_b(\zeta^t) = 0$ for $t = 0,\ldots, \delta - 2$. Then $Q_b(\zeta^s) \neq 0$ for any $s \in \{\delta-1,\ldots, n-1\}$. In other words, $Q_b(\zeta^s) = 0 \iff s = 0,\ldots, \delta-2$.
\end{theorem}

\begin{proof}
    Fix $s,b \in \{\delta-1,\ldots, n-1\}$. Consider the $\delta \times \delta$ matrix $M_{s,b} = (\zeta^{ac})_{a \in R, c \in C}$ where $R = \{0, 1, \ldots, \delta-2, s\}$ and $C = \{0, 1, \ldots, \delta - 2, b\}$. More explicitly, \[M_{s,b} = \begin{bmatrix}
        \zeta^{0\cdot 0} & \zeta^{0 \cdot 1} & \cdots & \zeta^{0 \cdot (\delta-2)} & \zeta^{0\cdot b}\\
        \zeta^{1\cdot 0} & \zeta^{1 \cdot 1} & \cdots & \zeta^{1 \cdot (\delta-2)} & \zeta^{1 \cdot b}\\
        \vdots & \vdots & \ddots & \vdots &\vdots\\
        \zeta^{s\cdot 0} & \zeta^{s \cdot 1} & \cdots & \zeta^{s \cdot (\delta-2)} & \zeta^{s \cdot b}
    \end{bmatrix}.\] Replacing the last column $Y^b$, by $Q_b(Y) = Y^b + \sum_{t = 0}^{\delta-2}v_tY^t$ does not change the determinant because we are only performing column operations. By hypothesis, $Q_b(\zeta^t) = 0$ for $t = 0, \ldots, \delta -2$. Hence, the last column is $[0,\ldots, 0, Q_b(\zeta^s)]^T$. Therefore, $\det M_{s,b} = \det(\zeta^{a,c})_{0\leq a,c \leq \delta-2} \cdot Q_b(\zeta^s)$. 
    
    Let $x_1 = \zeta^0, x_2 = \zeta^1, \ldots, x_{\delta-1} = \zeta^{\delta-2}, x_\delta = \zeta^s$. Then, \begin{align}
        Q_b(\zeta^s) &= \frac{\det M_{s,b}}{\det(\zeta^{a,b})_{0\leq a,b \leq \delta-2}}\label{eq:q = h 1}\\
        &= \frac{(-1)^{\lfloor \frac{\delta}{2}\rfloor}\begin{vmatrix}
            x_1^b & x_2^b & \cdots & x_\delta^b\\
            x_1^{\delta-2} & x_2^{\delta-2} & \cdots & x_\delta^{\delta-2}\\
            \vdots & \vdots & \ddots & \vdots\\
            x_1 & x_2 & \cdots & x_\delta\\
            1 & 1 & \cdots & 1
        \end{vmatrix}}{(-1)^{\lfloor \frac{\delta-1}{2}\rfloor}\begin{vmatrix}
            x_1^{\delta-2} & x_2^{\delta-2} & \cdots & x_{\delta-1}^{\delta-2}\\
            \vdots & \vdots & \ddots & \vdots\\
            x_1 & x_2 & \cdots & x_{\delta-1}\\
            1 & 1 & \cdots & 1
        \end{vmatrix}}\label{eq:q = h 2}\\
        &= (-1)^{\delta-1}\frac{\begin{vmatrix}
            x_1^b & x_2^b & \cdots & x_\delta^b\\
            x_1^{\delta-2} & x_2^{\delta-2} & \cdots & x_\delta^{\delta-2}\\
            \vdots & \vdots & \ddots & \vdots\\
            x_1 & x_2 & \cdots & x_\delta\\
            1 & 1 & \cdots & 1
        \end{vmatrix}}{\begin{vmatrix}
            x_1^{\delta-1} & x_2^{\delta-1} & \cdots & x_{\delta}^{\delta-1}\\
            x_1^{\delta-2} & x_2^{\delta-2} & \cdots & x_{\delta}^{\delta-2}\\
            \vdots & \vdots & \ddots & \vdots\\
            x_1 & x_2 & \cdots & x_{\delta}\\
            1 & 1 & \cdots & 1
        \end{vmatrix}} \cdot \frac{\begin{vmatrix}
            x_1^{\delta-1} & x_2^{\delta-1} & \cdots & x_{\delta}^{\delta-1}\\
            x_1^{\delta-2} & x_2^{\delta-2} & \cdots & x_{\delta}^{\delta-2}\\
            \vdots & \vdots & \ddots & \vdots\\
            x_1 & x_2 & \cdots & x_{\delta}\\
            1 & 1 & \cdots & 1
        \end{vmatrix}}{\begin{vmatrix}
            x_1^{\delta-2} & x_2^{\delta-2} & \cdots & x_{\delta-1}^{\delta-2}\\
            \vdots & \vdots & \ddots & \vdots\\
            x_1 & x_2 & \cdots & x_{\delta-1}\\
            1 & 1 & \cdots & 1
        \end{vmatrix}}\label{eq:q = h 3}\\
        &= (-1)^{\delta-1} \cdot s_{\lambda = (b - \delta +1, 0, \ldots, 0)}(x_1,\ldots, x_\delta) \cdot (-1)^{\delta-1} \cdot \prod_{t=0}^{\delta-2}(x_\delta - x_{t+1})\label{eq:q = h 4}\\
        &= \begin{vmatrix}
            h_{b - (\delta-1)} & h_{b - (\delta-2)} & \cdots & h_b\\
            h_{-1} & h_0 & \cdots & h_{\delta-2}\\
            \vdots & \vdots & \ddots & \vdots\\
            h_{-(\delta-1)} & h_{-(\delta-2)} & \cdots & h_0
        \end{vmatrix} \cdot \prod_{t=0}^{\delta-2}(x_\delta - x_{t+1})\label{eq:q = h 5}\\
        &= \begin{vmatrix}
            h_{b - (\delta-1)} & h_{b - (\delta-2)} & \cdots & h_b\\
            0 & 1 & \cdots & h_{\delta-2}\\
            \vdots & \vdots & \ddots & \vdots\\
            0 & 0 & \cdots & 1
        \end{vmatrix} \cdot \prod_{t=0}^{\delta-2}(x_\delta - x_{t+1})\label{eq:q = h 6}\\
        &= h_{b-\delta +1}(x_1,\ldots, x_\delta) \cdot \prod_{t=0}^{\delta-2}(x_\delta - x_{t+1})\label{eq:q = h 7}\\
        &= h_{b-\delta +1}(1,\ldots, \zeta^{\delta-2}, \zeta^s) \cdot \prod_{t=0}^{\delta-2}(\zeta^s - \zeta^t)\label{eq:q = h 8}.
    \end{align}
     Here, $s_\lambda(x_1,\ldots, x_\delta)$ is the Schur polynomial over the partition $\lambda$ and $h_{b-\delta+1}$ is the complete homogeneous symmetric polynomial of degree $b-\delta+1$.
     \Cref{eq:q = h 2} follows from \labelcref{eq:q = h 1} because the determinant is unchanged under column swaps and transposition up to sign. The first term in \labelcref{eq:q = h 4} follows from \labelcref{eq:q = h 3} by Jacobi's bialternant formula for Schur polynomials (also a special case of Weyl character formula)\cite{cauchy1815memoire, Jacobi1841} and the second term follows from taking a quotient of two Vandermonde determinants. \Cref{eq:q = h 5} follows from \labelcref{eq:q = h 4} by the Jacobi-Trudi formula \cite{Jacobi1841,trudi1862teoria}. For $s \in \{\delta-1,\ldots, n-1\}$, $\zeta^s - \zeta^t \neq 0$ for all $t = 0,\ldots, \delta-2$. Hence, $Q_b(\zeta^s) \neq 0$ if and only if $h_{b-\delta + 1}(1, \zeta, \ldots, \zeta^{\delta-2}, \zeta^s) \neq 0$ which follows from \Cref{lem:homogenous_poly_roots_of_unity}.
\end{proof}

\begin{remark}\label{rem:Q_b_delta=2}
    \Cref{thm:Q_b_has_delta-1_roots} is not true in the case of $\delta = 2$. Consider the following example. Let $\delta = 2$ then $n = r + 1$. We must choose $b \in \{1,\ldots ,r\}$. $Q_b(Y) = Y^b + v_0$ and $Q_b(1) = 0$ so $v_0 = -1$ and $Q_b(Y) = Y^b - 1$. Now for $s \in \{1,\ldots, r\}$, $Q_b(\zeta^s) = 0$ if $\zeta^{sb} = 1$. Let $ r= 3$ so $n = 4$ and choose $b = s = 2$ then $Q_2(\zeta^2) = \zeta^4 - 1 = 0$.

    Hence, if \(s\) is chosen such that $n \mid sb$, then $\zeta^s$ is a root.
    If $\delta = 2$ and we strengthen the hypothesis of \Cref{thm:Q_b_has_delta-1_roots} to choosing $b$ such that $\gcd(b, r + 1) = 1$, then the claim will follow.
    For more details on where the proof particularly breaks see \Cref{rem:homogeneous_delta=2}.
\end{remark}
\Cref{thm:Q_b_has_delta-1_roots} is over $\C$, but we need results over finite fields so we present the following reduction. In particular, the same result will hold over finite fields once we exclude finitely many prime characteristics.

\begin{corollary}\label{cor:Q_b_finite_field_version_outside_finitely_many_char}
Fix integers \(r\ge \delta\ge 3\), and set $ n = r + \delta -1$.
For each \[ m\in\{0,\dots,r-1\} \qquad\text{and}\qquad s\in\{\delta-1,\dots,n-1\},\] define \[ A_{m,s}(X):=h_m(1,X,X^2,\dots,X^{\delta-2},X^s)\in \mathbb{Z}[X]. \]
Let \(\Phi_n(X)\in \mathbb Z[X]\) denote the \(n\)th cyclotomic polynomial, and define
\[ \mathcal{M}_{r,\delta}:=\prod_{m=0}^{r-1}\ \prod_{s=\delta-1}^{n-1} \Res\!\left(\Phi_n(X),A_{m,s}(X)\right)\in \mathbb Z.\]
Then \(\mathcal{M}_{r,\delta}\neq 0\). Consequently, if \(\F_q\) is a finite field of characteristic \(p\) such that \(p\nmid \mathcal{M}_{r,\delta}\) and \(n \mid (q-1)\), then for every primitive \(n\)th root of unity \(\omega\in \F_q\), every \(b\in\{\delta-1,\dots,n-1\}\), and every \(s\in\{\delta-1,\dots,n-1\}\), the polynomial \[ Q_b(Y)=Y^b+\sum_{t=0}^{\delta-2}v_tY^t \] satisfying $Q_b(\omega^t)=0$ for $t = 0,\ldots, \delta-2$ also satisfies $Q_b(\omega^s) \neq 0$. In particular, for fixed \(r\) and \(\delta\), the conclusion of \Cref{thm:Q_b_has_delta-1_roots} holds over every finite field of characteristic outside a finite set of primes.
\end{corollary}

\begin{proof}
Let \(\zeta\in \mathbb C\) be a primitive \(n\)th root of unity. By \Cref{lem:homogenous_poly_roots_of_unity}, for every
\[
m\in\{0,\dots,r-1\},\qquad s\in\{\delta-1,\dots,n-1\},
\]
we have
\[
A_{m,s}(\zeta)
=
h_m(1,\zeta,\zeta^2,\dots,\zeta^{\delta-2},\zeta^s)\neq 0.
\]
\Cref{lem:homogenous_poly_roots_of_unity} applies to every primitive \(n\)-th root \(\eta\). Hence, \(A_{m,s}(\eta)\neq 0\) for every root \(\eta\) of \(\Phi_n\), so \(A_{m,s}\) and \(\Phi_n\) have no common root in \(\mathbb C\). Therefore, they are coprime in \(\mathbb Q[X]\) and 
\[
\Res\!\left(\Phi_n(X),A_{m,s}(X)\right)\neq 0
\]
for every such pair \((m,s)\). Thus \(\mathcal{M}_{r,\delta}\neq 0\).

Now let \(\F_q\) be a finite field of characteristic \(p\) such that \(p\nmid \mathcal M_{r,\delta}\) and \(n\mid(q-1)\), and let \(\omega\in \F_q\) be a primitive \(n\)th root of unity. Fix
\[
b\in\{\delta-1,\dots,n-1\},\qquad s\in\{\delta-1,\dots,n-1\},
\]
and let $m = b - \delta + 1$. Suppose, for contradiction, $Q_b(\omega^s) = 0$.

By the same determinant/Jacobi--Trudi computation used in the proof of \Cref{thm:Q_b_has_delta-1_roots}---which is purely algebraic and therefore valid over any field containing a primitive \(n\)th root of unity---we have
\[
Q_b(\omega^s)=0 \implies A_{m,s}(\omega)=0,
\]
because
\[
\prod_{t=0}^{\delta-2}(\omega^s-\omega^t)\neq 0
\]
in any field containing a primitive \(n\)-th root of unity.
Hence, $A_{m,s}(\omega) = 0 \in \F_q$ so $\omega$ is a root of the reduction $\overline{A_{m,s}}(X) \in \F_p[X]$.

We next show that \(\omega\) is also a root of the reduction of \(\Phi_n(X)\) modulo \(p\). Since \(n\mid(q-1)\), we have \(p\nmid n\). Reducing the factorization
\[
X^n-1=\prod_{d\mid n}\Phi_d(X)
\]
modulo \(p\), we obtain in \(\F_p[X]\)
\[
X^n-1=\prod_{d\mid n}\overline{\Phi_d}(X).
\]
Now \(\omega^n=1\), so \(\omega\) is a root of \(X^n-1\). Moreover, because \(\omega\) has exact multiplicative order \(n\), it is not a root of \(X^d-1\) for any proper divisor \(d\mid n\). Since every root of \(\overline{\Phi_d}(X)\) is also a root of \(X^d-1\), it follows that \(\omega\) cannot be a root of \(\overline{\Phi_d}(X)\) for any proper divisor \(d<n\). Therefore, \(\omega\) must be a root of \(\overline{\Phi_n}(X)\).

Thus, \(\omega\) is a common root of the reductions modulo \(p\) of \(A_{m,s}(X)\) and \(\Phi_n(X)\). By \Cref{prop:resultant_mod_p},
\[
p\mid \Res\!\left(\Phi_n(X),A_{m,s}(X)\right).
\]
But \(\Res(\Phi_n(X), A_{m,s}(X))\) is one of the factors of \(\mathcal{M}_{r,\delta}\), so this forces
\(
p\mid \mathcal{M}_{r,\delta},
\)
contrary to hypothesis. Therefore, our assumption was false, and we conclude that
\[
Q_b(\omega^s)\neq 0
\qquad
(s=\delta-1,\dots,n-1).
\]
This proves the corollary.
\end{proof}

\begin{remark}\label{rem:Q_b_delta=2_finite_field}
    As noted in \Cref{rem:Q_b_delta=2}, \Cref{thm:Q_b_has_delta-1_roots} does not hold for $\delta = 2$ without the additional hypothesis of $\gcd(b,r+1) = 1$. However, that additional assumption is all that is necessary even in the finite field setting.
    More formally, if $\delta = 2$, then $Q_b(Y) = Y^b - 1$ so its roots among the $n$th roots of unity are \[\omega^0, \omega^{\frac{n}{\gcd(b,n)}},\omega^{\frac{2n}{\gcd(b,n)}}, \ldots, \omega^{\frac{(\gcd(b,n)-1)n}{\gcd(b,n)}}.\] There are exactly $\gcd(b,n)$ roots so we must impose $\gcd(b,r+1) = 1$.
     
    Let \(b\in\{1,\dots,n-1\}\), \(s\in\{1,\dots,n-1\}\), and set \(m=b-1\). We claim, if \(\gcd(b,n)=1\), then 
    \[
        \Res\!\left(\Phi_n(X),h_m(1,X^s)\right)=1.
    \] Since
    \[
        h_m(1,X^s)=1+X^s+\cdots+X^{ms}=1+X^s+\cdots+X^{(b-1)s},
    \]
    and \(\Phi_n(X)\) is monic with roots \(\zeta_n^a\) for \(a\in(\Z/n\Z)^\times\), we have
    \[
        \Res\!\left(\Phi_n(X),h_m(1,X^s)\right)
    =
    \prod_{a\in(\Z/n\Z)^\times}\left(1+\zeta_n^{as}+\cdots+\zeta_n^{(b-1)as}\right).
    \]
    Because \(1\le s\le n-1\), we have \(\zeta_n^{as}\neq 1\), so
    \[
    1+\zeta_n^{as}+\cdots+\zeta_n^{(b-1)as}
    =
    \frac{1-\zeta_n^{asb}}{1-\zeta_n^{as}}.
    \]
    Thus
    \[
    \Res\!\left(\Phi_n(X),h_m(1,X^s)\right)
    =
    \frac{\prod_{a\in(\Z/n\Z)^\times}(1-\zeta_n^{asb})}
    {\prod_{a\in(\Z/n\Z)^\times}(1-\zeta_n^{as})}.
    \]
    Since \(\gcd(b,n)=1\), multiplication by \(b\) permutes \((\Z/n\Z)^\times\), so the numerator and denominator are equal. Therefore the resultant is \(1\) and it follows that $\mathcal{M}_{r,2} = 1$ so there are no characteristics $p$ which we need to exclude.

    In \cite{golowich2025quantum}, the authors chose $r+1$ to be prime which is stronger than necessary because with $r+1$ prime, $\gcd(b,r+1) = 1$ is immediate.
\end{remark}
\medskip
We now give an example of a pair of values \(r, \delta\) for which there exists a prime characteristic $p$ such that $Q_b(\omega^s) = 0$ for some $b \in \{\delta-1,\ldots, n-1\}$ and some $s \in \{\delta-1,\ldots, n-1\}$.

\begin{example}\label{example: Q_b_equal_0}
    Let $\delta = 3$ and $r = 9$. Then $n = 11$. Let $q = 23$ so $n \mid (q-1)$ and let $\omega = 2$ be the primitive $11$th root of unity in $\F_q$. Let $b = 4$ and $s = 5$.

    By definition,
    \[
        Q_4(Y) = Y^4 + v_1Y + v_0
    \] and $Q_4(1) = Q_4(\omega) = 0$. Solving for $v_1$ and $v_0$, we obtain $Q_4(Y) = Y^4 + 8Y + 14$. Now substituting $\omega^5 = 2^5 = 9$, we get
    \[
        Q_4(9) = 9^4 + 72 + 14 = 81^2 + 17 = 12^2 + 17 = 161 = 0 \in \F_{23}.
    \]
    Hence, for characteristic $p = 23$, there exists $b \in \{\delta-1,\ldots, n-1\}$ and $s \in \{\delta-1,\ldots, n-1\}$ such that $Q_b(\omega^s) = 0$.

    Now we show that $23$ is the only such characteristic for these values of $r$ and $\delta$. We need to compute $\mathcal{M}_{9,3}$. Since $n = 11$ is prime, \(\Phi_{11}(X) = X^{10} + X^9 + \cdots + 1\) and \(A_{m,s}(X) = h_m(1,X,X^s)\) so by \Cref{cor:cyclotomic_norm_formula}, \[\Res(\Phi_n(X),A_{m,s}(X)) = N_{\Q(\zeta_n)/\Q}\left(A_{m,s}(\zeta_n)\right).\] Using SageMath, of the $81$ pairs $(m,s)$ with $0 \leq m \leq 8$ and $2 \leq s \leq 10$, $36$ pairs had a resultant of $23$ and the remaining $45$ pairs had a resultant of $1$. The $36$ pairs are \((m,s)\) for
    \[
        m \in \{1,2,3,5,6,7\}, \qquad s\in \{3,4,5,7,8,9\}.
    \]
    We will compute 
    \[
        \Res \!\left(\Phi_{11}(X), A_{2,5}(X)\right) = N_{\Q(\zeta_{11})/\Q}\left(A_{2,5}(\zeta_{11})\right)
    \] to show the computation for one of the $81$ pairs with the others proceeding similarly.

    Let $\zeta = \zeta_{11}$. Then
    \[
        \alpha = A_{2,5}(\zeta) = 1 + \zeta +  \zeta^2 + \zeta^5 + \zeta^6 + \zeta^{10} \in \Q(\zeta)
    \] and \[\mathrm{Gal}(\Q(\zeta)/\Q) \cong (\Z/11\Z)^{\times} = \{1,\ldots, 10\}.\] For each $a \in \{1,\ldots, 10\}$, the corresponding automorphism is $\sigma_a : \Q(\zeta) \to \Q(\zeta)$ where $\sigma_a(\zeta) = \zeta^a$ so the conjugates of $\alpha$ are exactly \[\alpha_a:= \sigma_a(\alpha) = A_{2,5}(\zeta^a) = 1 + \zeta^a +  \zeta^{2a} + \zeta^{5a} + \zeta^{6a} + \zeta^{10a}.\]
    Thus, \[N_{\Q(\zeta_{11})/\Q}\left(A_{2,5}(\zeta_{11})\right) = \prod_{a=1}^{10} \alpha_a = 23.\] One could also see this by computing the minimal polynomial of $\alpha$, \(\prod_{a=1}^{10}(X-\alpha_a),\) and noting that the constant of this polynomial is exactly the norm. 
\end{example}

\subsection{\texorpdfstring{Bound on the minimum distance of $(r,\delta)$ Quantum Tamo--Barg codes}{Bound on the minimum distance of (r,delta) Quantum Tamo--Barg codes}}

We begin with a small lemma whose properties we will use when proving a bound on the minimum distance of the code.

\begin{restatable}{lemma}{PsiConvexMonotone}\label{lem:psi_convex_monotone}
Let \(r\ge \delta\ge 2\), set \(n=r+\delta-1\), and define
\[
\psi(t):=(r-t)_+\left(n-(\delta-1)t\right)_+,
\qquad y_+:=\max\{y,0\}.
\]
Then \(\psi\) is nonincreasing and convex on \([0,\infty)\).
\end{restatable}

\begin{proof}
Set
\[
a:=\delta-1,
\qquad
T:=\min\left\{r,\frac{n}{a}\right\}.
\]
Then \(\psi(t)>0\) precisely for \(0\le t<T\), and on this interval
\[
\psi(t)=(r-t)(n-at)=rn-(n+ar)t+at^2.
\]
Hence,
\[
\psi'(t)=-(n+ar)+2at,
\qquad
\psi''(t)=2a>0.
\]
Thus, \(\psi'\) is strictly increasing on \([0,T)\), so \(\psi\) is convex on \([0,T)\).

For \(t\ge T\), at least one of the two factors \((r-t)_+\) and \((n-at)_+\) is zero, so \(\psi(t) = 0\).
Thus, \(\psi\) is constant on \([T,\infty)\), and its right derivative at \(T\) is \(\psi'_+(T)=0\).
We now check that the derivative does not jump downward at \(T\). The left derivative is
\[
\psi'_-(T)=-(n+ar)+2aT.
\]
There are two cases. If \(T=r\), then \(r\le n/a\), equivalently \(ar\le n\). Hence,
\[
\psi'_-(T)=\psi'_-(r)=-(n+ar)+2ar=ar-n\le 0.
\]
If \(T=n/a\), then \(n/a\le r\), equivalently \(n\le ar\). Hence,
\[
\psi'_-(T)
=
\psi'_-\left(\frac{n}{a}\right)
=
-(n+ar)+2n
=
n-ar
\le 0.
\]
Therefore, in both cases,
\[
\psi'_-(T)\le 0=\psi'_+(T)
\]
so the derivative is nondecreasing through the clipping point. Since \(\psi'\) is increasing on \([0,T)\) and is equal to \(0\) on \([T,\infty)\), the derivative is nondecreasing on all of \([0,\infty)\). Therefore \(\psi\) is convex on \([0,\infty)\).

We now show that \(\psi\) is nonincreasing. Since \(\psi'\) is increasing on \([0,T)\) and \(\psi'_-(T)\le 0\), we have \(\psi'(t)\le 0\) for all \(0 \le t < T\). For \(t\ge T\), \(\psi(t)=0\) is constant. Hence \(\psi\) is nonincreasing on \([0,\infty)\).
\end{proof}

\begin{theorem}\label{thm:qtb_distance}
    Consider $(r,\delta)$-QTB code $\mathcal{Q} = \mathrm{CSS}(C,C)$ defined in \Cref{def:r.delta.qTB.qlrc}. If $\delta \geq 3$, assume $\mathrm{char}(\F_q) \nmid \mathcal{M}_{r,\delta}$ or if $\delta = 2$, assume $r+1$ is prime.
    %and assume that $r + \delta - 1$ is prime such that \Cref{prop:uncertainty} holds. 
    Then $\mathcal{Q}$ has distance at least \begin{equation}\label{eqn:qtb_distance_bound}\frac{q-1}{2}\left(\frac{1}{\delta-1} + \frac{r}{r+\delta-1} - \sqrt{\left(\frac{r}{r+\delta-1} - \frac{1}{\delta-1}\right)^2 + \frac{4r}{(\delta-1)(r+\delta-1)}\cdot \frac{\ell-1}{q-1}}\right).\end{equation}
\end{theorem}
\begin{proof}
    Fix an arbitrary $\ev(f(X)) \in C\setminus C^\perp$. Since  $C = \ev(\F_q[X]^S)$ for \[S = \left([\ell] \setminus S_{-}\right)\cup ([q-1]\cap S_{+})\] we may write $f(X) = g(X) + h(X)$ where $g(X) \in \mathbb{F}_q[X]^{[\ell]\setminus (S_{-}\cup S_{+})}$ and $h(X) \in \mathbb{F}_q[X]^{ [q-1]\cap S_{+}}$. By \Cref{lem:low.wt.parity.checks}, $h$ is piecewise of degree at most $\delta - 1$. Since $\ev(h) \in C^\perp$ and $\ev(f) \not \in C^\perp$ we must have $g \neq 0$.

    Let $n = r + \delta -1$. Choose integers $\Gamma_i = \{0, 1, \ldots, \delta -2, i\}$ that are distinct modulo $n$ so $i \in \{\delta -1, \ldots, n-1\}$. There are exactly $r$ such $\Gamma_i$s. Then for a fixed primitive root of unity $\omega \in \Omega_{n}$, we construct the Vandermonde-type matrix \[V_i = \begin{bmatrix}
        1 & 1 & 1 & \cdots & 1 & 1\\
        1 & \omega & \omega^2 & \cdots & \omega^{\delta-2} &  \omega^{i}\\
        \vdots & \vdots & \vdots & \ddots & \vdots & \vdots\\
        1 & \omega^{\delta-2} & \omega^{(\delta-2)2} & \cdots & \omega^{(\delta-2)(\delta-2)} & \omega^{(\delta-2)i}
    \end{bmatrix}.\] Clearly $\mathrm{rank}(V_i)  = \delta - 1$ so there is a nontrivial vector $v_i = (v_{i,0},\ldots, v_{i, \delta-2}, v_{i,\delta-1})$ such that $V_i v_i = 0$. If $v_{i,\delta-1} = 0$, then the first $\delta-1$ columns of $V_i$ would be linearly dependent, but the first $\delta-1$ columns of $V_i$ form a Vandermonde matrix. Hence, $v_{i,\delta-1} \neq 0$ and by normalizing, we may assume $v_{i,\delta-1} = 1$. Consider the polynomial \[Q_i(Y) = Y^i + \sum_{t=0}^{\delta-2}v_{i,t}Y^t.\] By construction, $Q_i(\omega^\gamma) = 0$ for $\gamma = 0,\ldots, \delta -2$. Let $\alpha \Omega_{n}\in \mathbb{F}_q^*/\Omega_{n}$. For a coset element $x\in \alpha\Omega_{n}$, \[\omega^{-i}(\omega^ix)^j + \sum_{t=0}^{\delta-2} v_{i,t}\omega^{-t}(\omega^{t}x)^j= x^j \left(\omega^{(j-1)i} + \sum_{t=0}^{\delta-2}v_{i,t}\omega^{(j-1)t} \right) = x^jQ_i(\omega^{j-1}) = 0\] for $j =1,\ldots, \delta-1$. It follows, for all $x \in \mathbb{F}_q^*$, \[\omega^{-i}h(\omega^ix)  + \sum_{t=0}^{\delta-2} v_{i,t}\omega^{-t}h(\omega^{t}x)= 0.\] Define \[g_i^{(1)}(X) = \omega^{-i}g(\omega^iX) + \sum_{t=0}^{\delta-2} v_{i,t}\omega^{-t}g(\omega^{t}X).\] Note that $\deg g_i^{(1)} \leq \deg g \leq \ell -1$. Additionally, if $x\in \mathbb{F}_q^*$, is such that $f(\omega^{\gamma}x) = 0$ for all $\gamma \in \Gamma_i$, then \[0 = \omega^{-i}f(\omega^ix) + \sum_{t=0}^{\delta-2} v_{i,t}\omega^{-t}f(\omega^{t}x) = g_i^{(1)}(x) + \omega^{-i}h(\omega^ix) + \sum_{t=0}^{\delta-2} v_{i,t}\omega^{-t}h(\omega^{t}x) = g_i^{(1)}(x).\] Hence, any $x \in \mathbb{F}_q^*$ such that $f(\omega^{\gamma}x) = 0$ for all $\gamma \in \Gamma_i$, is a root of $g_i^{(1)}(X)$.

    We define \[G(X) = \prod_{i = \delta-1}^{n-1} g_i^{(1)}(X),\] which has degree at most $r(\ell-1)$. It remains to show that $G \neq 0$ for which it suffices to show that each $g_i^{(1)} \neq 0$. 
    %Notice that $|Q_i| \leq \delta$ so by \Cref{prop:uncertainty}, we have \[|\ev(Q_i)|_{\Omega_{r+\delta-1}} \geq (r+\delta-1) + 1 - \delta = r.\] This implies $Q_i$ can have at most $\delta-1$ roots in $\Omega_{r+\delta-1}$, but as we already noted 
    Fix $i \in \{\delta-1,\ldots, n-1\}$. Writing $g(X) = \sum_{j \in [\ell]\setminus(S_-\cup S_+)} g_jX^j$, we see that \begin{align*}
        g_i^{(1)}(X) &= \sum_{j \in [\ell]\setminus(S_-\cup S_+)} g_j \left(\omega^{-i}(\omega^iX)^j + \sum_{t=0}^{\delta-2} v_{i,t}\omega^{-t}(\omega^{t}X)^j \right)\\
        &= \sum_{j \in [\ell]\setminus(S_-\cup S_+)} g_j \left(\omega^{(j-1)i} + \sum_{t=0}^{\delta-2} v_{i,t}\omega^{(j-1)t}\right)X^j\\
        &= \sum_{j \in [\ell]\setminus(S_-\cup S_+)} g_j Q_i(\omega^{j-1})X^j.
    \end{align*} Since $j \not \in S_+$, by \Cref{cor:Q_b_finite_field_version_outside_finitely_many_char} (and \Cref{rem:Q_b_delta=2_finite_field} if $\delta = 2$), $Q_i(\omega^{j-1}) \neq 0$ so $g_i^{(1)}(X) \neq 0$. Since $i$ was arbitrary, $G \neq 0$.
    
    We now bound the number of roots of $G$ in terms of the number of roots of $f$. Fix a coset $A = \alpha \Omega_{n}$ and let $|\ev(f)|_{A}$ denote the Hamming weight of the restriction of $\ev(f)$ to $A$. Let \[N_{A} = |\{x \in A \mid f(x) = f(\omega x) = \cdots = f(\omega^{\delta-2}x) = 0\}|\] so $N_{A}$ counts the number of starting points of a consecutive block of $\delta-1$ zeros. For a fixed $x \in A$ which is counted by $N_{A}$, $g_i^{(1)}(x)$ vanishes precisely when $f(\omega^i x) = 0$. Each nonzero position blocks at most $\delta-1$ possible starts of a $(\delta-1)$-zero block so \[N_A \geq (n - (\delta-1)|\ev(f)|_A)_+.\] Now fix an \(x\in A\) counted by \(N_A\). Then
    \[
    f(x)=f(\omega x)=\cdots=f(\omega^{\delta-2}x)=0.
    \]
    For an allowed value \(i\in\{\delta-1,\ldots,n-1\}\), the polynomial \(g_i^{(1)}\)
    vanishes at \(x\) whenever
    \(
    f(\omega^i x)=0.
    \)
    Among the \(r\) allowed values of \(i\), at least \((r-|\ev(f)|_A)_+\) satisfy this
    condition, since there are only \(|\ev(f)|_A\) nonzero positions in the coset \(A\).
    Therefore the total root multiplicity contribution from the coset \(A\) is at
    least
    \[
    (r-|\ev(f)|_A)_+N_A
    \ge
    (r-|\ev(f)|_A)_+(n-(\delta-1)|\ev(f)|_A)_+
    =
    \psi(|\ev(f)|_A).
    \]
    
    Summing over all cosets, we obtain
    \[
    |\{\text{roots of }G\text{ in }\F_q^*\text{, counted with multiplicity}\}|
    \ge
    \sum_{A\in\F_q^*/\Omega_n}\psi(|\ev(f)|_A).
    \]
    By the convexity of \(\psi\) as proven in \Cref{lem:psi_convex_monotone}, Jensen's inequality gives
    \[
    \sum_{A\in\F_q^*/\Omega_n}\psi(|\ev(f)|_A)
    \ge
    \frac{q-1}{n}
    \psi\left(\frac{n|\ev(f)|}{q-1}\right).
    \]
    Since \(G\) is nonzero and
    \(
    \deg G\le r(\ell-1),
    \)
    we get
    \[
    \frac{q-1}{n}
    \psi\left(\frac{n|\ev(f)|}{q-1}\right)
    \le
    r(\ell-1) \iff
    \psi\left(\frac{n|\ev(f)|}{q-1}\right)
    \le
    rn\frac{\ell-1}{q-1}.
    \]
    By \Cref{lem:psi_convex_monotone}, the function \(\psi\) is nonincreasing on
    \([0,\infty)\). Moreover, on the interval where \(\psi(t)>0\), it is given by
    \[
    \psi(t)=(r-t)(n-(\delta-1)t).
    \]
    Let \(y\) be the smaller solution of
    \[
    (r-y)(n-(\delta-1)y)=rn\frac{\ell-1}{q-1}.
    \]
    Then \(\psi(t)>rn\frac{\ell-1}{q-1}\) for \(0\le t<y\), while
    \(\psi(t)\le rn\frac{\ell-1}{q-1}\) for \(t\ge y\). Since
    \[
    \psi\left(\frac{n|\ev(f)|}{q-1}\right)\le rn\frac{\ell-1}{q-1},
    \]
    we conclude that
    \[
    \frac{n|\ev(f)|}{q-1}\ge y.
    \]
    Let
    \[
    \lambda:=\frac{\ell-1}{q-1}.
    \]
    Solving the quadratic in terms of $y$, substituting, and rearranging gives
    \[
    |\ev(f)|
    \ge
    \frac{q-1}{n}\cdot
    \frac{
    n+(\delta-1)r-
    \sqrt{(n+(\delta-1)r)^2-4(\delta-1)rn\left(1-\frac{\ell-1}{q-1}\right)}
    }{2(\delta-1)}.
    \]
    Simplifying, we obtain
    \[
    |\ev(f)|
    \ge
    \frac{q-1}{2}
    \left(
    \frac{1}{\delta-1}
    +
    \frac{r}{n}
    -
    \sqrt{
    \left(\frac{r}{n}-\frac{1}{\delta-1}\right)^2
    +
    \frac{4r}{(\delta-1)n}
    \cdot\frac{\ell-1}{q-1}
    }
    \right)
    \]
    and substituting $n = r+\delta-1$ gives \labelcref{eqn:qtb_distance_bound}, as desired.
    \end{proof}
    We make some important remarks regarding the relation of \labelcref{eqn:qtb_distance_bound} to the bound presented in \cite[Theorem 62]{golowich2025quantum} and illustrate the necessity of \Cref{thm:Q_b_has_delta-1_roots} and \Cref{cor:Q_b_finite_field_version_outside_finitely_many_char}.
    \begin{remark}
    \label{rem:qtb.dist.bound.comparison}
        Setting $\delta = 2$ in \labelcref{eqn:qtb_distance_bound} and shifting $r+1$ to $r$ recovers the exact bound presented in \cite[Theorem 62]{golowich2025quantum}.
    \end{remark}

    \begin{remark}
        Since we require \Cref{cor:Q_b_finite_field_version_outside_finitely_many_char} and \Cref{rem:Q_b_delta=2_finite_field} to hold for $(r,\delta)$, we must first fix $\delta$, then $r$, and then compute $\mathcal{M}_{r,\delta}$. After that we choose $q$, a prime power such that $(r+\delta-1) \mid (q-1)$ and $\mathrm{char}(\F_q)\nmid \mathcal{M}_{r,\delta}$. There are finitely many characteristics to exclude and, by Dirichlet's theorem on arithmetic progressions \cite{Dirichlet1837}, there are infinitely many primes of the form $p \equiv 1 \pmod{r+\delta-1}$. Hence, there are  infinitely many prime powers $q$ such that $(r+\delta-1) \mid (q-1)$ and $\mathrm{char}(\F_q)\nmid \mathcal{M}_{r,\delta}$. In particular, one may take $q = p$ for infinitely many such primes.
    \end{remark}

    \begin{remark}\label{rem:Q_b_prime_uncertainty_vs_composite}
        There is a shorter proof of \Cref{cor:Q_b_finite_field_version_outside_finitely_many_char} in the case that \(n=r+\delta-1\) is prime and \Cref{cor:uncertainty_prime_finite_field} holds over \(\F_q\). Indeed, for
        \[
        Q_b(Y)=Y^b+\sum_{t=0}^{\delta-2} v_tY^t,
        \]
        we have \(|Q_b|\le \delta.\)
        Hence, by \Cref{cor:uncertainty_prime_finite_field},
        \[
        |\ev(Q_b)|_{\Omega_n}\ge n+1-\delta=r,
        \]
        so the number of roots of \(Q_b\) in \(\Omega_n\) is at most \(n-r = \delta -1\).
        Since \(Q_b\) was constructed to vanish at $\omega^t$ for $t = 0,\ldots, \delta-2$ it follows that these are its only roots in \(\Omega_n\).
        
        However, this argument is insufficient for our purposes for two reasons. First, it requires \(n\) to be prime, whereas \Cref{thm:Q_b_has_delta-1_roots} and \Cref{cor:Q_b_finite_field_version_outside_finitely_many_char} hold for arbitrary \(n=r+\delta-1\), outside finitely many characteristics. Second, the composite-order uncertainty principle from \Cref{cor:uncertainty_composite_finite_field} is generally too weak to recover the same conclusion. Indeed, if \(d_1<d_2\) are the consecutive divisors of \(n\) such that \(d_1\leq |Q_b| \leq d_2\), then it yields only
        \[
        |\ev(Q_b)|_{\Omega_n}\ge \frac{n}{d_1d_2}\left(d_1+d_2-|Q_b|\right),
        \]
        which in general does not imply that the number of roots of \(Q_b\) in \(\Omega_n\) is at most \(\delta-1\).
        
        This distinction becomes especially important in the hierarchical setting. In order to derive a distance bound which meaningfully utilizes each level, one would like to apply the analogue of \Cref{cor:Q_b_finite_field_version_outside_finitely_many_char} at each level \(l\), with \(n_l = r_l + \delta_l - 1\).
        Using \Cref{cor:uncertainty_prime_finite_field} would, therefore, force one to assume that each \(n_l\) is prime, which is prohibitively restrictive. In particular, once $n_1$ is prime, the hierarchy necessarily collapses to one level. By contrast, \Cref{cor:Q_b_finite_field_version_outside_finitely_many_char} avoids any primality assumption on the \(n_l\), and for this reason is better suited to the hierarchical constructions developed later.
    \end{remark}
    
    We obtain the asymptotic distance of \((r,\delta)\) QTBs as an immediate corollary of \Cref{thm:qtb_distance}:
        \begin{corollary}\label{cor:asymptotic_qtb_distance}
            Let \(\delta\ge 2\) and \(r\ge \delta\), and set \(n=r+\delta-1\).
            Assume the following admissibility condition:
            \[
            \begin{cases}
            \operatorname{char}(\F_q)\nmid \mathcal M_{r,\delta}, & \text{if } \delta\ge 3,\\
            r+1 \text{ is prime}, & \text{if } \delta=2.
            \end{cases}
            \]
            Then, for every
            \[
            0<R<\frac{r-\delta+1}{r+\delta-1},
            \]
            there exists an explicit family of QLRCs over finite fields \(\F_q\) satisfying
            \(n\mid(q-1)\) and the admissibility condition above, with locality
            \((r,\delta)\), rate at least \(R\), and relative distance at least\begin{multline}\label{eqn:asymptotic_qtb_distance} \frac{1}{2}\Biggl(\frac{1}{\delta-1} + \frac{r}{r+\delta-1} \\- \sqrt{\left(\frac{r}{r+\delta-1} - \frac{1}{\delta-1}\right)^2 + \frac{2r}{(\delta-1)(r+\delta-1)}\ \left(1 + R\cdot \frac{r+\delta-1}{r-\delta + 1}\right)}\Biggr).\end{multline}
        \end{corollary}
        \begin{proof}
            By \Cref{lem:dim.r.delta.TB}, the dimension of $\mathcal{Q}$ is \[k = 1 + (2\ell -q)\left(1 - \frac{2(\delta-1)}{r + \delta-1}\right) + \epsilon\] for some $\epsilon \in [-2(\delta-1),2(\delta-1)]$. Let $k = R(q-1)$. Then \begin{align*}
                \frac{R(q-1) - 1 - \epsilon}{1 - \frac{2(\delta-1)}{r + \delta-1}} = 2\ell -2 - q + 2 = 2(\ell-1) - (q-2).
            \end{align*} Dividing by $(q-1)$, \begin{align*}
                \frac{R - \frac{1+\epsilon}{q-1}}{1 - \frac{2(\delta-1)}{r + \delta-1}} = 2\left(\frac{\ell-1}{q-1}\right) - \frac{q-2}{q-1} \implies \frac{\ell-1}{q-1} = \frac{\frac{q-2}{q-1} + \frac{R - \frac{1+\epsilon}{q-1}}{1 - \frac{2(\delta-1)}{r + \delta-1}}}{2}
            \end{align*} 
            and \begin{align}\label{eqn:ell-1/q-1}
                \frac{\ell-1}{q-1} &= \frac{\frac{q-2}{q-1} + \frac{R - \frac{1+\epsilon}{q-1}}{1 - \frac{2(\delta-1)}{r + \delta-1}}}{2}\nonumber\\ &= \frac{1}{2} - \frac{1}{2(q-1)}\left(1 + (1+\epsilon)\frac{r + \delta - 1}{r -\delta + 1}\right) + \frac{R}{2}\cdot \frac{r+\delta-1}{r - \delta + 1}.
            \end{align}Hence, the relative distance is at least \begin{multline*}
                \frac{1}{2}\Biggl(\frac{1}{\delta-1} + \frac{r}{r+\delta-1} \\- \sqrt{\left(\frac{r}{r+\delta-1} - \frac{1}{\delta-1}\right)^2 + \frac{2r}{(\delta-1)(r+\delta-1)}\cdot \left(1 + R \frac{r+\delta-1}{r-\delta + 1}\right)}\Biggr)
            \end{multline*} as $q \to \infty$.
        \end{proof}
        
        Note that for fixed $\delta \geq 2$, as $r \to \infty$, \labelcref{eqn:asymptotic_qtb_distance} tends to\[
            \frac{1}{2}\left(1 + \frac{1}{\delta-1}  - \sqrt{ \left(1 - \frac{1}{\delta-1}\right)^2 + \frac{2(1+R)}{\delta-1}}\right) = \frac{1}{2}\left(\frac{\delta}{\delta-1}  - \sqrt{ \left(\frac{\delta-2}{\delta-1}\right)^2 + \frac{2(1+R)}{\delta-1}}\right).\]
\subsection{Non-vanishing of complete homogeneous symmetric polynomial}
In this section, we present the technical lemma used to conclude the non-vanishing of \(Q_b\) in \Cref{thm:Q_b_has_delta-1_roots}.

\begin{lemma}\label{lem:homogenous_poly_roots_of_unity}
    Let \(d \geq 3\) and let \(n \geq 2d-1\). Let $\zeta \in \C$ be a primitive \(n\)th root of unity and write \(h_t(x_1,\ldots, x_d)\) for the complete homogeneous symmetric polynomial of degree \(t\) in \(d\) variables. Then, for every \(0 \leq t \leq n-d\) and \(d-1 \leq s \leq n-1\), we have \[h_{t}(1,\zeta,\ldots,\zeta^{d-2},\zeta^s) \neq 0.\]
\end{lemma}
\begin{proof}
    It suffices to prove the result for \(\zeta = e^{2\pi i /n}\) as any other primitive \(n\)th root of unity is obtained from this one by a Galois automorphism of \(\Q(\zeta)\) and such automorphisms preserve non-vanishing.
    
    Set $W = (1,\zeta,\ldots, \zeta^{d-2})$ and fix \(s \in \{d-1,\ldots, n-1\}\). Define \[F_s(T) = \sum_{j\geq 0}h_j(W,\zeta^s)T^j.\]

     Using the generating function for complete homogeneous symmetric polynomials and the identity \[\prod_{j=0}^{n-1}(1-\zeta^jT)=1-T^n,\]
     we get \[
     F_s(T)= \frac{1}{(1-\zeta^sT)\displaystyle\prod_{u=0}^{d-2}(1-\zeta^uT)}=
     %\frac{1}{(1-T)(1-\zeta T)\cdots(1-\zeta^{\delta-2}T)(1-\zeta^sT)}=
     \frac{\displaystyle\prod_{\substack{u=d-1\\ u\neq s}}^{n-1}(1-\zeta^uT)}{1-T^n}\]
     The numerator has degree $(n-1)-(d-2) - 1 = n-d < n$. Hence, for every $0\le j\le n-d$, the coefficient of $T^j$ is unaffected by the factor $(1-T^n)^{-1}$, and therefore \[ h_j(W,\zeta^s) = [T^j]\prod_{\substack{u=d-1\\ u\neq s}}^{n-1}(1-\zeta^u T).\]
     Writing \[ \prod_{\substack{u=d-1\\ u\neq s}}^{n-1}(1-\zeta^u T) =\sum_{j=0}^{n-d} c_j(s)T^j,\] it follows that $c_j(s) = h_j(W,\zeta^s)$.
    
     We next record a symmetry of the coefficients $c_j(s)$. Let $\Lambda_s = (\zeta^{d-1},\ldots, \widehat{\zeta^s}, \ldots, \zeta^{n-1})$ be the vector of length $n-d$ comprised of the consecutive roots $\zeta^{d-1},\ldots, \zeta^{n-1}$, but excluding $\zeta^s$. Since \[ \prod_{\lambda\in\Lambda_s}(1-\lambda T) =\sum_{j=0}^{n-d}(-1)^j e_j(\Lambda_s)T^j,\] we have \[ c_j(s)=(-1)^j e_j(\Lambda_s)\] where $e_j$ is the elementary symmetric polynomial.
     Because every $\lambda\in\Lambda_s$ satisfies $|\lambda|=1$, we have $\lambda^{-1}=\overline{\lambda}$, and hence \[ e_{n-d-j}(\Lambda_s) = e_{n-d}(\Lambda_s)\,e_j(\Lambda_s^{-1}) = e_{n-d}(\Lambda_s)\,\overline{e_j(\Lambda_s)}.\] Therefore \[ c_{n-d-j}(s) = (-1)^{n-d-j}e_{n-d-j}(\Lambda_s) = (-1)^{n-d} e_{n-d}(\Lambda_s)\,\overline{c_j(s)}. \] In particular, \[ c_{n-d-j}(s)=0 \iff c_j(s)=0 \] because \[e_{n-d}(\Lambda_s) = \prod_{\lambda \in \Lambda_s}\lambda =  \prod_{\substack{u=d-1\\ u\neq s}}^{n-1} \zeta^u = \zeta^{\sum_{u=d-1, u\neq s}^{n-1} u} \neq 0.\] Since $c_j(s)=h_j(W,\zeta^s)$, it follows that \[
     h_{n-d-j}(W,\zeta^s)=0 \iff h_j(W,\zeta^s)=0. \]
     Thus, it suffices to prove that the polynomial $h_t(W,\zeta^s) \neq 0$ for \[ 0\le t\le \left\lfloor \frac{n-d}{2}\right\rfloor. \]
    
     Fix such a $t$. If \(t = 0\), then $h_0(W,\zeta^s) = 1$ so there is nothing to prove. Hence, assume \(t \geq 1\) and consider the polynomial \[ f_t(Y):=h_t(W,Y)=h_t(1,\zeta,\dots,\zeta^{d-2},Y). \]
     By the defining recurrence for complete homogeneous symmetric polynomials,
     \[ f_t(Y)=\sum_{u=0}^{t} h_{t-u}(W)\,Y^u. \]
     Additionally, by definition \begin{align*}
          h_j(W) = h_j(1,\ldots, \zeta^{d-2}) = \sum_{a_0 + \cdots + a_{d-2} = j} (1)^{a_0}(\zeta)^{a_1} \cdots (\zeta^{d-2})^{a_{d-2}},
      \end{align*} but we can also interpret this combinatorially. We have $d-1$ objects, labeled $0,\ldots, d-2$ and we choose $a_0$ copies of $0$, $a_1$ copies of $1$, \ldots, $a_{d-2}$ copies of $d-2$ so in total we have \[j = a_0 + a_1 + \cdots + a_{d-2}\] objects. Hence, we have $j$ objects each of size at most $d -2$ and this is the standard interpretation of the Gaussian binomial \[\binom{j + d -2}{j}_\zeta.\]
      Now we can simplify the Gaussian binomial by substituting $\zeta = e^{2 \pi i /n}$: \begin{align*}
         \binom{j + d -2}{j}_\zeta 
         &= \prod_{k = 1}^{j} \frac{1-\zeta^{d - 2 + k}}{1-\zeta^k}\\
         &= \prod_{k = 1}^{j} \frac{1-(e^{2 \pi i /n})^{d - 2 + k}}{1-(e^{2 \pi i /n})^k}\\
         &= \zeta^{\frac{j(d-2)}{2}}\prod_{k=1}^j \frac{\sin\left(\frac{\pi(d - 2 +k)}{n}\right)}{\sin\left(\frac{\pi k}{n}\right)}\\
         &= \zeta^{\frac{j(d-2)}{2}}\frac{\prod_{k=d-1}^{d-2+j}\sin\left(\frac{\pi k}{n}\right)}{\prod_{k=1}^j\sin\left(\frac{\pi k}{n}\right)}\\
         &= \zeta^{\frac{j(d-2)}{2}}\frac{\prod_{k=1}^{d-2+j}\sin\left(\frac{\pi k}{n}\right)}{\prod_{k=1}^j\sin\left(\frac{\pi k}{n}\right) \prod_{k=1}^{d-2} \sin\left(\frac{\pi k}{n}\right)}\\
         &= \zeta^{\frac{j(d-2)}{2}}\prod_{k=1}^{d-2} \frac{\sin\left(\frac{\pi(j +k)}{n}\right)}{\sin\left(\frac{\pi k}{n}\right)}
         \end{align*} 
     and get \[h_j(W) = \zeta^{\frac{j(d-2)}{2}} \rho_j\] for \[\rho_j = \prod_{k=1}^{d-2}\frac{\sin\left(\frac{\pi(j +k)}{n}\right)}{\sin\left(\frac{\pi k}{n}\right)}.\] The coefficient $\rho_j > 0$ because $0 \leq j \leq n-d$ and $1 \leq j + k \leq  n-d + d-2 \leq n-2$ so both sine terms in each fraction are positive.

     Substituting this expression into $f_t(Y)$ gives \[ f_t(Y) = \sum_{u=0}^t h_{t-u}(W) \, Y^u =\sum_{u=0}^t \zeta^{\frac{(t-u)(d-2)}{2}} \rho_{t-u}\, Y^u = \zeta^{\frac{t(d-2)}{2}}\rho_t \sum_{u=0}^t \zeta^{\frac{-u(d-2)}{2}}\frac{\rho_{t-u}}{\rho_t}\,Y^u.\] Let \[ P_t(z):=\sum_{u=0}^{t} p_u z^u, \qquad p_u:=\frac{\rho_{t-u}}{\rho_t}\] so \[f_t(Y) =\zeta^{\frac{t(d-2)}{2}}\rho_t P_t(\zeta^{-\frac{d-2}{2}}Y).\] Thus the zeros of $f_t(x)$ are exactly the zeros of $P_t(z)$, rotated by the unit scalar $\zeta^{(d-2)/2}$.
    
     We now show that the coefficients of $P_t$ are positive reals such that \[1= p_0> p_1 > \cdots > p_t > 0.\] Indeed, \[\frac{\rho_{j+1}}{\rho_j} = \frac{\sin\left(\frac{\pi(j +d-1)}{n}\right)}{\sin\left(\frac{\pi (j+1)}{n}\right)}\] by telescoping. If $0\le j\le t-1$, then \[2j + d \leq 2t + d -2\] and since $t\leq \lfloor (n-d)/2\rfloor$, we have \[2t + d -2\leq n -2 < n.\] Therefore, \[ 0<\frac{(2j+d)\pi}{2n}<\frac{\pi}{2}, \] and using \[ \sin A-\sin B = 2\cos\!\left(\frac{A+B}{2}\right)\sin\!\left(\frac{A-B}{2}\right), \] with \[ A=\frac{\pi (j+d-1)}{n},\qquad B=\frac{\pi(j+1)}{n},\] we obtain \[ \sin\!\left(\frac{\pi(j+d-1)}{n}\right) - \sin\!\left(\frac{\pi(j+1)}{n}\right) = 2\cos\!\left(\frac{\pi(2j+d)}{2n}\right) \sin\!\left(\frac{\pi(d-2)}{2n}\right) >0. \] Hence, $\rho_{j+1}>\rho_j$ for $0 \leq j \leq t-1$. It follows that \[\rho_t > \rho_{t-1} > \cdots > \rho_0>0\] and dividing by $\rho_t$, since $p_u = \rho_{t-u}/\rho_t$, we obtain \[1 = p_0 > p_1 > \cdots > p_t> 0.\]
    
     Now apply the Enestr\"om--Kakeya theorem \cite[Theorem 4]{Gardner2014} to $P_t(z)$. Since every $p_u$ is a positive real and $1 = p_0 > p_1 > \cdots > p_t > 0$, every zero $z^*$ of $P_t$ satisfies \[ |z^*| \geq \min_{0 \leq u \leq t-1} \frac{p_u}{p_{u+1}} = \min_{0 \leq u \leq t-1} \frac{\rho_{t-u}}{\rho_{t-(u+1)}} = \min_{1 \leq u \leq t}\frac{\rho_{u}}{\rho_{u-1}} > 1 .\] Thus every zero of $P_t$ lies strictly outside the unit disk. Since \[f_t(Y) =\zeta^{\frac{t(d-2)}{2}}\rho_t P_t(\zeta^{-\frac{d-2}{2}}Y)\] the same is true for $f_t(Y)$: if $z^*$ is a zero of $P_t$, then $\zeta^{(d-2)/2}z^*$ is a zero of $f_t$ and \[|\zeta^{(d-2)/2}z^*| = |\zeta|^{\frac{d-2}{2}} |z^*| = |z^*| > 1.\]
    
     Since $|\zeta^s|=1$, $\zeta^s$ cannot be a zero of $f_t$. Hence \[h_t(W,\zeta^s) = f_t(\zeta^s) \neq 0\] for $0 \leq t \leq \lfloor (n-d)/2\rfloor$ and $d - 1 \leq s \leq n-1$. By the symmetry \[h_{n-d-t}(W,\zeta^s) \neq 0 \iff h_t(W,\zeta^s) \neq 0,\] for $0 \leq t \leq \lfloor (n-d)/2\rfloor$, this extends to all $0\leq t \leq n-d$. Therefore, $h_t(W,\zeta^s) \neq 0$ for $0 \leq t \leq n-d$ and $d - 1 \leq s \leq n-1$. Since \[h_{t}(1,\zeta, \ldots, \zeta^{d-2}, \zeta^s) = h_t(W, \zeta^s)\] we conclude $h_{t}(1,\zeta, \ldots, \zeta^{d-2}, \zeta^s)\neq 0$, as desired.
\end{proof}

\begin{remark}\label{rem:homogeneous_delta=2}
    \Cref{lem:homogenous_poly_roots_of_unity} is not true if $d = 2$. Let \(d = 2\) and suppose $n = 6$ and $t = 2$. Then, \[h_2(1,\zeta^s) = 1 + \zeta^s + \zeta^{2s} = \frac{\zeta^{3s}-1}{\zeta^s - 1}\] and for $s = 2$, this quantity is zero.

    The proof particularly breaks when we define $\rho_j$. For $d = 2$, all $\rho_j = 1$ so all $p_u = 1$ and $P_t(z) = 1 + z + \cdots + z^t$. Since, all the coefficients of $P_t$ are equal to $1$, every root of $P_t$ has modulus $1$ by Enestr\"om-Kakeya. In our setting, all the roots of $P_t$ are the nontrivial $(t+1)$th roots of unity. Therefore, if $s$ is chosen such that $n\mid s(t+1)$, $\zeta^s$ is a root.

    In the case of $d = 2$, if we strengthen the hypothesis of \Cref{lem:homogenous_poly_roots_of_unity} to choosing \(t\) such that \(\gcd(t+1,n) = 1\), then the claim will follow.
\end{remark}
    
\section{\texorpdfstring{$h$-level Quantum Tamo--Barg codes}{h-level Quantum Tamo--Barg codes}}
\label{sec:r_delta_hierarchical_qTB}

We now extend the $(r,\delta)$ quantum Tamo--Barg code to an $h$-level quantum hierarchical LRC. Let $r_l,\delta_l \in \mathbb{Z}^+$ for $l = 1, \ldots, h$. Consider the following sets 
\begin{align*}
    S_{l,+} &= \bigcup_{j_l=1}^{\delta_l-1} (j_l+(r_l+\delta_l-1)\Z)\\
    S_{l,-} &= \bigcup_{j_l=1}^{\delta_l-1} (-j_l+(r_l+\delta_l-1)\Z)\end{align*}
and define \begin{align*}
    S_{+} = \bigcup_{l=1}^h S_{l,+} \quad\text{and} \quad
    S_{-} = \bigcup_{l=1}^h S_{l,-}.
\end{align*} 
\begin{definition}[$h$-level quantum Tamo--Barg code]\label{def:h_level_qtb} Given a prime $p$ and $m\in \mathbb{Z}^+$, let $q = p^m$. Given locality and distance parameters $(r_l,\delta_l)_{l=1,\ldots,h}$ such that \[r_1 \geq \cdots \geq r_h \geq \delta_1 \geq \cdots \geq \delta_h \geq 2\] and $n_h\mid n_{h-1}\mid \cdots \mid n_1\mid (q-1)$ where $n_l := r_l + \delta_l-1$, and an integer $q/2\leq \ell \leq q-1$, the $h$-level quantum Tamo--Barg code is defined to be the CSS code $\mathcal{Q} = \mathrm{CSS}(C,C)$ with $C = \ev(\mathbb{F}_q[X]^S)$ for \begin{align}S &= ([\ell] \cap ([q-1]\setminus S_{-})) \cup ([q-1] \cap S_{+}).\end{align}
\end{definition}
An immediate consequence of the definition of the sets $S_{l,+}$ and $S_{l,-}$ and \Cref{def:h_level_qtb} is the disjointness of the sets.
\begin{lemma}\label{lem:disjointness}
    Suppose we are given locality and distance parameters $(r_l,\delta_l)_{l=1,\ldots,h}$ such that \[r_1 \geq \cdots \geq r_h \geq \delta_1 \geq \cdots \geq \delta_h \geq 2\] and $r_{l+1} + \delta_{l+1} -1 \mid r_l + \delta_l - 1$ for $l = 1,\ldots, h-1$ and $r_1+\delta_1-1\mid q-1$. Then for $1 \leq u \leq v \leq h$, $S_{u,+} \cap S_{v, -} = \emptyset = S_{u,-} \cap S_{v,+}$.
\end{lemma}
\begin{proof}
    We break the proof into cases.
    \begin{enumerate}
        \item Let $n_l = r_l + \delta_1-1$. Suppose $u= v$. $|S_{u,+}\bmod n_u| = |S_{u,-}\bmod n_u| = \delta_u - 1$. Since \[2(\delta_u - 1) < r_u + \delta_u - 1 \iff \delta_u - 1 < r_u\] is satisfied by hypothesis, $S_{u,+} \cap S_{u,-} = \emptyset$.
        \item Suppose $u < v$. By hypothesis, $r_u \geq r_v \geq \delta_u \geq \delta_v$ and $r_v + \delta_v - 1 \mid r_u + \delta_u - 1$. Fix $s \in S_{v,-}$ so $s \equiv -j_v \mod(r_v + \delta_v - 1)$. For the sake of contradiction, suppose $s \in S_{u, +}$. Then $s \equiv j_u \mod(r_u + \delta_u -1)$. Since $r_v + \delta_v - 1\mid r_u + \delta_u - 1$, $ s \equiv j_u \mod r_v + \delta_v - 1$. Hence, $r_v + \delta_v - 1\mid j_u+j_v$ and $j_u + j_v \geq 2$ so $j_u + j_v$ is a nonzero multiple of $r_v + \delta_v - 1$. However, $j_u + j_v \leq \delta_u + \delta_v - 2$ and \[\delta_u + \delta_v - 2 < r_v + \delta_v -1 \iff \delta_u - 1 < r_v \iff \delta_u \leq r_v\] which is a contradiction.

        Now fix $s \in S_{v,+}$ and suppose for the sake of contradiction that $s \in S_{u,-}$. By similar reasoning, $s \equiv j_v \mod (r_v + \delta_v - 1)$ and $s \equiv -j_u \mod (r_v + \delta_v - 1)$. Again $r_v + \delta_v - 1\mid j_u + j_v$ so we have a contradiction.
    \end{enumerate}
    Thus, $1 \leq u \leq v \leq h$, $S_{u,+} \cap S_{v, -} = \emptyset = S_{u,-} \cap S_{v,+}$.
\end{proof}
We now present the hierarchical analogue of \Cref{lem:qtb.construction}. For \(l=1,\ldots,h\), define
\[
B_l:=\ev(\F_q[X]^{[q-1]\setminus S_{l,-}}),
\qquad
B_l^\perp:=\ev(\F_q[X]^{[q-1]\cap S_{l,+}}).
\]
Then
\[
B=\bigcap_{l=1}^h B_l
=
\ev(\F_q[X]^{[q-1]\setminus S_-}).
\]
Moreover, since the dual of an intersection is the sum of the duals,
\[
B^\perp
=
\left(\bigcap_{l=1}^h B_l\right)^\perp
=
\sum_{l=1}^h B_l^\perp.
\]
Equivalently, because the spaces \(B_l^\perp\) are monomial evaluation spaces,
\[
B^\perp
=
\sum_{l=1}^h \ev(\F_q[X]^{[q-1]\cap S_{l,+}})
=
\ev(\F_q[X]^{[q-1]\cap S_+}).
\]
\begin{lemma}\label{lem:hqtb.construction}
For $q, (r_l,\delta_l)_{l=1,\ldots, h}$, and $\ell$ as defined in \Cref{def:h_level_qtb} let 
\begin{equation}
    A = \ev(\F_q[X]^{[\ell]}),
\end{equation}
\begin{equation}
    B = \ev(\F_q[X]^{[q-1]\setminus S_{-}}).
\end{equation}
Then $A \cap B$ is an $h$-level TB code. Furthermore, 
\begin{equation}
    A^\perp = \ev(\F_q[X]^{[q-\ell]\setminus \{0\}}) \text { and }
\end{equation}
\begin{equation}\label{eqn:bperp_hqtb}
    B^\perp = \ev(\F_q[X]^{[q-1]\cap S_{+}}) \subseteq B.
\end{equation}
If $\ell \geq q/2$, then $A^\perp \subseteq A$. Letting $C = (A \cap B) + B^\perp$, we obtain that $\mathrm{CSS}(C,C)$ is an $h$-level QTB code with 
\begin{equation}
    C^\perp = (A^\perp \cap B) + B^\perp = \ev(\F_q[X]^T) \subseteq C
\end{equation}   
where
\begin{equation}
    T = \left(([q-\ell]\setminus \{0\}) \cap ([q-1]\setminus S_{-})\right) \cup ([q-1]\cap S_{+}).
\end{equation}
\end{lemma}
\begin{proof}

It suffices to show \(B^\perp\subseteq B\) because the remaining claims follow from \Cref{lem:qtb.construction}.  By
\Cref{lem:disjointness}, we have
\(
S_{u,+}\cap S_{v,-}=\emptyset
\)
for all \(1\le u,v\le h\).  Hence
\(
S_+\cap S_-=\emptyset.
\)
Thus every exponent in \([q-1]\cap S_+\) lies in \([q-1]\setminus S_-\), and so
\[
B^\perp
=
\ev(\F_q[X]^{[q-1]\cap S_+})
\subseteq
\ev(\F_q[X]^{[q-1]\setminus S_-})
=
B.
\]
\end{proof}

We now compute the dimension of the $h$-level quantum Tamo--Barg code defined in \Cref{def:h_level_qtb}.
\begin{lemma}\label{lem:dim_qhTB}
The $h$-level quantum Tamo--Barg code $\mathcal{Q} = \mathrm{CSS}(C,C)$ with parameters $q$, $\ell$, and $(r_l,\delta_l)_{l=1,\ldots,h}$ as defined in \Cref{def:h_level_qtb} has dimension
\begin{align*}k &= 1 + |\{q - \ell \leq i \leq \ell - 1 : i \not\in (S_{+} \cup S_{-})\}|\\
&= 1 + (2\ell - q)\left(1-2\sum_{l=1}^{h-1}\frac{\delta_l-\delta_{l+1}}{r_l+\delta_l-1} - 2\frac{\delta_h-1}{r_h+\delta_h-1}\right) + \epsilon\end{align*} for some \[\epsilon \in \left[-2(\delta_1-1),2(\delta_1-1)\right].\]
\end{lemma}

\begin{proof}
We use a similar counting argument as in the proof of \Cref{lem:dim.r.delta.TB}. Recall $S \subseteq [q-1]$ as defined in \Cref{def:h_level_qtb}, $T \subseteq S$ as defined in \Cref{lem:hqtb.construction} such that $C = \ev(\F_q[X]^{S})$ and $C^{\perp} = \ev(\F_q[X]^{T})$. Then, dimension of $\mathcal{Q}$ can be computed as follows:
\begin{align*}
    \dim(\mathcal{Q}) &= \dim(C) - \dim(C^{\perp}) \\
    &= |S \setminus T| \\
    &= \left| \{0\} \cup \{q-\ell \leq i \leq \ell-1 : i \not\in (S_{+} \cup S_{-})\}\right|.
\end{align*} This proves the first equality. For the second equality, we inspect \[\{q-\ell \leq i \leq \ell-1 : i \not\in (S_{+} \cup S_{-})\}.\] If $(r_1+\delta_1-1) \mid (2\ell-q)$ then this set has size exactly \[(2\ell - q)\left(1-2\sum_{l=1}^{h-1}\frac{\delta_l-\delta_{l+1}}{r_l+\delta_l-1} - 2\frac{\delta_h-1}{r_h+\delta_h-1}\right)\] because we exclude $2(\delta_h-1)$ elements in every block of size $(r_h+\delta_h-1)$ and as we move up the layers, we exclude another $2(\delta_l-\delta_{l+1})$ elements in every block of size $(r_l+\delta_l-1)$ for $l = 1,\ldots,h-1$. Note, this exclusion will leave a nontrivial number of elements precisely because \[2(\delta_h-1) < r_h + \delta_h -1 \iff \delta_h - 1< r_h \iff \delta_h \leq r_h\] and \begin{align*}
    2(\delta_l-\delta_{l+1}) < r_l + \delta_l -1 \iff (\delta_l-1) - 2(\delta_{l+1}-1) <r_l \iff \delta_l \leq r_l + 2(\delta_{l+1}-1)
\end{align*} which is satisfied by the assumption on $(r_l,\delta_l)_{l=1,\ldots,h}$. If $(r_1+\delta_1-1) \nmid (2\ell-q)$, we differ from \[(2\ell - q)\left(1-2\sum_{l=1}^{h-1}\frac{\delta_l-\delta_{l+1}}{r_l+\delta_l-1} - 2\frac{\delta_h-1}{r_h+\delta_h-1}\right)\] by at most \[2\sum_{l=1}^{h-1}(\delta_l-\delta_{l+1}) + 2(\delta_h-1) = 2(\delta_1-1)\] in absolute value and the second equality follows.
\end{proof}

We give a small example to illustrate the sizes of $[q-1]\cap S_{+}$ and $[q-1] \cap S_{-}$ which appear in the proof of \Cref{lem:dim_qhTB}.
\begin{example} Let $q = 49$ and let $(21,4), (10,3), (5,2)$ be the locality and distance parameters. It is easy to see that \[r_1\geq r_2\geq r_3 \geq \delta_1\geq \delta_2\geq \delta_3 \geq 2\] and \begin{align*}
    r_3+\delta_3-1 &= 6\\
    r_2 + \delta_2-1 &= 12\\
    r_1+\delta_1-1 &= 24
\end{align*} so the divisibility condition follows. Now we compute $[q-1]\cap S_{+}$ and $[q-1]\cap S_{-}$. The constituent sets are \begin{align*}
    [q-1] \cap S_{1,-} &= \bigcup_{j_1 =1}^3 (-j_1 + 24 \Z) = \{21,\underline{22,\boxed{23}},45,\underline{46,\boxed{47}}\}\\
    [q-1] \cap S_{1,+} &= \bigcup_{j_1 =1}^3 (j_1 + 24 \Z)= \{\underline{\boxed{1},2},3,\underline{\boxed{25},26},27\}\\
    [q-1] \cap S_{2,-} &= \bigcup_{j_2 =1}^2 (-j_2 + 12 \Z) = \{10,\boxed{11},\underline{22,\boxed{23}},34,\boxed{35},\underline{46,\boxed{47}}\}\\
    [q-1] \cap S_{2,+} &= \bigcup_{j_2 =1}^2 (j_2 + 12 \Z)= \{\underline{\boxed{1},2},\boxed{13},14,\underline{\boxed{25},26},\boxed{37},38\}\\
    [q-1] \cap S_{3,-} &= \bigcup_{j_3 =1}^1 (-j_3 + 6 \Z) = \{5,\boxed{11},17,\boxed{23},29,\boxed{35},41,\boxed{47}\}\\
    [q-1] \cap S_{3,+} &= \bigcup_{j_3 =1}^1 (j_3 + 6 \Z)= \{\boxed{1},7,\boxed{13},19,\boxed{25},31,\boxed{37},43\}
\end{align*} so \begin{align*}
    |[q-1]\cap S_{+}| = \left\vert\bigcup_{l=1}^3 [q-1]\cap S_{l,+}\right\vert = 2 + 4 + 8 = \frac{48}{24}(4-3) + \frac{48}{12}(3-2) + \frac{48}{6}(2-1)
\end{align*} and an analogous statement holds for $[q-1] \cap S_{-}$. Thus, \[|[q-1] \cap (S_- \cup S_+)| = 2(q-1)\left(\frac{\delta_1-\delta_2}{r_1+\delta_1-1} + \frac{\delta_2-\delta_3}{r_2+\delta_2-1} + \frac{\delta_3-1}{r_3+\delta_3-1}\right).\] This example shows why in level $l$, we only contribute an extra $(\delta_l-\delta_{l+1})$ elements per length $r_l + \delta_l-1$ blocks as $\delta_{l+1}-1$ of the $\delta_l-1$ elements have already been accounted for from levels $l+1,\ldots, h$.
\end{example}

We now exhibit the hierarchical locality of the code $\mathcal{Q} = \mathrm{CSS}(C,C)$ as defined in \Cref{def:h_level_qtb} by stating and proving statements similar to \Cref{lem:low.wt.parity.checks} and \Cref{cor:locality_qtb}.
\begin{lemma}\label{lem:low.wt.parity.checks.hqtb}
    Let
    \(
    B^\perp=\ev(\F_q[X]^{[q-1]\cap S_+})
    \)
    be as in \Cref{lem:hqtb.construction}.  If \(f\in B^\perp\), then $f$ can be written as
    \[
    f=f_1+\cdots+f_h,
    \qquad
    f_l\in B_l^\perp.
    \]
    For each level \(l\), each \(f_l\in B_l^\perp\) has the following local
    description: for every coset \(\alpha\Omega_{n_l}\), there exists a polynomial
    \(P_{\alpha,l}(X)\in\F_q[X]\) with
    \[
    \deg P_{\alpha,l}\le \delta_l-1,
    \qquad
    P_{\alpha,l}(0)=0,
    \]
    such that
    \[
    f_l(\alpha\omega)=P_{\alpha,l}(\omega)
    \qquad
    \text{for all }\omega\in\Omega_{n_l}.
    \] The converse is also true.
\end{lemma} It follows immediately that each level $l$ contributes $\delta_l-1$ independent local parity checks per level-$l$ coset. Additionally, this lemma reduces to \Cref{lem:low.wt.parity.checks} when $h=1$.
\begin{proof}
By dimension counting, it suffices to show the forward direction.

$(\implies)$ The equality
\[
B^\perp=\sum_{l=1}^h B_l^\perp
\]
was shown in \Cref{lem:hqtb.construction}.  Hence every \(f\in B^\perp\) can be
written as \(f=f_1+\cdots+f_h\) with \(f_l\in B_l^\perp\).

The final statement is exactly the one-level description of \(B_l^\perp\) from
\Cref{lem:low.wt.parity.checks}, applied with parameters
\((r_l,\delta_l)\) and \(n_l=r_l+\delta_l-1\).
\end{proof}

\begin{corollary}\label{cor:locality_hqtb}
The $h$-level $(r_l,\delta_l)_{l=1,\ldots,h}$ quantum Tamo--Barg code in \Cref{def:h_level_qtb} is an $h$-level $((r_1,\delta_1),\ldots,(r_h,\delta_h))$-QHLRC.
\end{corollary}

\begin{proof}
Let $\mathcal Q=\mathrm{CSS}(C,C)$ be the code from \Cref{def:h_level_qtb}, and for each $l$ write $n_l:=r_l+\delta_l-1$.

We prove by induction on $l$ that for every level-$(l-1)$ repair group $A_{l-1}$, the punctured code $\mathcal Q|_{A_{l-1}}$ is an $(h-l+1)$-level $((r_l,\delta_l),\ldots,(r_h,\delta_h))$-QHLRC. Here, for $l=1$, we interpret $A_0=\{1,\ldots,q-1\}$ and $\mathcal Q|_{A_0}=\mathcal Q$.

\textbf{Base step: $l=h$.}
Fix a level-$(h-1)$ repair group $A_{h-1}$. Inside $A_{h-1}$, the level-$h$ repair groups are precisely the cosets $A_h=\alpha\Omega_{n_h}\subseteq A_{h-1}$.

By \Cref{lem:hqtb.construction}, we have $C^\perp \supseteq B^\perp=\ev(\F_q[X]^{[q-1]\cap S_+})$, and by \Cref{lem:low.wt.parity.checks.hqtb}, for each $j=1,\ldots,\delta_h-1$, the code $B_h^\perp \subseteq B^\perp$ contains the function
\[
f_{\alpha,h,j}(x)=
\begin{cases}
x^j,&x\in A_h,\\
0,&x\notin A_h.
\end{cases}
\]
Restricting to the coordinates in $A_h$, the functions $f_{\alpha,h,j}(x)$ give $\delta_h-1$ checks in $(C|_{A_h})^\perp$. Writing $A_h=\{x_1,\ldots,x_{n_h}\}$, we obtain a Vandermonde-type matrix
\[
H_{A_h}=
\begin{bmatrix}
x_1 & x_2 & \cdots & x_{n_h}\\
x_1^2 & x_2^2 & \cdots & x_{n_h}^2\\
\vdots & \vdots & \ddots & \vdots\\
x_1^{\delta_h-1} & x_2^{\delta_h-1} & \cdots & x_{n_h}^{\delta_h-1}
\end{bmatrix}.
\]
After scaling the $i$th column by $x_i^{-1}$, this becomes a Vandermonde matrix, so any $\delta_h-1$ columns are linearly independent. Hence, $d(C|_{A_h})\ge \delta_h$ and since \[(C|_{A_{h-1}})^\perp = \sigma_{A_{h-1}}(C^\perp) \subseteq \pi_{A_{h-1}}(C^\perp) \subseteq \pi_{A_{h-1}}(C) = C|_{A_{h-1}},\] we have \[d((C|_{A_{h-1}})^\perp) \geq d(C|_{A_{h-1}}) \geq d( (C|_{A_{h-1}})|_{A_h}) = d(C|_{A_h}) \geq \delta_h.\] Therefore, by \Cref{thm:css_qlrc_iff_clrc}, $\mathcal Q|_{A_{h-1}}=\mathrm{CSS}(C|_{A_{h-1}},C|_{A_{h-1}})$ is an $(r_h,\delta_h)$-QLRC. This proves the base case.

\textbf{Induction step.}
Assume $1\le l<h$, and assume that for every level-$l$ repair group $A_l$, the punctured code $\mathcal Q|_{A_l}$ is an $(h-l)$-level $((r_{l+1},\delta_{l+1}),\ldots,(r_h,\delta_h))$-QHLRC.

Now fix a level-$(l-1)$ repair group $A_{l-1}$. Inside $A_{l-1}$, the level-$l$ repair groups are the cosets $A_l=\alpha\Omega_{n_l}\subseteq A_{l-1}$. By \Cref{lem:low.wt.parity.checks.hqtb}, for each such $A_l$ and each $j=1,\ldots,\delta_l-1$, the code $B_l^\perp \subseteq B^\perp\subseteq C^\perp$ contains the function
\[
f_{\alpha,l,j}(x)=
\begin{cases}
x^j,&x\in A_l,\\
0,&x\notin A_l.
\end{cases}
\]
Restricting to $A_l$, these give $\delta_l-1$ checks in $(C|_{A_l})^\perp$. As in the base case, their parity-check matrix is Vandermonde after column scaling, so any $\delta_l-1$ columns are linearly independent and therefore $d(C|_{A_l})\ge \delta_l$. By the same reasoning as before $d((C|_{A_{l-1}})^\perp) \geq \delta_l$ so by \Cref{thm:css_qlrc_iff_clrc}, $ \mathcal Q|_{A_{l-1}}$ is an $(r_l,\delta_l)$-QLRC.

Moreover, by the induction hypothesis, for every such level-$l$ repair group $A_l$, the punctured code \(\mathcal Q|_{A_l}\) is an $(h-l)$-level $((r_{l+1},\delta_{l+1}),\ldots,(r_h,\delta_h))$-QHLRC. Therefore, \(\mathcal Q|_{A_{l-1}}\) satisfies the two conditions of \Cref{def:h_level_qhlrc} with top-level locality parameters \((r_l,\delta_l)\). Hence, \(\mathcal Q|_{A_{l-1}}\) is an $(h-l+1)$-level $((r_l,\delta_l),\ldots,(r_h,\delta_h))$-QHLRC.

Finally, taking $l=1$ and $A_0=\{1,\ldots,q-1\}$, we conclude that $\mathcal Q$ is indeed an $h$-level $((r_1,\delta_1),\ldots,(r_h,\delta_h))$-QHLRC.
\end{proof}

\subsection{Bound on the minimum distance of hierarchical Quantum Tamo--Barg codes}
We finish our analysis of the quantum Tamo--Barg HLRC by proving a distance bound using \Cref{thm:qtb_distance}. For clarity, we will first derive a distance bound for the case $h = 2$. The general $h$ derivation will naturally follow, but the notation will become quite heavy. Recall \Cref{lem:psi_convex_monotone} whose properties we will use repeatedly in the distance proofs.
\PsiConvexMonotone*
\subsubsection{Two-level bound}
Now we present the distance bound of the $2$-level hierarchical QTB code. The proof proceeds by iterating the one-level root-counting argument. At the bottom level, for each admissible shift \(i_2\), we apply a \(Q\)-operator to \(f\) and obtain a transformed polynomial \(g^{(2)}_{i_2}\). The \(Q_b\)-nonvanishing theorem guarantees that these transforms are nonzero on the nondual part of \(f\), while the defining vanishing of \(Q_b\) forces \(g^{(2)}_{i_2}\) to vanish whenever \(f\) has a suitable block of \(\delta_2-1\) zeros in a bottom-level repair group together with one additional zero.

We then apply the same idea at the top level. For each admissible top-level shift \(i_1\), we apply a level-\(1\) \(Q\)-operator to \(g^{(2)}_{i_2}\), producing \(g^{(1)}_{i_1,i_2}\). Taking the product over all choices of \(i_1\) and \(i_2\) gives a nonzero polynomial \(G\), whose degree is at most \(r_1r_2(\ell-1)\). Thus, an upper bound on the number of roots of \(G\) comes from its degree.

The lower bound on the number of roots is obtained by counting zero incidences level by level. Inside each bottom coset \(B\), if \(w_B\) is the weight of \(f\) on \(B\), the bottom-level argument contributes at least
\[
\psi_2(w_B)=(r_2-w_B)_+\left(n_2-(\delta_2-1)w_B\right)_+
\]
zero incidences. Averaging these incidences over the bottom cosets inside a fixed top coset gives an upper bound on the average weight of the transforms \(g^{(2)}_{i_2}\) on that top coset. The top-level root-counting argument then contributes
\[
\psi_1(t)=(r_1-t)_+\left(n_1-(\delta_1-1)t\right)_+
\]
roots as a function of this averaged weight. Finally, Jensen's inequality is used twice, once across the choices of \(i_2\) and once across the top cosets, to express the resulting lower bound only in terms of the total weight \(|\ev(f)|\). Comparing this lower bound with \(\deg G\le r_1r_2(\ell-1)\) yields the desired distance estimate.

\begin{theorem}[Two-level hierarchical QTB distance bound]\label{thm:two_level_hqtb_distance}
Let \(\mathcal Q=\mathrm{CSS}(C,C)\) be a two-level hierarchical QTB code with parameters \((r_1,\delta_1), (r_2,\delta_2)\)
and set
\[
n_1:=r_1+\delta_1-1,
\qquad
n_2:=r_2+\delta_2-1.
\]
Assume
\(
n_2\mid n_1\mid(q-1)
\).
Assume also that \Cref{cor:Q_b_finite_field_version_outside_finitely_many_char} holds at both levels; concretely, for every level \(l=1,2\), either \(\delta_l\ge 3\) and
\(
\mathrm{char}(\F_q)\nmid \mathcal M_{r_l,\delta_l},
\)
or \(\delta_l=2\) and \(n_l=r_l+1\) is prime.

For \(l=1,2\), define
\[
\psi_l(t):=(r_l-t)_+\left(n_l-(\delta_l-1)t\right)_+.
\]
Define
\[
\Theta_2(t):=
\frac{q-1}{n_1}\,r_2\,
\psi_1\left(
n_1-\frac{n_1}{r_2n_2}\psi_2(t)
\right),
\qquad 0\le t\le n_2.
\]
Let
\[
\tau_2:=
\inf\left\{
t\in[0,n_2]:
\Theta_2(t)\le r_1r_2(\ell-1)
\right\}.
\]
Then
\[
d(\mathcal Q)\ge \frac{q-1}{n_2}\tau_2.
\]
Equivalently, every nonzero codeword \(\ev(f)\in C\setminus C^\perp\) of weight
\[
w:=|\ev(f)|
\]
satisfies
\[
r_1r_2(\ell-1)
\ge
\frac{q-1}{n_1}\,r_2\,
\psi_1\left(
n_1-\frac{n_1}{r_2n_2}
\psi_2\left(\frac{n_2}{q-1}w\right)
\right).
\]
\end{theorem}

\begin{proof}
Fix \(\ev(f)\in C\setminus C^\perp\). We decompose
\[
f=g+P_1+P_2,
\]
where \(g\) is supported outside \(S_{+}\cup S_{-}\), \(P_2\) is supported on \(S_{2,+}\), and \(P_1\) is supported on \(S_{1,+}\setminus S_{2,+}\). Since \(\ev(f)\notin C^\perp\), we have \(g\neq 0\).

For \(l=1,2\), let
\[
I_l:=\{\delta_l-1,\dots,n_l-1\}.
\]
For each \(i_l\in I_l\), let
\[
Q^{(l)}_{i_l}(Y)
=
Y^{i_l}+\sum_{t=0}^{\delta_l-2}v^{(l)}_{i_l,t}Y^t
\]
be the corresponding \(Q\)-polynomial at level \(l\). Let \(\omega_l\) be a primitive \(n_l\)th root of unity. Define the linear operator
\[
\mathcal L^{(l)}_{i_l}[p](X)
:=
\omega_l^{-i_l}p(\omega_l^{i_l}X)
+
\sum_{t=0}^{\delta_l-2}
v^{(l)}_{i_l,t}\omega_l^{-t}p(\omega_l^tX).
\]
On monomials, this operator acts diagonally:
\[
\mathcal L^{(l)}_{i_l}[X^j]
=
Q^{(l)}_{i_l}(\omega_l^{j-1})X^j.
\]

Define the level-2 transforms
\(
g^{(2)}_{i_2}:=\mathcal L^{(2)}_{i_2}[f],
\, i_2\in I_2,
\)
and the level-1 transforms
\(
g^{(1)}_{i_1,i_2}:=\mathcal L^{(1)}_{i_1}[g^{(2)}_{i_2}],
\,i_1\in I_1,\ i_2\in I_2.
\)
Finally define the aggregate polynomial
\[
G(X):=\prod_{i_2\in I_2}\prod_{i_1\in I_1}g^{(1)}_{i_1,i_2}(X).
\]

We first show that \(G\neq 0\). The operator \(\mathcal L^{(2)}_{i_2}\) ensures that the coefficient of every monomial whose exponent lies in \(S_{2,+}\) is zero, so
\[
g^{(2)}_{i_2}
=
\mathcal L^{(2)}_{i_2}[g]
+
\mathcal L^{(2)}_{i_2}[P_1].
\]
The two summands have disjoint supports. Since every exponent in the support of \(g\) does not lie in \(S_{2,+}\), \Cref{cor:Q_b_finite_field_version_outside_finitely_many_char} at level \(2\) implies
\(
\mathcal L^{(2)}_{i_2}[g]\neq 0
\). Hence \(g^{(2)}_{i_2}\neq 0\).

Next, \(\mathcal L^{(2)}_{i_2}[P_1]\) is still supported on \(S_{1,+}\) so applying the operator \(\mathcal L^{(1)}_{i_1}\) ensures that it is zero. Therefore,
\[
g^{(1)}_{i_1,i_2}
=
\mathcal L^{(1)}_{i_1}\mathcal L^{(2)}_{i_2}[g].
\]
Every exponent in the support of \(g\) avoids both \(S_{1,+}\) and \(S_{2,+}\). Thus, \Cref{cor:Q_b_finite_field_version_outside_finitely_many_char} at levels \(1\) and \(2\) imply
\(
g^{(1)}_{i_1,i_2}\neq 0
\)
so
\(
G \neq 0.
\)

Moreover, each operator preserves degree, so
\[
\deg g^{(1)}_{i_1,i_2}\le \deg g\le \ell-1.
\]
Since \(|I_1|=r_1\) and \(|I_2|=r_2\), we have
\[
\deg G \le r_1r_2(\ell-1).
\]

We now lower bound the number of roots of \(G\). Let \(A\) range over the cosets of \(\Omega_{n_1}\) in \(\F_q^*\), and let \(B\subset A\) range over the cosets of \(\Omega_{n_2}\) contained in \(A\). Write
\(
w_B:=|\ev(f)|_B.
\)
Fix a bottom coset \(B\) and let
\[
N_B:=
\left|\left\{
x\in B:
f(x)=f(\omega_2x)=\cdots=f(\omega_2^{\delta_2-2}x)=0
\right\}\right|.
\]
Each nonzero value of \(f\) on \(B\) can block at most \(\delta_2-1\) such starting points, so
\[
N_B\ge \left(n_2-(\delta_2-1)w_B\right)_+.
\]
For any such starting point \(x\), the number of \(i_2\in I_2\) for which
\(
f(\omega_2^{i_2}x)=0
\)
is at least \((r_2-w_B)_+\). For each such pair \((x,i_2)\), the definition of
\(\mathcal L^{(2)}_{i_2}\) gives
\(
g^{(2)}_{i_2}(x)=0.
\)
Therefore,
\[
\sum_{i_2\in I_2}|\{x\in B:g^{(2)}_{i_2}(x)=0\}|
\ge
\psi_2(w_B).
\]
Summing over all bottom cosets \(B\subset A\), we obtain
\[
\sum_{i_2\in I_2}|\{x\in A:g^{(2)}_{i_2}(x)=0\}|
\ge
\sum_{B\subset A}\psi_2(w_B).
\]
Equivalently, if
\[
w^{(2)}_{A,i_2}:=|\ev(g^{(2)}_{i_2})|_A,
\]
then
\[
\sum_{i_2\in I_2}w^{(2)}_{A,i_2}
\le
r_2n_1-\sum_{B\subset A}\psi_2(w_B).
\]
Thus,
\begin{equation}\label{eqn:upper_bound_average}
\frac{1}{r_2}\sum_{i_2\in I_2}w^{(2)}_{A,i_2}
\le
n_1-\frac{1}{r_2} \sum_{B\subset A}\psi_2(w_B).
\end{equation}

Now fix \(A\) and \(i_2\). Applying the same incidence count at level \(1\) to the family
\(
\{g^{(1)}_{i_1,i_2}\}_{i_1\in I_1}
\)
gives
\[
\sum_{i_1\in I_1}|\{x\in A:g^{(1)}_{i_1,i_2}(x)=0\}|
\ge
\psi_1(w^{(2)}_{A,i_2}).
\]
Therefore, the root multiplicity contributed by \(A\) to \(G\) is at least
\[
\sum_{i_2\in I_2}\psi_1(w^{(2)}_{A,i_2}).
\]
By \Cref{lem:psi_convex_monotone}, \(\psi_1\) is convex, so Jensen's inequality gives
\[
\sum_{i_2\in I_2}\psi_1(w^{(2)}_{A,i_2})
\ge
r_2\psi_1\left(
\frac{1}{r_2}\sum_{i_2\in I_2}w^{(2)}_{A,i_2}
\right).
\]
Since \(\psi_1\) is nonincreasing, substituting \labelcref{eqn:upper_bound_average} gives
\[
\sum_{i_2\in I_2}\psi_1(w^{(2)}_{A,i_2})
\ge
r_2\psi_1\left(n_1-\frac{1}{r_2}\sum_{B\subset A}\psi_2(w_B)\right).
\]
Now sum over all top cosets \(A\). The function
\[
S\longmapsto r_2\psi_1\left(n_1-\frac{S}{r_2}\right)
\]
is convex because we have composed $\psi_1$, a convex function, with an affine function which maps $S \mapsto n_1 - S/r_2$. It is also nondecreasing because as $S$ increases, the argument $n_1 - S/r_2$ decreases and $\psi_1$ is nonincreasing. Hence Jensen's inequality gives
\[
|\{\text{roots of }G\}|
\ge
\frac{q-1}{n_1}r_2
\psi_1\left(
n_1-\frac{1}{r_2}\cdot
\frac{n_1}{q-1}\sum_A \sum_{B\subset A}\psi_2(w_B)
\right).
\]
But,
\[
\sum_A \sum_{B\subset A}\psi_2(w_B)=\sum_B\psi_2(w_B),
\]
where \(B\) now ranges over all bottom cosets in \(\F_q^*\). By Jensen's inequality and the convexity of \(\psi_2\),
\[
\sum_B\psi_2(w_B)
\ge
\frac{q-1}{n_2}
\psi_2\left(\frac{n_2}{q-1}|\ev(f)|\right).
\]
Using that \(\psi_1\) is nonincreasing, we obtain
\[
|\{\text{roots of }G\}|
\ge
\frac{q-1}{n_1}r_2
\psi_1\left(
n_1-\frac{n_1}{r_2n_2}
\psi_2\left(\frac{n_2}{q-1}|\ev(f)|\right)
\right).
\]
Since \(G\neq 0\), its number of roots in \(\F_q^*\), counted with multiplicity, is at most its degree. Therefore
\[
r_1r_2(\ell-1)
\ge
\frac{q-1}{n_1}r_2
\psi_1\left(
n_1-\frac{n_1}{r_2n_2}
\psi_2\left(\frac{n_2}{q-1}|\ev(f)|\right)
\right).
\]
This proves the displayed inequality.

Finally, \(\Theta_2(t)\) is nonincreasing in \(t\). Hence the inequality
\[
r_1r_2(\ell-1)
\ge
\Theta_2\left(\frac{n_2}{q-1}|\ev(f)|\right)
\]
implies
\[
|\ev(f)|\ge \frac{q-1}{n_2}\tau_2.
\]
Taking the minimum over all \(\ev(f)\in C\setminus C^\perp\) proves the distance bound.
\end{proof}

\subsubsection{\texorpdfstring{$h$-level bound}{h-level bound}}
We will extend the distance bound to the general $h$-level hierarchical QTB code by iterating the same argument as in \Cref{thm:two_level_hqtb_distance}. We first present a lemma that will aid in the recursive argument.

\begin{lemma}
\label{lem:recursive_psi_jensen}
For \(l=1,\ldots,h\), set
\(
\psi_l(t):=(r_l-t)_+\left(n_l-(\delta_l-1)t\right)_+.
\)
Define
\(
\Psi_h(t):=\psi_h(t),
\)
and for \(l=h-1,h-2,\ldots,1\), define
\[
\Psi_l(t):=
\psi_l\left(
n_l-\frac{n_l}{r_{l+1}n_{l+1}}\Psi_{l+1}(t)
\right).
\]
Then each \(\Psi_l\) is convex and nonincreasing on \([0,n_h]\). Consequently,
\[
\Theta_h(t):=
\frac{q-1}{n_1}
\left(\prod_{l=2}^h r_l\right)\Psi_1(t)
\]
is also convex and nonincreasing on \([0,n_h]\).

Moreover, for each level-\(l\) coset \(B_l\), let
\(
w_{B_l}:=|\ev(f)|_{B_l}.
\)
Let \(I_l=\{\delta_l-1,\dots,n_l-1\}\) and let \(\mathcal R_l(B_l)\) denote the total number of zero incidences contributed
inside \(B_l\) as follows: for each choice of lower-level indices \[(i_{l+1},\ldots, i_h) \in I_{l+1}\times \cdots \times I_h,\] apply the level-\(l\) incidence count to the family \(\{g_{i_l,\ldots, i_h}^{(l)}\}_{i_l \in I_l}\) and sum over all choices of \((i_{l+1},\ldots, i_h).\) Then
\[
\mathcal R_l(B_l)
\ge
\left(\prod_{u=l+1}^{h} r_u\right)
\Psi_l\left(\frac{n_h}{n_l}w_{B_l}\right),
\]
where the empty product is interpreted as \(1\).
\end{lemma}

\begin{proof}
We first prove the analytic claim. By \Cref{lem:psi_convex_monotone}, for each \(l\), the function
\(
\psi_l(t)
\)
is convex and nonincreasing on \([0,\infty)\). We show by downward induction that each \(\Psi_l\) is convex and
nonincreasing. The base case \(\Psi_h=\psi_h\) is immediate from \Cref{lem:psi_convex_monotone}.

Assume \(\Psi_{l+1}\) is convex and nonincreasing. Define
\[
A_l(t):=
n_l-\frac{n_l}{r_{l+1}n_{l+1}}\Psi_{l+1}(t).
\]
Since \(\Psi_{l+1}\) is convex and nonincreasing, the function \(A_l\) is
concave and nondecreasing. Moreover, \(0 \leq \Psi_{l+1}(t) \leq r_{l+1}n_{l+1}\) so \(0 \leq A_l(t) \leq n_l\). Hence, all inputs to \(\psi_l\) lie in the interval where \Cref{lem:psi_convex_monotone} applies. Because \(\psi_l\) is convex and nonincreasing, the
composition rule for convex functions implies that
\(
\Psi_l(t)=\psi_l(A_l(t))
\)
is convex. Also, since \(A_l\) is nondecreasing and \(\psi_l\) is nonincreasing,
\(\Psi_l\) is nonincreasing. Indeed, if \(t_1 \leq t_2\) then \(A_l(t_1)\leq A_l(t_2)\). Applying \(\psi_l\) reverses the inequality: \[\Psi_l(t_1) = \psi_l(A_l(t_1)) \geq \psi_l(A_l(t_2)) = \Psi_l(t_2).\] This proves the induction. The same properties for
\(\Theta_h\) follow because \(\Theta_h\) is a positive scalar multiple of
\(\Psi_1\).

We now prove the incidence bound by downward induction on \(l\).
For \(l=h\), fix a bottom-level coset \(B_h\). The one-level incidence count at
level \(h\) gives
\[
\mathcal R_h(B_h)
\ge
\psi_h(w_{B_h})
=
\Psi_h(w_{B_h}),
\]
which is the desired statement because \(n_h/n_h=1\) and the product
\(\prod_{u=h+1}^h r_u\) is empty.

Assume now that the claim holds at level \(l+1\), and fix a level-\(l\) coset
\(B_l\). Decompose \(B_l\) into its level-\((l+1)\) child cosets:
\[
B_l=\bigsqcup_{a=1}^{n_l/n_{l+1}}B_{l+1,a}.
\]
Set
\(
w_a:=|\ev(f)|_{B_{l+1,a}}.
\)
Let
\(
\Lambda_{l+1}:=I_{l+1}\times I_{l+2}\times\cdots\times I_h
\)
so
\[
|\Lambda_{l+1}|=\prod_{u=l+1}^{h}r_u.
\]
For \(\boldsymbol i = (i_{l+1},\ldots,i_h)\in\Lambda_{l+1}\) and a child coset \(B_{l+1,a}\), define
\[
w_{\boldsymbol i,a}^g:=
|\ev(g^{(l+1)}_{i_{l+1},\ldots,i_h})|_{B_{l+1,a}},
\]
the number of nonzero positions of the transformed polynomial
\(g^{(l+1)}_{i_{l+1},\ldots,i_h}\) on \(B_{l+1,a}\).

By the induction hypothesis applied to the child coset \(B_{l+1,a}\), the total
number of zero incidences in \(B_{l+1,a}\), summed over all lower-level operator
choices \(\boldsymbol i\in\Lambda_{l+1}\), is at least
\[
\left(\prod_{u=l+2}^{h}r_u\right)
\Psi_{l+1}\left(\frac{n_h}{n_{l+1}}w_a\right).
\]
Equivalently,
\[
\sum_{\boldsymbol i\in\Lambda_{l+1}}
\left(n_{l+1}-w_{\boldsymbol i,a}^g\right)
\ge
\left(\prod_{u=l+2}^{h}r_u\right)
\Psi_{l+1}\left(\frac{n_h}{n_{l+1}}w_a\right).
\]
Dividing by \(|\Lambda_{l+1}|\), we get
\[
\frac{1}{|\Lambda_{l+1}|}
\sum_{\boldsymbol i\in\Lambda_{l+1}}
w_{\boldsymbol i,a}^g
\le
n_{l+1}
-
\frac{1}{r_{l+1}}
\Psi_{l+1}\left(\frac{n_h}{n_{l+1}}w_a\right).
\]
Now define the total effective weight of \(g^{(l+1)}_{i_{l+1},\ldots,i_h}\) inside \(B_l\) by
\[
w_{\boldsymbol i}^g:=
|\ev(g^{(l+1)}_{i_{l+1},\ldots,i_h})|_{B_l}
=
\sum_{a=1}^{n_l/n_{l+1}}w_{\boldsymbol i,a}^g.
\]
Averaging over \(\boldsymbol i\in\Lambda_{l+1}\), we obtain
\begin{align}\label{eqn:upper_bound_average_h_level}
\frac{1}{|\Lambda_{l+1}|}
\sum_{\boldsymbol i\in\Lambda_{l+1}}w_{\boldsymbol i}^g
&=
\sum_{a=1}^{n_l/n_{l+1}}
\frac{1}{|\Lambda_{l+1}|}
\sum_{\boldsymbol i\in\Lambda_{l+1}}w_{\boldsymbol i,a}^g\nonumber\\
&\le
\sum_{a=1}^{n_l/n_{l+1}}
\left(
n_{l+1}
-
\frac{1}{r_{l+1}}
\Psi_{l+1}\left(\frac{n_h}{n_{l+1}}w_a\right)
\right)\nonumber\\
&=
n_l
-
\frac{1}{r_{l+1}}
\sum_{a=1}^{n_l/n_{l+1}}
\Psi_{l+1}\left(\frac{n_h}{n_{l+1}}w_a\right)\nonumber\\
&=
n_l
-
\frac{n_l}{r_{l+1}n_{l+1}}
\cdot
\frac{1}{n_l/n_{l+1}}
\sum_{a=1}^{n_l/n_{l+1}}
\Psi_{l+1}\left(\frac{n_h}{n_{l+1}}w_a\right).
\end{align}

For each fixed lower-level choice \(\boldsymbol i\), the level-\(l\) incidence
count gives at least
\(
\psi_l(w_{\boldsymbol i}^g)
\)
zero incidences inside \(B_l\), summed over the \(r_l\) choices of the
level-\(l\) operator. Therefore,
\[
\mathcal R_l(B_l)
\ge
\sum_{\boldsymbol i\in\Lambda_{l+1}}\psi_l(w_{\boldsymbol i}^g).
\]
Since \(\psi_l\) is convex, Jensen's inequality gives
\[
\sum_{\boldsymbol i\in\Lambda_{l+1}}\psi_l(w_{\boldsymbol i}^g)
\ge
|\Lambda_{l+1}|
\psi_l\left(
\frac{1}{|\Lambda_{l+1}|}
\sum_{\boldsymbol i\in\Lambda_{l+1}}w_{\boldsymbol i}^g
\right).
\]
And since \(\psi_l\) is nonincreasing, substituting the upper bound \labelcref{eqn:upper_bound_average_h_level} gives:
\[
\mathcal R_l(B_l)
\ge
\left(\prod_{u=l+1}^{h}r_u\right)
\psi_l\left(
n_l
-
\frac{n_l}{r_{l+1}n_{l+1}}
\cdot
\frac{1}{n_l/n_{l+1}}
\sum_{a=1}^{n_l/n_{l+1}}
\Psi_{l+1}\left(\frac{n_h}{n_{l+1}}w_a\right)
\right).
\]
Using the convexity of \(\Psi_{l+1}\), Jensen's inequality gives
\begin{equation}\label{eqn:h_level_average_Jensen_sub}
\frac{1}{n_l/n_{l+1}}
\sum_{a=1}^{n_l/n_{l+1}}
\Psi_{l+1}\left(\frac{n_h}{n_{l+1}}w_a\right)
\ge
\Psi_{l+1}\left(
\frac{1}{n_l/n_{l+1}}
\sum_{a=1}^{n_l/n_{l+1}}
\frac{n_h}{n_{l+1}}w_a
\right) = \Psi_{l+1}\left(
\frac{n_h}{n_l}
w_{B_l}
\right)
\end{equation}
since
\(
\sum_a w_a=w_{B_l}.
\)
Since \(\psi_l\) is nonincreasing, substituting \labelcref{eqn:h_level_average_Jensen_sub} can only increase the argument of \(\psi_l\), we obtain
\[
\mathcal R_l(B_l)
\ge
\left(\prod_{u=l+1}^{h}r_u\right)
\psi_l\left(
n_l-
\frac{n_l}{r_{l+1}n_{l+1}}
\Psi_{l+1}\left(\frac{n_h}{n_l}w_{B_l}\right)
\right).
\]
By the recursive definition of \(\Psi_l\), this is
\[
\mathcal R_l(B_l)
\ge
\left(\prod_{u=l+1}^{h}r_u\right)
\Psi_l\left(\frac{n_h}{n_l}w_{B_l}\right).
\]
This completes the induction.
\end{proof}

\begin{theorem}[\(h\)-level hierarchical QTB distance bound]\label{thm:h_level_hqtb_distance_bound}
Let \(\mathcal Q=\mathrm{CSS}(C,C)\) be an \(h\)-level hierarchical QTB code with parameters
\(
(r_1,\delta_1),\dots,(r_h,\delta_h),
\)
and set
\(
n_l:=r_l+\delta_l-1
\) for \(l = 1,\ldots,h.\)
Assume
\(
n_h\mid n_{h-1}\mid\cdots\mid n_1\mid(q-1).
\)
Also, assume that \Cref{cor:Q_b_finite_field_version_outside_finitely_many_char} holds at every level; concretely, for every \(l\), either \(\delta_l\ge 3\) and
\(
\operatorname{char}(\F_q)\nmid \mathcal M_{r_l,\delta_l},
\)
or \(\delta_l=2\) and \(n_l=r_l+1\) is prime.

For \(l=1,\dots,h\), define
\(
\psi_l(t):=(r_l-t)_+\left(n_l-(\delta_l-1)t\right)_+.
\)
Define functions \(\Psi_l:[0,n_h]\to \mathbb R_{\ge 0}\) recursively by
\(
\Psi_h(t):=\psi_h(t),
\)
and for \(l=h-1,h-2,\dots,1\),
\[
\Psi_l(t):=
\psi_l\left(
n_l-\frac{n_l}{r_{l+1}n_{l+1}}\Psi_{l+1}(t)
\right).
\]
Define
\[
\Theta_h(t):=
\frac{q-1}{n_1}
\left(\prod_{l=2}^h r_l\right)
\Psi_1(t).
\]
Let
\[
\tau_h:=
\inf\left\{
t\in[0,n_h]:
\Theta_h(t)\le
\left(\prod_{l=1}^h r_l\right)(\ell-1)
\right\}.
\]
Then
\[
d(\mathcal Q)\ge \frac{q-1}{n_h}\tau_h.
\]
Equivalently, every nonzero codeword \(\ev(f)\in C\setminus C^\perp\) of weight \(w\) satisfies
\[
\left(\prod_{l=1}^h r_l\right)(\ell-1)
\ge
\frac{q-1}{n_1}
\left(\prod_{l=2}^h r_l\right)
\Psi_1\left(\frac{n_h}{q-1}w\right).
\]
\end{theorem}

\begin{proof}
The proof is an iteration of the two-level argument in \Cref{thm:two_level_hqtb_distance}. We give the details needed to track the notation. Write
\[
f=g+P_1+\cdots+P_h,
\]
where \(g\) is supported outside \(S_+ \cup S_{-}=\bigcup_{l=1}^h S_{l,+} \cup \bigcup_{l=1}^h S_{l,-}\), and \(P_l\) is supported on
\[
S_{l,+}\setminus\bigcup_{k=l+1}^h S_{k,+}
\]
for \(l<h\), while \(P_h\) is supported on \(S_{h,+}\). Since \(\ev(f)\notin C^\perp\), we have \(g\neq 0\). For each level \(l\), set
\[
I_l:=\{\delta_l-1,\dots,n_l-1\}.
\]
For \(i_l\in I_l\), define the operator
\[
\mathcal L^{(l)}_{i_l}[p](X)
:=
\omega_l^{-i_l}p(\omega_l^{i_l}X)
+
\sum_{t=0}^{\delta_l-2}
v^{(l)}_{i_l,t}\omega_l^{-t}p(\omega_l^tX),
\]
so that
\[
\mathcal L^{(l)}_{i_l}[X^j]
=
Q^{(l)}_{i_l}(\omega_l^{j-1})X^j.
\]

We recursively define transformed polynomials by ascending through the hierarchy. First, for \(i_h\in I_h\), set
\[
g^{(h)}_{i_h}:=\mathcal L^{(h)}_{i_h}[f].
\]
For \(l=h-1,h-2,\dots,1\), and for a tuple
\[
(i_l,i_{l+1},\dots,i_h)\in I_l\times I_{l+1}\times\cdots\times I_h,
\]
define
\[
g^{(l)}_{i_l, i_{l+1},\dots,i_h}
:=
\mathcal L^{(l)}_{i_l}
\left[g^{(l+1)}_{i_{l+1},\dots,i_h}\right].
\]
Finally, define
\[
G(X):=
\prod_{i_h\in I_h}
\prod_{i_{h-1}\in I_{h-1}}
\cdots
\prod_{i_1\in I_1}
g^{(1)}_{i_1,\dots,i_h}(X).
\]

We first show that \(G\neq 0\). At each level \(l\), the operator \(\mathcal L^{(l)}_{i_l}\) annihilates the positive-residue part. Since the operators act diagonally on monomials, the support of \(g\) remains disjoint from all positive-residue supports throughout the iteration. After all levels have been applied, we have
\[
g^{(1)}_{i_1,\dots,i_h}
=
\mathcal L^{(1)}_{i_1}
\mathcal L^{(2)}_{i_2}
\cdots
\mathcal L^{(h)}_{i_h}[g].
\]
For every exponent \(j\) in the support of \(g\), and for every level \(l\), we have
\(
j\notin S_{l,+}.
\)
Therefore,
\(
Q^{(l)}_{i_l}(\omega_l^{j-1})\neq 0
\)
by \Cref{cor:Q_b_finite_field_version_outside_finitely_many_char} at level \(l\). Thus, every nonzero coefficient of \(g\) remains nonzero after applying the product of diagonal operators so
\(
g^{(1)}_{i_1,\dots,i_h}\neq 0
\)
for every tuple \((i_1,\dots,i_h)\). This implies \(G\neq 0\). Moreover, each operator preserves degree, so
\[
\deg g^{(1)}_{u_1,\dots,u_h}\le \deg g\le \ell-1.
\]
Since \(|I_l|=r_l\), we obtain
\[
\deg G
\le
\left(\prod_{l=1}^h r_l\right)(\ell-1).
\]

We now count roots. By \Cref{lem:recursive_psi_jensen}, applied to each level-\(1\) coset \(B_1\),
we have
\[
|\{\text{roots of }G\}|
\ge
\left(\prod_{l=2}^h r_l\right)
\sum_{B_1}
\Psi_1\left(\frac{n_h}{n_1}|\ev(f)|_{B_1}\right).
\]
Since \(\Psi_1\) is convex, Jensen's inequality gives
\[
\sum_{B_1}
\Psi_1\left(\frac{n_h}{n_1}|\ev(f)|_{B_1}\right)
\ge
\frac{q-1}{n_1}
\Psi_1\left(
\frac{1}{(q-1)/n_1}
\sum_{B_1}
\frac{n_h}{n_1}|\ev(f)|_{B_1}
\right)
\]
and
\[
\frac{1}{(q-1)/n_1}
\sum_{B_1}
\frac{n_h}{n_1}|\ev(f)|_{B_1} = \frac{n_h}{q-1}|\ev(f)|.
\]
Thus,
\[
|\{\text{roots of }G\}|
\ge
\frac{q-1}{n_1}
\left(\prod_{l=2}^h r_l\right)
\Psi_1\left(\frac{n_h}{q-1}|\ev(f)|\right).
\]

Since \(G\neq 0\), the number of roots of \(G\), counted with multiplicity, is at most its degree. Hence,
\[
\left(\prod_{l=1}^h r_l\right)(\ell-1)
\ge
\frac{q-1}{n_1}
\left(\prod_{l=2}^h r_l\right)
\Psi_1\left(\frac{n_h}{q-1}|\ev(f)|\right).
\]
By construction, the function
\[
\Theta_h(t):=
\frac{q-1}{n_1}
\left(\prod_{l=2}^h r_l\right)
\Psi_1(t)
\]
is nonincreasing in \(t\). Therefore the preceding inequality implies
\[
|\ev(f)|\ge \frac{q-1}{n_h}\tau_h.
\]
Taking the minimum over all \(\ev(f)\in C\setminus C^\perp\) proves the theorem.
\end{proof}

\Cref{thm:two_level_hqtb_distance} and \Cref{thm:h_level_hqtb_distance_bound} give implicit distance bounds which we now make explicit in the following corollary. The final expression will be similar to \labelcref{eqn:qtb_distance_bound}. 
\begin{corollary}\label{cor:h_level_nested_square_root}
Assume the hypotheses of \Cref{thm:h_level_hqtb_distance_bound}. For each
\(l=1,\ldots,h\), set
\(
a_l:=\delta_l-1.
\)
Define the clipped inverse function \(\operatorname{Inv}_l:\mathbb R\to[0,n_l]\) by
\[
\operatorname{Inv}_l(\theta)
:=
\begin{cases}
0, & \theta\ge r_ln_l,\\[4pt]
T_l, & \theta\le 0,\\[4pt]
\displaystyle
\frac{
n_l+a_lr_l
-
\sqrt{(n_l+a_lr_l)^2-4a_l(r_ln_l-\theta)}
}{2a_l},
& 0<\theta<r_ln_l,
\end{cases}
\]
where
\[
T_l:=\min\left\{r_l,\frac{n_l}{a_l}\right\}.
\]
Define \(y_1,\ldots,y_h\) recursively by
\[
y_1:=
\operatorname{Inv}_1\left(
r_1n_1\frac{\ell-1}{q-1}
\right),
\]
and, for \(l=2,\ldots,h\), by
\[
y_l:=
\operatorname{Inv}_l\left(
\frac{r_ln_l}{n_{l-1}}(n_{l-1}-y_{l-1})
\right).
\]
Then
\begin{equation}\label{eqn:hqtb_distance_bound}
    d(\mathcal Q)\ge \frac{q-1}{n_h}y_h.
\end{equation}

\end{corollary}

\begin{proof}
Recall that
\[
\psi_l(t)=(r_l-t)_+(n_l-a_lt)_+.
\]
On the interval where \(\psi_l(t)>0\), we have
\[
\psi_l(t)=(r_l-t)(n_l-a_lt)
=
r_ln_l-(n_l+a_lr_l)t+a_lt^2.
\]
Solving
\(
\psi_l(t)=\theta
\)
for \(t\) gives
\[
t=
\frac{
n_l+a_lr_l
-
\sqrt{(n_l+a_lr_l)^2-4a_l(r_ln_l-\theta)}
}{2a_l}.
\]
This is the smaller root, and it is the relevant one because \(\psi_l\) is nonincreasing on
\([0,n_l]\). With the clipping convention in the definition of \(\operatorname{Inv}_l\), we have
\begin{equation}\label{eqn:psi_inv_relation_implication}
\psi_l(t)\le \theta
\qquad\Longrightarrow\qquad
t\ge \operatorname{Inv}_l(\theta).
\end{equation}

By \Cref{thm:h_level_hqtb_distance_bound}, every nonzero codeword of weight \(w\) satisfies
\[
\left(\prod_{l=1}^h r_l\right)(\ell-1)
\ge
\frac{q-1}{n_1}
\left(\prod_{l=2}^h r_l\right)
\Psi_1\left(\frac{n_h}{q-1}w\right),
\]
where
\[
\Psi_h(t)=\psi_h(t),
\]
and, for \(l=h-1,\ldots,1\),
\[
\Psi_l(t)
=
\psi_l\left(
n_l-\frac{n_l}{r_{l+1}n_{l+1}}\Psi_{l+1}(t)
\right).
\]
Canceling the common factor \(\prod_{l=2}^h r_l\), we get
\[
\Psi_1\left(\frac{n_h}{q-1}w\right)
\le
r_1n_1\frac{\ell-1}{q-1}.
\]
Set
\[
t_0:=\frac{n_h}{q-1}w.
\]
Then
\[
\psi_1\left(
n_1-\frac{n_1}{r_2n_2}\Psi_2(t_0)
\right)
\le
r_1n_1\frac{\ell-1}{q-1}.
\]
By \labelcref{eqn:psi_inv_relation_implication}, this implies
\[
n_1-\frac{n_1}{r_2n_2}\Psi_2(t_0)
\ge y_1.
\]
Equivalently,
\[
\Psi_2(t_0)
\le
\frac{r_2n_2}{n_1}(n_1-y_1).
\]

Applying the same argument at level \(2\), we obtain
\[
n_2-\frac{n_2}{r_3n_3}\Psi_3(t_0)
\ge y_2,
\]
and hence
\[
\Psi_3(t_0)
\le
\frac{r_3n_3}{n_2}(n_2-y_2).
\]
Continuing recursively, after level \(h-1\) we get
\[
\Psi_h(t_0)
\le
\frac{r_hn_h}{n_{h-1}}(n_{h-1}-y_{h-1}).
\]
Since \(\Psi_h(t_0)=\psi_h(t_0)\), another application of \labelcref{eqn:psi_inv_relation_implication} gives
\(
t_0\ge y_h.
\)
Substituting back \(t_0=\frac{n_h}{q-1}w\), we obtain
\[
w\ge \frac{q-1}{n_h}y_h.
\]
Taking the minimum over all nonzero codewords gives
\[
d(\mathcal Q)\ge \frac{q-1}{n_h}y_h,
\]
as desired.
\end{proof}

For $h = 2$, we can unravel the bound in \labelcref{eqn:hqtb_distance_bound} to see that \begin{equation}\label{eqn:two_level_hqtb_distance_bound}
    d(\mathcal{Q}) \geq  \frac{q-1}{n_2}\cdot\frac{n_2 + (\delta_2-1)r_2 - \sqrt{(n_2 + (\delta_2-1)r_2)^2 - 4(\delta_2-1)r_2n_2\frac{y_1}{n_1}}}{2(\delta_2-1)}
\end{equation} is the explicit bound for \Cref{thm:two_level_hqtb_distance} where \begin{equation*}
    y_1 = \frac{n_1 + (\delta_1-1)r_1 - \sqrt{(n_1 + (\delta_1-1)r_1)^2 - 4(\delta_1-1)r_1n_1\left(1-\frac{\ell-1}{q-1}\right)}}{2(\delta_1-1)}.
\end{equation*}  Additionally, $y_1$ is exactly the one-level bound in \Cref{thm:qtb_distance}.

\begin{remark}
    In order to apply the distance bound \Cref{thm:h_level_hqtb_distance_bound}, we need \Cref{cor:Q_b_finite_field_version_outside_finitely_many_char} to hold at each level (and potentially \Cref{rem:Q_b_delta=2_finite_field} at the bottom level). In particular, we must first fix $(r_l,\delta_l)_{l= 1,\ldots, h}$ such that \[r_1\geq r_2\geq\cdots\geq r_h \geq \delta_1\geq \delta_2\geq \cdots \geq \delta_h\geq 2\] and $n_h\mid n_{h-1}\mid \cdots \mid n_1$. Then we compute $\mathcal{M}_{r_l,\delta_l}$ for each $l = 1,\ldots, h$ and obtain the following set of characteristics to exclude: \[\mathcal{P} = \bigcup_{l=1}^h \{p : p \mid\mathcal{M}_{r_l,\delta_l}\}.\] Finally, we choose $q$, a prime power such that $n_1\mid (q-1)$ and $\mathrm{char}(\F_q) \not \in \mathcal{P}$.
\end{remark}

% Since any $h$-level $(r_l,\delta_l)_{l=1,\ldots, h}$ quantum Tamo--Barg is itself an $(r_1,\delta_1)$ quantum Tamo--Barg (with additional structure), we get the following distance bound immediately from \Cref{thm:qtb_distance}:
% \begin{corollary}\label{cor:hqtb_distance_loose}
%     Consider the $h$-level $(r_l,\delta_l)_{l=1,\ldots,h}$ quantum Tamo--Barg code defined in \Cref{def:h_level_qtb}. If $\delta_1 \geq 3$, assume $\mathrm{char}(\F_q) \nmid \mathcal{M}_{r_1,\delta_1}$ or if $\delta_1 = 2$, assume $r_1+1$ is prime. Then $\mathcal{Q}$ has distance at least \begin{equation}\label{eqn:hqtb_distance_bound_one_level}
%         \frac{q-1}{2}\left(\frac{1}{\delta_1-1} + \frac{r_1}{n_1} -\sqrt{\left(\frac{r_1}{n_1} - \frac{1}{\delta_1-1}\right)^2 + \frac{4r_1}{(\delta_1-1)n_1}\cdot \frac{\ell-1}{q-1}}\right).
%     \end{equation}
% \end{corollary}

The recursive bound in \Cref{thm:h_level_hqtb_distance_bound} is a direct
multilevel analogue of the one-level root-counting proof. However, the hierarchy does not necessarily improve the distance when compared to the one-level $(r_1,\delta_1)$ QTB code. Indeed, the explicit form in \Cref{cor:h_level_nested_square_root} gives
\[
d(\mathcal Q)\ge \frac{q-1}{n_h}y_h,
\]
whereas the top-level bound is
\[
d(\mathcal Q)\ge \frac{q-1}{n_1}y_1.
\]
We now explain why the recursive estimate cannot be expected to beat the
top-level estimate.

\begin{corollary}\label{cor:h_level_bound_loss}
 For \(l=1,\ldots,h\), set \(a_l := \delta_l-1\) and \(n_l := r_l + a_l\).
Recall that
\[
\psi_l(t)=(r_l-t)_+(n_l-a_lt)_+.
\]
The recursion defining \(y_l\) says that, for \(l\ge 2\),
\[
y_l
=
\operatorname{Inv}_l\left(
\frac{r_ln_l}{n_{l-1}}(n_{l-1}-y_{l-1})
\right).
\]
We claim that \[\frac{y_h}{n_h} \leq \frac{y_{h-1}}{n_{h-1}} \leq \cdots \leq \frac{y_1}{n_1}.\]
\end{corollary}
\begin{proof}
Set \(x = y_{l-1}/n_{l-1}\) so we have 
\(
y_l
=
\operatorname{Inv}_l\left(r_ln_l(1-x)\right).
\)
We claim that \(y_l/n_l \leq x\).
Since \(\operatorname{Inv}_l(\theta)\) is the smallest threshold beyond which
\(\psi_l(t)\le \theta\), it suffices to show
\[
\psi_l(n_lx)\le r_ln_l(1-x).
\]
If one of the clipped factors in \(\psi_l(n_lx)\) is zero, then this is immediate. Otherwise,
\[
\psi_l(n_lx)=(r_l-n_lx)(n_l-a_ln_lx).
\]
Dividing by \(n_l\), it is enough to prove
\(
(r_l-n_lx)(1-a_lx)\le r_l(1-x).
\)
Subtracting the left-hand side from the right-hand side gives
\[
r_l(1-x)-(r_l-n_lx)(1-a_lx)
=
x\left(r_l(a_l-1)+n_l-a_ln_lx\right).
\]
In the nonzero range of \(\psi_l(n_lx)\), we have \(1-a_lx>0\), and hence,
\(
a_ln_lx<n_l.
\)
Therefore,
\[
r_l(a_l-1)+n_l-a_ln_lx\ge 0.
\]
Thus,
\(
\psi_l(n_lx)\le r_ln_l(1-x),
\)
and consequently
\(
y_l/n_l\le y_{l-1}/n_{l-1}.
\)
Iterating this inequality gives
\[
\frac{y_h}{n_h}\le
\frac{y_{h-1}}{n_{h-1}}
\le \cdots \le
\frac{y_1}{n_1}.
\]
% Hence the recursive multilevel bound
% \[
% d(\mathcal Q)\ge \frac{q-1}{n_h}y_h
% \]
% is generally no stronger than the top-level bound
% \[
% d(\mathcal Q)\ge \frac{q-1}{n_1}y_1.
% \]
\end{proof}
The loss comes from two sources. First, at each level the proof averages zero
incidences using Jensen's inequality, which discards information about how the
zeros are distributed among lower-level repair groups. Second, the aggregate
polynomial obtained by multiplying all transformed polynomials has degree
\[
\left(\prod_{l=1}^h r_l\right)(\ell-1),
\]
so each additional level increases the degree bound by a factor of \(r_l\). This
degree growth can overwhelm the extra zero incidences forced by the hierarchy. It remains an interesting open problem to develop a multilevel distance proof
which uses the reduced nondual support
\(
[\ell]\setminus(S_+\cup S_-)
\)
more efficiently. %One possible approach is to replace the product polynomial by a determinant polynomial whose degree depends on the actual residue support profile of the nondual part of the codeword, analogous to the determinant method used in the folded setting.

We end this section with an example of a two-level hierarchical QTB code and its computed dimension and distance.

\begin{example}\label{ex:two_level_qtb_vs_one_level_true_distance}
Consider the two-level QTB code with parameters
\((r_1,\delta_1)=(9,4)\) and
\((r_2,\delta_2)=(4,3)\).
Then $n_1 = 12$ and $n_2 = 6$
so \(n_2\mid n_1\). The bad-characteristic products from
\Cref{cor:Q_b_finite_field_version_outside_finitely_many_char} are
\[
\mathcal M_{4,3}=2^8
\quad \text{ and }\quad
\mathcal M_{9,4}=2^{76}3^{44}13^{20}37^4.
\]
Thus the excluded characteristics are
\(
\{2,3,13,37\}.
\)
The smallest prime power \(q\) satisfying \(12\mid(q-1)\) and whose characteristic
is not excluded is
\(
q=25.
\)

We compare the two-level hierarchical QTB code with the one-level QTB code having
top-level parameters \((r,\delta)=(9,4)\), over the same field \(\F_{25}\) and
with the same value of \(\ell\). The exact CSS distance was computed using the
shortened-support criterion
\[
d=\min\left\{
|U|:
\dim(C\cap \F_q^U)>\dim(C^\perp\cap \F_q^U)
\right\}.
\]
The results for \(\ell=13,\ldots,24\) are as follows:
\[
\begin{array}{c|cc|cc|c}
\ell
& k_{\mathrm{1lev}} & d_{\mathrm{1lev}}
& k_{\mathrm{hier}} & d_{\mathrm{hier}}
& d_{\mathrm{bound}}
\\
\hline
13 & 2 & 9 & 2 & 7 & 4\\
14 & 2 & 9 & 2 & 7 & 3\\
15 & 2 & 9 & 2 & 7 & 3\\
16 & 2 & 9 & 2 & 7 & 3\\
17 & 4 & 8 & 2 & 7 & 2\\
18 & 6 & 7 & 2 & 7 & 2\\
19 & 8 & 6 & 4 & 4 & 2\\
20 & 10 & 5 & 4 & 4 & 2\\
21 & 12 & 4 & 4 & 4 & 1\\
22 & 12 & 4 & 4 & 4 & 1\\
23 & 12 & 4 & 4 & 4 & 1\\
24 & 12 & 4 & 4 & 4 & 1
\end{array}
\]
In this example, the two-level hierarchical code does not improve the true global
distance over a one-level code with the same top level parameters. It sometimes has the same distance, but it is often worse. In the table, $d_{\mathrm{bound}}$ is the distance bound of the one-level $(r_1,\delta_1)$ QTB code as in \labelcref{eqn:qtb_distance_bound}. \Cref{cor:h_level_bound_loss} shows the $h$-level distance bound \labelcref{eqn:hqtb_distance_bound} is at best equal to \labelcref{eqn:qtb_distance_bound}. However, the example codes do beat the proven bound so the bound is not tight. This illustrates that hierarchy should be viewed primarily as improving the repair structure, but not automatically improving the global minimum distance.
\end{example}

\section{\texorpdfstring{Folded $(r,\delta)$ Quantum Tamo--Barg Codes}{Folded (r,delta) Quantum Tamo--Barg Codes}}\label{sec:folded_qtb}
We now present a variant of the quantum Tamo--Barg (QTB) codes in \Cref{def:r.delta.qTB.qlrc} by using a \textit{folding} operation. This code generalizes the folded quantum Tamo--Barg code which appears in \cite{golowich2025quantum}.
    \begin{definition}[Folded $(r,\delta)$ Quantum Tamo--Barg code]\label{def:fqtb}
        Let $\mathcal{Q} = \mathrm{CSS}(C, C)$ be the QTB code with parameters $q, r, \delta, \ell$. Given an additional folding parameter $s \mid (q-1)/(r+\delta-1)$, we define the \textbf{folded Quantum Tamo--Barg (fQTB) code} $\widetilde{\mathcal{Q}}$ to be the quantum code with local dimension $q^s$ and block length $(q-1)/s$ obtained as follows. Fix a generator $\omega_{q-1} \in \mathbb{F}_q^*$ and for every $i \in [(q-1)/s]$, we block together $s$ components at positions $F_{\omega_{q-1}^{i\cdot s}} = \{\omega_{q-1}^{i\cdot s},\ldots, \omega_{q-1}^{i\cdot s+s-1}\}$ in $\mathcal{Q}$ into a single component of $\widetilde{\mathcal{Q}}$. Let $\widetilde{C}$ denote the $\mathbb{F}_q$-linear codes obtained by folding $C$ so that $\widetilde{\mathcal{Q}} = \mathrm{CSS}(\widetilde{C}, \widetilde{C})$.
    \end{definition}
    For a given $x = \omega_{q-1}^b \in \mathbb{F}_q^*$, let $F_x \in \widetilde{\mathbb{F}}_q^*$ be the unique element of $\widetilde{\mathbb{F}}_q^*$ that contains $x$. Formally, \[F_x = \{\omega_{q-1}^{\lfloor b/s\rfloor\cdot s + j} : j \in [s]\}.\]

    Folding a code by definition preserves the rate, so a fQTB has the same rate as the associated unfolded QTB. We computed this rate in \Cref{lem:dim.r.delta.TB}. We will now show that fQTBs are also QLRCs. First, we need a small lemma.
    \begin{lemma}\label{lem:coset_block_intersection}
        Assume $s \mid (q-1)/(r+\delta-1)$. Let $\alpha\Omega_{r+\delta-1} \subseteq \mathbb{F}_q^*$ be a coset of $\Omega_{r+\delta-1}$ and let $F_{\omega_{q-1}^{i\cdot s}}$ be a block of positions as defined in \Cref{def:fqtb}. Then $|\alpha\Omega_{r+\delta-1} \cap F_{\omega_{q-1}^{i\cdot s}} | \leq 1$ for all $i$.
    \end{lemma}
    \begin{proof}
        Fix $i \in [(q-1)/s]$ and let $\alpha = \omega_{q-1}^b$ for some $b \in \mathbb{Z}$. Then \[\alpha \Omega_{r+\delta-1} = \left\{\omega_{q-1}^{b + t\cdot \frac{q-1}{r+\delta-1}} : t\in [r+\delta-1]\right\}.\] Suppose there exist $t_1 \neq t_2$ such that \[\omega_{q-1}^{b + t_1\cdot \frac{q-1}{r+\delta-1}}, \omega_{q-1}^{b + t_2\cdot \frac{q-1}{r+\delta-1}} \in F_{\omega_{q-1}^{i\cdot s}}.\] Then we must have \[b + t_1\cdot \frac{q-1}{r+\delta-1} \equiv i\cdot s + e_1 \mod (q-1)\quad \text{and}\quad b + t_2\cdot \frac{q-1}{r+\delta-1} \equiv i\cdot s + e_2 \mod (q-1)\] for some $e_1,e_2 \in [s]$. Subtracting, we obtain \[(t_1 - t_2) \cdot \frac{q-1}{r+\delta-1} \equiv e_1 - e_2 \mod (q-1)\] with $e_1-e_2 \in \{-(s-1),\ldots,(s-1)\}$. Since $s\mid \frac{q-1}{r+\delta-1}$, we get $s \mid (e_1-e_2)$ which is not possible unless $e_1 = e_2$. Thus \[(t_1 - t_2) \cdot \frac{q-1}{r+\delta-1} \equiv 0 \mod (q-1) \implies t_1-t_2 \equiv 0 \mod (r+\delta-1),\] but we assumed $t_1 \neq t_2$ and $0 \leq t_1, t_2 < r+\delta-1$. This is a contradiction. Hence, each element of $\alpha \Omega_{r+\delta-1}$ resides in a distinct folded block.
    \end{proof}
    \begin{corollary}\label{cor:locality_fqtb}
        Assume $s \mid (q-1)/(r+\delta-1)$. Let $\mathcal Q=\mathrm{CSS}(C,C)$ be the QTB code from \Cref{def:r.delta.qTB.qlrc}, and let $\widetilde{\mathcal Q}=\mathrm{CSS}(\widetilde C,\widetilde C)$ be its folded version. Then $\widetilde{\mathcal Q}$ is a QLRC with parameters $(r,\delta)$ i.e., every folded repair group has size $r+\delta-1$ and can correct $\delta-1$ erasures using only symbols inside the group.
    \end{corollary}

    \begin{proof}
    By \Cref{thm:css_qlrc_iff_clrc}, it suffices to prove that $\widetilde C$ is a classical LRC with parameters $(r,\delta)$.
    
    Fix a coset $\alpha\Omega_{r+\delta-1}$. Define the folded repair group
    \[
    \widetilde R_{\alpha} \;:=\; \{F_x : x\in \alpha\Omega_{r+\delta-1}\}.
    \]
    By \Cref{lem:coset_block_intersection}, the map $x\mapsto F_x$ is injective on $\alpha\Omega_{r+\delta-1}$,
    so
    \[
    |\widetilde R_{\alpha}| = |\alpha\Omega_{r+\delta-1}| = r+\delta-1.
    \]
    
    Write \(n:=r+\delta-1\). Let
    \(\alpha=\omega_{q-1}^b\), and let \(e_0\in\{0,\ldots,s-1\}\)
    be the residue of \(b\) modulo \(s\). The elements of
    \(\alpha\Omega_n\) are
    \[
        \omega_{q-1}^{b+t(q-1)/n},
        \qquad
        t=0,\ldots,n-1.
    \]
    Since \(s\mid (q-1)/n\), all of these elements occur in their folded
    blocks with the same internal offset \(e_0\). Thus the unfolded
    local repair on the coset \(\alpha\Omega_n\) recovers the
    \(e_0\)-th scalar component of the erased folded symbols.

    We now repeat this argument for every internal offset. For
    \(e\in\{0,\ldots,s-1\}\), define
    \[
        A_{\alpha,e}
        :=
        \left\{
        \omega_{q-1}^{b-e_0+e+t(q-1)/n}
        :
        t=0,\ldots,n-1
        \right\}.
    \]
    Then
    \(
        A_{\alpha,e}
        =
        \omega_{q-1}^{b-e_0+e}\Omega_n
    \)
    is a coset of \(\Omega_n\). Moreover, the elements of
    \(A_{\alpha,e}\) are precisely the \(e\)-th scalar components of
    the folded blocks in \(\widetilde R_{\alpha}\), since
    \[
        F_{\omega_{q-1}^{b-e_0+e+t(q-1)/n}}
        =
        F_{\omega_{q-1}^{b+t(q-1)/n}}
    \]
    for every \(t\).

    By \Cref{cor:locality_qtb}, applied to the unfolded code \(C\),
    each coset \(A_{\alpha,e}\) has \(\delta-1\) independent local
    parity checks and can correct any \(\delta-1\) erased scalar
    positions using only scalar positions in \(A_{\alpha,e}\). If at
    most \(\delta-1\) folded symbols in \(\widetilde R_{\alpha}\) are
    erased, then for each fixed offset \(e\), at most \(\delta-1\)
    scalar positions are erased in the coset \(A_{\alpha,e}\). Applying
    the unfolded local recovery independently for
    \(e=0,\ldots,s-1\) recovers all \(s\) scalar components of every
    erased folded symbol in \(\widetilde R_{\alpha}\).
    Thus $\widetilde C$ is a classical LRC. Therefore
    $\widetilde{\mathcal Q}=\mathrm{CSS}(\widetilde C,\widetilde C)$ is a QLRC with the same locality parameters.
    \end{proof}
\subsection{Distance bound}
Lastly, we will prove a distance bound for the fQTB. A trivial bound follows from \Cref{thm:qtb_distance}: if we denote the distance of the fQTB as $\widetilde{d}$, then $\widetilde{d} \geq d/s$, but we will shortly present a more involved proof for a stronger bound. Before we do so, we state a lemma that we require in the proof.
\begin{lemma}\label{lem:det_poly_root}
    For every \(v\in \mathbb N\), the determinant polynomial \(\det\in \mathbb F_q[(Y_{ij})_{i,j\in[v]}]\) has order of vanishing \(v-t\) at every matrix \((x_{ij})_{i,j\in[v]}\) of rank \(t\).
\end{lemma}

\begin{proof}
Let \(X\in M_v(\mathbb F_q)\) have rank \(t\). If \(t=v\), then
\(\det(X)\neq 0\), so the order of vanishing of \(\det\) at \(X\) is
\(0=v-t\).

Now suppose \(t<v\). Since \(X\) has rank \(t\), there exist invertible
matrices \(P,Q\in \operatorname{GL}_v(\mathbb F_q)\) such that
\[
PXQ=
X_0:=
\begin{bmatrix}
I_t & 0\\
0 & 0
\end{bmatrix}.
\]
The change of variables \(Y\mapsto PYQ\) is an invertible linear change of
coordinates on \(M_v(\mathbb F_q)\), and
\[
\det(PYQ)=\det(P)\det(Q)\det(Y),
\]
where \(\det(P)\det(Q)\neq 0\). Hence this change of variables does not affect
the order of vanishing of the determinant polynomial. Therefore, it suffices to
compute the order of vanishing of \(\det\) at \(X_0\).

Write a matrix near \(X_0\) as
\[
X_0+Z=
\begin{bmatrix}
I_t+A & B\\
C & D
\end{bmatrix},
\]
where \(A,B,C,D\) are matrices of indeterminates of the appropriate sizes, and
where \(D\) is \((v-t)\times (v-t)\). In the formal power series
ring in the entries of \(A,B,C,D\), the matrix \(I_t+A\) is invertible, since
\[
\det(I_t+A)=1+\text{terms of positive degree}.
\]
Thus, by the Schur complement formula,
\[
\det(X_0+Z)
=
\det(I_t+A)\det\left(D-C(I_t+A)^{-1}B\right).
\]

Now \(\det(I_t+A)\) is a unit with constant term \(1\). Also, the entries of
\(D\) have degree \(1\), while the entries of \(C(I_t+A)^{-1}B\) have degree at
least \(2\), because each term contains one entry from \(C\) and one entry from
\(B\). Therefore,
\[
D-C(I_t+A)^{-1}B
=
D+\text{terms of degree at least }2.
\]
It follows that
\[
\det\left(D-C(I_t+A)^{-1}B\right)
=
\det(D)+\text{terms of degree at least }v-t+1.
\]
Since \(\det(D)\) is a nonzero homogeneous polynomial of degree \(v-t\), and since
\(\det(I_t+A)\) has constant term \(1\), the first nonzero homogeneous part of
\(\det(X_0+Z)\) has degree \(v-t\). Thus, \(\det\) has multiplicity \(v-t\) at \(X_0\), and hence also at the
original matrix \(X\).
\end{proof}

Now we present a more involved argument than the unfolded distance proof that gives a significant boost to the relative distance. This argument follows the structure of the proof of \cite[Theorem 63]{golowich2025quantum}, but this proof has to deal with more general parameters.

\begin{theorem}\label{thm:fqtb_distance}
    Let $\widetilde{\mathcal{Q}}$ be the code from \Cref{def:fqtb} with parameters $q, r, \delta, \ell, s$ such that $r+\delta-1$ is prime and the uncertainty principle in \Cref{cor:uncertainty_prime_finite_field} holds for $r+\delta-1$ over $\mathbb{F}_q$. Then $\widetilde{\mathcal{Q}}$ has distance at least \[d = \frac{q-1}{s}\left(1 - \frac{\ell-1}{q-1} - \epsilon\right)\] for \[\epsilon = \max_{1 \leq m_f \leq r+\delta-1} \min \left\{\left(1 - \frac{\ell-1}{q-1}\right) \frac{m_f-1}{r+\delta-1}, \max_{\max\{1,m_f-(\delta-1)\} \leq m_g \leq m_f}\frac{\delta-1}{m_g} + \frac{m_g-1}{s}\right\}.\]
\end{theorem}
\begin{proof}
    Fix an arbitrary $f(x) = \sum_{j \in [q-1]}f_jX^j \in \F_q[X]$ such that  $\ev(f(X)) \in C \setminus C^\perp$ with associated folded codeword $\widetilde{\ev}(f(X)) \in \widetilde{C}\setminus \widetilde{C}^\perp$. Our goal is to show $|\widetilde{\ev}(f)| \geq d$.

    Since  $C = \ev(\F_q[X]^S)$ for \[S = \left([\ell] \setminus S_{-}\right)\cup ([q-1]\cap S_{+})\] we may write $f(X) = g(X) + h(X)$ where $g(X) = \sum_j g_jX^j \in \mathbb{F}_q[X]^{[\ell]\setminus (S_{-}\cup S_{+})}$ and $h(X) = \sum_j h_jX^j  \in \mathbb{F}_q[X]^{ [q-1]\cap S_{+}}$.

    Let $M_f = \{i \in [r+\delta-1]: \exists j \equiv i \mod r+\delta-1 \text{ with } f_j \neq 0\}$ and $m_f = |M_f|$. Similarly, define $M_g = \{i \in [r+\delta-1]: \exists j \equiv i \mod r+\delta-1 \text{ with } g_j \neq 0\}$ and $m_g = |M_g|$. Note that \[\max\{1,m_f-(\delta-1)\}\leq m_g\leq m_f \leq r+\delta-1.\] We show two lower bounds on $|\widetilde{\ev}(f)|$; the first is tighter when $m_f$ is small and the second is tighter when $m_g$ is large.

    \begin{enumerate}
        \item Using the uncertainty principle in \Cref{cor:uncertainty_prime_finite_field}, we will bound $|\ev(f)|$ and then apply the fact $|\widetilde{\ev}(f)|\geq |\ev(f)|/s$. For every $\alpha \in \mathbb{F}_q^*$, on the restriction to inputs in the coset $\alpha \Omega_{r+\delta-1}$, $f$ agrees with $f \pmod{X^{r+\delta-1} - \alpha^{r+\delta-1}}$. Clearly, $f \pmod{X^{r+\delta-1} - \alpha^{r+\delta-1}}$ is a polynomial of degree $< r+\delta-1$ whose coefficients are supported within $M_f$. Therefore, \[|f \pmod{X^{r+\delta-1} - \alpha^{r+\delta-1}}| \leq m_f\] and by \Cref{cor:uncertainty_prime_finite_field}, either $\ev(f)|_{\alpha\Omega_{r+\delta-1}} = 0$ or \[|\ev(f)|_{\alpha\Omega_{r+\delta-1}} \geq r+\delta-1 + 1 - m_f = r + \delta-m_f.\]

        If $M_f \subseteq \{1, \ldots, \delta-1\}$, then $\ev(f) \in C^\perp$ by \Cref{lem:qtb.construction}. This contradicts the assumption on $\ev(f)$ so there is some $i \in M_f\setminus\{1,\ldots, \delta-1\}$. Knowing this, $\ev(f)|_{\alpha \Omega_{r+\delta-1}} = 0$ if and only if $f\pmod{X^{r+\delta-1} - \alpha^{r+\delta-1}} = 0$ which is only possible if the $i$th coefficient of $f\pmod{X^{r+\delta-1} - \alpha^{r+\delta-1}}$: \[\sum_{j \in [\frac{q-1}{r+\delta-1}]}f_{i + (r+\delta-1)j}(\alpha^{r+\delta-1})^j\] is zero. The polynomial \[\sum_{j \in [\frac{q-1}{r+\delta-1}]}f_{i + (r+\delta-1)j}Y^j\] has degree $ \leq (\ell-1)/(r+\delta-1)$ because $f_{i+(r+\delta-1)j} =0$ if $i + (r+\delta-1)j \geq \ell$ by definition of $C$. Therefore, it has $\leq (\ell-1)/(r+\delta-1)$ roots which implies there are $\leq (\ell-1)/(r+\delta-1)$ cosets $\alpha \Omega_{r+\delta-1}$ for which $\ev(f)|_{\alpha \Omega_{r+\delta-1}} = 0$.

        Hence,\begin{align*}
            |\ev(f)| &= \sum_{\alpha\Omega_{r+\delta-1} \in \mathbb{F}_q^*/\Omega_{r+\delta-1}} |\ev(f)|_{\alpha \Omega_{r+\delta-1}}\\
            &\geq \left(\frac{q-1}{r+\delta-1} - \frac{\ell-1}{r+\delta-1}\right)(r+\delta-1 + 1-m_f)\\
            &= (q-1)\left(1 - \frac{\ell-1}{q-1}\right)\left(1 - \frac{m_f-1}{r+\delta-1}\right) \ .
        \end{align*} Now, by definition, \[|\widetilde{\ev}(f)| \geq \frac{|\ev(f)|}{s} \geq \frac{q-1}{s}\left(1 - \frac{\ell-1}{q-1}\right)\left(1 - \frac{m_f-1}{r+\delta-1}\right)\]
        
        \item By \Cref{lem:low.wt.parity.checks}, $h$ is piecewise of degree at most $\delta - 1$. Since $\ev(h) \in C^\perp$ and $\ev(f) \not \in C^\perp$ we must have $g \neq 0$.

        Let $\det \in \mathbb{F}_q[(Y_{i,j})_{i,j \in [v]}]$ denote the determinant polynomial which takes as input a $v \times v$ matrix of variables $(Y_{i,j})_{i,j \in [v]}$ over $\mathbb{F}_q$ and outputs the determinant of the matrix.
        Define the matrices \begin{align*}
        A(X) &= (\omega_{r+\delta-1}^{i}\omega_{q-1}^j X)_{i, j \in [m_g]},\\ A_g(X) &= (\omega_{r+\delta-1}^{-i}g(\omega_{r+\delta-1}^{i}\omega_{q-1}^j X))_{i, j \in [m_g]},\\
        A_h(X) &= (\omega_{r+\delta-1}^{-i}h(\omega_{r+\delta-1}^{i}\omega_{q-1}^j X))_{i, j \in [m_g]},\\
        \text{and } A_f(X) &= (\omega_{r+\delta-1}^{-i}f(\omega_{r+\delta-1}^{i}\omega_{q-1}^j X))_{i, j \in [m_g]}\end{align*}
        and the polynomial $G(X) = \det(A_g(X)) \in \mathbb{F}_q[X]$. It follows that $\deg G \leq m_g\cdot \deg g \leq m_g(\ell-1)$. Now we show that $G(X)$ is a nonzero polynomial. We know $g(X) = \sum_{\alpha \in [\ell]\setminus(S_{-}\cup S_{+})}g_\alpha X^\alpha$ so \begin{align*}
            A_g(X) &= (\omega_{r+\delta-1}^{-i}g(\omega_{r+\delta-1}^{i}\omega_{q-1}^j X))_{i, j \in [m_g]}\\
            &=\sum_{\alpha \in [\ell]\setminus(S_{-}\cup S_{+})} g_\alpha X^\alpha (\omega_{r+\delta-1}^{(\alpha -1)\cdot i}\omega_{q-1}^{\alpha \cdot j})_{i, j \in [m_g]}\\
            &= \sum_{\alpha \in [\ell]\setminus(S_{-}\cup S_{+})} g_\alpha X^\alpha \begin{bmatrix}
                \omega_{r+\delta-1}^{(\alpha-1) \cdot 0}\\
                \omega_{r+\delta-1}^{(\alpha-1) \cdot 1}\\
                \vdots\\
                \omega_{r+\delta-1}^{(\alpha-1) \cdot (m_g-1)}
            \end{bmatrix} \cdot \begin{bmatrix}
                \omega_{q-1}^{\alpha \cdot 0} & \omega_{q-1}^{\alpha \cdot 1} & \cdots & \omega_{q-1}^{\alpha \cdot (m_g-1)}
            \end{bmatrix}\\
            &= \sum_{u\in M_g}\begin{bmatrix}
                \omega_{r+\delta-1}^{(u-1) \cdot 0}\\
                \omega_{r+\delta-1}^{(u-1) \cdot 1}\\
                \vdots\\
                \omega_{r+\delta-1}^{(u-1)\cdot(m_g-1)}
            \end{bmatrix} \cdot \sum_{\substack{\alpha \in [\ell]\setminus(S_{-}\cup S_{+})\\\alpha \equiv u \bmod {(r+\delta-1)}}} g_\alpha X^\alpha \begin{bmatrix}
                \omega_{q-1}^{\alpha \cdot 0} & \omega_{q-1}^{\alpha \cdot 1} & \cdots & \omega_{q-1}^{\alpha \cdot (m_g-1)}
            \end{bmatrix}.
        \end{align*}
        $A_g(X)$ is an $m_g \times m_g$ matrix and $G(X) = \det(A_g(X))$ so in order to show $G(X)$ is nonzero, it suffices to show that $A_g(X)$ has full rank. It is then sufficient to show that the set of vectors \[\left\{\begin{bmatrix}
                \omega_{r+\delta-1}^{(u-1) \cdot 0}\\
                \omega_{r+\delta-1}^{(u-1) \cdot 1}\\
                \vdots\\
                \omega_{r+\delta-1}^{(u-1) \cdot(m_g-1)}
            \end{bmatrix} : u\in M_g\right\}\] and \[\left\{\sum_{\substack{\alpha \in [\ell]\setminus(S_{-}\cup S_{+})\\\alpha \equiv u \bmod {(r+\delta-1)}}} g_\alpha X^\alpha \begin{bmatrix}
                \omega_{q-1}^{\alpha \cdot 0} & \omega_{q-1}^{\alpha \cdot 1} & \cdots & \omega_{q-1}^{\alpha \cdot (m_g-1)}
            \end{bmatrix} : u \in M_g\right\}\] 
        are linearly independent over $\mathbb{F}_q[X]$. The first set of vectors form the columns of an $m_g\times m_g$ Vandermonde matrix which has full rank. If there is a nontrivial $\mathbb{F}_q[X]$ linear dependency in the second set of the vectors then taking the highest-degree term of the associated polynomials over $X$ gives a nontrivial dependency among the vectors $[\omega_{q-1}^{\alpha \cdot j}]_{j\in[m_g]}$ for $m_g$ distinct values of $\alpha$. This forms another $m_g \times m_g$ Vandermonde matrix. Therefore, both sets of vectors are linearly independent so $G(X)$ is indeed nonzero.
        
        It remains to bound the number of roots of $G(X)$. For a given $x= \omega_{q-1}^b \in \mathbb{F}_q^*$, recall that $F_x$ is the index of the folded component of $\widetilde{C}$ that contains the component $x$ of $C$.

        If $b \pmod{s} \in \{0,\ldots, s-m_g\}$, then because $s \mid \frac{q-1}{r+\delta-1}$ it follows that for each $i \in [m_g]$, \[\{\omega_{r+\delta-1}^{i} \omega_{q-1}^j x: j \in [m_g]\} = \{\omega_{q-1}^{b + j + i \cdot \frac{q-1}{r+\delta-1}}: j \in [m_g]\} \subseteq F_{\omega_{r+\delta-1}^{i}x}.\] The $i$th row of $A_f(x)$ consists of $m_g$ elements, all members of $F_{\omega_{r+\delta-1}^{i}x}$, so the $i$th row of $A_f(x)$ consists of $m_g$ out of the $s$ components of $\widetilde{\ev}(f)_{F_{\omega_{r+\delta-1}^{i}x}}$.

        Let $Z_x = \{i \in [m_g]: \widetilde{\ev}(f)_{F_{\omega_{r+\delta-1}^{i}x}} = 0\}$. If $b \pmod{s} \in \{0,\ldots, s-m_g\}$, where this set is understood to be empty when $s < m_g$, then for every $i \in Z_x$, by definition, \begin{align*}(0)_{j \in [m_g]} &= (\omega_{r+\delta-1}^{-i}f(\omega_{r+\delta-1}^{i}\omega_{q-1}^j x))_{j\in[m_g]}\\ &= (\omega_{r+\delta-1}^{-i}g(\omega_{r+\delta-1}^{i}\omega_{q-1}^j x))_{j\in[m_g]} + (\omega_{r+\delta-1}^{-i}h(\omega_{r+\delta-1}^{i}\omega_{q-1}^j x))_{j\in[m_g]}\end{align*} so \[(\omega_{r+\delta-1}^{-i}g(\omega_{r+\delta-1}^{i}\omega_{q-1}^j x))_{j\in[m_g]} =  (-\omega_{r+\delta-1}^{-i}h(\omega_{r+\delta-1}^{i}\omega_{q-1}^j x))_{j\in[m_g]}.\]

        Fix $j \in [m_g]$ and consider the following column vector whose rows are restricted to $Z_x$, \[\mathbf{c}_j = (\omega_{r+\delta-1}^{-i}h(\omega_{r+\delta-1}^{i} \omega_{q-1}^j x))_{i \in Z_x}  \in \mathbb{F}_q^{|Z_x|}.\] The points $\omega_{r+\delta-1}^{i}\omega_{q-1}^j x$, $i \in Z_x$ lie in the coset $\omega_{q-1}^jx \Omega_{r+\delta-1}$ because $\omega_{r+\delta-1}^{i} \in \Omega_{r+\delta-1}$ for all $i$. By \Cref{lem:low.wt.parity.checks}, on this coset, there exists a polynomial $P_{j,x}(Y) = \sum_{t = 1}^{\delta-1} p_{j,x,t}Y^t$ such that $\deg P_{j,x} \leq \delta - 1$, $P_{j,x}(0) = 0$, and for every $\omega \in \Omega_{r+\delta-1}$, $h(\omega_{q-1}^j x \cdot \omega) = P_{j,x}(\omega)$. Therefore, \[\mathbf{c}_j = (\omega_{r+\delta-1}^{-i}P_{j,x}(\omega_{r+\delta-1}^{i}))_{i\in Z_x} = \sum_{t=1}^{\delta-1} p_{j,x,t}\cdot(\omega_{r+\delta-1}^{(t-1)\cdot i})_{i \in Z_x}\] so $\mathbf{c}_j$ lies in the span of the $\delta - 1$ vectors $(\omega_{r+\delta-1}^{0 \cdot i})_{i \in Z_x}, \ldots,(\omega_{r+\delta-1}^{(\delta-2)\cdot i})_{i \in Z_x}$. Since $j$ was arbitrary, this is true for every $\mathbf{c}_j$, $j \in [m_g]$. Thus, the rank of \[(-\omega_{r+\delta-1}^{-i}h(\omega_{r+\delta-1}^{i}\omega_{q-1}^j x))_{i \in Z_x, j\in[m_g]} = (\omega_{r+\delta-1}^{-i}g(\omega_{r+\delta-1}^{i}\omega_{q-1}^j x))_{i \in Z_x, j\in[m_g]}\] is at most $\delta - 1$. Since $[m_g] = Z_x \cup ([m_g] \setminus Z_x)$, the span of rows in $A_g(x)$ indexed by $[m_g] \setminus Z_x$ has dimension at most $m_g - |Z_x|$ and the span of the rows in $A_g(x)$ indexed by $Z_x$, as we just showed, has dimension at most $\delta - 1$. Hence, $\mathrm{rank}(A_g(x)) \leq m_g - |Z_x| + \delta-1$. By \Cref{lem:det_poly_root}, $G(X)$ has a root of multiplicity $\geq |Z_x| - \delta + 1$ at $X = x$.

        Summing over all $x = \omega_{q-1}^b \in \mathbb{F}_q^*$, it follows that the number of roots (with multiplicity) of $G(X)$ is at least \begin{align*}
            &\sum_{\substack{b \in [q-1]\\b\bmod{s} \in \{0,\ldots, s-m_g\}}}(|Z_{\omega_{q-1}^b}| - \delta + 1)\\ &\geq \sum_{x \in \mathbb{F}_q^*} (|Z_x| - \delta + 1) - \frac{(q-1)}{s} \cdot m_g \cdot (m_g-1)\\
            &= \sum_{x \in \mathbb{F}_q^*}|Z_x| - (q-1)\left(\frac{m_g(m_g-1)}{s} + \delta - 1\right)\\
            &= (q-1 - |\widetilde{\ev}(f)|s)m_g - (q-1)\left(\frac{m_g(m_g-1)}{s} + \delta - 1\right)\\
            &= (q-1)m_g\left(1 - \frac{|\widetilde{\ev}(f)|\cdot s}{q-1} - \frac{m_g-1}{s} - \frac{\delta-1}{m_g}\right).
        \end{align*} Since $\deg G(X) \leq m_g(\ell - 1)$, $G(X)$ has at most $m_g(\ell - 1)$ roots so \[(q-1)m_g\left(1 - \frac{|\widetilde{\ev}(f)|\cdot s}{q-1} - \frac{m_g-1}{s} - \frac{\delta-1}{m_g}\right) \leq m_g(\ell - 1)\] and rearranging gives \[|\widetilde{\ev}(f)| \geq \frac{q-1}{s} \left(1 - \frac{\ell - 1}{q-1} - \frac{m_g-1}{s} - \frac{\delta - 1}{m_g}\right).\]
        \end{enumerate}
       Combining the two bounds we obtain \[d = \frac{q-1}{s}\left(1 - \frac{\ell-1}{q-1} - \epsilon\right)\] for \[\epsilon = \max_{1 \leq m_f \leq r+\delta-1} \min \left\{\left(1 - \frac{\ell-1}{q-1}\right) \frac{m_f-1}{r+\delta-1}, \max_{\max\{1,m_f-(\delta-1)\} \leq m_g \leq m_f}\frac{\delta-1}{m_g} + \frac{m_g-1}{s}\right\},\] as desired.
       
    \end{proof}

    \begin{remark}
    The second case in the proof is a generalization of \cite[Claim 68]{golowich2025quantum}. However, for us arguing the rank drop was more involved because the rows of $A_g(X)$ corresponding to $Z_x$ were not all equal unlike in the proof of \cite[Claim 68]{golowich2025quantum}. Instead, we looked at the columns of $A_g(X)$ restricted to $Z_x$ and showed that its column rank must be $\leq \delta - 1$ from which we could claim the multiplicity of a root of $G(X)$.
    \end{remark}

    \begin{remark}
        Since we require the uncertainty principle to hold for $r + \delta-1$ over $\mathbb{F}_q$, we must first fix $\delta$ then $r$ and let $q$ be a prime power that lies outside the finite set of characteristics where the uncertainty principle for $r+\delta-1$ does not hold.
    \end{remark}

    \begin{remark}
        We show that in the case of $s = 1$, the bound presented for the unfolded QTB code in \Cref{thm:qtb_distance} is tighter than the bound for the folded QTB presented in \Cref{thm:fqtb_distance}. Let $n = r+\delta-1$. Indeed when $s = 1$, \[\max_{\max\{1,m_f-(\delta-1)\} \leq m_g \leq m_f}\frac{\delta-1}{m_g} + m_g-1 \geq 1,\] but \[\left(1 - \frac{\ell-1}{q-1}\right) \frac{m_f-1}{r+\delta-1} < 1\] so \[\epsilon = \left(1 - \frac{\ell-1}{q-1}\right) \frac{n-1}{n}\] and the distance bound becomes \(\frac{q-\ell}{n}\). It remains to show \[\frac{q-\ell}{n} < \frac{q-1}{2}
        \left(
        \frac{1}{\delta-1}
        +
        \frac{r}{n}
        -
        \sqrt{
        \left(\frac{r}{n}-\frac{1}{\delta-1}\right)^2
        +
        \frac{4r}{(\delta-1)n}
        \cdot\frac{\ell-1}{q-1}
        }
        \right)\]

        Dividing both sides by \(q-1\), it remains to show
        \[
        \frac{q-\ell}{(q-1)n}
        <
        \frac{1}{2}
                \left(
                \frac{1}{\delta-1}
                +
                \frac{r}{n}
                -
                \sqrt{
                \left(\frac{r}{n}-\frac{1}{\delta-1}\right)^2
                +
                \frac{4r}{(\delta-1)n}
                \cdot\frac{\ell-1}{q-1}
                }
                \right).
        \]
        Equivalently, we want to show
        \[
        \sqrt{
                \left(\frac{r}{n}-\frac{1}{\delta-1}\right)^2
                +
                \frac{4r}{(\delta-1)n}
                \cdot\frac{\ell-1}{q-1}
                }
        <
        \frac{1}{\delta-1}
        +
        \frac{r}{n}
        -
        \frac{2(q-\ell)}{(q-1)n}.
        \]
        The right-hand side is positive, since \(\ell< q\) implies
        \[
        \frac{2(q-\ell)}{(q-1)n}\leq \frac{2}{n},
        \]
        and hence
        \[
        \frac{1}{\delta-1}
        +
        \frac{r}{n}
        -
        \frac{2(q-\ell)}{(q-1)n}
        \geq
        \frac{1}{\delta-1}
        +
        \frac{r-2}{n}
        >0,
        \]
        using \(r\geq \delta\geq 2\). Therefore, it suffices to square both sides.
        
        After squaring, the desired inequality becomes
        \[
        \left(\frac{r}{n}-\frac{1}{\delta-1}\right)^2
        +
        \frac{4r}{(\delta-1)n}
        \cdot
        \frac{\ell-1}{q-1}
        <
        \left(
        \frac{1}{\delta-1}
        +
        \frac{r}{n}
        -
        \frac{2(q-\ell)}{(q-1)n}
        \right)^2.
        \]
        Subtracting the left-hand side from the right-hand side gives
        \[
        \begin{aligned}
        &
        \left(
        \frac{1}{\delta-1}
        +
        \frac{r}{n}
        -
        \frac{2(q-\ell)}{(q-1)n}
        \right)^2
        -
        \left[
        \left(\frac{r}{n}-\frac{1}{\delta-1}\right)^2
        +
        \frac{4r}{(\delta-1)n}
        \cdot
        \frac{\ell-1}{q-1}
        \right]
        \\
        &\qquad =
        \frac{4(q-\ell)}
        {(q-1)(\delta-1)n^2}
        \left(
        r(r-1)
        -
        (\delta-1)\frac{\ell-1}{q-1}
        \right).
        \end{aligned}
        \]
        This quantity is strictly positive since \(\ell<q\) and \(r\geq \delta\); we have
        \[
        r(r-1)\geq \delta(\delta-1)>\delta-1.
        \]
        Thus the squared inequality holds strictly, and hence
        \[
        \frac{q-\ell}{n}
        <
        \frac{q-1}{2}
                \left(
                \frac{1}{\delta-1}
                +
                \frac{r}{n}
                -
                \sqrt{
                \left(\frac{r}{n}-\frac{1}{\delta-1}\right)^2
                +
                \frac{4r}{(\delta-1)n}
                \cdot\frac{\ell-1}{q-1}
                }
                \right).
        \]
        Therefore, for \(s=1\), the distance bound from \Cref{thm:qtb_distance} is tighter than the specialization of the bound from \Cref{thm:fqtb_distance}.
    \end{remark}

    \subsection{Asymptotic distance bound}
    
    In order to obtain an asymptotic formulation of the distance, we desire to remove the $m_f$ and $m_g$ dependency in the distance bound. To that end, we present a technical lemma that will allow us to bound the $\epsilon$ in \Cref{thm:fqtb_distance}.

    \begin{lemma}\label{lem:epsilon_bound}
    Let
    \[
    \epsilon
    =
    \max_{1 \leq m_f \leq r+\delta-1}
    \min \left\{
    \left(1 - \frac{\ell-1}{q-1}\right) \frac{m_f-1}{r+\delta-1},\;
    \max_{\max\{1,m_f-(\delta-1)\} \leq m_g \leq m_f}
    \left(\frac{\delta-1}{m_g} + \frac{m_g-1}{s}\right)
    \right\}.
    \]
    Let \(n=r+\delta-1\) and \(\lambda = 1 - \frac{\ell-1}{q-1}\). If \(s = c n^2\) for \(c \geq 2\), then
    \[
    \epsilon
    \leq
    \lambda\frac{\delta-2}{n}
    +
    \left(1 + \frac{1}{c}\right)\sqrt{\frac{\lambda(\delta-1)}{n}}.
    \]
    \end{lemma}

    \begin{proof}
    Let \(n := r + \delta - 1\) and \(\lambda:=1-\frac{\ell-1}{q-1}.\)
    Define
    \[
    \phi(m):=\frac{\delta-1}{m}+\frac{m-1}{s}.
    \]
    Then
    \[
    \epsilon
    =
    \max_{1\le m_f\le n}
    \min\left\{
    \lambda\frac{m_f-1}{n},
    \;
    \max_{\max\{1,m_f-(\delta-1)\}\le m_g\le m_f}\phi(m_g)
    \right\}.
    \]
    
    We first show that \(\phi\) is nonincreasing on \([1,n]\). Indeed,
    \[
    \phi'(m)
    =
    -\frac{\delta-1}{m^2}+\frac{1}{s}.
    \]
    Since \(s=cn^2\) with \(c\ge 2\), we have
    \[
    s\ge n^2\ge \frac{n^2}{\delta-1}.
    \]
    Thus
    \[
    \frac1s\le \frac{\delta-1}{n^2}\le \frac{\delta-1}{m^2}
    \qquad
    (1\le m\le n),
    \]
    and hence
    \(
    \phi'(m)\le 0
    \)
    on \([1,n]\). Therefore, for fixed \(m_f\), the maximum of \(\phi(m_g)\)
    % \[
    % \max\{1,m_f-(\delta-1)\}\le m_g\le m_f
    % \]
    is attained at the left endpoint:
    \[
    \max_{\max\{1,m_f-(\delta-1)\}\le m_g\le m_f}\phi(m_g)
    =
    \phi(\max\{1,m_f-(\delta-1)\}).
    \]
    Hence
    \[
    \epsilon
    =
    \max_{1\le m_f\le n}
    \min\left\{
    \lambda\frac{m_f-1}{n},
    \;
    \phi(\max\{1,m_f-(\delta-1)\})
    \right\}.
    \]
    
    We now split the maximum over \(m_f\) into two ranges.
    
    \smallskip
    \noindent\textbf{Small-support range: \(1\le m_f<\delta\).}
    In this range,
    \(
    m_f-1\le \delta-2.
    \)
    Therefore
    \[
    \min\left\{
    \lambda\frac{m_f-1}{n},
    \;
    \phi(\max\{1,m_f-(\delta-1)\})
    \right\}
    \le
    \lambda\frac{m_f-1}{n}
    \le
    \lambda\frac{\delta-2}{n}.
    \]
    Thus the contribution of all \(m_f<\delta\) is at most
    \(
    \lambda\frac{\delta-2}{n}.
    \)
    
    \smallskip
    \noindent\textbf{Large-support range: \(\delta\le m_f\le n\).}
    In this range,
    \(
    m_f-(\delta-1)=m_f-\delta+1\ge 1.
    \)
    Therefore,
    \(
    \max\{1,m_f-(\delta-1)\}=m_f-\delta+1.
    \)
    Thus the large-support contribution is
    \[
    \max_{\delta\le m_f\le n}
    \min\left\{
    \lambda\frac{m_f-1}{n},
    \;
    \frac{\delta-1}{m_f-\delta+1}+\frac{m_f-\delta}{s}
    \right\}.
    \]
    Set
    \(
    u:=m_f-\delta+1.
    \)
    Then
    \(
    1\le u\le r,
    \)
    and
    \(
    m_f-1=u+\delta-2.
    \)
    Hence
    \[
    \lambda\frac{m_f-1}{n}
    =
    \lambda\frac{u+\delta-2}{n}
    =
    \lambda\frac{\delta-2}{n}
    +
    \lambda\frac{u}{n},
    \]
    while
    \[
    \frac{\delta-1}{m_f-\delta+1}+\frac{m_f-\delta}{s}
    =
    \frac{\delta-1}{u}+\frac{u-1}{s}.
    \]
    Therefore the large-support contribution is at most
    \[
    \max_{1\le u\le r}
    \min\left\{
    \lambda\frac{\delta-2}{n}+\lambda\frac{u}{n},
    \;
    \frac{\delta-1}{u}+\frac{u-1}{s}
    \right\}.
    \]
    Using
    \[
    \min\{A+B,C\}\le A+\min\{B,C\},
    \]
    with
    \[
    A=\lambda\frac{\delta-2}{n},
    \qquad
    B=\lambda\frac{u}{n},
    \qquad
    C=\frac{\delta-1}{u}+\frac{u-1}{s},
    \]
    we get
    \[
    \max_{1\le u\le r}
    \min\left\{
    \lambda\frac{\delta-2}{n}+\lambda\frac{u}{n},
    \;
    \frac{\delta-1}{u}+\frac{u-1}{s}
    \right\}
    \le
    \lambda\frac{\delta-2}{n}
    +
    \max_{1\le u\le r}
    \min\left\{
    \lambda\frac{u}{n},
    \;
    \frac{\delta-1}{u}+\frac{u-1}{s}
    \right\}.
    \]
    Furthermore, \((u-1)/s < u/s\),
    and enlarging the range from \(1\le u\le r\) to \(1\le u\le n\) gives
    \[
    \max_{1\le u\le r}
    \min\left\{
    \lambda\frac{u}{n},
    \;
    \frac{\delta-1}{u}+\frac{u-1}{s}
    \right\}
    \le
    \max_{1\le u\le n}
    \min\left\{
    \lambda\frac{u}{n},
    \;
    \frac{\delta-1}{u}+\frac{u}{s}
    \right\}.
    \]
    Combining the small-support and large-support ranges, we obtain
    \[
    \epsilon
    \le
    \lambda\frac{\delta-2}{n}
    +
    \max_{1\le u\le n}
    \min\left\{
    \lambda\frac{u}{n},
    \;
    \frac{\delta-1}{u}+\frac{u}{s}
    \right\}.
    \]
    
    It remains to bound the final maximum. We split into two cases.
    
    \smallskip
    \noindent\textbf{Case 1: \(\lambda\le \frac1n\).}
    By definition,
    \[
    \epsilon
    \le
    \max_{1\le m_f\le n}\lambda\frac{m_f-1}{n}
    \le
    \lambda.
    \]
    Since \(\delta\ge 2\), we have
    \[
    \lambda\le \frac1n
    \implies
    \lambda\le \sqrt{\frac{\lambda}{n}}
    \le
    \sqrt{\frac{\lambda(\delta-1)}{n}}.
    \]
    Therefore
    \[
    \lambda
    \le
    \lambda\frac{\delta-2}{n}
    +
    \left(1+\frac1c\right)\sqrt{\frac{\lambda(\delta-1)}{n}},
    \]
    and the claim follows in this case.
    
    \smallskip
    \noindent\textbf{Case 2: \(\lambda>\frac1n\).}
    Set
    \[
    B(u):=\lambda\frac{u}{n},
    \qquad
    C(u):=\frac{\delta-1}{u}+\frac{u}{s}.
    \]
    Then \(B(u)\) is increasing. Also,
    \[
    C'(u)
    =
    -\frac{\delta-1}{u^2}+\frac1s
    \le 0
    \qquad
    (1\le u\le n),
    \]
    because \(s \geq n^2/(\delta-1)\).
    Hence \(C(u)\) is nonincreasing on \([1,n]\).
    
    The two curves \(B(u)\) and \(C(u)\) meet at the positive real number
    \[
    u^*
    =
    \sqrt{\frac{\delta-1}{\frac{\lambda}{n}-\frac1s}},
    \]
    because the equation \(B(u)=C(u)\) is
    \[
    \lambda\frac{u}{n}
    =
    \frac{\delta-1}{u}+\frac{u}{s}
    \iff
    \left(\frac{\lambda}{n}-\frac1s\right)u^2=\delta-1.
    \]
    Since \(B\) is increasing and \(C\) is nonincreasing, we have
    \(
    \max_{1\le u\le n}\min\{B(u),C(u)\}
    \le
    B(u^*).
    \)
    Thus,
    \[
    \max_{1\le u\le n}
    \min\left\{
    \lambda\frac{u}{n},
    \;
    \frac{\delta-1}{u}+\frac{u}{s}
    \right\}
    \le
    \lambda\frac{u^*}{n}.
    \]
    Substituting the value of \(u^*\), we obtain
    \[
    \lambda\frac{u^*}{n}
    =
    \frac{\lambda}{n}
    \sqrt{
    \frac{\delta-1}{\frac{\lambda}{n}-\frac1s}
    }
    =
    \sqrt{\frac{\lambda(\delta-1)}{n}}
    \cdot
    \frac{1}{\sqrt{1-\frac{n}{\lambda s}}}.
    \]
    Since \(\lambda>1/n\), we have
    \[
    \frac{n}{\lambda s}
    <
    \frac{n^2}{s}
    =
    \frac1c.
    \]
    Therefore
    \[
    \lambda\frac{u^*}{n}
    \le
    \sqrt{\frac{\lambda(\delta-1)}{n}}
    \cdot
    \frac{1}{\sqrt{1-\frac1c}}.
    \]
    For \(c\ge 2\), we have \(1/c\le 1/2\), and
    \[
    \frac{1}{\sqrt{1-x}}\le 1+x
    \qquad
    (0\le x\le 1/2).
    \]
    Applying this with \(x=1/c\), we get
    \[
    \frac{1}{\sqrt{1-\frac1c}}
    \le
    1+\frac1c.
    \]
    Hence
    \[
    \max_{1\le u\le n}
    \min\left\{
    \lambda\frac{u}{n},
    \;
    \frac{\delta-1}{u}+\frac{u}{s}
    \right\}
    \le
    \left(1+\frac1c\right)
    \sqrt{\frac{\lambda(\delta-1)}{n}}.
    \]
    
    Combining this with the previous reduction, we obtain
    \[
    \epsilon
    \le
    \lambda\frac{\delta-2}{n}
    +
    \left(1+\frac1c\right)
    \sqrt{\frac{\lambda(\delta-1)}{n}},
    \]
    as desired.
    \end{proof}
    Now we get the following corollary for the asymptotic bound using \Cref{lem:epsilon_bound}.
    \begin{corollary}\label{cor:fqtb_distance_asymptotic}

    Let $\widetilde{\mathcal Q}$ be the folded QTB code from \Cref{thm:fqtb_distance}. If
    \[
    s \geq 2(r+\delta-1)^2,
    \]
    then $\widetilde{\mathcal Q}$ has distance at least
    \[
    d \geq \frac{q-1}{s}\left(
    \left(1-\frac{\ell-1}{q-1}\right)\frac{r+1}{r+\delta-1}
    -\left(1+\frac{(r+\delta-1)^2}{s}\right)
    \sqrt{\frac{\delta-1}{r+\delta-1}\left(1-\frac{\ell-1}{q-1}\right)}
    \right).
    \]
    Moreover, for every $0 < R < (r-\delta+1)/(r+\delta-1)$, there exists an explicit family of QLRCs of locality $(r,\delta)$, rate $\geq R$, and relative distance at least
    \begin{align*}
  \left(\frac{1}{2}-\frac{R}{2}\cdot\frac{r+\delta-1}{r-\delta+1}\right)\frac{r+1}{r+\delta-1}
    -\frac{3}{2}\sqrt{
    \frac{\delta-1}{r+\delta-1}
    \left(
    \frac{1}{2}-\frac{R}{2}\cdot\frac{r+\delta-1}{r-\delta+1}
    \right)
    }.
    \end{align*}
    \end{corollary}
    
    \begin{proof}
    The first claim follows immediately from \Cref{lem:epsilon_bound}, since if
    $s = c(r+\delta-1)^2$,
    then
    \[
    1+\frac{1}{c}=1+\frac{(r+\delta-1)^2}{s}.
    \]
    For the second claim, using \labelcref{eqn:ell-1/q-1} and taking \(q\to\infty\), we obtain
    \[
    1-\frac{\ell-1}{q-1}
    =
    \frac{1}{2}-\frac{R}{2}\cdot\frac{r+\delta-1}{r-\delta+1}.
    \]
    Therefore the asymptotic relative distance is at least
    \begin{align*}
    &\left(1-\frac{\ell-1}{q-1}\right)\frac{r+1}{r+\delta-1}
    -\left(1+\frac{(r+\delta-1)^2}{s}\right)
    \sqrt{\frac{\delta-1}{r+\delta-1}\left(1-\frac{\ell-1}{q-1}\right)}\\
    &=
    \left(\frac{1}{2}-\frac{R}{2}\cdot\frac{r+\delta-1}{r-\delta+1}\right)\frac{r+1}{r+\delta-1}
    -\left(1+\frac{(r+\delta-1)^2}{s}\right)
    \sqrt{
    \frac{\delta-1}{r+\delta-1}
    \left(
    \frac{1}{2}-\frac{R}{2}\cdot\frac{r+\delta-1}{r-\delta+1}
    \right)
    }.
    \end{align*}
    Since \(s \geq 2(r+\delta-1)^2\), we have
    \[
    1+\frac{(r+\delta-1)^2}{s}\le \frac{3}{2},
    \]
    which yields the claimed bound.
    \end{proof}

    \subsection{Composite-order distance relaxation}
    Using \Cref{cor:uncertainty_composite_finite_field}, we can relax the primality requirement on $r + \delta -1$ at the cost of a slightly weaker bound. This will be particularly useful when we generalize the fQTB code to its hierarchical analogue. 

    \begin{corollary}\label{cor:fqtb_distance_composite_uncertainty}
    Let \(\widetilde{\mathcal Q}\) be the folded QTB code from \Cref{def:fqtb} with parameters \(q,r,\delta,\ell, s\)
    and set \(n:= r +\delta -1\).
    Assume that $n \mid (q-1)$ and $s \mid (q-1)/n$. Also assume that the finite-field composite-order uncertainty principle from
    \Cref{cor:uncertainty_composite_finite_field} holds for order \(n\) over \(\F_q\).
    
    For \(1\le m\le n\), define
    \[
    u_n(m):=
    \frac{n}{d_1d_2}(d_1+d_2-m),
    \]
    where \(d_1<d_2\) are consecutive divisors of \(n\) satisfying
    \(
    d_1\le m\le d_2.
    \)
    Equivalently, \(u_n(m)\) is Meshulam's composite-order uncertainty lower bound for a polynomial with coefficient support size at most \(m\).
    
    Let \(\lambda:=1-(\ell-1)/(q-1).\)
    Then \(\widetilde{\mathcal Q}\) has distance at least
    \[
    d(\widetilde{\mathcal Q})
    \ge
    \frac{q-1}{s}\left(\lambda-\epsilon_{\mathrm{comp}}\right),
    \]
    where
    \[
    \epsilon_{\mathrm{comp}}
    =
    \max_{1\le m_f\le n}
    \min\left\{
    \lambda\left(1-\frac{u_n(m_f)}{n}\right),
    \;
    \max_{\max\{1,m_f-(\delta-1)\}\le m_g\le m_f}
    \left(
    \frac{\delta-1}{m_g}+\frac{m_g-1}{s}
    \right)
    \right\}.
    \]
    \end{corollary}
    
    \begin{proof}
    The proof is the same two-case argument as in \Cref{thm:fqtb_distance}, except that the prime-order uncertainty principle, \Cref{cor:uncertainty_prime_finite_field}, is replaced by the composite-order uncertainty principle, \Cref{cor:uncertainty_composite_finite_field}.
    
    Fix \(\ev(f)\in C\setminus C^\perp\) and write \(f = g +h \)
    where
    \[
    g\in \F_q[X]^{[\ell]\setminus(S_-\cup S_+)},
    \qquad
    h\in \F_q[X]^{[q-1]\cap S_+}.
    \]
    Let $M_f = \{i \in [r+\delta-1]: \exists j \equiv i \mod r+\delta-1 \text{ with } f_j \neq 0\}$ and $m_f = |M_f|$ and similarly define \(M_g\) and \(m_g\) for \(g\).
    
    We first modify Case \(1\) of \Cref{thm:fqtb_distance}. For a coset
    \(
    \alpha\Omega_n\subseteq \F_q^*,
    \)
    the restriction of \(f\) to \(\alpha\Omega_n\) is represented by
    \[
    f \pmod{X^n-\alpha^n},
    \]
    a polynomial of degree \(<n\) whose coefficient support has size at most \(m_f\). Hence, by the composite-order uncertainty principle, either
    \(
    \ev(f)|_{\alpha\Omega_n}=0,
    \)
    or
    \(
    |\ev(f)|_{\alpha\Omega_n}\ge u_n(m_f).
    \)
    As in the proof of \Cref{thm:fqtb_distance}, since \(\ev(f)\notin C^\perp\), there is at least one residue class
    \(
    i\in M_f\setminus\{1,\dots,\delta-1\}.
    \)
    For this residue class, the coefficient of \(X^i\) in
    \(
    f \pmod{X^n-\alpha^n}
    \)
    is a polynomial in \(\alpha^n\) of degree at most \((\ell-1)/n\). Therefore this coefficient vanishes for at most \((\ell-1)/n\) cosets \(\alpha\Omega_n\). Thus the number of cosets on which \(f\) does not vanish identically is at least
    \[
    \frac{q-1}{n}-\frac{\ell-1}{n}
    =
    \frac{q-1}{n}\left(1-\frac{\ell-1}{q-1}\right)
    =
    \frac{q-1}{n}\lambda.
    \]
    Consequently,
    \[
    |\ev(f)|
    \ge
    \frac{q-1}{n}\lambda\cdot u_n(m_f).
    \]
    Since folding can reduce Hamming weight by at most a factor of \(s\), we get
    \[
    |\widetilde{\ev}(f)|
    \ge
    \frac{q-1}{s}\lambda\frac{u_n(m_f)}{n}.
    \]
    Equivalently, Case \(1\) gives
    \[
    |\widetilde{\ev}(f)|
    \ge
    \frac{q-1}{s}
    \left[
    \lambda
    -
    \lambda\left(1-\frac{u_n(m_f)}{n}\right)
    \right].
    \]
    Case \(2\), the determinant-polynomial argument in \Cref{thm:fqtb_distance}, is unchanged. It gives
    \[
    |\widetilde{\ev}(f)|
    \ge
    \frac{q-1}{s}
    \left(
    \lambda
    -
    \frac{\delta-1}{m_g}
    -
    \frac{m_g-1}{s}
    \right).
    \]
    Thus for fixed \(m_f\), every codeword has folded weight at least
    \[
    \frac{q-1}{s}
    \left(
    \lambda
    -
    \min\left\{
    \lambda\left(1-\frac{u_n(m_f)}{n}\right),
    \;
    \max_{\max\{1,m_f-(\delta-1)\}\le m_g\le m_f}
    \left(
    \frac{\delta-1}{m_g}+\frac{m_g-1}{s}
    \right)
    \right\}
    \right).
    \]
    Taking the worst case over \(1\le m_f\le n\) gives the claimed expression for
    \(\epsilon_{\mathrm{comp}}\).
    \end{proof}
\section{Folded Hierarchical Quantum Tamo--Barg Codes}\label{sec:folded_hqtb}

    We now fold the hierarchical quantum Tamo--Barg codes from
    \Cref{sec:r_delta_hierarchical_qTB}. Let
    \(
    n_l:=r_l+\delta_l-1
    \) for $l = 1,\ldots, h$
    and assume
    \(
    n_h\mid n_{h-1}\mid\cdots\mid n_1\mid(q-1).
    \)
    Let \(\omega_{q-1}\in\F_q^*\) be a generator. We choose a folding parameter
    \(
    s\mid (q-1)/n_1.
    \)
    Since \(n_l\mid n_1\) for every \(l\), this also implies
    \(
    s\mid (q-1)/n_l.
    \)
    
    For \(i=0,\ldots,(q-1)/s-1\), define the folded block
    \[
    F_{\omega_{q-1}^{is}}
    :=
    \{\omega_{q-1}^{is},\omega_{q-1}^{is+1},\ldots,\omega_{q-1}^{is+s-1}\}.
    \]
    Equivalently, for \(x=\omega_{q-1}^b\in\F_q^*\), let
    \[
    F_x
    =
    \{\omega_{q-1}^{\lfloor b/s\rfloor s+j}:0\le j\le s-1\}.
    \]
    
    \begin{definition}[Folded Hierarchical Quantum Tamo--Barg code]
    \label{def:fhqtb}
    Let
    \(
    \mathcal Q=\mathrm{CSS}(C,C)
    \)
    be the \(h\)-level quantum Tamo--Barg code from
    \Cref{def:h_level_qtb}. The \textbf{folded Hierarchical Quantum Tamo--Barg (fHQTB) code}
    is the CSS code
    \(
    \widetilde{\mathcal Q}:=\mathrm{CSS}(\widetilde C,\widetilde C),
    \)
    where \(\widetilde C\subseteq(\F_q^s)^{(q-1)/s}\) is obtained from \(C\) by
    grouping the coordinates indexed by each folded block \(F_{\omega_{q-1}^{is}}\) into a
    single coordinate over the alphabet \(\F_q^s\).
    \end{definition}
    
    Folding preserves the rate, which we computed in \Cref{lem:dim_qhTB}. We next check that the folded code retains the hierarchical locality structure.
    
    \begin{lemma}\label{lem:h_folded_coset_block_intersection}
    Fix \(l\in\{1,\ldots,h\}\). Let \(\alpha\Omega_{n_l}\) be a coset of
    \(\Omega_{n_l}\) in \(\F_q^*\). Then every folded block \(F_{\omega_{q-1}^{is}}\)
    intersects \(\alpha\Omega_{n_l}\) in at most one point.
    \end{lemma}
    
    \begin{proof}
    Write
    \[
    \alpha=\omega_{q-1}^b,
    \qquad
    \Omega_{n_l}=\langle \omega_{q-1}^{(q-1)/n_l}\rangle.
    \]
    Thus
    \[
    \alpha\Omega_{n_l}
    =
    \left\{
    \omega_{q-1}^{b+t\frac{q-1}{n_l}}:0\le t\le n_l-1
    \right\}.
    \]
    Suppose two distinct elements of this coset lie in the same folded block. Then
    for some \(t_1\neq t_2\) and some \(e_1,e_2\in\{0,\ldots,s-1\}\),
    \[
    b+t_1\frac{q-1}{n_l}\equiv is+e_1\pmod{q-1},
    \]
    and
    \[
    b+t_2\frac{q-1}{n_l}\equiv is+e_2\pmod{q-1}.
    \]
    Subtracting gives
    \[
    (t_1-t_2)\frac{q-1}{n_l}\equiv e_1-e_2\pmod{q-1}.
    \]
    Since \(s\mid (q-1)/n_l\), the left-hand side is divisible by \(s\). Hence
    \(e_1-e_2\) is divisible by \(s\). But
    \[
    -(s-1)\le e_1-e_2\le s-1,
    \]
    so \(e_1=e_2\). Therefore
    \[
    (t_1-t_2)\frac{q-1}{n_l}\equiv 0\pmod{q-1},
    \]
    which implies
    \[
    t_1\equiv t_2\pmod{n_l}.
    \]
    Since \(0\le t_1,t_2\le n_l-1\), this forces \(t_1=t_2\), a contradiction.
    Thus the intersection has size at most one.
    \end{proof}
    
    \begin{corollary}\label{cor:locality_fhqtb}
    The folded \(h\)-level quantum Tamo--Barg code
    \(\widetilde{\mathcal Q}\) from \Cref{def:fhqtb} is an
    \(h\)-level
    \(
    ((r_1,\delta_1),\ldots,(r_h,\delta_h))
    \text{-QHLRC}.
    \)
    \end{corollary}
    
    \begin{proof}
    Let \(l\in\{1,\ldots,h\}\). For a coset \(\alpha\Omega_{n_l}\), define the
    folded level-\(l\) repair group
    \[
    \widetilde R_{\alpha,l}
    :=
    \{F_x:x\in \alpha\Omega_{n_l}\}.
    \]
    By \Cref{lem:h_folded_coset_block_intersection}, the map
    \(
    x\mapsto F_x
    \)
    is injective on \(\alpha\Omega_{n_l}\). Therefore
    \[
    |\widetilde R_{\alpha,l}|=|\alpha\Omega_{n_l}|=n_l=r_l+\delta_l-1.
    \]
    
    We now show that \(\widetilde R_{\alpha,l}\) can correct \(\delta_l-1\) erasures.
    Let the elements of \(\alpha\Omega_{n_l}\) be written as
    \[
    \omega_{q-1}^{b+t(q-1)/n_l},
    \qquad
    t=0,\ldots,n_l-1.
    \]
    Because \(s\mid(q-1)/n_l\), all these elements occur in folded blocks with the
    same internal offset \(b\bmod s\). More generally, for each internal offset
    \(e\in\{0,\ldots,s-1\}\), the \(e\)-th scalar components of the folded blocks
    in \(\widetilde R_{\alpha,l}\) form a coset of \(\Omega_{n_l}\).
    
    By the unfolded hierarchical locality proved in \Cref{cor:locality_hqtb}, on
    each such coset there are \(\delta_l-1\) independent local parity checks coming
    from \(C^\perp\), and any \(\delta_l-1\) erased scalar positions in that coset
    can be recovered. Applying this scalar recovery independently to each of the
    \(s\) internal offsets recovers all \(s\) components of any \(\delta_l-1\)
    erased folded symbols in \(\widetilde R_{\alpha,l}\).
    
    The nesting of the folded repair groups follows from the nesting
    \(
    \Omega_{n_h}\subseteq\cdots\subseteq\Omega_{n_1}
    \)
    and the injectivity of the folding map on every level-\(l\) coset. Hence
    \(\widetilde{\mathcal Q}\) is an
    \(((r_1,\delta_1),\ldots,(r_h,\delta_h))\) QHLRC.
\end{proof}
\subsection{Distance bound}
We now prove a distance bound for folded hierarchical QTB codes. The proof is
the hierarchical analogue of the two-case argument for folded QTB codes. The
first case uses the composite-order uncertainty principle on top-level cosets,
while the second case applies the determinant-polynomial argument independently
at each level of the hierarchy and retains the strongest resulting bound.

For \(l=1,\ldots,h\), set
\(
n_l:=r_l+\delta_l-1,
\)
and assume
\(
n_h\mid n_{h-1}\mid\cdots\mid n_1\mid(q-1).
\)
Let
\[
S_+=\bigcup_{l=1}^h S_{l,+},
\qquad
S_-=\bigcup_{l=1}^h S_{l,-}.
\]
For each level \(l\), define the level-\(l\) positive residue set
\[
H_l
:=
\{a\in [n_l]:\exists e\in S_+\text{ such that }e\equiv a\pmod{n_l}\}.
\]
Equivalently,
\[
H_l
=
\bigcup_{u=1}^h
\left\{
a\in[n_l]:
a\equiv j_u\pmod{\gcd(n_l,n_u)}
\text{ for some }1\le j_u\le \delta_u-1
\right\}.
\]
Since the \(n_l\)'s form a divisibility chain, this set is easy to compute by
reducing the hierarchical positive-residue set \(S_+\) modulo \(n_l\). Set
\(
\kappa_l:=|H_l|.
\)

For \(1\le m\le n_1\), define
\[
u_{n_1}(m)
:=
\frac{n_1}{d_1d_2}(d_1+d_2-m),
\]
where \(d_1<d_2\) are consecutive divisors of \(n_1\) satisfying
\(
d_1\le m\le d_2.
\)
This is the composite-order uncertainty lower bound.

\begin{theorem}[Folded hierarchical QTB distance bound]
\label{thm:fhqtb_distance}
Let \(\widetilde{\mathcal Q}=\mathrm{CSS}(\widetilde C,\widetilde C)\) be the
folded \(h\)-level QTB code from \Cref{def:fhqtb}. Assume that the
finite-field composite-order uncertainty principle from
\Cref{cor:uncertainty_composite_finite_field} holds for order \(n_1\) over
\(\F_q\). Let
\[
\lambda:=1-\frac{\ell-1}{q-1}.
\]

For \(c = \ev(f) \in C\setminus C^\perp\) where \(f=g+P\), define
\[
M_{f,1}
:=
\{a\in[n_1]:
\exists e\equiv a\pmod{n_1}\text{ with }f_e\neq 0\},
\]
and for each level \(l\),
\[
M_{g,l}
:=
\{a\in[n_l]:
\exists e\equiv a\pmod{n_l}\text{ with }g_e\neq 0\}.
\]
Write
\[
m_{f,1}:=|M_{f,1}|,
\qquad
m_{g,l}:=|M_{g,l}|.
\]
Then every folded codeword \(\widetilde{\ev}(f) \in \widetilde{C}\setminus\widetilde{C}^\perp\) satisfies
\[
|\widetilde{\ev}(f)|
\ge
\frac{q-1}{s}
\left(
\lambda
-
\min\left\{
\lambda\left(1-\frac{u_{n_1}(m_{f,1})}{n_1}\right),
\;
\min_{1\le l\le h}
\left(
\frac{\kappa_l}{m_{g,l}}
+
\frac{m_{g,l}-1}{s}
\right)
\right\}
\right).
\]
Consequently,
\[
d(\widetilde{\mathcal Q})
\ge
\frac{q-1}{s}\left(\lambda-\epsilon_{\mathrm{hier}}\right),
\]
where
\[
\epsilon_{\mathrm{hier}}
:=
\max_{\textup{feasible profiles}}
\min\left\{
\lambda\left(1-\frac{u_{n_1}(m_{f,1})}{n_1}\right),
\;
\min_{1\le l\le h}
\left(
\frac{\kappa_l}{m_{g,l}}
+
\frac{m_{g,l}-1}{s}
\right)
\right\}.
\]
Here ``feasible profiles'' means all tuples
\[
(m_{f,1},m_{g,1},\ldots,m_{g,h})
\]
arising from some nonzero \(f=g+P\) with
\(
\ev(f)\in C\setminus C^\perp
\) and \[
g\in \F_q[X]^{[\ell]\setminus(S_-\cup S_+)},
\qquad
P\in \F_q[X]^{[q-1]\cap S_+}.
\]
\end{theorem}

\begin{proof}
Fix
\(
\widetilde{\ev}(f)\in \widetilde C\setminus \widetilde C^\perp,
\)
and let \(\ev(f)\in C\setminus C^\perp\) be the corresponding unfolded codeword.
Decompose
\[
f(X)=g(X)+P(X),
\]
where
\[
g\in \F_q[X]^{[\ell]\setminus(S_-\cup S_+)},
\qquad
P\in \F_q[X]^{[q-1]\cap S_+}.
\]
Since \(\ev(P)\in C^\perp\) and \(\ev(f)\notin C^\perp\), we have \(g\neq 0\). We prove two lower bounds for \(|\widetilde{\ev}(f)|\).

\smallskip
\noindent\textbf{Case 1: the composite-order uncertainty bound.}
Consider a top-level coset
\(
\alpha\Omega_{n_1}\subseteq \F_q^*.
\)
The restriction of \(f\) to this coset is represented by
\(
f \pmod{X^{n_1}-\alpha^{n_1}},
\)
a polynomial of degree \(<n_1\) with coefficient support size at most
\(m_{f,1}\). By \Cref{cor:uncertainty_composite_finite_field}, either
\(
\ev(f)|_{\alpha\Omega_{n_1}}=0,
\)
or
\[
|\ev(f)|_{\alpha\Omega_{n_1}}
\ge
u_{n_1}(m_{f,1}).
\]

Next, we bound the number of top-level cosets on which \(f\) vanishes
identically. Let
\[
H_1
=
\{a\in [n_1]:\exists e\in S_+\text{ with }e\equiv a\pmod{n_1}\}.
\]
If
\(
M_{f,1}\subseteq H_1,
\)
then every exponent occurring in \(f\) lies in \(S_+\), because membership in
\(H_1\) means that the entire congruence class modulo \(n_1\) is contained in
\(S_+\). Hence
\(
\ev(f)\in C^\perp,
\)
contradicting our assumption. Therefore there exists
\(
i\in M_{f,1}\setminus H_1.
\)
For this residue class, the coefficient of \(X^i\) in
\(f \pmod{X^{n_1}-\alpha^{n_1}}\) is
\[
\sum_j f_{i+n_1j}(\alpha^{n_1})^j.
\]
Because \(i\notin H_1\), no exponent congruent to \(i\pmod{n_1}\) lies in
\(S_+\). Thus these coefficients come only from the \([\ell]\)-part of the
code, and the above polynomial in \(\alpha^{n_1}\) has degree at most\((\ell-1)/n_1.\)
Consequently, it vanishes for at most \((\ell-1)/n_1\) top-level cosets.
Hence \(f\) is nonzero on at least
\[
\frac{q-1}{n_1}-\frac{\ell-1}{n_1}
=
\frac{q-1}{n_1}\lambda
\]
top-level cosets. Therefore,
\[
|\ev(f)|
\ge
\frac{q-1}{n_1}\lambda\,u_{n_1}(m_{f,1}).
\]
Since folding decreases Hamming weight by at most a factor of \(s\),
\[
|\widetilde{\ev}(f)|
\ge
\frac{q-1}{s}\lambda\frac{u_{n_1}(m_{f,1})}{n_1}.
\]
Equivalently,
\[
|\widetilde{\ev}(f)|
\ge
\frac{q-1}{s}
\left[
\lambda
-
\lambda\left(1-\frac{u_{n_1}(m_{f,1})}{n_1}\right)
\right].
\]

\smallskip
\noindent\textbf{Case 2: determinant-polynomial bounds.}
We show that the determinant-polynomial argument can be applied at every
level \(l\). Fix \(l\in\{1,\ldots,h\}\). Let \(\omega_l\) be a primitive
\(n_l\)-th root of unity and let \(\omega_{q-1}\) be a generator of \(\F_q^*\).
Set
\(
m_l:=m_{g,l}.
\)
Define
\[
A_{g,l}(X)
=
\left(
\omega_l^{-i}g(\omega_l^i\omega_{q-1}^jX)
\right)_{0\le i,j\le m_l-1},
\]
and define $A_{f,l}(X)$ and $A_{P,l}$ similarly. Let \(G_l(X):=\det(A_{g,l}(X)).\) 

We first show that \(G_l(X)\neq 0\). Decompose \(g\) by residue classes modulo
\(n_l\):
\[
g(X)=\sum_{a\in M_{g,l}}X^a g_a(X^{n_l}).
\]
Then
\[
A_{g,l}(X)
=
\sum_{a\in M_{g,l}}
\left(\omega_l^{(a-1)i}\right)_{0\le i\le m_l-1}
\cdot
\left(
\sum_{\substack{e\equiv a\pmod{n_l}}}
g_e X^e
(\omega_{q-1}^{ej})_{0\le j\le m_l-1}
\right).
\]
The column vectors
\[
\left(\omega_l^{(a-1)i}\right)_{0\le i\le m_l-1},
\qquad a\in M_{g,l},
\]
form a Vandermonde matrix and are linearly independent. The corresponding row
vectors are also linearly independent over \(\F_q[X]\): a nontrivial relation
would give, by taking the highest-degree term in \(X\), a nontrivial
Vandermonde relation among
\[
(\omega_{q-1}^{ej})_{0\le j\le m_l-1}
\]
for distinct exponents \(e\). Thus \(A_{g,l}(X)\) has full rank over
\(\F_q(X)\), and hence
\(
G_l(X)\neq 0.
\)
Moreover,
\(
\deg G_l\le m_l(\ell-1).
\)

We now bound the number of roots of \(G_l\). Fix
\(
x=\omega_{q-1}^b\in\F_q^*.
\)
If
\(
b\bmod s\in\{0,\ldots,s-m_l\}, 
\)  where this set is understood to be empty when $s < m_l$,
then for each \(0\le i\le m_l-1\), the points
\[
\omega_l^i\omega_{q-1}^j x,
\qquad
j=0,\ldots,m_l-1,
\]
lie in the folded block \(F_{\omega_l^ix}\).
Define
\[
Z_{x,l}
:=
\{0\le i\le m_l-1:
\widetilde{\ev}(f)_{F_{\omega_l^ix}}=0\}.
\]
For \(i\in Z_{x,l}\), the \(i\)-th row of
\[
A_{f,l}(x):=A_{g,l}(x)+A_{P,l}(x)
\]
is zero, where
\[
A_{P,l}(X)
=
\left(
\omega_l^{-i}P(\omega_l^i\omega_{q-1}^jX)
\right)_{0\le i,j\le m_l-1}.
\]
Therefore, on the rows indexed by \(Z_{x,l}\),
\[
A_{g,l}(x)=-A_{P,l}(x).
\]

We bound the rank of \(A_{P,l}(x)\) restricted to the rows \(Z_{x,l}\). Fix a
column \(j\). The points
\[
\omega_l^i\omega_{q-1}^jx,
\qquad i\in Z_{x,l},
\]
lie in the level-\(l\) coset
\(
\omega_{q-1}^j x\Omega_{n_l}.
\)
Since the exponents of \(P\) lie in \(S_+\), their residues modulo \(n_l\) lie
in \(H_l\). Thus, on this coset, \(P\) is a linear combination of monomials
\(\omega^a\) for \(a \in H_l\) where \(\omega\in\Omega_{n_l}\). After the row scaling by \(\omega_l^{-i}\),
each column of the restricted matrix lies in the span of the \(\kappa_l\)
vectors
\[
(\omega_l^{(a-1)i})_{i\in Z_{x,l}},
\qquad a\in H_l.
\]
Hence
\[
\rank\left(A_{P,l}(x)|_{Z_{x,l},\,[m_l]}\right)\le \kappa_l.
\]
It follows that
\[
\rank A_{g,l}(x)
\le
m_l-|Z_{x,l}|+\kappa_l.
\]
By \Cref{lem:det_poly_root}, \(G_l(X)\) has a root at \(X=x\) of multiplicity at
least
\(
|Z_{x,l}|-\kappa_l.
\)

Summing over all \(x=\omega_{q-1}^b\in\F_q^*\), and using the same folding-window
count as in the one-level folded proof, the total number of roots of \(G_l\),
counted with multiplicity, is at least
\[
(q-1)m_l
\left(
1
-
\frac{|\widetilde{\ev}(f)|\,s}{q-1}
-
\frac{m_l-1}{s}
-
\frac{\kappa_l}{m_l}
\right).
\]
Since \(G_l\neq 0\) and \(\deg G_l\le m_l(\ell-1)\), we obtain
\[
(q-1)m_l
\left(
1
-
\frac{|\widetilde{\ev}(f)|\,s}{q-1}
-
\frac{m_l-1}{s}
-
\frac{\kappa_l}{m_l}
\right)
\le
m_l(\ell-1).
\]
Rearranging gives the level-\(l\) determinant bound
\[
|\widetilde{\ev}(f)|
\ge
\frac{q-1}{s}
\left(
\lambda
-
\frac{\kappa_l}{m_{g,l}}
-
\frac{m_{g,l}-1}{s}
\right).
\]

Since this holds for every level \(l=1,\ldots,h\), we may keep the strongest of
these determinant bounds:
\[
|\widetilde{\ev}(f)|
\ge
\frac{q-1}{s}
\left(
\lambda
-
\min_{1\le l\le h}
\left[
\frac{\kappa_l}{m_{g,l}}
+
\frac{m_{g,l}-1}{s}
\right]
\right).
\]

Combining this with the uncertainty-case bound, we obtain
\[
|\widetilde{\ev}(f)|
\ge
\frac{q-1}{s}
\left(
\lambda
-
\min\left\{
\lambda\left(1-\frac{u_{n_1}(m_{f,1})}{n_1}\right),
\;
\min_{1\le l\le h}
\left(
\frac{\kappa_l}{m_{g,l}}
+
\frac{m_{g,l}-1}{s}
\right)
\right\}
\right).
\]
Finally, taking the worst case over all feasible profiles gives the stated lower bound.
\end{proof}

\begin{remark}
When \(h=1\), we have \(H_1=\{1,\ldots,\delta_1-1\}\), so
\(\kappa_1=\delta_1-1\), and the theorem recovers the composite-order folded QTB
distance bound presented in \Cref{cor:fqtb_distance_composite_uncertainty}. For \(h>1\), each level \(l\) gives a valid determinant
estimate, and the theorem keeps the best one. This is a genuine multilevel
generalization of the determinant case. Nevertheless, the uncertainty case
remains top-level because it is applied to restrictions of \(f\) on cosets of
\(\Omega_{n_1}\).
\end{remark}

\section{Decoding Quantum Tamo--Barg codes}
\label{sec:decoding}
In this section we present a classical decoding algorithm, $\mathrm{Dec}_C$, to decode a $(r,\delta)$ Quantum Tamo--Barg code $\mathcal{Q} = \mathrm{CSS}(C,C)$ with parameters $q, \ell, r,\delta$ as defined in \Cref{def:r.delta.qTB.qlrc}. Our decoding algorithm $\mathrm{Dec}_C$ can efficiently correct errors of weight strictly less than $e$. The algorithm $\mathrm{ListDec}_{\mathrm{RS}(q,\ell)}$ used in \Cref{alg:decoding_r_delta_qtb} decodes a Reed-Solomon code as stated in \Cref{thm:RS.list.dec}. Recall that \begin{align*}
    S_{\pm} = \bigcup_{j=1}^{\delta-1} (\pm j + (r+\delta-1)\Z)
\end{align*} and $B^\perp = \ev(\F_q[X]^{[q-1]\cap S_{+}})$ is the space of piecewise polynomials of degree at most $\delta-1$ and no constant term, as shown in \Cref{lem:low.wt.parity.checks}.
Let $n:= r+ \delta-1$. Define the sets $\Gamma_i = \{0,\ldots, \delta-2, i\}$ for $i \in \{\delta-1,\ldots, n-1\}$. Also recall the polynomial $Q_i$ as used in the proof of \Cref{thm:qtb_distance}. The algorithm takes a corrupted codeword $a$ as input and outputs a codeword $c' \in C$ such that $\mathrm{dis}(c'-a,C^\perp)$ is less than $e$. The performance of \Cref{alg:decoding_r_delta_qtb} is summarized in \Cref{thm:dec.r.delta.qtb}. 

We first dispose of the degenerate case \(\ell=q-1\). In this case, the
quantity \(e\) in \labelcref{eqn:assumption_on_e} is equal to \(0\). Thus there are no nonzero error patterns of weight \(<e\), and the decoding guarantee in \Cref{thm:dec.r.delta.qtb} is vacuous. If one instead formulates decoding for errors of weight \(\le e\), then the only such error is the zero error, and the decoder simply returns the received word if it is in the code or it returns failure. Therefore, in the nontrivial decoding analysis and in \Cref{alg:decoding_r_delta_qtb}, we assume \(\ell\le q-2.\)

\begin{algorithm}[ht]
\caption{Classical decoding algorithm for a $(r,\delta)$ Quantum Tamo--Barg code $\mathcal{Q} = \CSS(C,C)$ with parameters $q,\ell,r,\delta$.}
\label{alg:decoding_r_delta_qtb}
\SetKwInOut{Input}{Input}
  \SetKwInOut{Output}{Output}
  
  \SetKwFunction{FnDec}{$\mathrm{Dec}_C$}
  \SetKwFunction{FnListDecRS}{$\mathrm{ListDec}_{\mathrm{RS}(q,\ell)}$}
  \SetKwProg{Fn}{Function}{:}{}

  \Input{Received word $a:\F_q^*\rightarrow\F_q$ with $\mathrm{dis}(a,C)< e$}
  \Output{$c'\in C$ such that $\mathrm{dis}(c'-a,C^\perp)<e$}
  
  \Fn{\FnDec{$a$}}{
    $\mathcal{L}\gets\emptyset$ \\
    \For{$i \in \{\delta-1,\ldots, r+\delta-2\}$}{
      Fix a primitive root of unity $\omega \in \Omega_{r+\delta-1}$ and construct the Vandermonde matrix $V_i = [\omega^{(j-1)\gamma}]_{1 \leq j \leq \delta-1,\gamma \in \Gamma_i}$ and let $v_i = (v_{i,0}, \ldots, v_{i,\delta-2}, 1)$ be a nontrivial vector such that $V_i v_i = 0$.\\
      Define $a_i:\F_q^*\rightarrow\F_q$ by $a_i(X) = \sum_{t=0}^{\delta-2} v_{i,t}\omega^{-t}a(\omega^{t}X) + \omega^{-i}a(\omega^iX)$ \\
      $\mathcal{L}_i\gets \mathrm{ListDec}_{\mathrm{RS}(q,\ell)}(a_i)$ \\
      \For{$g_i(X)=\sum_{j\in[\ell]}g_{i,j}X^j\in\mathcal{L}_i$}{
        \If{$g_{i,j}=0$ for every $j\in S_{-} \cup S_+$}{
          Add $\ev\left(g(X):=\sum_{j\in[\ell]\setminus (S_- \cup S_+)}Q_i(\omega^{j-1})^{-1}g_{i,j}X^j\right)$ to $\mathcal{L}$
        }
      }
    }   \KwRet{$\mathrm{argmin}_{\ev(g)\in\mathcal{L}}\mathrm{dis}(\ev(g)-a,B^\perp)$}
  }
\end{algorithm}

\begin{theorem}\label{thm:dec.r.delta.qtb}
    Let $\mathcal{Q}$ be the $(r,\delta)$ QTB code from \Cref{def:r.delta.qTB.qlrc}. If $\delta \geq 3$, assume $\mathrm{char}(\F_q) \nmid \mathcal{M}_{r,\delta}$ or if $\delta = 2$, assume $r+1$ is prime. Then $\mathcal{Q}$ can be decoded from errors of weight $< e$ for  \begin{equation}\label{eqn:assumption_on_e}
        e = \frac{q-1}{4}\left(\frac{1}{\delta-1} + \frac{r}{r+\delta-1} - \sqrt{\left(\frac{r}{r+\delta-1} - \frac{1}{\delta-1}\right)^2 + \frac{4r}{(\delta-1)(r+\delta-1)}\cdot \frac{\ell}{q-1}}\right)
    \end{equation}
    in $(q^{O(\delta)}\operatorname{poly}(r,q))$ time.
\end{theorem}

\begin{remark}\label{rem:ell_to_ell_minus_one_difference}
    The error bound \(e\) in \Cref{thm:dec.r.delta.qtb} is slightly less than half the distance bound in \Cref{thm:qtb_distance}; one can see this by replacing \(\ell\) with \(\ell -1\) in the expression for \(e\).
\end{remark}

The following lemma helps prove that \Cref{alg:decoding_r_delta_qtb} runs in polynomial time. We show that computing $\mathrm{dis} (\ev(g)-a,B^\perp)$ can be done efficiently. Let \[\langle Y,\ldots, Y^{\delta-1}\rangle = \left\{P(Y) \mid P(Y) = \sum_{j=1}^{\delta-1} y_j Y^j, \quad y_j \in \F_q\right\}.\]

\begin{lemma}\label{lem:B_perp_distance_runtime}
Let
\(
B^\perp=\ev(\F_q[X]^{[q-1]\cap S_+})
\)
be as in \Cref{lem:qtb.construction}. Given a function
\(b:\F_q^*\to\F_q\), the quantity
\(
\operatorname{dis}(b,B^\perp)
\)
can be computed by solving independent nearest-neighbor problems on the cosets
of \(\Omega_n\), where \(n=r+\delta-1\). More precisely,
\[
\operatorname{dis}(b,B^\perp)
=
\sum_{\alpha\Omega_n}
\left(
n-
\max_{P\in\langle Y,\ldots,Y^{\delta-1}\rangle}
\left|
\{\omega\in\Omega_n:b(\alpha\omega)=P(\omega)\}
\right|
\right).
\]
In particular, by exhaustive search over the \(q^{\delta-1}\) polynomials in
\(\langle Y,\ldots,Y^{\delta-1}\rangle\), this can be computed in time
\(
O\left((q-1)q^{\delta-1}\delta\right).
\)
Thus for fixed \(\delta\), the computation is polynomial in \(q\).
\end{lemma}
\begin{proof}
By \Cref{lem:low.wt.parity.checks}, a function lies in \(B^\perp\) if and only
if, on every coset \(\alpha\Omega_n\), it is represented by a polynomial
\[
P_\alpha(Y)\in\langle Y,\ldots,Y^{\delta-1}\rangle.
\]
Moreover, the choices of the polynomials \(P_\alpha\) on distinct cosets are
independent. Therefore minimizing the Hamming distance to \(B^\perp\) separates
over cosets:
\[
\operatorname{dis}(b,B^\perp)
=
\sum_{\alpha\Omega_n}
\min_{P\in\langle Y,\ldots,Y^{\delta-1}\rangle}
|\{\omega\in\Omega_n:b(\alpha\omega)\neq P(\omega)\}|.
\]
For a fixed coset, this is equal to
\[
n-
\max_{P\in\langle Y,\ldots,Y^{\delta-1}\rangle}
|\{\omega\in\Omega_n:b(\alpha\omega)=P(\omega)\}|.
\]
This proves the formula.

For the runtime claim, the local space
\(\langle Y,\ldots,Y^{\delta-1}\rangle\) has dimension \(\delta-1\) over
\(\F_q\), so it contains \(q^{\delta-1}\) polynomials. For each such polynomial,
we can evaluate it on the \(n\) points of \(\Omega_n\) and count agreements in
time \(O(n\delta)\). Since there are
\((q-1)/n\) cosets, this gives the stated runtime.
\end{proof}

\begin{proof}[Proof of \Cref{thm:dec.r.delta.qtb}]
  Let $\mathcal{Q} = \CSS(C,C)$ be the $(r,\delta)$ Quantum Tamo--Barg code with parameters $q$, $\ell$, $r$, $\delta$. By the preceding discussion, we may assume \(\ell \leq q-2\) so \(e>0\). To show \Cref{alg:decoding_r_delta_qtb} efficiently decodes $\mathcal{Q}$, it is sufficient to show that any input corrupted word $a$ can be written as $a=c+b$ where $c \in C$ and $b:\mathbb{F}_q^*\rightarrow \mathbb{F}_q$ is some corruption of Hamming weight $|b| < e$, and outputs $c' \in C$ such that $c'-c \in C^\perp$. 
  
  Assume $c = \ev(f)$ is a codeword in $C$. Let $n : = r + \delta-1$. Fix a coset $\alpha\Omega_{n} \in \mathbb{F}_q^*/\Omega_{n} $. Let $\Gamma_i = \{0,\ldots,\delta-2,i\}$ for $i \in \{\delta-1,\ldots, r+\delta-2\}$. Construct the Vandermonde matrix $V_i = [\omega^{(j-1)\gamma}]_{1 \leq j \leq \delta-1,\gamma \in \Gamma_i}$ which has rank $\delta-1$ and let $v_i= (v_{i,0}, \ldots, v_{i,\delta-2}, 1)$ be a nontrivial vector such that $V_i v_i = 0$. 
  
  For a given $i \in \{\delta-1,\ldots, r+\delta-2\}$, we have \[a_i(x) = \sum_{t=0}^{\delta-2} v_{i,t}\omega^{-t}a(\omega^{t}x) + \omega^{-i}a(\omega^ix)\] equal to \[c_i(x) = f_i(x) = \sum_{t=0}^{\delta-2} v_{i,t}\omega^{-t}f(\omega^{t}x) + \omega^{-i}f(\omega^ix)\] at every point $x$ for which $\omega^{\gamma}x \not\in \supp(b)$ for all $\gamma \in \Gamma_i$. Within a given coset $A := \alpha\Omega_{n} \in \mathbb{F}_q^*/\Omega_{n}$, the number of such points which are not in the support of $b$ is at least \(\psi(|b|_A)\) where \[ \psi(t):=(r-t)_+(n-(\delta-1)t)_+ \] as in \Cref{lem:psi_convex_monotone}. Indeed, let
    \[
    N_A:=
    |\{x\in A:x,\omega x,\ldots,\omega^{\delta-2}x\notin\supp(b)\}|.
    \]
    Each error position in \(A\) can lie in at most \(\delta-1\) of the consecutive
    blocks
    \(
    \{x,\omega x,\ldots,\omega^{\delta-2}x\},
    \)
    so
    \[
    N_A\ge (n-(\delta-1)|b|_A)_+.
    \]
    For each such \(x\), all \(|b|_A\) errors lie among the remaining \(r\) possible
    positions indexed by \(i\in\{\delta-1,\ldots,n-1\}\). Hence at least
    \((r-|b|_A)_+\) choices of \(i\) remain valid. Therefore the contribution of the
    coset \(A\) is at least
    \[
    (r-|b|_A)_+(n-(\delta-1)|b|_A)_+=\psi(|b|_A).
    \]
    The function \(\psi\) is convex on \([0,\infty)\). Hence Jensen's inequality gives
    \[
    \sum_{A\in\F_q^*/\Omega_n}\psi(|b|_A)
    \ge
    \frac{q-1}{n}
    \psi\left(\frac{n|b|}{q-1}\right).
    \]
    Since \(|b|< e\), it follows that
    \[
    \psi\left(\frac{n|b|}{q-1}\right)
    \ge
    \psi\left(\frac{ne}{q-1}\right).
    \]
    Therefore,
    \[
    \sum_{A\in\F_q^*/\Omega_n}\psi(|b|_A)
    \ge
    \frac{q-1}{n}
    \psi\left(\frac{ne}{q-1}\right).
    \]
    We now justify that the argument of \(\psi\) remains in the nonzero decreasing
    range for the value of \(e\) in \labelcref{eqn:assumption_on_e}. Set
    \[
    \eta:=\frac1{\delta-1},\qquad
    \xi:=\frac rn,\qquad
    \mu:=\frac{\ell}{q-1},
    \]
    and
    \[
    H(Y):=
    \eta+\xi-
    \sqrt{(\xi-\eta)^2+4\eta\xi Y}.
    \]
    Then \labelcref{eqn:assumption_on_e} says
    \[
    e=\frac{q-1}{4}H(\mu).
    \]
    Since \(H\) is decreasing in \(Y\) and \(0 \le Y <1\), we have
    \[
    H(\mu)\le H(0)
    =
    \eta+\xi-|\xi-\eta|
    =
    2\min\{\eta,\xi\}.
    \]
    Thus
    \[
    e
    \le
    \frac{q-1}{2}\min\left\{\frac1{\delta-1},\frac rn\right\}.
    \]
    Multiplying by \(n/(q-1)\), we obtain
    \[
    \frac{ne}{q-1}
    \le
    \frac12\min\left\{\frac{n}{\delta-1},r\right\}.
    \]
    Therefore
    \[
    \frac{ne}{q-1}
    <
    \min\left\{\frac{n}{\delta-1},r\right\},
    \]
    so \(\frac{ne}{q-1}\) lies in the nonzero range of \(\psi\). 
    Therefore, 
    \[
    \frac{q-1}{n}
    \psi\left(\frac{ne}{q-1}\right) = \frac{(\delta-1) n}{q-1}
    \left(\frac{r(q-1)}{n}-e\right)^2
    -
    \left((r-1)(\delta-1)-r\right)
    \left(\frac{r(q-1)}{n}-e\right).
    \]
  Averaging over all $i \in \{\delta-1,\ldots, r+\delta-2\}$, there must be some $i$ such that 
  \begin{align}\label{eqn:average_agreement_a_f}
      &|\{x \in \mathbb{F}_q^* : a_i(x) = f_i(x)\}|\nonumber\\ &\geq \frac{(\delta-1)n}{r(q-1)}\left(\frac{r(q-1)}{n}-e\right)^2 - \frac{((r-1)(\delta-1)-r)}{r}\left(\frac{r(q-1)}{n} - e\right).
  \end{align}

  Recall that by \Cref{lem:qtb.construction}, since  $C = \ev(\F_q[X]^S)$ for \[S = \left([\ell] \setminus S_{-}\right)\cup ([q-1]\cap S_{+})\] we may write $f(X) = g(X) + h(X)$ where $g(X) \in \mathbb{F}_q[X]^{[\ell]\setminus (S_{-}\cup S_{+})}$ and $h(X) \in \mathbb{F}_q[X]^{ [q-1]\cap S_{+}}$. Define $g_i, h_i$ analogously to $f_i$ so $f_i = g_i + h_i$. Consider the polynomial \[Q_i(Y) = \sum_{t=0}^{\delta-2}v_{i,t}Y^t + Y^{i}.\] By construction, $Q_i(\omega^\gamma) = 0$ for $\gamma = 0,\ldots, \delta -2$. For a coset element $x\in \alpha\Omega_{n}$, \[\sum_{t=0}^{\delta-2} v_{i,t}\omega^{-t}(\omega^{t}x)^j + \omega^{-i}(\omega^ix)^j = x^j \left(\sum_{t=0}^{\delta-2}v_{i,t}\omega^{(j-1)t} + \omega^{(j-1)i}\right) = x^jQ_i(\omega^{j-1}) = 0\] for $j =1,\ldots, \delta-1$. It follows that \[h_i(x) = \sum_{t=0}^{\delta-2} v_{i,t}\omega^{-t}h(\omega^{t}x) + \omega^{-i}h(\omega^ix) = 0\] for all $x \in \mathbb{F}_q^*$ so $f_i = g_i$. Therefore, \labelcref{eqn:average_agreement_a_f} is equivalent to \begin{align}\label{eqn:average_agreement_a_g}
      &|\{x \in \mathbb{F}_q^* : a_i(x) = g_i(x)\}|\nonumber\\ &\geq \frac{(\delta-1)n}{r(q-1)}\left(\frac{r(q-1)}{n}-e\right)^2 - \frac{((r-1)(\delta-1)-r)}{r}\left(\frac{r(q-1)}{n} - e\right).
  \end{align} The coefficients of $g_i$ are given by \[g_{i,j} = Q_i(\omega^{j-1})g_j\] so $g$ and $g_i$ have coefficients of the same support by \Cref{cor:Q_b_finite_field_version_outside_finitely_many_char}. In particular, $\deg g_i = \deg g < \ell$ so $\ev(g_i) \in \mathrm{RS}(q,\ell)$ is a Reed-Solomon codeword. Therefore, \labelcref{eqn:average_agreement_a_g} says $a_i = \ev(g_i) + b_i$ is a corrupted Reed-Solomon codeword with \begin{align*}
      |b_i| &= q-1 -  |\{x \in \mathbb{F}_q^* : a_i(x) = g_i(x)\}|\\
      & \leq q-1 - \frac{(\delta-1)n}{r(q-1)}\left(\frac{r(q-1)}{n}-e\right)^2 + \frac{((r-1)(\delta-1)-r)}{r}\left(\frac{r(q-1)}{n} - e\right)\\
      & = (q-1)\left(1-\frac{(\delta-1)n}{r}\left(\frac{r}{n} - \frac{e}{q-1}\right)^2 + \frac{((r-1)(\delta-1)-r)}{r}\left(\frac{r}{n} - \frac{e}{q-1}\right)\right)
  \end{align*} 
  %for \[E = \left(\frac{r}{n} - \frac{e}{q-1}\right).\]

  By \Cref{thm:RS.list.dec}, the output of running $\mathrm{ListDec}_{\mathrm{RS}(q,\ell)}(a_i)$ is a list containing $\ev(g_i)$ as long as \[1-\frac{(\delta-1)n}{r}\left(\frac{r}{n} - \frac{e}{q-1}\right)^2 + \frac{((r-1)(\delta-1)-r)}{r}\left(\frac{r}{n} - \frac{e}{q-1}\right) < 1- \sqrt{\frac{\ell}{q-1}}\] which simplifies to needing \begin{equation}\label{eqn:condition_on_e} e < \frac{q-1}{2}\left(\frac{1}{\delta-1} + \frac{r}{n} - \sqrt{\left(\frac{r}{n} - \frac{1}{\delta-1}\right)^2 + \frac{4r}{(\delta-1)n}\cdot \sqrt{\frac{\ell}{q-1}}}\right).\end{equation} Let \[\eta = \frac{1}{\delta-1}, \quad \xi = \frac{r}{n},\quad \mu = \frac{\ell}{q-1},\] and \[H(Y) = \eta + \xi - \sqrt{(\xi-\eta)^2 + 4\eta\xi Y}.\] Substituting the variables, by the assumption of \eqref{eqn:assumption_on_e} on $e$, \eqref{eqn:condition_on_e} holds if we can show that $H(\mu) < 2H(\sqrt{\mu})$. Since $0 \leq \mu < 1$, letting $F(Y) = 2H(Y) - H(Y^2)$, it suffices to show $F(Y) > 0$ for $0 \leq Y < 1$. Differentiating, \[F'(Y) = -\frac{4\eta\xi}{\sqrt{(\xi-\eta)^2 + 4\eta\xi Y}}+\frac{4\eta\xi Y}{\sqrt{(\xi-\eta)^2 + 4\eta \xi Y^2}}\] and $F'(Y) < 0$ is equivalent to 
  \[\frac{Y}{\sqrt{(\xi-\eta)^2 + 4\eta \xi Y^2}} < \frac{1}{\sqrt{(\xi-\eta)^2 + 4\eta \xi Y}} \iff (1-Y^2)(\xi-\eta)^2 + 4\eta \xi Y^2(1-Y) > 0\] which is true for $0 \leq Y < 1$. Hence, $F$ is decreasing on $[0,1)$ and $F(1) = 0$ so $F > 0$ on $[0,1)$.
  
  Therefore, $\mathcal{L}_i$ in \Cref{alg:decoding_r_delta_qtb} will contain $g_i$ and after iteration $i$, $\mathcal{L}$ will contain $\ev(g)$. If \Cref{alg:decoding_r_delta_qtb} outputs $\ev(g') \in \mathcal{L}$, then $g' \in \mathbb{F}_q[X]^{[\ell]\setminus(S_- \cup S_+)}$ and \[\mathrm{dis}(\ev(g') - a, B^\perp) \leq \mathrm{dis}(\ev(g) - a, B^\perp) \leq |\ev(g + h) - a| = |\ev(f)-a| = |b| < e.\] Since $B^\perp \subseteq C^\perp$, we have $\mathrm{dis}(\ev(g) - a, C^\perp) < e$ and $\mathrm{dis}(\ev(g') - a, C^\perp)< e$ so $\mathrm{dis}(\ev(g) - \ev(g'), C^\perp) < 2e$. By definition of $e$ and \Cref{thm:qtb_distance}, $\mathcal{Q}$ has distance $\min_{c \in C\setminus C^\perp} |c| \geq 2e$ so $\ev(g) - \ev(g') \in C^\perp$. Thus, $\mathrm{Dec}_C(a)$ outputs some $\ev(g') \in \ev(g) + C^\perp$, as desired.

  The calls to $\mathrm{ListDec}_{\mathrm{RS}(q,\ell)}$ run in $q^{O(1)}$ time. \Cref{lem:B_perp_distance_runtime} implies the last line of \Cref{alg:decoding_r_delta_qtb} runs in $O(|\mathcal{L}|\cdot (q-1)q^{\delta-1}\delta) = (q^{O(\delta)}\operatorname{poly}(r,q))$ time. The rest of \Cref{alg:decoding_r_delta_qtb} also runs in $(q^{O(\delta)}\operatorname{poly}(r,q))$ time so the entire decoding procedure takes $(q^{O(\delta)}\operatorname{poly}(r,q))$ time.
\end{proof}

\begin{remark}
The decoding algorithm in \Cref{alg:decoding_r_delta_qtb} efficiently decodes $(r,\delta)$ Quantum Tamo--Barg (QTB) codes. Moreover, extending \Cref{alg:decoding_r_delta_qtb} similarly to \cite[Algorithm 2]{golowich2025quantum}, we can obtain a decoding algorithm for the folded quantum Tamo--Barg codes presented in \Cref{sec:folded_qtb}.
\end{remark}

\section{Discussion and Conclusions}
As demonstrated in \cite{gottesman1997stabilizercodesquantumerror, Grassl_1997} quantum erasures are easier to recover than other errors. Furthermore, detected errors can be converted into erasures on quantum hardware \cite{Grassl_1997,Wu_2022,Teoh_2023}. Erasure recovery is also believed to aid building efficient quantum storage. To this end, we develop the theory of quantum erasure recovery by extending the study of local recovery. In particular, we extend the literature on quantum locally recoverable codes. We provide constructions of $(r,\delta)$ QLRCs that are CSS codes by effectively constructing the parity-check matrices of underlying classical codes. We also demonstrate explicit examples of $(r,\delta)$ QLRCs by constructing $(r,\delta)$ quantum Tamo--Barg codes. 

Furthermore, we explore the idea of ``hierarchical local recovery'' in the quantum setting. Classical hierarchical locally recoverable codes were studied to improve efficiency of local recovery as well as recover more erasures. With a similar goal in mind, we define $h$-level quantum hierarchical locally recoverable codes for any integer $h \geq 2$. In the case that $h=1$, we reduce to the definition of $(r,\delta)$ QLRCs presented in \cite{galindo2024quantumrdeltalocallyrecoverablecodes}. Given locality and distance parameters $(r_1,\delta_1),\ldots,(r_h,\delta_h)$, we construct $h$-level QHLRCs that are CSS codes. The random constructions that we present rely on carefully choosing the parity-check matrices for associated classical codes. We also present a Singleton-like bound on the minimum distance of $h$-level QHLRCs. When $h=1$, we obtain the Singleton-like bound presented in \cite{galindo2024quantumrdeltalocallyrecoverablecodes}. We also present explicit constructions of $h$-level QHLRCs, the $h$-level quantum Tamo--Barg codes. In addition, we introduce folded $(r,\delta)$ quantum Tamo--Barg codes and folded quantum Hierarchical Tamo--Barg codes. %\RK{Change code name to match the one used in Section 7.} 
%However, constructing folded $h$-level quantum Tamo–Barg codes remains open largely due to the complexity of computing a good bound on their minimum distance.
%, of which the one-level folded quantum Tamo–Barg code ($h=1$) arises as a special case. 
Finally, we present an algorithm which decodes the $(r,\delta)$ quantum Tamo--Barg codes efficiently and can be extended to the folded codes.

While this work contributes to the growing literature on quantum locally recoverable codes (QLRCs), an area that has recently attracted significant attention, several important questions remain open. The classical Tamo--Barg codes are designed such that their minimum distance achieves the Singleton-like bound for classical locally recoverable codes. In contrast, the quantum Tamo--Barg codes introduced in this work do not attain the corresponding quantum Singleton-like bound in either the QLRC or QHLRC settings. Constructing QLRCs and QHLRCs that meet their respective quantum Singleton-like bounds therefore remains an open problem. Moreover, devising new proof techniques to prove tighter lower bounds on the minimum distances derived in this work is also of interest. Finally, a key avenue for future exploration lies in understanding the practical implications of these theoretical advances—specifically, how such quantum codes can be effectively integrated into emerging quantum hardware architectures.

%\bibliography{qHLRC_bibliography}
\printbibliography

@article{Calderbank_1996,
   title={Good quantum error-correcting codes exist},
   volume={54},
   ISSN={1094-1622},
   url={http://dx.doi.org/10.1103/PhysRevA.54.1098},
   DOI={10.1103/physreva.54.1098},
   number={2},
   journal={Physical Review A},
   publisher={American Physical Society (APS)},
   author={Calderbank, A. R. and Shor, Peter W.},
   year={1996},
   month=aug, pages={1098–1105} }

@article{Steane_1996, volume={452},
   title = {Multiple-particle interference and quantum error correction}, 
   author = {Steane, Andrew},
   ISSN={1471-2946},
   url={http://dx.doi.org/10.1098/rspa.1996.0136},
   DOI={10.1098/rspa.1996.0136},
   number={1954},
   journal={Proceedings of the Royal Society of London. Series A: Mathematical, Physical and Engineering Sciences},
   publisher={The Royal Society},
   year={1996},
   month=nov, pages={2551–2577} }

@book{Cover2006,
  at = {2008-03-31 06:17:47},
  author = {Cover, Thomas M. and Thomas, Joy A.},
  howpublished = {Hardcover},
  id = {1877660},
  isbn = {0471241954},
  month = {7},
  publisher = {Wiley-Interscience},
  title = {Elements of Information Theory 2nd Edition (Wiley Series in Telecommunications and Signal Processing)},
  year = 2006
}

@article{Grassl_1997,
   title={Codes for the quantum erasure channel},
   volume={56},
   ISSN={1094-1622},
   url={http://dx.doi.org/10.1103/PhysRevA.56.33},
   DOI={10.1103/physreva.56.33},
   number={1},
   journal={Physical Review A},
   publisher={American Physical Society (APS)},
   author={Grassl, M. and Beth, Th. and Pellizzari, T.},
   year={1997},
   month=jul, pages={33–38} }

@misc{gottesman1997stabilizercodesquantumerror,
      title={Stabilizer Codes and Quantum Error Correction}, 
      author={Daniel Gottesman},
      year={1997},
      eprint={quant-ph/9705052},
      archivePrefix={arXiv},
      primaryClass={quant-ph},
      url={https://arxiv.org/abs/quant-ph/9705052}, 
}

@article{Wu_2022,
   title={Erasure conversion for fault-tolerant quantum computing in alkaline earth Rydberg atom arrays},
   volume={13},
   ISSN={2041-1723},
   url={http://dx.doi.org/10.1038/s41467-022-32094-6},
   DOI={10.1038/s41467-022-32094-6},
   number={1},
   journal={Nature Communications},
   publisher={Springer Science and Business Media LLC},
   author={Wu, Yue and Kolkowitz, Shimon and Puri, Shruti and Thompson, Jeff D.},
   year={2022},
   month=aug }

@article{Teoh_2023,
   title={Dual-rail encoding with superconducting cavities},
   volume={120},
   ISSN={1091-6490},
   url={http://dx.doi.org/10.1073/pnas.2221736120},
   DOI={10.1073/pnas.2221736120},
   number={41},
   journal={Proceedings of the National Academy of Sciences},
   publisher={Proceedings of the National Academy of Sciences},
   author={Teoh, James D. and Winkel, Patrick and Babla, Harshvardhan K. and Chapman, Benjamin J. and Claes, Jahan and de Graaf, Stijn J. and Garmon, John W. O. and Kalfus, William D. and Lu, Yao and Maiti, Aniket and Sahay, Kaavya and Thakur, Neel and Tsunoda, Takahiro and Xue, Sophia H. and Frunzio, Luigi and Girvin, Steven M. and Puri, Shruti and Schoelkopf, Robert J.},
   year={2023},
   month=oct }

@misc{sharma2024quantumlocallyrecoverablecodes,
      title={Quantum Locally Recoverable Codes via Good Polynomials}, 
      author={Sandeep Sharma and Vinayak Ramkumar and Itzhak Tamo},
      year={2024},
      eprint={2411.01504},
      archivePrefix={arXiv},
      primaryClass={cs.IT},
      url={https://arxiv.org/abs/2411.01504}, 
}

@misc{zhou2025optimalquantumrdeltalocallyrepairable,
      title={Optimal Quantum $(r,\delta)$-Locally Repairable Codes via Classical Ones}, 
      author={Kun Zhou and Meng Cao},
      year={2025},
      eprint={2507.18175},
      archivePrefix={arXiv},
      primaryClass={quant-ph},
      url={https://arxiv.org/abs/2507.18175}, 
}

@ARTICLE{KPLK14,
  author={Kamath, Govinda M. and Prakash, N. and Lalitha, V. and Kumar, P. Vijay},
  journal={IEEE Trans. on Inf. Theory}, 
  title={Codes With Local Regeneration and Erasure Correction}, 
  year={2014},
  volume={60},
  number={8},
  pages={4637-4660},
  doi={10.1109/TIT.2014.2329872}}

@article{LRC_intro,
  title={On the locality of codeword symbols},
  author={Gopalan, Parikshit and Huang, Cheng and Simitci, Huseyin and Yekhanin, Sergey},
  journal={IEEE Trans. on Inf. Theory},
  volume={58},
  number={11},
  pages={6925--6934},
  year={2012},
  publisher={IEEE}
}

@ARTICLE{tamo.barg.lrc,
  author={Tamo, Itzhak and Barg, Alexander},
  journal={IEEE Transactions on Information Theory}, 
  title={A Family of Optimal Locally Recoverable Codes}, 
  year={2014},
  volume={60},
  number={8},
  pages={4661-4676},
  doi={10.1109/TIT.2014.2321280}}

@ARTICLE{ballentine.barg.vladut.hlrc,
  author={Ballentine, Sean and Barg, Alexander and Vlăduţ, Serge},
  journal={IEEE Transactions on Information Theory}, 
  title={Codes With Hierarchical Locality From Covering Maps of Curves}, 
  year={2019},
  volume={65},
  number={10},
  pages={6056-6071},
  doi={10.1109/TIT.2019.2919830}}

@misc{galindo2024quantumrdeltalocallyrecoverablecodes,
      title={Quantum $(r,\delta)$-locally recoverable codes}, 
      author={Carlos Galindo and Fernando Hernando and Helena Martín-Cruz and Ryutaroh Matsumoto},
      year={2024},
      eprint={2412.16590},
      archivePrefix={arXiv},
      primaryClass={cs.IT},
      url={https://arxiv.org/abs/2412.16590}, 
}

@inproceedings{golowich2025quantum,
  title={Quantum locally recoverable codes},
  author={Golowich, Louis and Guruswami, Venkatesan},
  booktitle={Proceedings of the 2025 Annual ACM-SIAM Symposium on Discrete Algorithms (SODA)},
  pages={5512--5522},
  year={2025},
  organization={SIAM}
}

@article{Gundersen2025puncturing,
   title={Puncturing Quantum Stabilizer Codes},
   volume={6},
   ISSN={2641-8770},
   url={http://dx.doi.org/10.1109/JSAIT.2025.3562287},
   DOI={10.1109/jsait.2025.3562287},
   journal={IEEE Journal on Selected Areas in Information Theory},
   publisher={Institute of Electrical and Electronics Engineers (IEEE)},
   author={Gundersen, Jaron Skovsted and Christensen, René Bødker and Grassl, Markus and Popovski, Petar and Wisniewski, Rafał},
   year={2025},
   pages={74–84} }

@INPROCEEDINGS{sasidharan2015hlrc,
  author={Sasidharan, Birenjith and Agarwal, Gaurav Kumar and Kumar, P. Vijay},
  booktitle={2015 IEEE International Symposium on Information Theory (ISIT)}, 
  title={Codes with hierarchical locality}, 
  year={2015},
  volume={},
  number={},
  pages={1257-1261},
  doi={10.1109/ISIT.2015.7282657}}

@book{Macdonald_2008, place={Oxford, NY}, title={Symmetric functions and hall polynomials}, publisher={Clarendon Press; Oxford University Press}, author={Macdonald, I. G.}, year={2008}}

@article{Grasslpuncturingstabilizer,
	author = {Markus Grassl},
	doi = {10.1080/23799927.2020.1850530},
	eprint = {https://doi.org/10.1080/23799927.2020.1850530},
	journal = {International Journal of Computer Mathematics: Computer Systems Theory},
	number = {4},
	pages = {243--259},
	publisher = {Taylor \& Francis},
	title = {Algebraic quantum codes: linking quantum mechanics and discrete mathematics},
	url = {https://doi.org/10.1080/23799927.2020.1850530},
	volume = {6},
	year = {2021}}

@article{Goldstein_Guralnick_Isaacs_2005, 
    title={Inequalities for finite group permutation modules}, 
    volume={357}, 
    DOI={10.1090/s0002-9947-05-03927-9}, 
    number={10}, 
    journal={Transactions of the American Mathematical Society}, 
    author={Goldstein, Daniel and Guralnick, Robert M. and Isaacs, I. M.}, 
    year={2005}, 
    month={05}, 
    pages={4017–4042}}

@INPROCEEDINGS{GS98,
  author={Guruswami, V. and Sudan, M.},
  booktitle={Proceedings 39th Annual Symposium on Foundations of Computer Science (Cat. No.98CB36280)}, 
  title={Improved decoding of Reed-Solomon and algebraic-geometric codes}, 
  year={1998},
  volume={},
  number={},
  pages={28-37},
  doi={10.1109/SFCS.1998.743426}}

@misc{li2025improvedboundsoptimalconstructions,
      title={Improved bounds and optimal constructions of pure quantum locally recoverable codes}, 
      author={Yang Li and Shitao Li and Gaojun Luo and San Ling},
      year={2025},
      eprint={2512.07256},
      archivePrefix={arXiv},
      primaryClass={cs.IT},
      url={https://arxiv.org/abs/2512.07256}, 
}

@misc{galindo2026quantumrdeltalocallyrecoverablebch,
      title={Quantum $(r,\delta)$-Locally Recoverable BCH and Homothetic-BCH Codes}, 
      author={Carlos Galindo and Fernando Hernando and Ryutaroh Matsumoto},
      year={2026},
      eprint={2601.22567},
      archivePrefix={arXiv},
      primaryClass={cs.IT},
      url={https://arxiv.org/abs/2601.22567}, 
}

@misc{li2025optimalquantumlrcshermitian,
      title={On optimal quantum LRCs from the Hermitian construction and $t$-designs}, 
      author={Yang Li and Shitao Li and Huimin Lao and Gaojun Luo and San Ling},
      year={2025},
      eprint={2508.13553},
      archivePrefix={arXiv},
      primaryClass={cs.IT},
      url={https://arxiv.org/abs/2508.13553}, 
}

@misc{bu2025quantumlocallyrecoverablecode,
      title={Quantum locally recoverable code with intersecting recovery sets}, 
      author={Kaifeng Bu and Weichen Gu and Xiang Li},
      year={2025},
      eprint={2501.10354},
      archivePrefix={arXiv},
      primaryClass={quant-ph},
      url={https://arxiv.org/abs/2501.10354}, 
}

@misc{luo2023boundsconstructionsquantumlocally,
      title={Bounds and Constructions of Quantum Locally Recoverable Codes from Quantum CSS Codes}, 
      author={Gaojun Luo and Bocong Chen and Martianus Frederic Ezerman and San Ling},
      year={2023},
      eprint={2312.11115},
      archivePrefix={arXiv},
      primaryClass={cs.IT},
      url={https://arxiv.org/abs/2312.11115}, 
}

@misc{galindo2025optimalquantumlocallyrecoverable,
      title={Optimal quantum locally recoverable codes from matrix-product construction}, 
      author={Carlos Galindo and Fernando Hernando and Carlos Munuera and Diego Ruano},
      year={2025},
      eprint={2310.15703},
      archivePrefix={arXiv},
      primaryClass={cs.IT},
      url={https://arxiv.org/abs/2310.15703}, 
}

@incollection{Gardner2014,
  author    = {Gardner, Robert B. and Govil, N. K.},
  title     = {Enestr{\"o}m--Kakeya Theorem and Some of Its Generalizations},
  booktitle = {Current Topics in Pure and Computational Complex Analysis},
  editor    = {Joshi, Santosh and Dorff, Michael and Lahiri, Indrajit},
  pages     = {171--199},
  publisher = {Springer India},
  address   = {New Delhi},
  year      = {2014},
  isbn      = {978-81-322-2113-5},
  doi       = {10.1007/978-81-322-2113-5_8},
  url       = {https://doi.org/10.1007/978-81-322-2113-5_8}
}

@article{Jacobi1841,
author = {Jacobi, C.G.J.},
journal = {Journal für die reine und angewandte Mathematik},
language = {lat},
pages = {360-371},
title = {De functionibus alternantibus earumque divisione per productum e differentiis elementorum conflatum.},
url = {http://eudml.org/doc/147139},
volume = {22},
year = {1841},
}

@article{cauchy1815memoire,
  author = {Cauchy, Augustin-Louis},
  title = {M{\'e}moire sur les fonctions qui ne peuvent obtenir que deux valeurs {\'e}gales et de signes contraires par suite des transpositions op{\'e}r{\'e}es entre les variables qu'elles renferment},
  journal = {Journal de l'{\'E}cole polytechnique},
  volume = {10},
  number = {17},
  pages = {29--112},
  year = {1815},
  note = {OEuvres compl{\`e}tes, ser. 2, vol. 1, pp. 91--169}
}

@book{trudi1862teoria,
	author = {Trudi, N.},
	publisher = {Libreria scientifica e industriale di B. Pellerano},
	title = {Teoria de' determinanti e loro applicazioni},
	url = {https://books.google.com/books?id=sIVZc_PczqsC},
	year = {1862}}

@book{neukirch1999algebraic,
  title={Algebraic Number Theory},
  author={Neukirch, J{\"u}rgen},
  volume={322},
  year={1999},
  publisher={Springer-Verlag},
  address={Berlin, Heidelberg},
  series={Grundlehren der mathematischen Wissenschaften},
  isbn={978-3-540-65399-8},
  doi={10.1007/978-3-662-03983-0}
}

@book{cox2005using,
  title={Using Algebraic Geometry},
  author={Cox, David A and Little, John and O'Shea, Donal},
  volume={185},
  year={2005},
  publisher={Springer Science \& Business Media},
  address={New York},
  edition={2nd},
  series={Graduate Texts in Mathematics},
  isbn={978-0-387-20706-3}
}

@book{gelfand1994discriminants,
  title={Discriminants, Resultants, and Multidimensional Determinants},
  author={Gelfand, Izrail Moiseevich and Kapranov, Mikhail M and Zelevinsky, Andrei V},
  year={1994},
  publisher={Birkh{\"a}user},
  address={Boston},
  isbn={978-0-8176-3660-9}
}

@article{MESHULAM200663,
	author = {Roy Meshulam},
	doi = {https://doi.org/10.1016/j.ejc.2004.07.009},
	issn = {0195-6698},
	journal = {European Journal of Combinatorics},
	number = {1},
	pages = {63-67},
	title = {An uncertainty inequality for finite abelian groups},
	url = {https://www.sciencedirect.com/science/article/pii/S0195669804001453},
	volume = {27},
	year = {2006}}

@article{tao2003,
       author = {{Tao}, Terence},
        title = "{An uncertainty principle for cyclic groups of prime order}",
      journal = {arXiv Mathematics e-prints},
         year = 2003,
        month = aug,
          eid = {math/0308286},
        pages = {math/0308286},
          doi = {10.48550/arXiv.math/0308286},
archivePrefix = {arXiv},
       eprint = {math/0308286},
 primaryClass = {math.CA},
       adsurl = {https://ui.adsabs.harvard.edu/abs/2003math......8286T}
}

@article{Dirichlet1837,
  author = {Dirichlet, P. G. Lejeune},
  title = {Beweis des Satzes, dass jede unbegrenzte arithmetische Progression, deren erstes Glied und Differenz ganze Zahlen ohne gemeinschaftlichen Factor sind, unendlich viele Primzahlen enth{\"a}lt},
  journal = {Abhandlungen der K{\"o}niglichen Preu{\ss}ischen Akademie der Wissenschaften zu Berlin},
  year = {1837},
  pages = {45--71},
  volume = {48}
}

\end{document}